\documentclass{article}

\newcommand{\hi}{\mathcal{H}}

\newcommand{\his}{\mathcal{H}_{\mathcal{S}}}
\newcommand{\hisx}{\mathcal{H}_{\mathcal{S}}^x}
\newcommand{\scrhi}{\mathscr{H}}
\newcommand{\scrhis}{\mathscr{H}_{\mathcal{S}}}

\newcommand{\hir}{\mathcal{H}_{\mathcal{R}}}
\newcommand{\hix}{\mathcal{H}^x}
\newcommand{\hirx}{\mathcal{H}_{\mathcal{R}}^x}
\newcommand{\scrhir}{\mathscr{H}_{\mathcal{R}}}

\newcommand{\hisr}{\mathcal{H}_{\mathcal{S} \R
}}

\newcommand{\Poinc}{\mathcal{P}}
\newcommand{\Poincup}{\mathcal{P}_+^\uparrow}

\newcommand{\mink}{\mathbb{M}}

\newcommand{\Lup}{SO_+^\uparrow(1,d-1)}

\newcommand{\supp}{\text{supp }}

\newcommand{\At}{\text{Atiyah}}

\newcommand\indep{\protect\mathpalette{\protect\independenT}{\perp}}
\def\independenT#1#2{\mathrel{\rlap{$#1#2$}\mkern2mu{#1#2}}}

\newcommand\indepER{\protect\mathpalette{\protect\independenTER}{\perp}}
\def\independenTER#1#2{\mathrel{\rlap{$#1#2$}\mkern2mu{#1#2}^{\R,\pi}}}

\newcommand\indepERsigma{\protect\mathpalette{\protect\independenTERsigma}{\perp}}
\def\independenTERsigma#1#2{\mathrel{\rlap{$#1#2$}\mkern2mu{#1#2}^{\R,\pi}_{\sigma}}}

\newcommand\indepsigma{\protect\mathpalette{\protect\independenTsigma}{\perp}}
\def\independenTsigma#1#2{\mathrel{\rlap{$#1#2$}\mkern2mu{#1#2}_{\sigma}}}

\newcommand{\inner}[2]{\langle #1, #2 \rangle}

\newcommand{\Bor}{\text{Bor}}

\def\framestate{\mathscr{S}(\scrhir)}
\def\framestategauge{\mathscr{S}(\scrhi_{\mathcal{R}_G})}
\def\framestatex{\mathscr{D}(\mathcal{H}_\mathcal{R}^x)}

\def\state{\mathscr{S}(\scrhi)}
\def\normalstate{\mathscr{D}(\scrhi)}
\def\normalsystemstate{\mathscr{D}(\scrhis)}
\def\normalsystemframestate{\mathscr{D}(\scrhis\boxtimes \scrhir)}
\def\normalframestate{\mathscr{D}(\scrhir)}
\def\statex{\mathscr{D}(\mathcal{H}^x)}
\def\systemstate{\mathscr{S}(\scrhis)}
\def\systemstatex{\mathscr{D}(\mathcal{H}_\mathcal{S}^x)}
\def\homstate{\mathscr{D}(\hi)}
\def\homsystemstate{\mathscr{D}(\his)}
\def\homframestate{\mathscr{D}(\hir)}

\def\operationalstate{\mathfrak{S}_\R}

\newcommand{\bh}{\mathcal{B}(\hi)}

\newcommand{\bhs}{\mathcal{B}(\his)}
\newcommand{\bhsx}{\mathcal{B}(\hisx)}
\newcommand{\bhsr}{\mathcal{B}(\his \otimes \hir)}
\newcommand{\bhsrx}{\mathcal{B}(\hisx \otimes \hirx)}
\newcommand{\bhr}{\mathcal{B}(\hir)}
\newcommand{\bhrx}{\mathcal{B}(\hir^x)}
\newcommand{\bhx}{\mathcal{B}(\hi^x)}

\newcommand{\thi}{\mathcal{T}(\hi)}
\newcommand{\thr}{\mathcal{T}(\hir)}

\newcommand{\thx}{\mathcal{T}(\hi^x)}

\newcommand{\thsrx}{\mathcal{T}(\hisx \otimes \hirx)}
\newcommand{\thsx}{\mathcal{T}(\hisx)}

\newcommand{\eh}{\mathcal{E}(\hi)} 
\newcommand{\Eff}{\mathcal{E}}

\newcommand{\M}{\mathcal{M}}
\newcommand{\R}{\mathcal{R}}
\renewcommand{\S}{\mathcal{S}}

\newcommand{\A}{\mathcal{A}}
\newcommand{\B}{\mathcal{B}}
\newcommand{\C}{\mathcal{C}}
\newcommand{\D}{\mathcal{D}}

\newcommand{\Nn}{\mathbb{N}}
\newcommand{\id}{\mathbb{1}}

\newcommand{\E}{\mathsf{E}}
\newcommand{\F}{\mathsf{F}}

\newcommand{\esssup}{\mathrm{ess\,sup}}

\usepackage[page,title]{appendix}

\usepackage{tocloft}
\usepackage{multicol}
\usepackage{physics}
\usepackage{soul}
\usepackage{relsize}
\usepackage[T1]{fontenc}
\usepackage{lmodern}
\usepackage{slantsc}
\usepackage{bbm}
\usepackage{graphicx}
\usepackage{amsmath}
\usepackage{bbold}
\usepackage{amsfonts}
\usepackage{amssymb}
\usepackage{amsthm}
\usepackage{amsbsy}
\usepackage{mathrsfs}
\usepackage{varioref}
\usepackage{dsfont}
\usepackage{bm}
\usepackage{color}
\usepackage[nopar]{lipsum}
\usepackage{lipsum}
\usepackage{multicol}
\usepackage{tikz-cd}
\usepackage[T1]{fontenc}
\usepackage{authblk}
\usepackage{xcolor}
\usepackage[utf8]{inputenc} 
\usepackage{csquotes}
\usepackage{mathtools}
\usepackage{slashed}
\usepackage{subcaption}
\usepackage{makecell}
\usepackage{tabu}
\usepackage{textcomp}

\usepackage[compat=1.1.0]{tikz-feynman}
\usepackage{tikz-cd}
\usetikzlibrary{arrows}
\usepackage{pgfplots}
\usepackage{tikz-cd}
\usepackage{quiver}

\newtheorem{theorem}{Theorem}[section] 
\newtheorem{lemma}[theorem]{Lemma} 
\newtheorem{corollary}[theorem]{Corollary}
\newtheorem{proposition}[theorem]{Proposition}
\newtheorem{definition}[theorem]{Definition}
\newtheorem{example}[theorem]{Example}

\newtheorem{remark}[theorem]{Remark}

\newtheorem{notation}[theorem]{Notation}

\newcommand{\U}{\mathcal{U}} 

\newcommand{\V}{\mathcal{V}}

\DeclareMathOperator*{\esssupp}{ess\,supp}

\tikzset{
pattern size/.store in=\mcSize, 
pattern size = 5pt,
pattern thickness/.store in=\mcThickness, 
pattern thickness = 0.3pt,
pattern radius/.store in=\mcRadius, 
pattern radius = 1pt}
\makeatletter
\pgfutil@ifundefined{pgf@pattern@name@_nikqpx0v0}{
\pgfdeclarepatternformonly[\mcThickness,\mcSize]{_nikqpx0v0}
{\pgfqpoint{0pt}{0pt}}
{\pgfpoint{\mcSize+\mcThickness}{\mcSize+\mcThickness}}
{\pgfpoint{\mcSize}{\mcSize}}
{
\pgfsetcolor{\tikz@pattern@color}
\pgfsetlinewidth{\mcThickness}
\pgfpathmoveto{\pgfqpoint{0pt}{0pt}}
\pgfpathlineto{\pgfpoint{\mcSize+\mcThickness}{\mcSize+\mcThickness}}
\pgfusepath{stroke}
}}
\makeatother

\tikzset{
pattern size/.store in=\mcSize, 
pattern size = 5pt,
pattern thickness/.store in=\mcThickness, 
pattern thickness = 0.3pt,
pattern radius/.store in=\mcRadius, 
pattern radius = 1pt}
\makeatletter
\pgfutil@ifundefined{pgf@pattern@name@_ng5lqn19c}{
\pgfdeclarepatternformonly[\mcThickness,\mcSize]{_ng5lqn19c}
{\pgfqpoint{0pt}{-\mcThickness}}
{\pgfpoint{\mcSize}{\mcSize}}
{\pgfpoint{\mcSize}{\mcSize}}
{
\pgfsetcolor{\tikz@pattern@color}
\pgfsetlinewidth{\mcThickness}
\pgfpathmoveto{\pgfqpoint{0pt}{\mcSize}}
\pgfpathlineto{\pgfpoint{\mcSize+\mcThickness}{-\mcThickness}}
\pgfusepath{stroke}
}}
\makeatother
\tikzset{every picture/.style={line width=0.75pt}} 

\tikzset{
pattern size/.store in=\mcSize, 
pattern size = 5pt,
pattern thickness/.store in=\mcThickness, 
pattern thickness = 0.3pt,
pattern radius/.store in=\mcRadius, 
pattern radius = 1pt}
\makeatletter
\pgfutil@ifundefined{pgf@pattern@name@_5n7rcgqhm}{
\makeatletter
\pgfdeclarepatternformonly[\mcRadius,\mcThickness,\mcSize]{_5n7rcgqhm}
{\pgfpoint{-0.5*\mcSize}{-0.5*\mcSize}}
{\pgfpoint{0.5*\mcSize}{0.5*\mcSize}}
{\pgfpoint{\mcSize}{\mcSize}}
{
\pgfsetcolor{\tikz@pattern@color}
\pgfsetlinewidth{\mcThickness}
\pgfpathcircle\pgfpointorigin{\mcRadius}
\pgfusepath{stroke}
}}
\makeatother

\tikzset{
pattern size/.store in=\mcSize, 
pattern size = 5pt,
pattern thickness/.store in=\mcThickness, 
pattern thickness = 0.3pt,
pattern radius/.store in=\mcRadius, 
pattern radius = 1pt}
\makeatletter
\pgfutil@ifundefined{pgf@pattern@name@_7o1tsuk3p}{
\makeatletter
\pgfdeclarepatternformonly[\mcRadius,\mcThickness,\mcSize]{_7o1tsuk3p}
{\pgfpoint{-0.5*\mcSize}{-0.5*\mcSize}}
{\pgfpoint{0.5*\mcSize}{0.5*\mcSize}}
{\pgfpoint{\mcSize}{\mcSize}}
{
\pgfsetcolor{\tikz@pattern@color}
\pgfsetlinewidth{\mcThickness}
\pgfpathcircle\pgfpointorigin{\mcRadius}
\pgfusepath{stroke}
}}
\makeatother

\tikzset{
pattern size/.store in=\mcSize, 
pattern size = 5pt,
pattern thickness/.store in=\mcThickness, 
pattern thickness = 0.3pt,
pattern radius/.store in=\mcRadius, 
pattern radius = 1pt}
\makeatletter
\pgfutil@ifundefined{pgf@pattern@name@_zubggx6qk}{
\makeatletter
\pgfdeclarepatternformonly[\mcRadius,\mcThickness,\mcSize]{_zubggx6qk}
{\pgfpoint{-0.5*\mcSize}{-0.5*\mcSize}}
{\pgfpoint{0.5*\mcSize}{0.5*\mcSize}}
{\pgfpoint{\mcSize}{\mcSize}}
{
\pgfsetcolor{\tikz@pattern@color}
\pgfsetlinewidth{\mcThickness}
\pgfpathcircle\pgfpointorigin{\mcRadius}
\pgfusepath{stroke}
}}
\makeatother

\tikzset{
pattern size/.store in=\mcSize, 
pattern size = 5pt,
pattern thickness/.store in=\mcThickness, 
pattern thickness = 0.3pt,
pattern radius/.store in=\mcRadius, 
pattern radius = 1pt}
\makeatletter
\pgfutil@ifundefined{pgf@pattern@name@_mcbkrmesn}{
\makeatletter
\pgfdeclarepatternformonly[\mcRadius,\mcThickness,\mcSize]{_mcbkrmesn}
{\pgfpoint{-0.5*\mcSize}{-0.5*\mcSize}}
{\pgfpoint{0.5*\mcSize}{0.5*\mcSize}}
{\pgfpoint{\mcSize}{\mcSize}}
{
\pgfsetcolor{\tikz@pattern@color}
\pgfsetlinewidth{\mcThickness}
\pgfpathcircle\pgfpointorigin{\mcRadius}
\pgfusepath{stroke}
}}
\makeatother
\tikzset{every picture/.style={line width=0.75pt}} 

\tikzset{
pattern size/.store in=\mcSize, 
pattern size = 5pt,
pattern thickness/.store in=\mcThickness, 
pattern thickness = 0.3pt,
pattern radius/.store in=\mcRadius, 
pattern radius = 1pt}
\makeatletter
\pgfutil@ifundefined{pgf@pattern@name@_72x0qoqws}{
\makeatletter
\pgfdeclarepatternformonly[\mcRadius,\mcThickness,\mcSize]{_72x0qoqws}
{\pgfpoint{-0.5*\mcSize}{-0.5*\mcSize}}
{\pgfpoint{0.5*\mcSize}{0.5*\mcSize}}
{\pgfpoint{\mcSize}{\mcSize}}
{
\pgfsetcolor{\tikz@pattern@color}
\pgfsetlinewidth{\mcThickness}
\pgfpathcircle\pgfpointorigin{\mcRadius}
\pgfusepath{stroke}
}}
\makeatother

\tikzset{
pattern size/.store in=\mcSize, 
pattern size = 5pt,
pattern thickness/.store in=\mcThickness, 
pattern thickness = 0.3pt,
pattern radius/.store in=\mcRadius, 
pattern radius = 1pt}
\makeatletter
\pgfutil@ifundefined{pgf@pattern@name@_f1y3jgxey}{
\makeatletter
\pgfdeclarepatternformonly[\mcRadius,\mcThickness,\mcSize]{_f1y3jgxey}
{\pgfpoint{-0.5*\mcSize}{-0.5*\mcSize}}
{\pgfpoint{0.5*\mcSize}{0.5*\mcSize}}
{\pgfpoint{\mcSize}{\mcSize}}
{
\pgfsetcolor{\tikz@pattern@color}
\pgfsetlinewidth{\mcThickness}
\pgfpathcircle\pgfpointorigin{\mcRadius}
\pgfusepath{stroke}
}}
\makeatother

\tikzset{
pattern size/.store in=\mcSize, 
pattern size = 5pt,
pattern thickness/.store in=\mcThickness, 
pattern thickness = 0.3pt,
pattern radius/.store in=\mcRadius, 
pattern radius = 1pt}
\makeatletter
\pgfutil@ifundefined{pgf@pattern@name@_h3f2g35xw}{
\pgfdeclarepatternformonly[\mcThickness,\mcSize]{_h3f2g35xw}
{\pgfqpoint{0pt}{0pt}}
{\pgfpoint{\mcSize+\mcThickness}{\mcSize+\mcThickness}}
{\pgfpoint{\mcSize}{\mcSize}}
{
\pgfsetcolor{\tikz@pattern@color}
\pgfsetlinewidth{\mcThickness}
\pgfpathmoveto{\pgfqpoint{0pt}{0pt}}
\pgfpathlineto{\pgfpoint{\mcSize+\mcThickness}{\mcSize+\mcThickness}}
\pgfusepath{stroke}
}}
\makeatother

\tikzset{
pattern size/.store in=\mcSize, 
pattern size = 5pt,
pattern thickness/.store in=\mcThickness, 
pattern thickness = 0.3pt,
pattern radius/.store in=\mcRadius, 
pattern radius = 1pt}
\makeatletter
\pgfutil@ifundefined{pgf@pattern@name@_cebealfc3}{
\pgfdeclarepatternformonly[\mcThickness,\mcSize]{_cebealfc3}
{\pgfqpoint{0pt}{0pt}}
{\pgfpoint{\mcSize+\mcThickness}{\mcSize+\mcThickness}}
{\pgfpoint{\mcSize}{\mcSize}}
{
\pgfsetcolor{\tikz@pattern@color}
\pgfsetlinewidth{\mcThickness}
\pgfpathmoveto{\pgfqpoint{0pt}{0pt}}
\pgfpathlineto{\pgfpoint{\mcSize+\mcThickness}{\mcSize+\mcThickness}}
\pgfusepath{stroke}
}}
\makeatother

\tikzset{
pattern size/.store in=\mcSize, 
pattern size = 5pt,
pattern thickness/.store in=\mcThickness, 
pattern thickness = 0.3pt,
pattern radius/.store in=\mcRadius, 
pattern radius = 1pt}
\makeatletter
\pgfutil@ifundefined{pgf@pattern@name@_cl4zvaxox}{
\pgfdeclarepatternformonly[\mcThickness,\mcSize]{_cl4zvaxox}
{\pgfqpoint{0pt}{-\mcThickness}}
{\pgfpoint{\mcSize}{\mcSize}}
{\pgfpoint{\mcSize}{\mcSize}}
{
\pgfsetcolor{\tikz@pattern@color}
\pgfsetlinewidth{\mcThickness}
\pgfpathmoveto{\pgfqpoint{0pt}{\mcSize}}
\pgfpathlineto{\pgfpoint{\mcSize+\mcThickness}{-\mcThickness}}
\pgfusepath{stroke}
}}
\makeatother

\tikzset{
pattern size/.store in=\mcSize, 
pattern size = 5pt,
pattern thickness/.store in=\mcThickness, 
pattern thickness = 0.3pt,
pattern radius/.store in=\mcRadius, 
pattern radius = 1pt}
\makeatletter
\pgfutil@ifundefined{pgf@pattern@name@_yh1bf9qsy}{
\pgfdeclarepatternformonly[\mcThickness,\mcSize]{_yh1bf9qsy}
{\pgfqpoint{0pt}{-\mcThickness}}
{\pgfpoint{\mcSize}{\mcSize}}
{\pgfpoint{\mcSize}{\mcSize}}
{
\pgfsetcolor{\tikz@pattern@color}
\pgfsetlinewidth{\mcThickness}
\pgfpathmoveto{\pgfqpoint{0pt}{\mcSize}}
\pgfpathlineto{\pgfpoint{\mcSize+\mcThickness}{-\mcThickness}}
\pgfusepath{stroke}
}}
\makeatother

\newcommand{\euro}{\text{€}}

\tikzset{every picture/.style={line width=0.75pt}} 

\usepackage{amsmath}
\usepackage[makeroom]{cancel}

\usepackage{enumitem}

\usepackage{fancyhdr} 
\usepackage[lmargin=0.9in, rmargin=0.9in, tmargin=1in, bmargin=1in]{geometry} 
\usepackage{url} 
\usepackage{listings} 
\usepackage{amsmath,amssymb} 
\usepackage{array, booktabs} 
\usepackage{graphicx} 
\usepackage{float} 
\usepackage[export]{adjustbox} 

\usepackage[backend = bibtex,
            style = numeric,
            date = long,     
            sorting = none,
            maxcitenames = 3,   
            citestyle=numeric-comp,
            url=false]{biblatex} 
\usepackage{comment} 
\usepackage{afterpage} 

\usepackage{adjustbox}

\usepackage{hyperref}
\usepackage{cleveref}

\usepackage{pifont}

\author[1]{Samuel Fedida\thanks{sylf2@cam.ac.uk}}
\affil[1]{\emph{Centre for Quantum Information and Foundations, DAMTP, Centre for Mathematical Sciences, University of Cambridge, Wilberforce Road, Cambridge CB3 0WA, UK}}
\author[2]{Alberto Ibort\thanks{albertoi@math.uc3m.es}}
\author[2]{Arnau Mas-Dorca\thanks{arnau.mas@icmat.es}}
\affil[2]{\emph{Department of Mathematics, Univ. Carlos III de Madrid, Avda. de la Universidad 30, 28911, Legan\'es, Madrid;  ICMAT, Calle Nicolas Cabrera, 13-15. Campus de Cantoblanco, 28049 Madrid,  Spain}}

\title{\textbf{A groupoidal approach to quantum reference frames}}

\date{\today}

\begin{document}
\maketitle
\thispagestyle{empty}

\begin{abstract}
We develop the kinematical and operator-algebraic foundations of a groupoid-based relational quantum field theory on curved spacetimes. The starting point is the observation that the usual group-based quantum reference frame formalism is not directly suited to generic curved Lorentzian backgrounds as global symmetry groups are typically absent or too small. We formulate a notion of quantum reference frame for a continuous groupoid. This yields a groupoid relativization map and relational observables. We show that a localization limit recovers the ordinary non-relational description. 

We construct canonical sharp groupoid quantum reference frames, which form the groupoidal counterpart of the ideal quantum reference frames based on $L^2(G)$ in the group-based theory. We further prove that the groupoid quantum reference frame construction reduces to the standard operational quantum reference frame formalism for locally compact groups. The action groupoid quantum reference frames are torsor quantum reference frames only for specific classes of fields of positive operator-valued measures, and the torsor relativization map only applies to constant operator fields of system observables.

We review the foundations of relational quantum field theory in Minkowski spacetime. We prove new results that further link relational quantum field theory to Wightman quantum field theory. We show that covariant positive operator-valued measures are $\mu$-continuous with respect to quasi-invariant $\sigma$-finite positive Borel measures $\mu$, hence proving that relational quantum fields can indeed be understood as the smearing of pointwise-defined kernels with respect to the quantum reference frame's statistics. 

We develop relational quantum field theory in curved spacetime, where we argue that the correct replacement for the Poincar\'e group is the Poincar\'e groupoid of the spacetime, i.e. the gauge groupoid of its orthonormal frame bundle. We show that the groupoid framework suitably extends the formalism when curved spacetimes are considered. We also indicate how the framework extends further to internal gauge symmetry and relational gauge-covariant quantum field theory via Atiyah groupoids, providing a first step towards formulating a relational quantum Yang-Mills field theory. These results suggest that groupoids provide the natural kinematical language for operational quantum reference frames and relational quantum field theory beyond the homogeneous space setting.

\end{abstract}

\newpage

\tableofcontents

\newpage


\section{Introduction}
\label{sec:introduction}

In recent years, quantum reference frames (QRFs) have emerged as a mathematically precise way of implementing the basic relational insight that the description of a physical system should be given relative to another physical system \cite{rovelli_relational_1996,timpson_quantum_2008,gambini_montevideo_2009}, itself treated quantum mechanically \cite{bartlett_reference_2007,loveridge_quantum_2012,belenchia_quantum_2018,loveridge_symmetry_2018,giacomini_quantum_2019,loveridge_relative_2019,vanrietvelde_change_2020,de_la_hamette_quantum_2020,de_la_hamette_perspective-neutral_2021,kabel_quantum_2023,hoehn_quantum_2023,glowacki_operational_2023,castro-ruiz_relative_2025,kabel_quantum_2025,de_vuyst_gravitational_2025,carette_operational_2025}. In the operational approach \cite{loveridge_symmetry_2018,loveridge_relativity_2017,glowacki_quantum_2023}, a quantum reference frame is described by a unitary representation of a locally compact group together with a covariant positive operator-valued measure (POVM) on a homogeneous space of that group. This framework has proved flexible enough to describe relative observables, relative states, and transformations between quantum frames \cite{glowacki_quantum_2023,carette_operational_2025}, and it provides a natural arena in which to formulate symmetry principles operationally.

In the relativistic setting, Minkowski spacetime fits naturally into this picture. The proper orthochronous Poincaré group acts transitively on the relevant space of inertial frames, and the operational formalism of QRFs may be formulated in terms of systems of covariance for that group. This line of thought has led to a relational approach to quantum field theory in Minkowski spacetime \cite{fedida_foundations_2025}, in which relational observables and quantum fields arise as derived notions from the operational QRF formalism. In that flat spacetime setting, the relevant space of classical inertial reference frames may be identified with the trivial Lorentz bundle over Minkowski spacetime, and relational observables are obtained by averaging absolute fields against the localization statistics of a relativistic quantum frame in that frame bundle. The idea of ``quantising" the frame bundle has also been approached by Vanzella and Butterfield \cite{vanzella_frame-bundle_2024}, who assigned complex amplitudes (``wavefunctions") to tetrads. 

The situation changes drastically on a generic curved spacetime $(\mathcal{M},g)$. In general, the isometry group of $(\mathcal{M},g)$ is too small, and typically trivial, to support a useful homogeneous-space description of frames. Thus the standard group-based QRF formalism no longer applies.\footnote{We mention here an ongoing project by G{\l}owacki who is exploring a group-based extension of the QRF formalism to principal bundles \cite{glowacki_towards_2024,glowacki_approximate_2026}. The relativization map in that group-based approach to quantum reference frames applies to operator-valued functions on the bundle rather than to operators. Understanding the links between the approach of the present paper that deals with groupoids and observables (and recovers a notion of quantum field as a necessary emergent feature of the relativistic structure of the theory), and that explored by G{\l}owacki that uses groups and operator-valued functions (where one is relativising already-existing quantum fields rather than observables), is kept for future work.} Nevertheless, curved Lorentzian geometry still possesses a canonical kinematical object: the bundle
\begin{equation}
O_+^\uparrow(\mathcal{M},g)\to \mathcal{M}
\end{equation}
of oriented orthochronous orthonormal frames, and therefore its canonical gauge groupoid
\begin{equation}
O_+^\uparrow(\mathcal{M},g)\times_{SO_+^\uparrow(1,d-1)} O_+^\uparrow(\mathcal{M},g)\rightrightarrows \mathcal{M}.
\end{equation}
Equivalently, this groupoid may be described intrinsically as
\begin{equation}
\mathrm{Poin}(\mathcal{M},g)
=
\left\{
(y,T_{yx},x)\ \middle|\ x,y\in \mathcal{M},\;
T_{yx}:T_x\mathcal{M}\to T_y\mathcal{M},\;
T_{yx}^{*}g_y=g_x
\right\},
\end{equation}
whose arrows are linear isometries between tangent spaces at different spacetime points. Following \cite{ibort_notion_2025}, we refer to this Lie groupoid as the \emph{Poincar\'e groupoid} of the spacetime $(\mathcal{M},g)$.

The central thesis of this paper is that the Poincar\'e groupoid is the correct replacement for the Poincar\'e group in the formulation of operational quantum reference frames on curved spacetimes. More generally, we argue that when global symmetry is no longer available, one should not abandon covariance, but rather pass from \emph{groups} to \emph{groupoids}. This is a natural extension of relational quantum field theory on Minkowski spacetime for at least three reasons. First, the Poincar\'e groupoid exists on every Lorentzian spacetime, regardless of the existence of nontrivial global isometries. Second, in the flat case it reduces to the familiar Poincar\'e-group picture. Third, groupoids come equipped with a well-developed measure-theoretic and operator-algebraic formalism, making them suitable kinematical structures for quantum theory.

Our aim is therefore twofold. On the one hand, we extend the operational QRF formalism from transitive locally compact group actions to continuous groupoids. On the other, we apply that extension to the Poincar\'e groupoid of a Lorentzian spacetime, thereby obtaining a groupoid-based framework for relational quantum field theory beyond the homogeneous space setting. The resulting framework is broad enough to accommodate the notion of operational quantum reference frames on curved spacetimes; a relativization and relative observables without global symmetry groups; relational observables and fields on curved Lorentzian backgrounds and, on top of that, further extensions to internal gauge symmetry via Atiyah groupoids.

From a broader perspective, this paper should also be viewed as part of a larger groupoid-based programme in the foundations of quantum theory \cite{ciaglia_groupoidal_2024,connes_noncommutative_1994,landsman_mathematical_1998}. The Poincar\'e and Wigner groupoids have recently been proposed as natural kinematical structures for the description of particles on curved spacetimes \cite{ibort_notion_2025}. Here we show that they also provide a natural geometric setting for quantum reference frames and relational observables. This suggests that particles, frames, and fields may all admit a common treatment in a unified groupoid language.

This extension of QRFs from locally compact groups to groupoids should not be seen merely as a mathematical exercise. We argue that groupoids are fundamentally the correct object to consider in the context of quantum reference frames when one is concerned with physical symmetries. Some arguments, both conceptual and technical, to support this claim are as follows.
\begin{enumerate}
    \item Only \emph{local} Poincaré symmetry is well-tested. However, the Poincaré group, and more generally the isometry group of a spacetime, restricted to a region $U \subset \mathcal{M}$ (Minkowski for the Poincaré group) is not a group. Thus, we have, strictly speaking, never experimentally tested Poincaré (nor isometry) \emph{group} symmetry in a lab. However, the Poincaré groupoid is still a groupoid when restricted to such regions (and indeed, also works on curved spacetimes).
    \item The notion of symmetry is not always global: local symmetries, captured by groupoids, are just as important (if not more) than global symmetries, captured by groups but also by groupoids. Working with groupoids thus allows one to discuss both global and local symmetries.
    \item When one is interested in the symmetries of a curved spacetime, the isometry group of a generic spacetime is trivial. If, however, one wants to talk about the diffeomorphisms of that spacetime, then this becomes technically arduous: the diffeomorphism group of a spacetime is not locally compact and typically a very complicated object which cannot be straightforwardly integrated into the operational QRF formalism.
    \item To talk about the mass and spin of a quantum field in flat spacetime we typically use the Wigner classification of (projective) unitary representations of the Poincaré group with respect to which the quantum field is covariant. It was argued in \cite{ibort_notion_2025} that the same can and indeed should be done in curved spacetimes.
    \item The idea, in relational quantum physics, that ``true" observables of a system and a frame are those which are invariant under symmetry transformations has been argued from the point of view of the Wigner-Araki-Yanase (WAY) theorem \cite{busch_translation_2010,araki_measurement_1960,yanase_optimal_1961,ozawa_conservation_2002}. This result is often understood in the context of group symmetries (e.g. see \cite{loveridge_relational_2020} for a review of those in the context of QRFs), but really is more general than that. It involves the commutation of observables with unitaries, which can be taken to be unitary representations of groups. However, one can also consider invariance and unitary representations of more general symmetry structures, and in particular of groupoid symmetries.
    \item To discuss gauge field theories, one may be tempted to work with gauge groups of the form $\mathcal{G}=C^\infty(\mink,G)$. However, these are typically very difficult objects to work with: gauge groups are not locally compact and so cannot generically be treated in the standard QRF formalism (see \cite{fewster_semi-local_2025} for a recent attempt at studying gauge group symmetries in the context of QRFs).
    \item More generally, gauge theories are not fully described by $\mathcal{G}=C^\infty(\M,G)$, except in a chosen trivialization of a trivial principal bundle. In general, the gauge data consist of a principal $G$-bundle $P \to \M$, its associated bundles, and its gauge group $\text{Gau}(P)$ containing vertical bundle automorphisms. Nontrivial bundles encode topological sectors and holonomy data that are invisible in a global $C^\infty(\M,G)$ description. The Atiyah groupoid $(P \times P)/G \rightrightarrows \M$ is therefore the natural object for working with gauge theory and local gauge transformations.
\end{enumerate}

\bigskip

\subsection{Organisation of the paper}

The paper is organized as follows. In Sec.~\ref{sec:lie-groupoids}, we review the geometric and representation-theoretic background of measure groupoids, with special emphasis on action groupoids, Atiyah groupoids, and the Poincar\'e and Wigner groupoids associated with a Lorentzian spacetime.

In Sec.~\ref{sec:operational-qrf-groups}, we recall the operational theory of quantum reference frames for locally compact groups and torsors in the form most convenient for comparison with the groupoid setting.

Sec.~\ref{sec:groupoid-qrf} contains the main conceptual development of the paper: the definition of groupoid quantum reference frames, the construction of the corresponding relativization maps, and the proof that the formalism reduces to the usual one in the action-groupoid case.

In Sec.~\ref{sec:RQFT Minkowski}, we apply the formalism to relational quantum field theory, first in Minkowski spacetime. 

We then pass to generic curved spacetimes in Sec.~\ref{sec:RQFT curved}, where the relational field is defined intrinsically on the Poincar\'e groupoid without any global trivialization of the frame bundle. We show that, in the case where the spacetime is flat, both pictures concur.

Finally, in Sec.~\ref{sec:towards RQGFT}, we indicate how the same logic extends further to internal gauge symmetry via Atiyah groupoids, leading toward a groupoid-based formulation of relational gauge-covariant quantum field theory.

Sec.~\ref{sec:conclusions} concludes with a discussion of the main conceptual consequences of the paper and several open questions.

\subsection{Notation and conventions}

All Hilbert spaces under consideration are complex and separable. If $\hi$ is a Hilbert space, we denote by $\bh$ the von Neumann algebra of bounded operators on $\hi$, by $\bh^{sa} \subset \bh$ the self-adjoint bounded operators, by $\thi$ the trace-class operators, by $\homstate$ the density operators, and by $\eh \subset \bh^{sa}$ the convex set of effects. For a measurable space $X$, $\Bor(X)$ denotes its Borel $\sigma$-algebra.

Given a groupoid $\Gamma\rightrightarrows M$, the source and target maps are denoted by
\begin{equation}
s,t:\Gamma\to M.
\end{equation}
The identity arrow at $x\in M$ is denoted by $1_x$, and the inverse of $\alpha\in\Gamma$ by $\alpha^{-1}$. We write $\Gamma^{(2)}\subseteq \Gamma\times\Gamma$ for the set of composable pairs. If $\alpha:x\to y$, the left translation map is written
\begin{equation}
L_\alpha:\Gamma^{x}\to\Gamma^{y},
\qquad
\beta\mapsto \alpha\circ\beta.
\end{equation}

For a $d$-dimensional Lorentzian spacetime $(\M,g)$, the bundle of oriented orthochronous orthonormal frames is denoted by $O_+^\uparrow(\M,g)\to \M$, and its Poincar\'e groupoid by $\mathrm{Poin}(\mathcal{M},g)\rightrightarrows \M$. In the flat case, Minkowski spacetime is denoted by $\mink$ and equipped with metric $\eta$ of mostly-minus signature. The spacetime translation group is denoted by $T(1,d-1)$, the proper orthochronous Lorentz group by $SO_+^\uparrow(1,d-1)$, and the proper orthochronous Poincar\'e group by
\begin{equation}
\Poincup
=
T(1,d-1)\rtimes SO_+^\uparrow(1,d-1).
\end{equation}

Whenever direct integrals are used, measurability is understood in the standard sense of measurable fields of Hilbert spaces. All operator-valued integrals are understood ultraweakly in the relevant operator algebras. 

\section{Continuous groupoids and kinematical representations}
\label{sec:lie-groupoids}

In this section we collect the geometric and representation-theoretic background needed in the rest of the paper. Our main kinematical object is a continuous groupoid $\Gamma\rightrightarrows M$, to be interpreted as encoding admissible changes of local frame on a spacetime or, more generally, local symmetry data. Note that, for most examples of physical interest, one is interested in Lie groupoids, which are a special (smooth) type of continuous groupoids. The fundamental example for us is the Poincar\'e groupoid of a Lorentzian spacetime, together with its phase-space counterpart, the Wigner groupoid. We also recall the notion of a measurable field of Hilbert spaces and of a unitary representation of a continuous groupoid, since these will replace the familiar single-Hilbert-space representation theory of symmetry groups in the groupoid-based formulation of quantum reference frames.

\subsection{Lie groupoids, action groupoids and gauge groupoids}
\label{subsec:lie-groupoids-basic}

We begin by recalling some standard terminology; see for instance \cite{mackenzie_general_2005,landsman_mathematical_1998,ibort_introduction_2021}.

\begin{definition}
    A \emph{continuous groupoid} $\Gamma\rightrightarrows M$ consists of a topological space of arrows $\Gamma$, a topological space of objects $M$, continuous maps
\begin{equation}
s,t:\Gamma\to M
\end{equation}
called the \emph{source} and \emph{target} maps, a continuous unit map
\begin{equation}
u:M\to \Gamma,\qquad x\mapsto 1_x,
\end{equation}
a continuous inversion map
\begin{equation}
\iota:\Gamma\to\Gamma,\qquad \alpha\mapsto \alpha^{-1},
\end{equation}
and a continuous partially defined multiplication
\begin{equation}
m:\Gamma^{(2)}\to \Gamma,\qquad (\alpha,\beta)\mapsto \alpha\circ\beta,
\end{equation}
where
\begin{equation}
\Gamma^{(2)}:=\{(\alpha,\beta)\in \Gamma\times\Gamma\mid s(\alpha)=t(\beta)\},
\end{equation}
satisfying the usual associativity, identity and inverse axioms for groupoids. 
\end{definition}

If $\alpha\in\Gamma$ satisfies $s(\alpha)=x$ and $t(\alpha)=y$, we write $\alpha:x\to y$. For $x\in M$ we denote
\begin{equation}
\Gamma_x:=s^{-1}(x),\qquad \Gamma^x:=t^{-1}(x),\qquad \Gamma_x^x:=\Gamma_x\cap \Gamma^x.
\end{equation}
$\Gamma_x^x$, also denoted $\Gamma(x)$, is the \emph{isotropy group} at $x$. We write
\begin{equation}
    \Gamma^y_x := \Gamma_x \cap \Gamma^y =  s^{-1}(x) \cap t^{-1}(y) = \{\alpha \in \Gamma \mid s(\alpha) = x, \, t(\alpha) = y\} \subset \Gamma.
\end{equation}

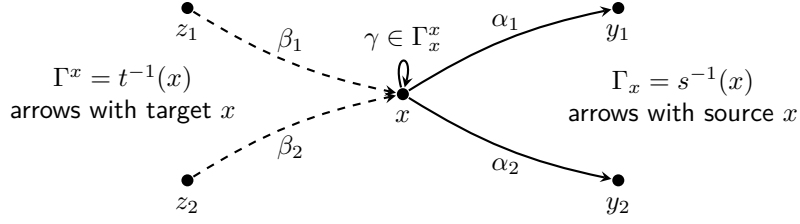
\begin{figure}[t]
\centering
\begin{tikzpicture}[>=stealth,scale=0.95]
  \node[circle,fill=black,inner sep=1.6pt,label=below:$x$] (x) at (0,0) {};
  \node[circle,fill=black,inner sep=1.6pt,label=below:$y_1$] (y1) at (3,1.2) {};
  \node[circle,fill=black,inner sep=1.6pt,label=below:$y_2$] (y2) at (3,-1.2) {};
  \node[circle,fill=black,inner sep=1.6pt,label=below:$z_1$] (z1) at (-3,1.2) {};
  \node[circle,fill=black,inner sep=1.6pt,label=below:$z_2$] (z2) at (-3,-1.2) {};

  \draw[->,thick] (x) to[bend left=10] node[above] {$\alpha_1$} (y1);
  \draw[->,thick] (x) to[bend right=10] node[below] {$\alpha_2$} (y2);
  \draw[->,thick,dashed] (z1) to[bend right=10] node[above] {$\beta_1$} (x);
  \draw[->,thick,dashed] (z2) to[bend left=10] node[below] {$\beta_2$} (x);
  \draw[->,thick] (x) to[loop above] node {$\gamma\in\Gamma_x^x$} (x);

  \node[align=center] at (3.9,0) {$\Gamma_x=s^{-1}(x)$\\arrows with source $x$};
  \node[align=center] at (-3.9,0) {$\Gamma^x=t^{-1}(x)$\\arrows with target $x$};
\end{tikzpicture}
\caption{The source fiber $\Gamma_x$, the target fiber $\Gamma^x$, and their
intersection $\Gamma_x^x=\Gamma_x\cap\Gamma^x$, the isotropy group at $x$.}
\label{fig:source-target-fibers}
\end{figure}

The orbit of $x$ is
\begin{equation}
\mathcal O_x:=t(\Gamma_x)=\{y\in M\mid \exists\,\alpha:x\to y\} \subset M.
\end{equation}
The groupoid is said to be \emph{transitive} if it has a single orbit. The isotropy groups $\Gamma(x)$ and $\Gamma(y)$ of points $x,y$ lying in the same orbit are isomorphic; for a transitive groupoid, we write $H$ for the common isomorphism class and refer to it as the isotropy group of $\Gamma$. These notions are represented pictorially in Figs.~\ref{fig:source-target-fibers} and \ref{fig:groupoids}.

An important example of continuous groupoids is that of Lie groupoids.

\begin{definition}
\label{def:lie-groupoid}
A \emph{Lie groupoid} is a continuous groupoid in which $\Gamma$ and $M$ are smooth manifolds, $s$ and $t$ are smooth surjective submersions, and all structure maps are smooth.
\end{definition}

\begin{example}[Action groupoid]
\label{ex:action-groupoid}
Let a Lie group $G$ act smoothly on a manifold $M$. The associated \emph{action groupoid}
\begin{equation}
G\ltimes M \rightrightarrows M
\end{equation}
has arrows $G\times M$, source and target maps
\begin{equation}
s(g,x)=x,\qquad t(g,x)=g\cdot x,
\end{equation}
unit map $x\mapsto (e,x)$, inverse
\begin{equation}
(g,x)^{-1}=(g^{-1},g\cdot x),
\end{equation}
and composition
\begin{equation}
(g',g\cdot x)\circ (g,x)=(g'g,x).
\end{equation}
This is the groupoid avatar of an ordinary global symmetry action on a manifold $M$.
\end{example}

\begin{example}[Atiyah groupoid]
\label{ex:gauge-groupoid}
Let $\pi:P\to M$ be a right principal $H$-bundle.  
The \emph{Atiyah groupoid} (sometimes called gauge groupoid) of $P$ is the Lie groupoid
\begin{equation}
\At(P):=(P\times P)/H \rightrightarrows M,
\end{equation}
whose arrows are classes
\begin{equation}
[p_x,p_y],\qquad p_x\in P_x,\ p_y\in P_y,
\end{equation}
under the diagonal right action
\begin{equation}
(p_x,p_y)\cdot k := (p_x\cdot k,p_y\cdot k),\qquad k\in H.
\end{equation}
The source and target maps are
\begin{equation}
s([p_x,p_y])=y,\qquad t([p_x,p_y])=x,
\end{equation}
the unit at $x$ is
\begin{equation}
1_x=[p_x,p_x],
\end{equation}
the inverse is
\begin{equation}
[p_x,p_y]^{-1}=[p_y,p_x],
\end{equation}
and the composition law is
\begin{equation}
[p_z,p_y]\circ [p_y,p_x]=[p_z,p_x].
\end{equation}

Equivalently, an arrow of $\At(P)$ may be described as a $H$-equivariant map
\begin{equation}
\alpha_{yx}:P_x\to P_y \, .
\end{equation}
Indeed, given the class $[p_y,p_x]$, we define the map $\alpha_{yx}:P_x\to P_y$ as the assignment $\alpha_{yx}(p) = q$, with $p = p_x h$, and $q = p_y h$, for all $p\in P_x$.  This map is well defined because if we change the representative $[p'_y,p_x'] = [p_y,p_x]$, then there is $k\in H$ such that $p_x' = p_x k$ and $p_y' = p_y k$. Then given $p \in P_x$, $p = p_xh = p_x k k^{-1}h = p_x' (k^{-1}h)$, and $\alpha'_{yx}(p) = p_y' (k^{-1}h) = p_y h = q$.

\end{example}

The two examples above are the basic geometric models that will reappear throughout the paper. Action groupoids encode ordinary global symmetries, while Atiyah groupoids encode local frame changes associated with a principal bundle even in the absence of a nontrivial global symmetry group. A useful concept we will use in the context of Atiyah groupoids is that of bisections.

\begin{definition}
    A \emph{bisection} of a Lie groupoid $\Gamma \rightrightarrows M$ is a section $b : M \to \Gamma$ such that $s \circ b = \text{id}_M$ and $\varphi_b := t \circ b \in \mathrm{Diff}(M)$. The set of bisections forms a group \cite{mackenzie_lie_2005} denoted $\text{Bis}(\Gamma)$ that acts in a natural way on $M$ by means of $(b,x) \mapsto \varphi_b(x)$ .   
\end{definition}

\begin{remark}
    For the Atiyah groupoid of a principal bundle $\pi : P \to M$, $\text{Bis}(\At(P)) \cong \text{Aut}(P)$ (the group of bundle automorphisms $\varphi \colon P \to P$), and the subgroup $\text{Bis}_{\id}(\At(P))$ of bisections $b$ covering the identity, that is $\varphi_b = \mathrm{id}_M$, is isomorphic to the subgroup of gauge transformations  $\text{Gau}(P)$, that is the subgroup of $\text{Aut}(P)$ covering the identity.
\end{remark}

\begin{remark}\label{rem:bisections_action_groupoid}
For the action groupoid $\Gamma = G \times M$ corresponding to the action of the Lie group $G$ on a smooth manifold $M$, Example \ref{ex:action-groupoid}, the group $G$ is a canonical subgroup of the group of bisections $\text{Bis}(G \times M)$.  Indeed, for any $g\in G$ we define the bisection $b_g (x) = (g,x)$, for all $x\in M$.   Then, clearly, $\varphi_{b_g}(x) =gx$, and the original action of $G$ on $M$ is recovered as the action of the subgroup $G$ of the group of bisections. 
\end{remark}

\subsection{The Poincaré groupoid of a Lorentzian spacetime}
\label{subsec:poincare-groupoid}

Let $(\mathcal{M},g)$ be a $d$-dimensional Lorentzian spacetime. To simplify notation, we shall assume that $(\mathcal{M},g)$ is oriented and time-oriented, so that its bundle of oriented orthochronous orthonormal frames
\begin{equation}
O_+^\uparrow(\mathcal{M},g)\to \mathcal{M}
\end{equation}
is a principal $SO_+^\uparrow(1,d-1)$-bundle. All constructions below admit an obvious variant using the full orthonormal frame bundle and the full Lorentz group.

\begin{definition}[\cite{ibort_notion_2025}]
\label{def:poincare-groupoid-gauge}
The \emph{Poincar\'e groupoid} of $(\mathcal{M},g)$ is the gauge groupoid
\begin{equation}
\mathrm{Poin}(\mathcal{M},g):=
O_+^\uparrow(\mathcal{M},g)\times_{SO_+^\uparrow(1,d-1)} O_+^\uparrow(\mathcal{M},g)
\rightrightarrows \mathcal{M}.
\end{equation}
\end{definition}

There is an equivalent intrinsic description which is often more convenient.

\begin{proposition}[\cite{ibort_notion_2025}]
\label{prop:poincare-groupoid-intrinsic}
The Poincar\'e groupoid $\mathrm{Poin}(\mathcal{M},g)$ is canonically isomorphic to the Lie groupoid whose arrows are triples
\begin{equation}
(y,T_{yx},x),
\end{equation}
where $x,y\in \mathcal{M}$ and
\begin{equation}
T_{yx}:T_x\mathcal{M}\to T_y\mathcal{M}
\end{equation}
is an orientation- and time-orientation-preserving linear isometry, i.e.
\begin{equation}
T_{yx}^{*}g_y=g_x.
\end{equation}
and $T_{yx}$ maps the future-component of the light cone at $x$ into the future-component of the light cone at $y$.
The source, target, identity, inverse and composition maps are given by
\begin{equation}
s(y,T_{yx},x)=x,\qquad t(y,T_{yx},x)=y,
\end{equation}
\begin{equation}
1_x=(x,\id_{T_x\mathcal{M}},x),
\end{equation}
\begin{equation}
(y,T_{yx},x)^{-1}=(x,T_{yx}^{-1},y),
\end{equation}
and
\begin{equation}
(z,T_{zy},y)\circ (y,T_{yx},x)=(z,T_{zy}\circ T_{yx},x).
\end{equation}
\end{proposition}

\begin{proof}
Let $[u_y,u_x]\in O_+^\uparrow(\mathcal{M},g)\times_{SO_+^\uparrow(1,d-1)} O_+^\uparrow(\mathcal{M},g)$, where
\begin{equation}
u_x:\mathbb R^{1,d-1}\to T_x\mathcal{M},\qquad u_y:\mathbb R^{1,d-1}\to T_y\mathcal{M}
\end{equation}
are oriented orthochronous orthonormal frames. The map
\begin{equation}
[u_y,u_x]\longmapsto u_y\circ u_x^{-1}:T_x\mathcal{M}\to T_y\mathcal{M}
\end{equation}
is a well-defined linear isometry: replacing $(u_y,u_x)$ by $(u_y h, u_x h)$ for $h\in SO_+^\uparrow(1,d-1)$ gives $u_y h\circ (u_x h)^{-1} = u_y \circ u_x^{-1}$, so the map depends only on the equivalence class. Conversely, given a linear isometry $T_{yx}:T_x\mathcal{M}\to T_y\mathcal{M}$ and any frame $u_x$ at $x$, the frame
\begin{equation}
u_y:=T_{yx}\circ u_x
\end{equation}
is orthonormal at $y$, and the class $[u_y,u_x]$ depends only on $T_{yx}$. The remaining structure maps are immediate.
\end{proof}

\begin{remark}
\label{rem:poincare-no-connection}
No connection is required in Definition \ref{def:poincare-groupoid-gauge}. The Levi-Civita connection determines a distinguished subgroupoid, namely the holonomy groupoid obtained by parallel transport along curves, but for the purposes of quantum reference frames we shall work with the full Poincar\'e groupoid. This reflects the fact that the kinematical comparison of local frames is more primitive than any particular dynamical rule selecting preferred transports.
\end{remark}

\begin{proposition}[\cite{ibort_notion_2025}]
\label{prop:poincare-transitive}
The Lie groupoid $\mathrm{Poin}(\mathcal{M},g)\rightrightarrows \mathcal{M}$ is transitive. For every $x\in \mathcal{M}$, the isotropy group $\mathrm{Poin}(\mathcal{M},g)_x^x$
is isomorphic to $SO_+^\uparrow(1,d-1)$ after choosing an oriented orthochronous orthonormal frame at $x$.
\end{proposition}

\begin{proof}
For any $x,y\in \mathcal{M}$, the tangent spaces $(T_x\mathcal{M},g_x)$ and $(T_y\mathcal{M},g_y)$ are isometric Lorentzian vector spaces of the same signature, hence there exists a linear isometry $T_{yx}:T_x\mathcal{M}\to T_y\mathcal{M}$. This gives an arrow $x\to y$, so the groupoid is transitive. When $x=y$, the arrows are precisely the linear isometries of $(T_x\mathcal{M},g_x)$ preserving the chosen orientation and time-orientation, hence the isotropy group is $SO_+^\uparrow(1,d-1)$.
\end{proof}

Thus $\mathrm{Poin}(\mathcal{M},g)$ should be viewed as the curved-spacetime analogue of the usual Poincar\'e symmetry group: its arrows compare orthonormal frames at different points, while its isotropy groups recover the local Lorentz symmetry. This groupoid coincides with the kinematical groupoid $R(\mathcal{M})$ introduced in \cite{ciaglia_groupoidal_2024} in the context of the groupoidal picture of quantum mechanics.

\begin{remark}
    For $\mathrm{Poin}(\M,g)$, bisections are pairs $(\varphi,T)$ with $\varphi \in \text{Diff}(\M)$ and $T : T\M \to T\M$ is a metric-preserving bundle map over $\varphi$.
\end{remark}

\begin{definition}
    The subgroup of isometry-induced bisections of the Poincar\'e groupoid is
\begin{equation}
    \mathrm{Bis}_{\mathrm{Isom}}
    (\mathrm{Poin}(\M,g)) :=
    \left\{b_\varphi \in \mathrm{Bis}(\mathrm{Poin}(\M,g)) \mid   \varphi\in\operatorname{Isom}_+^\uparrow(\M,g),
    b_\varphi(x)=(\varphi(x),d\varphi_x,x)  
    \right\}.
\end{equation}
This group is distinct from the isotropy group
$\mathrm{Poin}(\M,g)_x^x$, which consists of Lorentz transformations based at one
fixed point.
\end{definition}

\subsection{Flat spacetime and recovery of the ordinary Poincar\'e picture}
\label{subsec:flat-recovery}

We now highlight that in Minkowski spacetime the Poincar\'e groupoid collapses to the familiar action-groupoid description.

\begin{example}[Minkowski spacetime]
\label{ex:minkowski-poincare-groupoid}
Let $(\mathbb M,\eta)$ be $d$-dimensional Minkowski spacetime. Since $T\mathbb M$ is globally trivial,
\begin{equation}
T\mathbb M \cong \mathbb M\times \mathbb R^{1,d-1},
\end{equation}
every arrow of $\mathrm{Poin}(\mathbb M,\eta)$ may be written uniquely as
\begin{equation}
(y,\Lambda,x),
\qquad
x,y\in \mathbb M,\quad \Lambda\in SO_+^\uparrow(1,d-1),
\end{equation}
where $\Lambda$ acts as the corresponding Lorentz transformation in the global orthonormal frame. Hence
\begin{equation}
\mathrm{Poin}(\mathbb M,\eta)\cong \mathbb M\times SO_+^\uparrow(1,d-1)\times \mathbb M
\end{equation}
as a smooth manifold.
\end{example}

Let $\mathcal P_+^\uparrow= T(1,d-1)\rtimes SO_+^\uparrow(1,d-1)$
be the proper orthochronous Poincar\'e group acting on Minkowski spacetime by $(a,\Lambda)\cdot x=\Lambda x+a$.
The associated action groupoid is
\begin{equation}
\mathcal P_+^\uparrow\ltimes \mathbb M \rightrightarrows \mathbb M.
\end{equation}

\begin{proposition}
\label{prop:minkowski-action-groupoid}
There is a canonical isomorphism of Lie groupoids
\begin{equation}\label{eq:Poincare_Minkowski}
\mathrm{Poin}(\mathbb M,\eta)\cong \mathcal P_+^\uparrow\ltimes \mathbb M.
\end{equation}
Explicitly, the isomorphism is given by
\begin{equation}
\Phi:\mathrm{Poin}(\mathbb M,\eta)\to \mathcal P_+^\uparrow\ltimes \mathbb M,
\qquad
(y,\Lambda,x)\mapsto ((a,\Lambda),x),
\end{equation}
where
\begin{equation}
a:=y-\Lambda x.
\end{equation}
Its inverse is
\begin{equation}
\Phi^{-1}((a,\Lambda),x)=(\Lambda x+a,\Lambda,x).
\end{equation}
\end{proposition}

\begin{proof}
The formula for $\Phi$ is well defined because for fixed $x,y,\Lambda$ there is a unique translation vector $a$ satisfying $y=\Lambda x+a$. A direct computation shows that source and target maps are preserved:
\begin{equation}
s((y,\Lambda,x))=x,\qquad t((y,\Lambda,x))=y=(a,\Lambda)\cdot x.
\end{equation}
For the composition, let $a:=y-\Lambda x$ and $a':=z-\Lambda' y$. Then
\begin{equation}
\Phi\bigl((z,\Lambda',y)\circ(y,\Lambda,x)\bigr)
=
\Phi(z,\Lambda'\Lambda,x)
=
((z-\Lambda'\Lambda x,\Lambda'\Lambda),x)
=
((a'+\Lambda'a,\Lambda'\Lambda),x),
\end{equation}
which is exactly the product
\begin{equation}
((a',\Lambda'),y)\circ ((a,\Lambda),x)
\end{equation}
in the action groupoid. The formula for the inverse is immediate.
\end{proof}

Because of \ref{rem:bisections_action_groupoid}, the Poincar\'e group $\mathcal P_+^\uparrow$ is now a subgroup of the group of bisections of the Poincar\'e groupoid $\mathrm{Poin}(\mathbb M,\eta)$, whose action on Minkowski spacetime $\mathbb{M}$ is induced from the action of the group of bisections on $\mathbb{M}$ recovering the complete action-groupoid structure of Minkowski spacetime from its Poincar\'e groupoid, hence Prop.~\ref{prop:minkowski-action-groupoid} is the geometric core of the flat-spacetime consistency check stated in Sec.~\ref{sec:introduction}: once the operational formalism is developed for action groupoids, the ordinary Poincar\'e-group description is recovered automatically from the groupoid picture. 

\subsection{The Wigner groupoid}
\label{subsec:wigner-groupoid}

Besides the Poincar\'e groupoid, there is a closely related phase-space groupoid which will be useful when discussing particle sectors and, later on, relational quantum fields.

\begin{definition}
\label{def:wigner-groupoid}
Let $(\mathcal{M},g)$ be a Lorentzian spacetime. The \emph{Wigner groupoid} of $(\mathcal{M},g)$ is the Lie groupoid
\begin{equation}
\mathrm{Wig}(\mathcal{M},g)\rightrightarrows T^*\mathcal{M}
\end{equation}
whose arrows are triples
\begin{equation}
(p_y,T_{yx},p_x),
\end{equation}
where $p_x\in T_x^*M$, $p_y\in T_y^*M$, and $T_{yx}:T_x\mathcal{M}\to T_y\mathcal{M}$
is a linear isometry satisfying the compatibility condition
\begin{equation}
T_{yx}^*p_y=p_x,
\end{equation}
where $T_{yx}^*:T_y^*M\to T_x^*M$ denotes the dual (pullback) map.
The source and target maps are
\begin{equation}
s(p_y,T_{yx},p_x)=(x,p_x),\qquad t(p_y,T_{yx},p_x)=(y,p_y),
\end{equation}
the identities are
\begin{equation}
1_{(x,p_x)}=(p_x,\id_{T_x\mathcal{M}},p_x),
\end{equation}
and the composition law is
\begin{equation}
(p_z,T_{zy},p_y)\circ (p_y,T_{yx},p_x)
=
(p_z,T_{zy}\circ T_{yx},p_x).
\end{equation}
\end{definition}

Intuitively, $\mathrm{Wig}(\mathcal{M},g)$ refines $\mathrm{Poin}(\mathcal{M},g)$ by keeping track not only of local frames but also of cotangent vectors, i.e. momenta. It is therefore the natural groupoid analogue of the usual Wigner picture of particle kinematics.

For the orbit list below we assume $d\geq 3$. In $1+1$ dimensions the
future and past null shells each split into two orbits (right- and left-moving
rays), and every spacelike shell splits into two orbits distinguished by the
sign of the spatial component. The massive future and past shells remain
single orbits.

\begin{proposition}[\cite{ibort_notion_2025}]
\label{prop:wigner-orbits}
Let $\dim(\mathcal{M}) \geq 3$, then the orbits of $\mathrm{Wig}(\mathcal{M},g)$ in $T^*\mathcal{M}$ are the mass shells determined by the Lorentzian norm of the covector and, if a time-orientation is fixed, by its time orientation. More precisely, the standard orbits are:
\begin{align*}
\mathcal O_{m,+} &= \{(x,p)\in T^*\mathcal{M} \mid g_x^{-1}(p,p)=m^2>0,\; p\text{ future-directed}\},\\
\mathcal O_{0,+} &= \{(x,p)\in T^*\mathcal{M} \mid g_x^{-1}(p,p)=0,\; p\text{ future-directed}\},\\
\mathcal O_{0} &= \{(x,0)\in T^*\mathcal{M}\},\\
\mathcal O_{\mathrm{tach},m} &= \{(x,p)\in T^*\mathcal{M} \mid g_x^{-1}(p,p)=-m^2<0\}.
\end{align*}
There are analogous past-directed components for the massive and massless sectors.
\end{proposition}

\begin{proof}
An arrow $(p_y,T_{yx},p_x)$ exists precisely when $p_x=T_{yx}^*p_y$. Since $T_{yx}$ is a Lorentz isometry, it preserves the value of $g^{-1}(p,p)$ and, in the oriented/time-oriented setting, also preserves the future/past decomposition. Conversely, suppose $(x,p_x)$ and $(y,p_y)$ have the same Lorentzian norm and the same time orientation. By Proposition \ref{prop:poincare-transitive}, there exists at least one Lorentz isometry $T_0:T_x\mathcal{M}\to T_y\mathcal{M}$. The covectors $T_0^*p_y$ and $p_x$ then lie in the same mass shell of $(T_x\mathcal{M},g_x)$. Since $SO_+^\uparrow(1,d-1)$ acts transitively on each mass shell, there exists $\Lambda\in SO_+^\uparrow(1,d-1)\cong\mathrm{Poin}(\mathcal{M},g)_x^x$ with $\Lambda^*T_0^*p_y=p_x$. Setting $T_{yx}:=T_0\circ\Lambda$ produces the required arrow.
\end{proof}

\begin{remark}
\label{rem:wigner-particle-sectors}
The orbit decomposition of $\mathrm{Wig}(\mathcal{M},g)$ is the groupoid counterpart of the familiar classification of one-particle kinematics by mass shell \cite{wigner_unitary_1939}.  It is therefore natural to expect that irreducible or projective representations of appropriate orbit restrictions of $\mathrm{Wig}(\mathcal{M},g)$ will capture the local particle sectors of the theory (see \cite{ibort_notion_2025}). We will not need the full representation-theoretic classification in the present paper, but this picture helps motivate the role of the Wigner groupoid in the background geometry of the theory.
\end{remark}

\begin{example}[Flat spacetime]
\label{ex:wigner-flat}
In Minkowski spacetime, the trivialization of $T^*\mathbb M$ identifies $\mathrm{Wig}(\mathbb M,\eta)$ with the action groupoid of the proper orthochronous Poincar\'e group acting on phase space by the cotangent-lifted affine action:
\begin{equation}
(a,\Lambda)\cdot (x,p)=\bigl(\Lambda x+a,(\Lambda^{-1})^{*}p\bigr).
\end{equation}
Thus, in the flat case, the Wigner groupoid reduces to the standard global phase-space kinematics.
\end{example}

\subsection{Hilbert bundles and unitary representations of Lie groupoids}
\label{subsec:groupoid-representations}

To quantize a groupoid kinematics one must replace the usual notion of a unitary group representation on a single Hilbert space by a representation on a field of Hilbert spaces over the object manifold.

\begin{definition}
\label{def:measurable-hilbert-bundle}
A \emph{measurable field of Hilbert spaces} over a standard Borel space $X$ is a family 
\begin{equation}
\scrhi=\bigsqcup_{x\in X}\mathcal H_x \xrightarrow{\;\pi\;} X
\end{equation}
of separable Hilbert spaces together with a distinguished set of measurable sections, $\psi \colon x \mapsto \psi(x) \in \mathcal{H}_x$, satisfying the usual axioms of direct-integral theory (see, for instance,  \cite{dixmier_champs_1963,bos_continuous_2007}). The associated direct-integral Hilbert space is denoted
\begin{equation}
\hi = \int_X^\oplus \mathcal H_x\,d\mu(x),
\end{equation}
for any chosen $\sigma$-finite measure $\mu$ on $X$.
\end{definition}

When continuous or smooth structure is relevant, one may work with continuous or smooth Hilbert bundles, respectively. In the following, we shall indeed work with continuous groupoids, although most examples of interest in this paper will be with respect to Lie groupoids. 

\begin{definition}
\label{def:unitary-groupoid-representation}
Let $M$ be a manifold, $\Gamma\rightrightarrows M$ be a measurable groupoid and let $\scrhi =\bigsqcup_{x\in M}\mathcal H_x\to M$ be a measurable field of Hilbert spaces over $M$. A \emph{unitary representation} of $\Gamma$ on $\scrhi$ is a measurable assignment
\begin{equation}
\alpha:x\to y \in \Gamma
\quad\longmapsto\quad
U_\alpha:\mathcal H_x\to \mathcal H_y
\end{equation}
such that:
\begin{enumerate}
    \item each $U_\alpha$ is unitary;
    \item for every composable pair $(\alpha,\beta)\in \Gamma^{(2)}$,
    \begin{equation}
    U_{\alpha\circ\beta}=U_\alpha U_\beta;
    \end{equation}
    \item for every $x\in M$,
    \begin{equation}
    U_{1_x}=\id_{\mathcal H_x}.
    \end{equation}
\end{enumerate}
\end{definition}

This is simply a functor $U \colon \Gamma \to \mathsf{Hilbert}$, from the groupoid $\Gamma$ (viewed as a small category) to the category $\mathsf{Hilbert}$ of Hilbert spaces and bounded operators, together with the appropriate measurability condition (see \cite{ibort_notion_2025} for details).

\begin{definition}
\label{def:projective-groupoid-representation}
Let $\mathbb P\scrhi=\bigsqcup_{x\in M}\mathbb P(\mathcal H_x)$ be the associated field of projective Hilbert spaces. A \emph{projective representation} of $\Gamma$ on $\scrhi$ is a measurable assignment
\begin{equation}
\alpha:x\to y
\quad\longmapsto\quad
\varphi_\alpha:\mathbb P(\mathcal H_x)\to \mathbb P(\mathcal H_y)
\end{equation}
such that $\varphi_\alpha$ preserves probability amplitudes and
\begin{equation}
\varphi_{\alpha\circ\beta}=\varphi_\alpha\circ \varphi_\beta,
\qquad
\varphi_{1_x}=\id_{\mathbb P(\mathcal H_x)}.
\end{equation}
\end{definition}

\begin{proposition}
\label{prop:unitary-induces-projective}
Every unitary representation $U$ of $\Gamma$ on $\scrhi$ induces a projective representation $\mathbb P U$ on $\mathbb P\scrhi$ by
\begin{equation}
(\mathbb P U)_\alpha([\psi])=[U_\alpha\psi].
\end{equation}
\end{proposition}

\begin{proof}
This is immediate from the functoriality and unitarity of $U$.
\end{proof}

For the purposes of this paper, honest unitary representations will be the primary object. Projective representations are nevertheless important for physical reasons: they arise naturally from Wigner symmetry principles and, in many cases, can be encoded by central $U(1)$-extensions or multiplier cocycles of the underlying groupoid.

\begin{remark}
\label{rem:projective-vs-unitary}
The representation theory of the Poincar\'e and Wigner groupoids, including projective aspects and orbitwise classification results, provides the conceptual background for interpreting the fibers $\mathcal H_x$ as local quantum state spaces and the arrows $U_\alpha$ as coherent transport maps between them. In the present paper, however, we only require the functorial unitary framework of Definition \ref{def:unitary-groupoid-representation}; the more refined classification theory will be invoked only heuristically.
\end{remark}

\subsection{Continuous groupoids as measure groupoids}
\label{subsec:continuous-groupoid-preliminaries}

Let $\Gamma\rightrightarrows M$ be a second countable Hausdorff continuous groupoid with source and target maps
\begin{equation}
s,t:\Gamma\to M.
\end{equation}
We assume throughout that $\Gamma$ is endowed with a continuous left Haar system\footnote{This is an extra assumption for general continuous groupoids, though for Lie groupoids a smooth Haar system exists if $\Gamma$ and $M$ are second countable and Hausdorff \cite{renault_locally_1980}.}
\begin{equation}
\nu=\{\nu^x\}_{x\in M},
\end{equation}
that is, for each $x\in M$, $\nu^x : \Bor(\Gamma^x) \to [0,\infty]$ is a positive Radon measure on the target fiber, with support
\begin{equation}
    \supp \nu^x =\Gamma^x,
\end{equation}
depending continuously on $x$, and satisfying the left-invariance property
\begin{equation}
\int_{\Gamma^{s(\alpha)}} f(\alpha\circ \beta)\, d\nu^{s(\alpha)}(\beta)
=
\int_{\Gamma^{t(\alpha)}} f(\gamma)\, d\nu^{t(\alpha)}(\gamma),
\qquad
\alpha\in \Gamma,
\end{equation}
for every $f \in C_c(\Gamma)$ (also written as $\alpha \circ \nu^{s(\alpha)} = \nu^{t(\alpha)}$). This is shown pictorially in Fig.~\ref{fig:groupoids}.

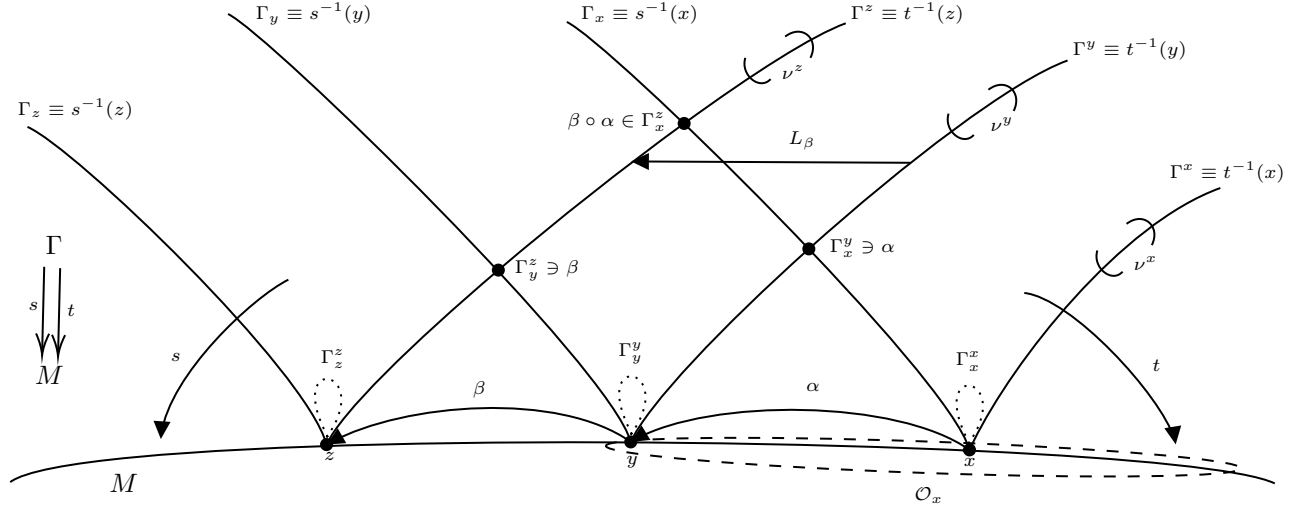
\begin{figure*}
    \centering
    \begin{tikzpicture}[x=0.75pt,y=0.75pt,yscale=-1,xscale=1]

\draw    (32.27,258.53) .. controls (60.46,237.39) and (339.36,234.46) .. (519.06,242.98) .. controls (594.31,246.54) and (652.17,252.12) .. (666.93,259.2) ;
\draw  [fill={rgb, 255:red, 0; green, 0; blue, 0 }  ,fill opacity=1 ] (188.27,239.7) .. controls (188.27,238.17) and (189.51,236.93) .. (191.03,236.93) .. controls (192.56,236.93) and (193.8,238.17) .. (193.8,239.7) .. controls (193.8,241.23) and (192.56,242.47) .. (191.03,242.47) .. controls (189.51,242.47) and (188.27,241.23) .. (188.27,239.7) -- cycle ;
\draw  [fill={rgb, 255:red, 0; green, 0; blue, 0 }  ,fill opacity=1 ] (340.93,238.37) .. controls (340.93,236.84) and (342.17,235.6) .. (343.7,235.6) .. controls (345.23,235.6) and (346.47,236.84) .. (346.47,238.37) .. controls (346.47,239.89) and (345.23,241.13) .. (343.7,241.13) .. controls (342.17,241.13) and (340.93,239.89) .. (340.93,238.37) -- cycle ;
\draw  [fill={rgb, 255:red, 0; green, 0; blue, 0 }  ,fill opacity=1 ] (510.93,242.37) .. controls (510.93,240.84) and (512.17,239.6) .. (513.7,239.6) .. controls (515.23,239.6) and (516.47,240.84) .. (516.47,242.37) .. controls (516.47,243.89) and (515.23,245.13) .. (513.7,245.13) .. controls (512.17,245.13) and (510.93,243.89) .. (510.93,242.37) -- cycle ;
\draw    (49.33,150.47) -- (48.31,192.93) ;
\draw [shift={(48.27,194.93)}, rotate = 271.37] [color={rgb, 255:red, 0; green, 0; blue, 0 }  ][line width=0.75]    (10.93,-3.29) .. controls (6.95,-1.4) and (3.31,-0.3) .. (0,0) .. controls (3.31,0.3) and (6.95,1.4) .. (10.93,3.29)   ;
\draw    (57.33,151.13) -- (56.31,193.6) ;
\draw [shift={(56.27,195.6)}, rotate = 271.37] [color={rgb, 255:red, 0; green, 0; blue, 0 }  ][line width=0.75]    (10.93,-3.29) .. controls (6.95,-1.4) and (3.31,-0.3) .. (0,0) .. controls (3.31,0.3) and (6.95,1.4) .. (10.93,3.29)   ;
\draw    (40.93,80.27) .. controls (64.27,90.27) and (180.93,202.27) .. (191.03,239.7) ;
\draw    (141.6,23.6) .. controls (162.27,33.6) and (333.6,200.93) .. (343.7,238.37) ;
\draw    (311.6,27.6) .. controls (332.93,37.6) and (503.6,204.93) .. (513.7,242.37) ;
\draw    (191.03,239.7) .. controls (198.29,223.71) and (233.95,189) .. (277.24,152.14) .. controls (306.96,126.83) and (340.28,100.51) .. (370.5,78.53) .. controls (404.83,53.55) and (435.15,34.17) .. (451.6,28.27) ;
\draw    (343.7,238.37) .. controls (352.16,219.72) and (390.02,180.08) .. (433.07,141.49) .. controls (481.07,98.47) and (535.52,56.77) .. (562.93,46.93) ;
\draw    (513.7,242.37) .. controls (514.36,240.91) and (515.08,239.4) .. (515.84,237.86) .. controls (535.74,197.58) and (589.52,128.91) .. (639.6,110.93) ;
\draw    (343.7,238.37) .. controls (307.67,212.79) and (230.37,218.68) .. (193.25,238.47) ;
\draw [shift={(191.03,239.7)}, rotate = 330.14] [fill={rgb, 255:red, 0; green, 0; blue, 0 }  ][line width=0.08]  [draw opacity=0] (8.93,-4.29) -- (0,0) -- (8.93,4.29) -- cycle    ;
\draw    (513.7,242.37) .. controls (477.67,216.79) and (383.73,216.28) .. (345.94,237.07) ;
\draw [shift={(343.7,238.37)}, rotate = 328.57] [fill={rgb, 255:red, 0; green, 0; blue, 0 }  ][line width=0.08]  [draw opacity=0] (8.93,-4.29) -- (0,0) -- (8.93,4.29) -- cycle    ;
\draw  [fill={rgb, 255:red, 0; green, 0; blue, 0 }  ,fill opacity=1 ] (274.48,152.14) .. controls (274.48,150.61) and (275.72,149.37) .. (277.24,149.37) .. controls (278.77,149.37) and (280.01,150.61) .. (280.01,152.14) .. controls (280.01,153.66) and (278.77,154.9) .. (277.24,154.9) .. controls (275.72,154.9) and (274.48,153.66) .. (274.48,152.14) -- cycle ;
\draw  [fill={rgb, 255:red, 0; green, 0; blue, 0 }  ,fill opacity=1 ] (430.31,141.49) .. controls (430.31,139.96) and (431.54,138.72) .. (433.07,138.72) .. controls (434.6,138.72) and (435.84,139.96) .. (435.84,141.49) .. controls (435.84,143.02) and (434.6,144.26) .. (433.07,144.26) .. controls (431.54,144.26) and (430.31,143.02) .. (430.31,141.49) -- cycle ;
\draw  [fill={rgb, 255:red, 0; green, 0; blue, 0 }  ,fill opacity=1 ] (367.73,78.53) .. controls (367.73,77) and (368.97,75.76) .. (370.5,75.76) .. controls (372.02,75.76) and (373.26,77) .. (373.26,78.53) .. controls (373.26,80.06) and (372.02,81.3) .. (370.5,81.3) .. controls (368.97,81.3) and (367.73,80.06) .. (367.73,78.53) -- cycle ;
\draw    (483.87,98.27) -- (346.87,97.61) ;
\draw [shift={(343.87,97.6)}, rotate = 0.27] [fill={rgb, 255:red, 0; green, 0; blue, 0 }  ][line width=0.08]  [draw opacity=0] (8.93,-4.29) -- (0,0) -- (8.93,4.29) -- cycle    ;
\draw  [dash pattern={on 0.84pt off 2.51pt}]  (188.27,239.7) .. controls (219.87,200.93) and (165.2,190.93) .. (193.8,239.7) ;
\draw  [dash pattern={on 0.84pt off 2.51pt}]  (340.93,238.37) .. controls (372.53,199.6) and (317.87,189.6) .. (346.47,238.37) ;
\draw    (171.87,156.93) .. controls (153.76,165.99) and (112.44,204.52) .. (108.16,234.21) ;
\draw [shift={(107.87,236.93)}, rotate = 273.81] [fill={rgb, 255:red, 0; green, 0; blue, 0 }  ][line width=0.08]  [draw opacity=0] (8.93,-4.29) -- (0,0) -- (8.93,4.29) -- cycle    ;
\draw    (541.2,163.6) .. controls (564.76,166.15) and (609.59,214.9) .. (616.48,236.71) ;
\draw [shift={(617.2,239.6)}, rotate = 260.54] [fill={rgb, 255:red, 0; green, 0; blue, 0 }  ][line width=0.08]  [draw opacity=0] (8.93,-4.29) -- (0,0) -- (8.93,4.29) -- cycle    ;
\draw  [dash pattern={on 4.5pt off 4.5pt}] (331.23,240.35) .. controls (331.39,236) and (402.87,234.99) .. (490.9,238.11) .. controls (578.93,241.23) and (650.16,247.29) .. (650.01,251.65) .. controls (649.85,256.01) and (578.37,257.02) .. (490.34,253.9) .. controls (402.31,250.78) and (331.08,244.71) .. (331.23,240.35) -- cycle ;
\draw  [dash pattern={on 0.84pt off 2.51pt}]  (510.27,243.03) .. controls (541.87,204.27) and (487.2,194.27) .. (515.8,243.03) ;
\draw    (503.2,73.6) .. controls (500.53,82.27) and (510.53,90.27) .. (515.87,84.27) ;
\draw    (577.87,142.27) .. controls (575.2,150.93) and (585.2,158.93) .. (590.53,152.93) ;
\draw    (401.2,45.6) .. controls (398.53,54.27) and (408.53,62.27) .. (413.87,56.27) ;
\draw    (523.2,61.6) .. controls (531.87,54.93) and (540.53,62.93) .. (535.87,72.27) ;
\draw    (592.53,128.93) .. controls (601.2,122.27) and (609.87,130.27) .. (605.2,139.6) ;
\draw    (421.2,34.93) .. controls (429.87,28.27) and (438.53,36.27) .. (433.87,45.6) ;

\draw (43.33,198.87) node [anchor=north west][inner sep=0.75pt]    {$M$};
\draw (48.67,133.53) node [anchor=north west][inner sep=0.75pt]    {$\Gamma $};
\draw (40,166.53) node [anchor=north west][inner sep=0.75pt]  [font=\scriptsize]  {$s$};
\draw (59.33,167.87) node [anchor=north west][inner sep=0.75pt]  [font=\scriptsize]  {$t$};
\draw (509.33,245.07) node [anchor=north west][inner sep=0.75pt]  [font=\scriptsize]  {$x$};
\draw (340,243.73) node [anchor=north west][inner sep=0.75pt]  [font=\scriptsize]  {$y$};
\draw (188.27,242.1) node [anchor=north west][inner sep=0.75pt]  [font=\scriptsize]  {$z$};
\draw (262.93,206.1) node [anchor=north west][inner sep=0.75pt]  [font=\scriptsize]  {$\beta $};
\draw (430.27,207.43) node [anchor=north west][inner sep=0.75pt]  [font=\scriptsize]  {$\alpha $};
\draw (284.1,145.1) node [anchor=north west][inner sep=0.75pt]  [font=\scriptsize]  {$\Gamma _{y}^{z} \ni \beta $};
\draw (310.43,71.3) node [anchor=north west][inner sep=0.75pt]  [font=\scriptsize]  {$\beta \circ \alpha \in \Gamma _{x}^{z}$};
\draw (421.43,81.1) node [anchor=north west][inner sep=0.75pt]  [font=\scriptsize]  {$L_{\beta }$};
\draw (80.67,253.07) node [anchor=north west][inner sep=0.75pt]    {$M$};
\draw (34.67,64.2) node [anchor=north west][inner sep=0.75pt]  [font=\scriptsize]  {$\Gamma _{z} \equiv s^{-1}( z)$};
\draw (154,16.2) node [anchor=north west][inner sep=0.75pt]  [font=\scriptsize]  {$\Gamma _{y} \equiv s^{-1}( y)$};
\draw (317.33,16.07) node [anchor=north west][inner sep=0.75pt]  [font=\scriptsize]  {$\Gamma _{x} \equiv s^{-1}( x)$};
\draw (452.67,16.07) node [anchor=north west][inner sep=0.75pt]  [font=\scriptsize]  {$\Gamma ^{z} \equiv t^{-1}( z)$};
\draw (564,34.07) node [anchor=north west][inner sep=0.75pt]  [font=\scriptsize]  {$\Gamma ^{y} \equiv t^{-1}( y)$};
\draw (612.67,96.07) node [anchor=north west][inner sep=0.75pt]  [font=\scriptsize]  {$\Gamma ^{x} \equiv t^{-1}( x)$};
\draw (186.67,190.2) node [anchor=north west][inner sep=0.75pt]  [font=\scriptsize]  {$\Gamma _{z}^{z}$};
\draw (336,186.87) node [anchor=north west][inner sep=0.75pt]  [font=\scriptsize]  {$\Gamma _{y}^{y}$};
\draw (111.33,192.07) node [anchor=north west][inner sep=0.75pt]  [font=\scriptsize]  {$s$};
\draw (604,195.4) node [anchor=north west][inner sep=0.75pt]  [font=\scriptsize]  {$t$};
\draw (484.61,259.83) node [anchor=north west][inner sep=0.75pt]  [font=\scriptsize]  {$\mathcal{O}_{x}$};
\draw (505.33,191.53) node [anchor=north west][inner sep=0.75pt]  [font=\scriptsize]  {$\Gamma _{x}^{x}$};
\draw (442.1,134.43) node [anchor=north west][inner sep=0.75pt]  [font=\scriptsize]  {$\Gamma _{x}^{y} \ni \alpha $};
\draw (594.67,143.73) node [anchor=north west][inner sep=0.75pt]  [font=\scriptsize]  {$\nu ^{x}$};
\draw (522.67,73.73) node [anchor=north west][inner sep=0.75pt]  [font=\scriptsize]  {$\nu ^{y}$};
\draw (418,47.73) node [anchor=north west][inner sep=0.75pt]  [font=\scriptsize]  {$\nu ^{z}$};

\end{tikzpicture}
    \caption{The groupoid $\Gamma \rightrightarrows M$, its arrows and associated spaces. When endowed with a continuous left Haar system $\nu=\{\nu^x\}_{x \in M}$, we have that $\nu^x : \Bor(\Gamma^x) \to [0,\infty]$ takes Borel subsets of the target fiber $\Gamma^x$.}
\label{fig:groupoids}
\end{figure*}

Let $\mu$ be a $\sigma$-finite quasi-invariant Radon measure on the base $M$. Then $\nu$ and $\mu$ determine a measure $\lambda$ on $\Gamma$ by
\begin{equation}
\int_\Gamma f(\alpha)\, d\lambda(\alpha)
:=
\int_M \left( \int_{\Gamma^x} f(\alpha)\, d\nu^x(\alpha)\right) d\mu(x),
\qquad
f\in C_c(\Gamma).
\end{equation}
Quasi-invariance means that $\lambda$ is mutually absolutely continuous with its push-forward under inversion. Equivalently, there exists a measurable modular function
\begin{equation}
\delta:\Gamma\to \mathbb R_{>0}
\end{equation}
such that
\begin{equation}
\frac{d\lambda^{-1}}{d\lambda}(\alpha)=\delta(\alpha)^{-1},
\end{equation}
where $\lambda^{-1}$ denotes the image measure of $\lambda$ under $\alpha\mapsto \alpha^{-1}$, $\iota_* \lambda = \lambda^{-1}$.

\begin{definition}
\label{def:measure-groupoid-section4}
The data $(\Gamma,\nu,\delta,\mu,\lambda)$
will be called the \emph{measure groupoid} associated with $(\Gamma\rightrightarrows M,\nu,\mu)$.
\end{definition}

Associated with the measure groupoid is the Hilbert space $L^2(\Gamma,\lambda)$
and the corresponding groupoid von Neumann algebra $V_\lambda(\Gamma)$ obtained from the left-regular convolution representation (see \cite{connes_noncommutative_1994,landsman_mathematical_1998,ciaglia_groupoidal_2024} for the general construction). We shall not need the full construction here, but only the fact that it exists canonically once the measure-groupoid data are fixed.

\begin{remark}
\label{rem:measure-groupoid-operator}
The measure-groupoid formalism provides a natural candidate for the ambient algebra of observables associated with the kinematics of $\Gamma$. In the present paper we work instead with operator fields on direct-integral Hilbert spaces, since this is the most transparent setting in which to formulate the operational quantum reference frame formalism. The two viewpoints are expected to be closely related; in the case of action groupoids the groupoid von Neumann algebra reduces to the usual crossed product \cite{renault_c-algebra_1980}.
\end{remark}

\begin{definition}
    \label{def:nu source compatible mu}
    Let $(\Gamma \rightrightarrows M,\nu,\mu)$ be a second countable Hausdorff continuous groupoid endowed with a continuous left Haar system. Then $\nu$ is said to be \emph{source-compatible with $\mu$} if for $\mu$-almost every $x \in M$,
    \begin{equation}
        s_* \nu^x \ll \mu
    \end{equation}
    where $s_* \nu^x(N) := \nu^x(s^{-1}(N) \cap \Gamma^x)$.
\end{definition}

\begin{proposition}
    \label{prop:transitive groupoid implies nu compatible}
    Let $(\Gamma \rightrightarrows M,\nu,\mu)$ be a transitive second countable Hausdorff continuous groupoid endowed with a continuous left Haar system $\nu$ and suppose $\mu$ is a nonzero $\sigma$-finite quasi-invariant Radon measure. Then $\nu$ is source-compatible with $\mu$.
\end{proposition}

\begin{proof}
    Let $x,y\in M$. By transitivity, there exists $\alpha:x\to y$. Left invariance of the Haar system gives
    \begin{equation}
        (L_\alpha)_*\nu^x=\nu^y,
    \end{equation}
    where $L_\alpha : \Gamma^{s(\alpha)} \to \Gamma^{t(\alpha)}$, $L_\alpha \beta = \alpha \circ \beta$, while $s\circ L_\alpha=s$ on $\Gamma^x$. Consequently,
\begin{equation}
    s_*\nu^y = s_*(L_\alpha)_*\nu^x = (s\circ L_\alpha)_*\nu^x = s_*\nu^x.
\end{equation}
    Thus there exists a measure $\kappa$ on $M$ such that $s_*\nu^x=\kappa$ for every $x\in M$. Let $N\in\Bor(M)$ satisfy $\mu(N)=0$. Then
    \begin{equation}
        \lambda(t^{-1}(N))= \int_M \nu^x(t^{-1}(N)\cap\Gamma^x)\,d\mu(x)=  \int_N\nu^x(\Gamma^x)\,d\mu(x)=0.
    \end{equation}
    But
    \begin{equation}
        \lambda^{-1}\bigl(s^{-1}(N)\bigr) =\lambda\bigl(t^{-1}(N)\bigr)=0.
    \end{equation}
    Mutual absolute continuity, and in particular $\lambda \ll \lambda^{-1}$ applied to $s^{-1}(N)$, therefore implies
    \begin{equation}
        0=\lambda\bigl(s^{-1}(N)\bigr)=\int_M s_*\nu^x(N)\,d\mu(x)=\int_M\kappa(N)\,d\mu(x).
    \end{equation}
    Since $\kappa(N)$ is a nonnegative constant and $\mu$ is nonzero, we obtain $\kappa(N)=0$. Thus $\kappa\ll\mu$.
\end{proof}

\subsection{Systems over a continuous groupoid}
\label{subsec:systems-over-groupoid}

Let
\begin{equation}
\pi_{\S} : \scrhis =\bigsqcup_{x\in M} \hisx \to M
\end{equation}
be a measurable field of separable Hilbert spaces over $M$, and let  $V:\Gamma \curvearrowright \scrhis$ be a unitary representation of $\Gamma$ on $\scrhis$, i.e. a measurable assignment
\begin{equation}
\alpha:x\to y \longmapsto V_\alpha:\hisx\to \his^y \, ,
\end{equation}
such that
\begin{equation}
V_{\alpha\circ \beta}=V_\alpha V_\beta,
\qquad
V_{1_x}=\id_{\hisx}.
\end{equation}
The associated direct-integral Hilbert space is
\begin{equation}
\his := \int_M^\oplus \hisx\, d\mu(x).
\end{equation}

Likewise, if
\begin{equation}
\pi_{\R}: \scrhir =\bigsqcup_{x\in M} \hirx\to M
\end{equation}
carries a unitary representation
\begin{equation}
U:\Gamma \curvearrowright \scrhir,
\end{equation}
we write
\begin{equation}
\hir := \int_M^\oplus \hirx\, d\mu(x).
\end{equation}

The joint system-frame object is represented by the fiber-wise tensor product bundle
\begin{equation}
\scrhis \boxtimes \scrhir
:=
\bigsqcup_{x\in M} (\hisx \otimes \hirx)\to M,
\end{equation}
with induced representation
\begin{equation}
W_\alpha:=V_\alpha\otimes U_\alpha:
\his^{s(\alpha)}\otimes \hir^{s(\alpha)}
\to
\his^{t(\alpha)}\otimes \hir^{t(\alpha)}.
\end{equation}
Its direct-integral Hilbert bundle is 
\begin{equation}
    \scrhis\boxtimes \scrhir : \bigsqcup_{x \in M} \hisx \otimes \hirx \to M
\end{equation}
with associated direct-integral Hilbert space
\begin{equation}
\hisr
=
\int_M^\oplus (\hisx\otimes \hirx)\, d\mu(x).
\end{equation}
Note that in general $\hisr \neq \his \otimes \hir$.

\begin{definition}
    Let $\pi : \scrhi =\bigsqcup_{x\in M} \hix \to M$ be a measurable field of separable Hilbert spaces over $M$. Let
\begin{equation}
    \A(\scrhi) := \int^\oplus_M \bhx d\mu(x)= \left\{A = \int^\oplus_M A_x \, d\mu(x) \mid x \mapsto A_x \text{ measurable}, \norm{A}_{\A(\scrhi)} := \esssup_x \norm{A_x}_{\bhx} < \infty \right\}.
\end{equation} This is a von Neumann algebra \cite{takesaki_tensor_1979} with a canonical semifinite normal trace $\Tr_{\hi}[\cdot] = \int_M \Tr_{\hix}[\cdot] d\mu$. 
\end{definition}
Given $\scrhis \boxtimes \scrhir$, we can similarly define
\begin{equation}
\A(\scrhis \boxtimes \scrhir) := \int_M^\oplus \B(\hisx\otimes \hirx)\, d\mu(x).
\end{equation}

Note that $\A(\scrhis) \subset \mathcal{B}(\his)$, $\A(\scrhir) \subset \mathcal{B}(\hir)$ and $\A(\scrhis\boxtimes\scrhir) \subset \mathcal{B}(\hisr)$, the latter two being typically independent of $\bhsr$. Kinematically, we have that some operators are ``special" as soon as we assume a groupoid symmetry structure: not every possible operator lies in these algebras.

\begin{definition}
    We call $ \omega =(\omega_x \in \statex) \in \A(\scrhi)$ \emph{measurable operator fields of states}. We write
    \begin{equation}
        \state := \left\{\int^\oplus_M \omega_x d\mu(x) \mid \omega_x \in \statex \; \text{ for } \mu-a.e. \, x \in M\right\} \subset \A(\scrhi) 
    \end{equation}
    as the set of such measurable operator fields of states.
\end{definition}

\begin{remark}
    Note that $A \omega \in \A(\scrhi)$ for all $A \in \A(\scrhi)$ and $\omega \in \state$.
\end{remark}

\begin{definition}
    Let $\pi : \scrhi =\bigsqcup_{x\in M} \hix \to M$ be a measurable field of separable Hilbert spaces over $M$. Let \cite{takesaki_tensor_1979}
    \begin{equation}
    \A(\scrhi)_* := \left\{\chi = \int^\oplus_M \chi_x \, d\mu(x) \mid x \mapsto \chi_x \in \thx \text{ measurable}, \norm{\chi}_1 := \int_M \norm{\chi_x}_1 \, d\mu(x) < \infty\right\}
\end{equation}
with trace norm $\norm{\chi}_1 \equiv \Tr_{\hi}[\abs{\chi}] = \int_M \Tr[\abs{\chi_x}] \, d\mu(x)$ and $\abs{\chi} = \int^\oplus \abs{\chi_x} \, d\mu(x)$, and
\begin{equation}
    \inner{A}{\chi} := \Tr[A \chi] = \int_M \Tr[A_x \chi_x] d\mu(x).
\end{equation}
    $\A(\scrhi)_*$ is the predual of $\A(\scrhi)$ \cite{takesaki_tensor_1979}.
\end{definition}

\begin{definition}
    \emph{Normal states of $\A(\scrhi)$} can then be defined as positive decomposable $\rho = \int^\oplus_M \rho_x d\mu(x) \in \A(\scrhi)_*$ with $\Tr_\hi[\rho]=1$. Write
    \begin{equation}
        \normalstate := \left\{\eta \in \A(\scrhi)_* \mid \norm{\eta}_1 = 1, \eta \geq 0\right\} \subset \A(\scrhi)_*
    \end{equation}
    for the set of such normal states.
\end{definition}

\begin{remark}
    Note that $\state \neq \normalstate$ in general; in particular, $\normalstate \subset \A(\scrhi)_* \not\subset \bh$ in general while $\state \subset \A(\scrhi) \subset \bh$. Moreover, for $\rho = (\rho_x) \in \normalstate$, we have, for $\mu$-almost every $x \in M$, $\rho_x \geq 0$, but if  $\Tr_{\hix}[\rho_x] = 1$ for $\mu$-almost every $x \in M$ then $\mu(M) = 1$. If $\mu(M) = \infty$ however (as is the case for the volume measure on Minkowski for example), a constant field of fiber-states $\rho = (\rho_x)$ with each $\rho_x \in \statex$ is not in $\A(\scrhi)_*$.
\end{remark}

\begin{remark}
    We will interpret states in $\state$ as characterising \emph{preparation states} (of the frame) while states in $\normalstate$ as characterising \emph{evaluation states} (of the system) extracting probabilities from (relational local) observables.
\end{remark}

Continuous groupoids become interesting when we want to control the continuity of $x \mapsto A_x$ of operator fields.

\begin{definition}
Let $\pi:\scrhi\to M$ be a continuous field of separable Hilbert spaces.
A bounded operator field $A \in \A(\scrhi)$ is called \emph{ultraweakly continuous} if, for every trace-norm continuous section
\begin{equation}
\chi:M\to\bigsqcup_{x\in M}\thx, \qquad x\mapsto\chi_x,
\end{equation}
the scalar function $x\mapsto \Tr_{\hix}[\chi_x A_x]$ is continuous.
We write $\A(\scrhi)_{\mathrm{uw}}$ for the space of such fields.
\end{definition}
Again, given $\scrhis \boxtimes \scrhir$, we can similarly define $\A(\scrhis \boxtimes \scrhir)_\mathrm{uw}$ as the space of ultraweakly continuous operator fields in $\A(\scrhis \boxtimes \scrhir)$.

\begin{definition}
\label{def:groupoid-invariant-operator-field}
An operator field $B=(B_x)_{x\in M}\in \A(\scrhis\boxtimes \scrhir)$ (resp. $\A(\scrhis \boxtimes \scrhir)_\mathrm{uw}$) is called \emph{$\Gamma$-invariant} if
\begin{equation}
W_\alpha\, B_{s(\alpha)}\, W_\alpha^\dagger = B_{t(\alpha)}
\qquad
\text{for all }\alpha\in \Gamma.
\end{equation}
The set of all such invariant fields will be denoted by $\A(\scrhis \boxtimes \scrhir)^\Gamma$ (resp. $\A(\scrhis \boxtimes \scrhir)_{\mathrm{uw}}^\Gamma$).
\end{definition}

This is the groupoid analogue of the invariant algebra $\bhsr^G$ in the group-based theory.

\begin{definition}
\label{def:isotropy-invariant-fields}
An operator field $A=(A_x)_{x\in M}\in \A(\scrhi)$ (resp. $\A(\scrhi)_{\mathrm{uw}}$) is called \emph{isotropy-invariant} if
\begin{equation}
V_\gamma\, A_x\, V_\gamma^\dagger = A_x
\qquad
\text{for every }\gamma\in \Gamma_x^x,\ x\in M.
\end{equation}
The set of isotropy-invariant fields is denoted by $\A(\scrhi)^{\mathrm{iso}}$ (resp. $\A(\scrhi)_{\mathrm{uw}}^{\mathrm{iso}}$).
\end{definition}
This is the groupoid analogue of the invariant algebra $\bh^H$ for $H$ a subgroup of $G$ in the group-based theory of homogeneous spaces $G/H$.
The point of Definition \ref{def:isotropy-invariant-fields} is that if $A$ is isotropy-invariant and $\alpha,\beta:x\to y$, then
\begin{equation}
V_\alpha A_x V_\alpha^\dagger = V_\beta A_x V_\beta^\dagger.
\end{equation}
Indeed, $\gamma:=\beta^{-1}\circ\alpha\in\Gamma_x^x$, so $V_\gamma A_x V_\gamma^\dagger=A_x$ by isotropy-invariance, which gives $V_\alpha A_x V_\alpha^\dagger = V_\beta V_\gamma A_x V_\gamma^\dagger V_\beta^\dagger = V_\beta A_x V_\beta^\dagger$. Thus such an operator may be transported along arrows without ambiguity, just as an $H$-invariant operator on a homogeneous space $G/H$ may be transported by the group action.

\subsection{Bisections}

Bisections can be used to define actions over the whole direct integral algebra, rather than fiber-wise, as follows.

\begin{definition}
    Let $(\Gamma\rightrightarrows M,\nu,\mu)$ be a Lie groupoid, and let $U : \Gamma \curvearrowright \scrhi$ be a unitary representation of $\Gamma$ with $\mathcal{H} = \int_M^\oplus \mathcal{H}_x d\mu (x)$. For all $b \in \mathrm{Bis}(\Gamma)$ such that $(\varphi_b)_* \mu \sim \mu$\footnote{That is, $(\varphi_b)_* \mu$ and $\mu$ are mutually absolutely continuous: $(\varphi_b)_* \mu(N) = 0 \Leftrightarrow \mu(N)=0$ for all $N \in \Bor(M)$.} where $\varphi := t \circ b$, we define \cite{kaniuth_induced_2012}
    \begin{equation}
        (\mathbf{U}_{b} \Psi)_x := \left(\frac{d((\varphi_b)_*\mu)}{d\mu}(x) \right)^{1/2} U_{b(\varphi^{-1}_b(x))} \Psi(\varphi_b^{-1}(x)) \in \hix, \qquad \Psi \in \hi.
    \end{equation}
    Then $\mathbf{U}_b \in \U(\hi)$ and $b \mapsto \mathbf{U}_b$ is a unitary representation of $\mathrm{Bis}(\Gamma)$. In particular, $\mathbf{U}_b^\dagger = \mathbf{U}_{b^{-1}}$, i.e.
    \begin{equation}
        (\mathbf{U}^\dagger_{b} \Psi)_x = \left(\frac{d((\varphi^{-1}_b)_*\mu)}{d\mu}(x) \right)^{1/2} U^\dagger_{b(x)} \Psi(\varphi_b(x)) \in \hix, \qquad \Psi \in \hi
    \end{equation}
    where we used the fact that $\varphi_{b^{-1}}=\varphi_b^{-1}$ and $b^{-1}(\varphi_b(x)) = b(x)^{-1}$ and $U_{b(x)^{-1}}=U_{b(x)}^\dagger$.
    If $\mathbf{U}_b \in \U(\hir)$ and $\mathbf{V}_b \in \U(\his)$, we define $\mathbf{W}_b \in \U(\hisr)$ as
    \begin{align}
        (\mathbf{W}_b \Xi)_x &:= \left(\frac{d((\varphi_b)_*\mu)}{d\mu}(x) \right)^{1/2} (V_{b(\varphi^{-1}_b(x))} \otimes U_{b(\varphi^{-1}_b(x))}) \Xi(\varphi_b^{-1}(x)) \\ &= \left(\frac{d((\varphi_b)_*\mu)}{d\mu}(x) \right)^{1/2} W_{b(\varphi^{-1}_b(x))} \Xi(\varphi_b^{-1}(x)) \in \hisx\otimes\hirx
    \end{align}
    for all $\Xi \in \hisr$.
\end{definition}

\begin{definition}
    A Lie groupoid $(\Gamma \rightrightarrows M,\nu,\mu)$ is called \emph{bisection-admissible} if $(\varphi_b)_* \mu \sim \mu$ for all $b \in \mathrm{Bis}(\Gamma)$.
\end{definition}

\begin{definition}
    An operator field $A \in \A(\scrhi)$ (resp. $\A(\scrhi)_{\mathrm{uw}}$) is called \emph{bisection-invariant} if
    \begin{equation}
        \mathbf{V}_b A \mathbf{V}_b^\dagger = A \qquad \text{for every } b \in \mathrm{Bis}(\Gamma).
    \end{equation}
    The set of bisection-invariant fields is denoted by $\A(\scrhi)^{\mathrm{Bis}(\Gamma)}$. Likewise, the set of ultraweakly continuous bisection-invariant fields is denoted by $\A(\scrhi)^{\mathrm{Bis}(\Gamma)}_\mathrm{uw}$. 
\end{definition}

\begin{lemma}
    \label{lem:Gamma invariance implies bisection invariance}
    Let $(\Gamma\rightrightarrows M,\nu,\mu)$ be a bisection-admissible Lie groupoid, and let $V : \Gamma \curvearrowright \scrhi$ be a unitary representation of $\Gamma$. Then $\A(\scrhi)^\Gamma \subset \A(\scrhi)^{\mathrm{Bis}(\Gamma)}$, i.e. $A \in \A(\scrhi)^\Gamma \Rightarrow A \in \A(\scrhi)^{\mathrm{Bis}(\Gamma)}$.
\end{lemma}

\begin{proof}
    Fix $A=(A_x)\in\A(\scrhi)^\Gamma$, so $V_\alpha A_{s(\alpha)}V_\alpha^\dagger=A_{t(\alpha)}$
for all $\alpha$. Let $b\in\mathrm{Bis}(\Gamma)$. For $\mu$-almost every $x$ and $y:=\varphi_b^{-1}(x)$,
\begin{equation}
(\mathbf V_b A\mathbf V_b^\dagger)_x
=\left(\frac{d((\varphi_b)_*\mu)}{d\mu}(x)\right)^{1/2}\left(\frac{d((\varphi_b^{-1})_*\mu)}{d\mu}(\varphi_b^{-1}(x))\right)^{1/2} V_{b(y)}\,A_{y}\,V_{b(y)}^\dagger
=V_{b(y)}\,A_{y}\,V_{b(y)}^\dagger,
\end{equation}
since $\frac{d((\varphi_b)_*\mu)}{d\mu}(x)\cdot \frac{d((\varphi_b^{-1})_*\mu)}{d\mu}(\varphi_b^{-1}(x))=1$. The arrow $b(y):y\to\varphi_b(y)=x$ then gives, by $\Gamma$-invariance, $V_{b(y)}A_yV_{b(y)}^\dagger=A_{t(b(y))}=A_x$. Hence $\mathbf V_b A\mathbf V_b^\dagger=A$ so $A\in\A(\scrhi)^{\mathrm{Bis}(\Gamma)}$.
\end{proof}

\subsection{States and continuity}

In the context of infinite-dimensional Hilbert spaces, topologies and different notions of continuity become important to distinguish. Just as strong, weak and ultraweak continuity of (unitary) representations can be defined for group representations, similar notions extend to groupoids.

In \cite{bos_continuous_2007}, weak and strong continuity of representations of continuous groupoids were analysed. However, although commonly discussed, the weak and strong topologies are not operational: they refer to convergence in inner products and in vectors, respectively. Instead, one should look at notions of continuity which have to do with convergence in probabilities, i.e. the ultraweak topology. To the best of our knowledge, this has not been discussed in the literature in the context of continuous groupoids. Here, we provide a sufficient exposition for the ensuing discussion, but encourage further work on the topic.

The ultraweak topology is arguably the relevant one in operational (quantum) physics. A sequence of operators $(A_n)$ in $\bh$ converge ultraweakly to $A \in \bh$ if $\Tr[\rho A_n] \to \Tr[\rho A]$ for all density operators $\rho \in \normalstate$. That is, ultraweak convergence is the convergence of expectation values.

Ultraweak continuity has to do with probabilities when evaluated (through the trace) on elements of the predual of the space of operators, i.e. typically $\thi$ as the predual of $\bh$. However, as we have seen, there are several spaces of preduals to consider here.

\begin{remark}
    There are at least three possible notions of ultraweak continuity that one may wish to define: a global one given by the topology $\sigma(\A(\scrhi),\A(\scrhi)_*)$, a fiber-wise one given by the topologies $\sigma(\mathcal{B}(\hix,\hi^y),\mathcal{T}(\hi^y,\hi^x))$ (i.e. linking different fibers together), and a fiber-wise one given by the topologies $\sigma(\bhx,\thx)$ (i.e. at the isotropy group level). 
\end{remark}

\begin{definition}
     Let $\pi : \scrhi =\bigsqcup_{x\in M} \hix \to M$ be a continuous field of separable Hilbert spaces over $M$ and let $(\Gamma,\nu,\delta,\mu,\lambda)$ be the continuous measure groupoid associated with $(\Gamma \rightrightarrows M,\nu,\mu)$. Let $U : \Gamma \curvearrowright \scrhi$ be a unitary representation on $\scrhi$. $U$ is said to be 
     \begin{enumerate}
         \item \emph{fiber-wise ultraweakly continuous} if for every normal predual test section $\chi$,\footnote{The predual of $\mathcal{B}(\hi^{s(\alpha)},\hi^{t(\alpha)})$ is $\mathcal{T}(\hi^{t(\alpha)},\hi^{s(\alpha)})$. A \emph{normal predual test section} is a field $\chi:\Gamma\to\bigsqcup_{\alpha\in\Gamma}
\mathcal T(\hi^{t(\alpha)},\hi^{s(\alpha)})$, with $\chi_\alpha\in\mathcal T(\hi^{t(\alpha)},\hi^{s(\alpha)})$, which is locally the trace-norm uniform limit of finite-rank fields of the form $\chi^{(n)}_\alpha = \sum_{j=1}^{N_n} \ket{\xi^{(j,n)}_{s(\alpha)}}\bra{\eta^{(j,n)}_{t(\alpha)}}$.}$\alpha \mapsto \Tr_{\hi^{t(\alpha)}}[U_\alpha \chi_{\alpha}]$ is continuous,\footnote{This reduces to ultraweak continuity for the group case ($M=\{*\}$). In such cases, $U : G \to U(\hi)$ is ultraweakly continuous iff for every trace-norm continuous test map $\chi : G \to \thi$, the map $g \mapsto \Tr[U_g \chi_g]$ is continuous. Since $\chi_\cdot$ is trace-norm continuous, this is easily shown to be equivalent to the continuity of the map $g \mapsto \Tr[U_g \chi]$ for a fixed $\chi \in \thi$.}
        \item \emph{isotropy ultraweakly continuous} if for every $x \in M$ and every $\chi \in \thx$, the map $\Gamma^x_x \ni \gamma  \mapsto \Tr_{\hix}[U_\gamma \chi]$ is continuous,
         \item \emph{weakly continuous} if for all continuous sections $\xi,\eta : M \to \scrhi$, $\alpha \mapsto \braket{U_\alpha \xi_{s(\alpha)}}{\eta_{t(\alpha)}}$ is continuous,
         \item \emph{strongly continuous} if for all continuous sections $\xi : M \to \scrhi$, $\alpha \mapsto U_\alpha \ket{\xi_{s(\alpha)}}$ is continuous.
     \end{enumerate}
\end{definition}

\begin{remark}
    Note that the arrow giving the unitary $U_\alpha\in\mathcal B(\hi^{s(\alpha)},\hi^{t(\alpha)})$ pairs with the trace class $\mathcal T(\hi^{t(\alpha)},\hi^{s(\alpha)})$, which reduces to the fiber predual $\thx$ only on the isotropy diagonal. It is never an element of the global algebra $\A(\scrhi)=\int^\oplus_M\bhx\,d\mu$
    for $\alpha$ off the isotropy diagonal. 
\end{remark}

\begin{theorem}
\label{thm:equivalences continuity}
    Let $\pi : \scrhi =\bigsqcup_{x\in M} \hix \to M$ be a continuous field of separable Hilbert spaces over $M$ and let $(\Gamma,\nu,\delta,\mu,\lambda)$ be the measure groupoid associated with $(\Gamma \rightrightarrows M,\nu,\mu)$. Let $U : \Gamma \curvearrowright \scrhi$ be a unitary representation on $\scrhi$. Then 
    \begin{enumerate}
        \item $U$ is strongly continuous iff $U$ is weakly continuous iff $U$ is fiber-wise ultraweakly continuous,
        \item if $U$ is fiber-wise ultraweakly continuous then it is isotropy ultraweakly continuous.
    \end{enumerate}
\end{theorem}

\begin{proof}
    \begin{enumerate}
        \item It is known \cite[Lemma 2.6]{bos_continuous_2007} that weak and strong continuity of unitary representations are equivalent. Thus, it is sufficient to show their equivalence with fiber-wise ultraweak continuity of unitary representations.

    Assume first that $U$ is fiber-wise ultraweakly continuous. For continuous sections $\xi,\eta$ of $\scrhi$, define the rank-one normal predual test section
\begin{equation}
\chi^{\xi,\eta}_\alpha
:=
\ket{\xi_{s(\alpha)}}\bra{\eta_{t(\alpha)}}
\in
\mathcal T(\hi^{t(\alpha)},\hi^{s(\alpha)}).
\end{equation}
Then
\begin{equation}
\Tr_{\hi^{t(\alpha)}}\!\bigl(U_\alpha \chi^{\xi,\eta}_\alpha\bigr)
=
\braket{\eta_{t(\alpha)}}{U_\alpha\xi_{s(\alpha)}}.
\end{equation}
Thus fiber-wise ultraweak continuity implies weak continuity.

Conversely assume that $U$ is weakly continuous. Let $\chi$ be a normal predual
test section and fix $\alpha_0\in\Gamma$. On some neighbourhood $K$ of $\alpha_0$, there are finite-rank continuous fields
\begin{equation}
\chi^{(n)}_\alpha
=
\sum_{j=1}^{N_n}
\ket{\xi_{j,n}(s(\alpha))}\bra{\eta_{j,n}(t(\alpha))},
\end{equation}
where each $\xi^{(j,n)},\eta^{(j,n)}$ is a continuous section of $\scrhi$, such that $\sup_{\alpha\in K}\|\chi_\alpha-\chi^{(n)}_\alpha\|_1\to 0$. Since $U$ is unitary, $\|U_\alpha\|=1$ for all $\alpha$, hence
\begin{equation}
\sup_{\alpha\in K}
\left|
\Tr_{\hi^{t(\alpha)}}\!\bigl(U_\alpha \chi_\alpha\bigr)
-
\Tr_{\hi^{t(\alpha)}}\!\bigl(U_\alpha \chi^{(n)}_\alpha\bigr)
\right|
\leq
\sup_{\alpha\in K}\|\chi_\alpha-\chi^{(n)}_\alpha\|_1
\to 0.
\end{equation}
For each $n$,
\begin{equation}
\Tr_{\hi^{t(\alpha)}}\!\bigl(U_\alpha \chi^{(n)}_\alpha\bigr)
=
\sum_{j=1}^{N_n}
\braket{
U_\alpha\xi_{j,n}(s(\alpha))
}{
\eta_{j,n}(t(\alpha))
},
\end{equation}
which is continuous by weak continuity. Therefore
\begin{equation}
\alpha\longmapsto
\Tr_{\hi^{t(\alpha)}}\!\bigl(U_\alpha \chi_\alpha\bigr)
\end{equation}
is a locally uniform limit of continuous functions, hence continuous. Thus
$U$ is fiber-wise ultraweakly continuous.
\item Isotropy ultraweak continuity is just fiber-wise ultraweak continuity for the cases where $s(\alpha) = t(\alpha)$.
    \end{enumerate}
\end{proof}

\begin{remark}
    Note that in the group case $M=\{*\}$, fiber-wise and isotropy ultraweak continuity reduce to the same notion of ultraweak continuity, and the three topologies discussed above reduce to the same one. In an action-groupoid model of a torsor, isotropy continuity is vacuous (since stabilisers are trivial so isotropy groups are trivial), while fiber-wise ultraweak continuity recovers the usual ultraweak continuity of the group representation.
\end{remark}

Thus, we can safely work interchangeably with weakly, strongly and fiber-wise ultraweakly continuous unitary representations of continuous (in particular Lie) groupoids. 

In the rest of this paper, we will sometimes be interested in sets of operator fields (respectively, of states) that form subsets of $\A(\scrhi)$ (respectively, of $\state$). We thus use the following notation.

\begin{notation}
    We write $\mathcal{O} \subset \A(\scrhi)$ for a set of measurable operator fields in $\A(\scrhi)$ and $\mathfrak{S} \subset \state$ for a (typically convex) set of measurable operator fields of states.
\end{notation}

\noindent\textit{Outlook toward the operational formalism.}
The geometric picture developed above suggests the following replacement of the standard symmetry-based quantum kinematics.

\begin{center}
\begin{tabular}{ccl}
locally compact group action & $\rightsquigarrow$ & continuous groupoids,\\
homogeneous space & $\rightsquigarrow$ & principal bundle, \\
single Hilbert space representation & $\rightsquigarrow$ & Hilbert bundle,\\ 
unitary group representation & $\rightsquigarrow$ & unitary groupoid representation,\\ 
global covariant POVM & $\rightsquigarrow$ & fiber-wise covariant POVM.
\end{tabular}
\end{center}

The first four entries have now been introduced. The fifth will be the starting point of the operational formalism developed later. In particular, Sec.~\ref{sec:operational-qrf-groups} recalls the usual group-based theory of operational quantum reference frames, and Sec.~\ref{sec:groupoid-qrf} reformulates it for Lie groupoids. The flat-space reduction provided by Proposition \ref{prop:minkowski-action-groupoid} will then imply that the groupoid formalism recovers the standard relativistic QRF framework on Minkowski spacetime.

\section{Operational quantum reference frames for locally compact groups}
\label{sec:operational-qrf-groups}

In this section we recall the operational formulation of quantum reference frames for locally compact groups \cite{loveridge_symmetry_2018,carette_operational_2025,glowacki_quantum_2023}, in the form needed later to compare it with the Lie-groupoid framework. We keep only those parts of the formalism that are essential for the passage from groups to action groupoids: systems of covariance, relativization maps, relative observables and states, and the localization limit.

Throughout this section, $G$ denotes a locally compact second countable Hausdorff topological group and $\Sigma$ a principal homogeneous left $G$-space, i.e. $\Sigma \cong G$.

\subsection{Systems, invariance and relativization}

Let $\mathcal{S}$ be a quantum system carrying an ultraweakly continuous unitary representation
\begin{equation}
U_\mathcal{S}:G\to \mathcal U(\mathcal H_\mathcal{S}).
\end{equation}
The compound system ``system + frame'' is described on the tensor product Hilbert space $\mathcal H_\mathcal{S}\otimes \mathcal H_\mathcal{R}$ with diagonal representation $U_\mathcal{S}(g)\otimes U_\mathcal{R}(g)$. The algebra of invariant bounded operators is
\begin{equation}
\bhsr^G
:=
\{A\in \bhsr \mid U_\mathcal{S}(g)\otimes U_\mathcal{R}(g)AU_\mathcal{S}(g)^\dagger\otimes U_\mathcal{R}(g)^\dagger=A,\ \forall g\in G\}.
\end{equation}

\paragraph*{Relativization for locally compact groups.} We first focus on the notion of relativization on groups themselves. We begin with the well-established \cite{loveridge_symmetry_2018,carette_operational_2025} definition for operational quantum reference frames.

\begin{definition}
\label{def:qrf-group}
A \emph{quantum reference frame} on the locally compact second countable
topological group $G$ is a triple $\mathcal{R}=(U_\mathcal{R},\E_\R,\mathcal H_\mathcal{R})$,
where:
\begin{enumerate}
    \item $\mathcal H_\mathcal{R}$ is a separable complex Hilbert space;
    \item $U_\mathcal{R}:G\to \mathcal U(\mathcal H_R)$ is an ultraweakly continuous unitary representation;
    \item $\E_\mathcal{R}:\Bor(G)\to \mathcal E(\mathcal H_R)$ is a normalized POVM satisfying the covariance condition
    \begin{equation}
    U_\mathcal{R}(g)\,\E_\R(\Delta)\,U_\mathcal{R}(g)^\dagger = \E_\R(g \cdot \Delta),
    \qquad g\in G,\ \Delta\in \Bor(G).
    \end{equation}
\end{enumerate}
The POVM $\E_\R$ will be called the \emph{frame observable}.
\end{definition}

\begin{remark}
\label{rem:projective-group-case}
In much of the QRF literature one allows projective unitary representations. For the purposes of the present paper, and in order to match the groupoid definitions of Sec.~\ref{sec:introduction}, we work with honest unitary representations. The projective variant can be incorporated later by passing to suitable central extensions or multiplier cocycles.
\end{remark}

\begin{definition}
\label{def:relativization-group}
Let $\R=(U_\mathcal{R},\E_\R,\mathcal H_\mathcal{R})$ be a quantum reference frame on $\Sigma = G$, and let $\mathcal{S}$ be a system with representation $U_\mathcal{S}$. The associated \emph{relativization map} is
\begin{equation}
\yen^\R: \bhs \to \bhsr,
\end{equation}
defined by the operator-valued integral \cite{glowacki_w-algebraic_2026}
\begin{equation}
\yen^\R(A)
:=
\int_{G} U_\S(g) A U_\S(g)^\dagger\otimes dE_\mathcal{R}(g).
\end{equation}
\end{definition}

\begin{remark}
    Note that $\overline{\yen^\R(\bhs)}^{uw} \subset \bhsr^G$. Thus relative observables are not merely invariant: they are those invariant observables on the compound system that arise from the chosen frame observable $\E_\R$ via relativization.
\end{remark}

A particularly useful special case arises when one conditions on a chosen state of the frame.
\begin{definition}
\label{def:restriction-group}
Let $\omega\in \mathscr{D}(\hir)$ be a state (density operator) of the frame. The \emph{restriction map}
\begin{equation}
\Gamma_\omega: \bhsr \to \bhs
\end{equation}
is defined by
\begin{equation}
\Tr[\rho\,\Gamma_\omega(B)] = \Tr[(\rho\otimes \omega)B],
\qquad
\rho\in \mathscr{D}(\his),\ B\in \bhsr.
\end{equation}
This gives rise to the \emph{restricted relativization map}
\begin{align}
    \yen^\R_\omega := \Gamma_\omega\circ \yen^\R : \bhs &\to \bhs \\  A &\mapsto \int_{G} U_\S(g) A U_\S(g)^\dagger \, d\mu^{\E_\R}_{\omega}(g),
\end{align}
where we write the Born probability measure
\begin{equation}
\mu^{E_\mathcal{R}}_\omega(X):=\Tr[\omega E_\mathcal{R}(X)].
\end{equation}
\end{definition}

\paragraph*{Relativization for torsors of locally compact groups.} We now examine relativization for torsors (i.e. principal homogeneous spaces) of locally compact groups. Any space $\Sigma$ with a continuous transitive $G$-action $\alpha : G \times \Sigma \to \Sigma$ can be written as $G/H_x$, where $H_x$ is the stabiliser of some point $x \in \Sigma$. This holds even in the case of principal homogeneous spaces i.e. for $G$-torsors, where $H_x = \{e\}$; colloquially, we might say that a $G$-torsor is like a group ``+ a point". Indeed, the equality $\Sigma = G$ does not hold anymore, although $\Sigma \cong G$.

\begin{remark}
    Note that for any normalized $G$-covariant POVM $\E : \Bor(\Sigma) \to \eh$ on a homogeneous $G$-space $\Sigma$, $\supp \E = \Sigma$. This is because $G$ acts transitively on $\Sigma$. 
\end{remark}

\begin{definition}
    Let $G$ be a locally compact second countable
topological group, $\Sigma$ be a principal $G$-space and $o \in \Sigma$. A principal\footnote{Since $\Sigma$ is principal, such quantum reference frames are called \emph{principal}.} \emph{quantum reference frame} on $\Sigma$ is a tuple $\R_o = (U_{\R_o},\E_{\R_o},\hi_{\R_o})$ where
\begin{enumerate}
    \item $\hi_{\R_o}$ is a separable complex Hilbert space;
    \item $U_{\R_o}: G \to \mathcal U(\hi_{\R_o})$ is an ultraweakly continuous unitary representation;
    \item $\E_{\R_o}:\Bor(\Sigma)\to \Eff(\hi_{\R_o})$ is a normalized POVM satisfying the covariance condition
    \begin{equation}
    U_{\R_o}(g)\,\E_{\R_o}(\Delta)\,U_{\R_o}(g)^\dagger = \E_{\R_o}(g \cdot \Delta),
    \qquad g\in G,\ \Delta\in \Bor(\Sigma).
    \end{equation}
\end{enumerate}
The POVM $\E_{\R_o}$ will be called the \emph{frame observable}.
\end{definition}

\begin{definition}
\label{def:relativization-torsor}
Let $\R_o=(U_{\R_o},\E_{\R_o},\hi_{\R_o})$ be a principal quantum reference frame on $\Sigma \cong G$, and let $\mathcal{S}$ be a system with representation $U_\mathcal{S}$. For $p \in \Sigma$, let $g_p \in G$ be the unique element with $p = g_p \cdot o$. The associated \emph{relativization map} is
\begin{equation}
\yen^{\R_o}:\bhs \to \bhsr^G,
\end{equation}
defined by the operator-valued integral \cite{glowacki_w-algebraic_2026}
\begin{equation}
\yen^{\R_o}(A)
:=
\int_{\Sigma} U_\S(g_p) A U_\S(g_p)^\dagger\otimes dE_{\mathcal{R}_o}(p).
\end{equation}
Given $\omega \in \mathscr{D}(\hi_{\R_o})$, the \emph{restricted relativization map} is given by
\begin{align}
    \yen^{\R_o}_{\omega} := \Gamma_{\omega}\circ \yen^{\R_o} : \bhs &\to \bhs \\  A &\mapsto \int_{\Sigma} U_\S(g_p) A U_\S(g_p)^\dagger \, d\mu^{\E_{\R_o}}_{\omega}(p),
\end{align}
\end{definition}

\begin{remark}
    For most purposes, we can, at least in the principal case considered here, drop the $o \in \Sigma$ in the above. We only included it here to contrast the locally compact group case and the homogeneous space case, to highlight that there is here a choice of $o \in \Sigma$ that is usually implicit, and to mirror with the construction of quantum reference frames for groupoids of Sec.~ \ref{sec:groupoid-qrf} where such considerations do become important. Indeed, a groupoid reduces to a group when $M=\{*\}$, but a group action on a principal homogeneous space is instead an action groupoid, which is different. We can however consider the case where $\Sigma=G$ and where $\Sigma \cong G$ without distinction for now.
\end{remark}

As mentioned, this construction can be generalised beyond torsors \cite{glowacki_quantum_2023,glowacki_w-algebraic_2026} to e.g. all homogeneous $G$-spaces, though we are interested in the principal ones here.

\begin{definition}
\label{def:special-frames}
Let $\R$ be a quantum reference frame on $\Sigma$.
\begin{enumerate}
    \item $\R$ is called \emph{sharp}\footnote{In particular, if $\Sigma$ is principal, these quantum reference frames are called ideal.} if $\E_\R$ is projection-valued.
    \item $\R$ is called \emph{localizable} if $\E_\R$ has the norm-$1$ property, i.e.
    \begin{equation}
    \|E_\mathcal{R}(\mathcal O)\|=1
    \end{equation}
    for every $\mathcal O \in \Bor(\Sigma)$ such that $\E_\R(\mathcal{O}) \neq 0$.
\end{enumerate}
\end{definition}

We will come back to the localizability property in Sec.~ \ref{subsec:localization-limit-groups}. In this paper, we will focus on principal QRFs. 

\begin{proposition}[\cite{carette_operational_2025}]
\label{prop:relativization-basic}
The relativization map $\yen^\R$ is linear, normal, unital and completely positive. If $\mathcal{R}$ is sharp, then $\yen^\R$ is multiplicative, hence a normal $*$-homomorphism. If $\R$ is principal and localizable, then $\yen^\R$ is injective.
\end{proposition}

\begin{remark}
    There is a dual picture of relativization that applies to states rather than operators \cite{carette_operational_2025}. We shall not expand on this viewpoint in this paper.
\end{remark}

\begin{remark}
    Note that for either the locally compact group or the torsor case, $G$-invariant operators are invariant under relativization: if $A \in \bhs^G$ then $\yen^\R(A) = A \otimes \id_{\bhr}$ and for all $\omega \in \homframestate$, $\yen^\R_\omega(A) = A$.
\end{remark}

\subsection{Localization limit and recovery of the absolute description}
\label{subsec:localization-limit-groups}

The localizability property is operationally crucial: it guarantees the existence of frame states producing probability measures that can be concentrated around arbitrary frame values with arbitrarily high precision \cite{glowacki_operational_2023,carette_operational_2025}. 

\begin{proposition}[\cite{heinonen_norm-1-property_2003}]
    \label{prop:norm 1 equivalences}
    Let $\Sigma$ be a principal homogeneous $G$-space and $\E : \Bor(\Sigma) \to \eh$ be a normalized POVM. The following are equivalent:
    \begin{enumerate}
        \item $\E$ has the norm-1 property.
        \item For every $X \in \Bor(\Sigma)$ for which $\E(X) \neq 0$, there is a sequence of unit vectors $\ket{\phi_n} \subset \hi$ for which
        \begin{equation}
            \lim_{n \to \infty} \expval{\E(X)}{\phi_n} = 1 \, .
        \end{equation}
        \item ($\epsilon$-decidability property) For every $X \in \Bor(\Sigma)$ such that $\E(X) \neq 0$ and for any $\epsilon > 0$ there exists a unit vector $\ket{\phi_\epsilon} \in \hi$ for which
        \begin{equation}
            \expval{\E(X)}{\phi_\epsilon} > 1-\epsilon \, .
        \end{equation}
    \end{enumerate}
\end{proposition}

An important consequence is the following.

\begin{corollary}[\cite{carette_operational_2025}]
    \label{cor:norm 1 sequence}
    If a covariant normalized POVM $\E : \Bor(\Sigma) \to \eh$ has the norm-1 property and $\Sigma$ is metrizable, then for any $x \in \Sigma$ there exists a sequence of pure states $(\omega_n) \subset \mathscr{D}(\hir)$ such that the sequence of probability measures $(\mu^\E_{\omega_n})$ converges weakly to the Dirac measure $\delta_x$.
\end{corollary}

We call such a sequence $(\omega_n) \subset \mathscr{D}(\hir)$ a \emph{localizing sequence centred at the identity}. The physical meaning of this for QRFs is that, once the frame is prepared in states increasingly localized around a classical frame value, the relative description converges to the ordinary non-relational one. We review this result in the case of localizable QRFs on principal homogeneous spaces.

\begin{theorem}[Fundamental theorem of operational quantum reference frames \cite{carette_operational_2025}]
\label{prop:localization-limit-groups}
Let $\mathcal{R}=(U_\mathcal{R},\E_\R,\mathcal H_\mathcal{R})$ be a localizable principal frame, and let $(\omega_n) \subset \mathscr{D}(\hir)$ be a localizing sequence centred at the identity. Then for every $A\in \mathcal B(\mathcal H_\mathcal{S})$,
\begin{equation}
\lim_{n\to\infty}\yen^\R_{\omega_n}(A)=A
\end{equation}
in the ultraweak topology. 
\end{theorem}

In this operational sense, the ordinary kinematics of the system is recovered as a localization limit of the relational description. This result is the reason that the standard, apparently non-relational description of a quantum system may be understood as an effective description relative to a sufficiently good external frame.

\subsection{The continuity of frame observables}
\label{subsec:localization and continuity}

The notion of localizability comes into contact with various notions of continuity of the frame observables, as we now explore. Let $\F : \Bor(\mink) \to \eh$ be any normalized POVM over Minkowski and $\V_0 \subset \mink$ be a bounded region of Minkowski spacetime. Hegerfeldt showed \cite{hegerfeldt_causality_1998} that if the generator $H$ of time translations is bounded from below (i.e. the spectrum condition holds) and if $\exists \rho \in \homstate$ such that $\Tr[\rho \F(\V_0)] = 1$ then either $\Tr[\rho \cdot t \F(\V_0)] = 1$ for all $t \in \mathbb{R}$ or $\Tr[\rho \cdot t \F(\V_0)] < 1$ for almost every $t \in \mathbb{R}$ (the set of such $t$’s is dense and open). For normalized covariant POVMs this is particularly problematic and, in Hegerfeldt's words \cite{hegerfeldt_causality_1998},
\begin{quote}
    \enquote{If at $t = 0$ a particle is strictly localized in a bounded region $\V_0$ then, unless it remains in $\V_0$ for all times, it cannot be strictly localized in a bounded region $\V$, however large, for any finite time interval thereafter, and the particle localization immediately develops infinite \enquote{tails}. The spreading is over all space except possibly for \enquote{holes} which, if any, will persist for all times $[\cdots]$.}
\end{quote}

A weaker requirement is that $\norm{\F(\V_0)} = 1$, i.e. that the localisation is arbitrarily well approximated to be in $\V_0$. This is the meaning of the norm-1 or localisability property of POVMs and QRFs. We link these notions of localizability to various continuity conditions of the POVMs.

\begin{definition}
    Let $\Sigma$ be a real finite-dimensional manifold, $\E : \Bor(\Sigma) \to \eh$ be a POVM and $\mu : \Bor(\Sigma) \to \mathbb{R}$ a regular measure. $\E$ is said to be \emph{absolutely continuous with respect to $\mu$}, written $\E \ll_\text{ac} \mu$, if there exists a number $c \in \mathbb{R}$ such that
    \begin{equation}
        \norm{\E(\Delta)} \leq c \mu(\Delta) \, \qquad \forall \Delta \in \Bor(\Sigma) \, .
    \end{equation}
\end{definition}

\begin{theorem}
    \label{thm:no go absolute continuity}
    Let $\Sigma$ be a topological space as above, $\mu : \Bor(\Sigma) \to \mathbb{R}$ a regular non-atomic measure and $\E : \Bor(\Sigma) \to \eh$ be a POVM. If $\E \ll_\text{ac} \mu$, then $\E$ does not have the norm-1 property.
\end{theorem}

\begin{proof}
    Beneduci and Schroeck showed (\cite{beneduci_note_2013}, Thm 2.10) that such an $\E$ can only have the norm-1 property if $\norm{\E(\{x\})} \neq 0$ for each $x \in \Sigma$ such that $\mu(\{x\}) < \infty$. But if $\mu$ is non-atomic then $\mu(\{x\}) = 0 < \infty$ for all $x \in \Sigma$ and so if $\E \ll_\text{ac} \mu$ then $\norm{\E(\{x\})} \leq c \cdot 0 = 0 \Rightarrow \E(\{x\}) = 0$ which shows that $\E$ cannot have the norm-1 property.
\end{proof}

In particular, the Lebesgue measure on Minkowski is non-atomic; thus, any absolutely continuous POVM with respect to that measure is not localizable. This covers POVMs which have a bounded operator-valued density with respect to the Lebesgue measure, i.e. POVMs of the form $\E(\Delta) = \int_\Delta K(x) d\mu(x)$ where $K(x) \in \bh$; and in particular coherent-state POVMs \cite{beneduci_note_2013}. A weaker notion to absolute continuity is that of $\mu$-continuity, which we now cover.

\begin{definition}
    Let $\Sigma$ be a real finite-dimensional manifold, $\E : \Bor(\Sigma) \to \eh$ be a POVM and $\mu : \Bor(\Sigma) \to \mathbb{C}$ be a real measure. We say that $\E$ is \emph{$\mu$-continuous}, written $\E \ll \mu$, if $\mu(\Delta) = 0 \Rightarrow \E(\Delta) = 0$ for all $\Delta \in \Bor(\Sigma)$. 
\end{definition}

Note that if $\E \ll_\text{ac} \mu$ then $\E \ll \mu$. Indeed, if $\exists c < \infty$ such that $\norm{\E(\Delta)} \leq c \mu(\Delta)$ for all $\Delta \in \Bor(\Sigma)$, then $\mu(\Delta)=0 \Rightarrow \norm{E(\Delta)} \leq c \cdot 0 = 0 \Rightarrow \E(\Delta) = 0$. The converse, however, need not hold: for example, the position PVM on the real line is $\mu$-continuous for $\mu$ the Lebesgue measure but not absolutely continuous with respect to $\mu$, since $\norm{\E((0,\epsilon))} = 1$ for any $\epsilon > 0$ while $\mu((0,\epsilon)) = \epsilon$, so we cannot have $1 \leq c \epsilon$ for finite $c$. We now show that covariant POVMs are necessarily $\mu$-continuous with respect to (quasi-)invariant measures.

\begin{theorem}
    \label{thm:covariance implies mu continuity}
    Let $G$ be a locally compact, second countable Hausdorff topological group which acts transitively on $\Sigma \cong G/H$, a locally compact second countable topological space with Borel $\sigma$-algebra $\Bor(\Sigma)$, where $H$ is a closed subgroup of $G$. Let $\mu$ be a quasi-invariant\footnote{A measure $\mu$ on $\Sigma$ is quasi-invariant if $\forall g \in G, \, \mu(g \cdot X) = 0 \Leftrightarrow \mu(X) = 0$ for all $X \in \Bor(\Sigma)$.} $\sigma$-finite positive Borel measure on $\Sigma$. If $\E : \Bor(\Sigma) \to \eh$ is a $G$-covariant normalized  POVM then $\E \ll \mu$, and for all $\chi \in \thi$, $\exists \dfrac{d\mu^{\E}_{\chi}}{d\mu} \in L^1(\Sigma;\mu)$. Moreover, if $\chi \geq 0$ then $\dfrac{d\mu^{\E}_{\chi}}{d\mu} \geq 0$ and if $\Tr[\chi] = 1$ then $\int_\Sigma \dfrac{d\mu^{\E}_{\chi}}{d\mu} d\mu = 1$.
\end{theorem}

\begin{proof}
    See App. \ref{proof mu continuity}.
\end{proof}

Thus, covariant POVMs on homogeneous space have Radon-Nikodym derivatives in $L^1$. We now show that if these covariant POVMs are absolutely continuous, then their Radon-Nikodym derivatives are in $L^1 \cap L^\infty$ -- a property which makes these densities \emph{nicer} in the sense that any function in $L^1 \cap L^\infty$ is also in every other $L^p$ space by interpolation. 

\begin{theorem}
    \label{thm: absolute continuity POVMs}
     Let $\Sigma$ be a locally compact second countable topological space with Borel $\sigma$-algebra $\Bor(\Sigma)$, and let $\mu$ be a $\sigma$-finite positive Borel measure on $\Sigma$. Then $\E \ll_\text{ac} \mu$ iff for all $\chi \in \thi$, $\exists \dfrac{d\mu^{\E}_{\chi}}{d\mu} \in L^\infty(\Sigma;\mu)$. Moreover, if $\chi \geq 0$ then $\dfrac{d\mu^{\E}_{\chi}}{d\mu} \geq 0$ and if $\Tr[\chi] = 1$ then $\int_\Sigma  \dfrac{d\mu^{\E}_{\chi}}{d\mu} d\mu = 1$.
\end{theorem}

\begin{proof}
    See App. \ref{proof absolute continuity}.
\end{proof}

Theorems~\ref{thm:no go absolute continuity} and
\ref{thm: absolute continuity POVMs} give a dichotomy. Covariance implies only $\mu$-continuity, and hence the existence of $L^1$ Radon-Nikodym densities for all Born measures. By contrast, absolute continuity $\E\ll_{\mathrm{ac}}\mu$ is equivalent to essential boundedness of all scalar Radon-Nikodym derivatives; the Banach-Steinhaus argument in App.~\ref{proof absolute continuity} then produces a single constant $c<\infty$ satisfying $\|\E(\Delta)\|\leq c\mu(\Delta)$. For a non-atomic measure this is incompatible with the norm-$1$ property. Consequently, under Definition~\ref{thm:no go absolute continuity}, there are no localizable absolutely continuous POVMs. Localizable covariant POVMs may nevertheless be $\mu$-continuous; along a localizing sequence their $L^1$ densities must become increasingly concentrated and cannot remain uniformly bounded in $L^\infty$.

\begin{remark}
    Several known examples of absolutely continuous covariant POVMs (e.g. coherent state POVMs \cite{beneduci_note_2013}) are not localizable.
\end{remark}

Beyond these tensions, Thm. \ref{thm:covariance implies mu continuity} is useful for us in the context of relativization, since it guarantees the existence of a Radon-Nikodym derivative $f^{\E_\R}_{\omega} = \frac{d\mu^{\E_\R}_{\omega}}{d\mu} \in L^1(\Sigma,\mu)$, i.e. we can write the integral as a smearing with respect to this function, that we call the \emph{frame smearing function} \cite{fedida_foundations_2025}, and the relevant quasi-invariant measure on the space. It will also be useful for showing the $\mu$-admissibility (see Sec.~\ref{subsec:groupoid-qrf-definition})  of groupoid QRF in the recovery of torsor QRFs, as we will see in Sec.~\ref{subsec:reduction to group qrfs}.

\section{Quantum reference frames for continuous groupoids}
\label{sec:groupoid-qrf}

In this section we formulate the operational notion of a quantum reference frame for a continuous groupoid. The guiding principle is simple: the group-based ingredients recalled in Sec.~\ref{sec:operational-qrf-groups} are replaced by their groupoid analogues. A single Hilbert space is replaced by a measurable field of Hilbert spaces, a unitary representation of a group by a unitary representation of a continuous groupoid, and a covariant POVM on a homogeneous space by a family of fiber-wise covariant POVMs on the target fibers of the groupoid.

The analytic background for this replacement is provided by the measure-groupoid formalism and the associated groupoid von Neumann algebra (see, for instance \cite{ciaglia_groupoidal_2024} and references therein). Conceptually, however, the operational construction itself only needs three ingredients: a continuous groupoid, a measurable field of Hilbert spaces carrying a unitary groupoid representation, and a covariant family of POVMs. The operator-algebraic structure associated with the corresponding measure groupoid will become more important in Secs.~\ref{sec:RQFT Minkowski}, \ref{sec:RQFT curved} and \ref{sec:towards RQGFT} where the framework is applied to relational quantum field theory.

\subsection{Groupoid quantum reference frames}
\label{subsec:groupoid-qrf-definition}

We now formalize the operational notion of a quantum reference frame for a continuous groupoid. 

\begin{definition}
    Let $\Gamma \rightrightarrows M$ be a continuous groupoid. We say that a field of POVMs $(\E^x)$ is \emph{measurable} if for every $\Delta \in \Bor(\Gamma)$, the field $x \mapsto \E^x(\Delta \cap \Gamma^x)$ is a measurable operator field.
\end{definition}

\begin{definition}
\label{def:groupoid-qrf-section4}
A (principal) \emph{quantum reference frame} for the continuous groupoid $\Gamma\rightrightarrows M$ based on $\Gamma$ is a triple $\R = (U,\E_\R,\scrhir)$,
where:
\begin{enumerate}
    \item $\scrhir=\bigsqcup_{x\in M}\hirx\to M$ is a continuous field of separable Hilbert spaces;
    \item $U$ is a fiber-wise ultraweakly continuous unitary representation of $\Gamma$ on $\scrhir$;
    \item $\E_\R=\{\E_\R^x\}_{x\in M}$ is a measurable field of normalized POVMs  $\E_\R^x:\Bor(\Gamma^x)\to \mathcal{E}(\hirx)$,
    satisfying the covariance condition
    \begin{equation}
    U_\alpha\, \E_\R^{s(\alpha)}(\Delta)\, U_\alpha^\dagger  =     \E_\R^{t(\alpha)}(L_\alpha \Delta),
    \qquad
    \alpha\in \Gamma,\ \Delta \in \Bor(\Gamma^{s(\alpha)}) \, ,
    \end{equation}
    where  $L_\alpha:\Gamma^{s(\alpha)}\to \Gamma^{t(\alpha)}$,  $\beta\mapsto L_\alpha (\beta) = \alpha\circ \beta$.
\end{enumerate}
The family $\E_\R$ is called the \emph{frame observable} of $\R$.
\end{definition}

The group-based formalism is recovered by replacing the family of fibers $\Gamma^x$ by a single space $\Sigma$ and fiber-wise covariance under left translation in the groupoid by global covariance under $G$.

\begin{remark}
    It is possible to define QRFs for groupoids based on spaces other than $\Gamma$. That is, the frame observables need not take Borel subsets of $\Gamma^x$ for $x \in M$, but Borel subsets of other spaces (perhaps $\Gamma^x/\Gamma^{x}_x$, perhaps some other groupoid spaces). Those generalise non-principal homogeneous space QRFs. One then expects the ensuing construction of the relativization map to follow through assuming isotropy-invariant operators. We do not focus on such QRFs in this paper, but these might be of interest if one wanted, for example, to define QRFs directly on the spacetime manifold $\mathcal{M}$ or on the space of $4$-metrics $GL(4,\mathbb{R})/SO_+^\uparrow(1,3)$.
\end{remark}

\begin{remark}
    One may also want to generalise the notion of groupoid QRFs to measurable groupoids, and relax the continuity assumptions of the field of Hilbert spaces and of the unitary representation. One challenge to overcome this would be to define the countable additivity of POVMs
    \begin{equation}
        \E_\R^x\left(\bigsqcup_{n=1}^\infty \Delta_n\right) = \sum_{n=1}^\infty \E_\R^x(\Delta_n), \qquad \Delta_n \in \Bor(\Gamma^x)
    \end{equation}
    in some suitable topology; we leave this for future work. 
\end{remark}


\begin{definition}
    Let $\R$ be a (principal) QRF for the continuous groupoid $\Gamma \rightrightarrows M$. For $\mu$-almost every $x \in M$ and $\pi_x : \Gamma^x \to M$ be a measurable projection, we define
    \begin{align}
        \F_{\R,\pi}^x : \Bor(M) &\to \Eff(\hirx) \\
        \F_{\R,\pi}^x(N) &:= \E_\R^x(\pi_x^{-1}(N))
    \end{align}
    and call $(\F_{\R,\pi}^x)$ the \emph{$\pi$-marginal frame POVM}.\footnote{It is a normalized POVM: $\F_{\R,\pi}^x(N)$ is non-negative for all $N \in \Bor(M)$, $\F_{\R,\pi}^x(M) = \E_{\R,\pi}^x(\Gamma^x) = \id_{\hirx}$, and countable additivity follows from $\E_\R^x$ being a POVM and by the measurability assumption of $\pi$.} In particular, if $\pi_x = s|_{\Gamma^x}$ (i.e. $\pi_x^{-1}(N) = s^{-1}(N) \cap \Gamma^x$ for all $N \in \Bor(M)$), we call $\F_{\R,s}^x$ the \emph{source marginal of the frame observable}.
\end{definition}

From there, we can also introduce several nice properties that the QRFs may or may not have.

\begin{definition}
\label{def:groupoid-qrf-sharp-localizable}
Let $\R=(\scrhir,U,\E_\R)$ be a quantum reference frame for the continuous groupoid $\Gamma \rightrightarrows M$. $\R$ is called
\begin{enumerate}
    \item $\mu$-admissible if the source marginal of the frame observable does not see $\mu$-null sets, i.e. given
    \begin{equation}
        \label{eqn:marginal povm s}
        \F_{\R,s}^x(N) := \E_\R^x(s^{-1}(N) \cap \Gamma^x), \qquad N \in \Bor(M),
    \end{equation}
    we have that $\F_{\R,s}^x$ is $\mu$-continuous: $\F_{\R,s}^x \ll \mu$ (i.e. $\mu(N)=0 \Rightarrow \F_{\R,s}^x(N)=0$).
    \item \emph{Lie} if $\Gamma$ is a Lie groupoid.
    \item \emph{sharp} if each $\E_\R^x$ is projection-valued.
    \item \emph{principal} if it is based on $\Gamma$ and non-principal if not.
    \item \emph{ideal} if it is both principal and sharp.
   \item \emph{localizable} if each $\E_\R^x$ has the norm-$1$ property, namely
\begin{equation}
  \|\E_\R^x(\Delta)\|=1
\end{equation}
for every $x\in M$ and every $\Delta\in\Bor(\Gamma^x)$ with
$\E_\R^x(\Delta)\neq0$.
    \item \emph{bisection-admissible} if $(\Gamma \rightrightarrows M,\nu,\mu)$ is Lie and bisection-admissible.
\end{enumerate}
\end{definition}

$\mu$-admissibility will be useful for the independence of the relativization map defined below on the choice of representative operator field. 

\begin{remark}
    If $\E_\R^x$ is projection-valued then it has the norm-1 property. Hence, if $\R$ is sharp (i.e. $\E_\R^x$ is projection-valued for $\mu$-almost every $x$) then it is localizable (i.e. $\E_\R^x$ has the norm-1 property for $\mu$-almost every $x$).
\end{remark}

\begin{lemma}
    \label{lem:covariance marginal}
    Let $\R$ be a principal quantum reference frame for $\Gamma \rightrightarrows M$. Then $\F_{\R,s} = (\F_{\R,s}^x)$, defined as in Eqn.~\eqref{eqn:marginal povm s}, is a measurable field of normalized POVMs satisfying the invariance condition
    \begin{equation}
        \label{eqn:source marginal invariance}
        U_\alpha \F_{\R,s}^{s(\alpha)}(N)U_{\alpha}^\dagger = \F_{\R,s}^{t(\alpha)}(N), \qquad \alpha \in \Gamma, N \in \Bor(M).
    \end{equation}
\end{lemma}

\begin{proof}
    This follows immediately from the fact that
    \begin{equation}
        U_\alpha \E_\R^{s(\alpha)}(\Delta) U_\alpha^\dagger = \E_\R^{t(\alpha)}(L_\alpha \Delta)
    \end{equation}
    for all $\alpha \in \Gamma$ and all $\Delta \in \Bor(\Gamma^{s(\alpha)})$: in particular, for any $N \in \Bor(M)$, $s^{-1}(N) \cap \Gamma^{s(\alpha)} \in \Bor(\Gamma^{s(\alpha)})$ since $s$ is continuous thus Borel and $\Gamma^x = t^{-1}(x)$ is closed. Since $s(\alpha \circ \beta) = s(\beta)$ and $L_\alpha : \Gamma^{s(\alpha)} \to \Gamma^{t(\alpha)}$ is a bijection, $L_\alpha(s^{-1}(N) \cap \Gamma^{s(\alpha)})=s^{-1}(N) \cap \Gamma^{t(\alpha)}$ so writing $\F_{\R,s}^{s(\alpha)}(N) := \E_\R^{s(\alpha)}(s^{-1}(N) \cap \Gamma^{s(\alpha)})$ the result follows.
\end{proof}

Note that this is the source marginal of the frame observable; it is not the usual marginal POVM that one manipulates, for example, in the case of principal bundles $P \to M$.

\begin{definition}
    Let $(\Gamma \rightrightarrows M,\nu,\mu)$ be a Lie groupoid, and $\mathscr{B} \leq \mathrm{Bis}(\Gamma)$ be a subgroup of the group of bisections of $\Gamma$. We say that the measurable projections $\pi=(\pi_x : \Gamma^x \to M)$ are \emph{$\mathscr{B}$-compatible} if for $\mu$-almost every $x \in M$, $\pi_{\varphi_b(x)} \circ L_{b(x)} = \varphi_b \circ \pi_x$ for all $b \in \mathscr{B}$. If $\pi=(\pi_x)$ is a family of $\mathscr{B}$-compatible measurable projections, we call $(\F_{\R,\pi}^x)$ \emph{$\mathscr{B}$-compatible marginal frame observables.}
\end{definition}

\begin{lemma}
    \label{lem:projection marginal covariance}
    Let $\R$ be a principal Lie quantum reference frame for $\Gamma \rightrightarrows M$ and $\mathscr{B} \leq \mathrm{Bis}(\Gamma)$. Then the $\mathscr{B}$-compatible marginal frame observables $\F_{\R,\pi} = (\F_{\R,\pi}^x)$ satisfy the covariance condition
    \begin{equation}
        \label{eqn:projection marginal covariance}
        U_{b(x)} \F_{\R,\pi}^{x}(N)U_{b(x)}^\dagger = \F_{\R,\pi}^{\varphi_b(x)}(\varphi_b(N)), \qquad b \in \mathscr{B}, x \in M, N \in \Bor(M).
    \end{equation}
\end{lemma}

\begin{proof}
    See App.~\ref{app:proof projection marginal covariance}.
\end{proof}

Thus, we see from Eqns.~\eqref{eqn:source marginal invariance} and \eqref{eqn:projection marginal covariance} that the source marginal of the frame observable is not $\mathrm{Bis}(\Gamma)$-compatible. On the other hand, $\F_{\R,t}$ is trivially $\mathrm{Bis}(\Gamma)$-compatible. Other interesting marginal frame observables will typically be $\mathscr{B}$-compatible for e.g. $\mathscr{B}$ the isometry subgroup of $\mathrm{Bis}(\mathrm{Poin}(\M,g))$ for $\Gamma = \mathrm{Poin}(\M,g) \rightrightarrows \M$. There are thus different notions of ``marginal POVMs over spacetime" to take into account. Now that we can talk about marginal POVMs, let us look at conditional probability measures and disintegration. We remind the reader that $\mu^{\E_\R^x}_{\omega_x}(\cdot) := \Tr[\omega_x \E_\R^x(\cdot)]$ and $\mu^{\F_{\R,\pi}^x}_{\omega_x}(\cdot) = \Tr[\omega_x \F_{\R,\pi}^x(\cdot)]$ are Born probability measures.

\begin{lemma}
\label{lem:disintegration exists Lie QRF}
Let $\R$ be a principal Lie groupoid QRF for
$\Gamma \rightrightarrows M$. Then, for $\mu$-almost every $x \in M$, every $\omega_x \in \mathscr{D}(\hirx)$, and every Borel measurable map $\pi_x : \Gamma^x \to M$, there exists a Markov kernel
\begin{equation}
    M \ni y
    \longmapsto
    \nu^{\R,\pi}_{\omega_x}(\,\cdot\mid y)
    \in \operatorname{Prob}(\Gamma^x)
\end{equation}
such that:
\begin{enumerate}
    \item for every $\Delta \in \Bor(\Gamma^x)$, the function $y \mapsto \nu^{\R,\pi}_{\omega_x}(\Delta\mid y)$ is Borel measurable,
    \item one has $\nu^{\R,\pi}_{\omega_x}         \bigl(\pi_x^{-1}(\{y\})\mid y\bigr) = 1$
    for $\mu^{\F_{\R,\pi}^x}_{\omega_x}$-almost every $y \in M$,
    
    \item for every $\Delta \in \Bor(\Gamma^x)$ and every
    $\Lambda \in \Bor(M)$,
    \begin{equation}
        \mu^{\E_\R^x}_{\omega_x}
        \bigl(\Delta\cap\pi_x^{-1}(\Lambda)\bigr)
        =
        \int_{\Lambda}
        \nu^{\R,\pi}_{\omega_x}(\Delta\mid y)\,
        d\mu^{\F_{\R,\pi}^x}_{\omega_x}(y).
    \end{equation}
\end{enumerate}
The kernel is unique up to $\mu^{\F_{\R,\pi}^x}_{\omega_x}$-almost-everywhere equality.
\end{lemma}

\begin{proof}
Since $\Gamma$ is a Lie groupoid, $\Gamma^x=t^{-1}(\{x\})$ is a second-countable smooth manifold and hence a Polish space. Likewise, $M$ is Polish. The measure $\mu^{\E_\R^x}_{\omega_x}$ is a Borel probability measure on $\Gamma^x$, and its push-forward under $\pi_x$ is
\begin{equation}
    (\pi_x)_*\mu^{\E_\R^x}_{\omega_x}(\Lambda)    =\mu^{\E_\R^x}_{\omega_x}    \bigl(\pi_x^{-1}(\Lambda)\bigr)    =\Tr_{\hirx}\left[\omega_x\E_\R^x\bigl(\pi_x^{-1}(\Lambda)\bigr)\right]=\Tr_{\hirx}\left[   \omega_x\F_{\R,\pi}^x(\Lambda)\right] = \mu^{\F_{\R,\pi}^x}_{\omega_x}(\Lambda)
\end{equation}
for every $\Lambda \in \Bor(M)$. The disintegration theorem for Borel probability measures on standard Borel
(or, in particular, Polish) spaces therefore yields the stated kernel, its
concentration on the fibers of $\pi_x$, the disintegration identity, and
essential uniqueness; see, for example, \cite[Sec.~10.4]{bogachev_conditional_2007}.
\end{proof}

\begin{remark}
    For every fixed $x\in M$, the preceding disintegration applies to any Borel measurable map $\pi_x:\Gamma^x\to M$; $\mathscr{B}$-compatibility is not required. The lemma does not, by itself, assert that the conditional probability measures may be selected jointly measurably as functions of $(x,y)$.
\end{remark}

\subsection{Localization}
\label{subsec:localization-limit-groupoids}

Before discussing relativization, let us first discuss more general notions of localization of groupoid QRFs. Unlike in the homogeneous space setting, where the normalization and covariance of the frame observables ensure that they are supported over all of $\Sigma$, this can fail to hold in more general settings. Indeed, the covariance properties now relate different $\E_\R^x$'s. We thus untangle several notions of localization and supports.

\begin{definition}
    Let $\R$ be a principal QRF for the continuous groupoid $\Gamma \rightrightarrows M$.
    \begin{enumerate}
        \item For $\mu$-almost every $x \in M$, we write
        \begin{equation}
            \supp \E_\R^x := \bigcup_{\omega_x \in \mathscr{D}(\hirx)} \supp \mu^{\E^x_\R}_{\omega_x} \subseteq \Gamma^x
        \end{equation}
        for the \emph{support of $\E_\R^x$}.
        \item For $\mu$-almost every $x \in M$ and projections $\pi_x : \Gamma^x \to M$, we write
        \begin{equation}
            \supp \F_{\R,\pi}^x := \bigcup_{\omega_x \in \mathscr{D}(\hirx)}\supp \mu^{\F^x_{\R,\pi}}_{\omega_x} \subseteq M
        \end{equation}
        for the \emph{support of $\F_{\R,\pi}^x$}.
        \item Given $\omega \in \framestate$, we write
        \begin{equation}
            \supp(\R,\omega) := \esssupp_{x \in M}  \mu^{\E_\R^x}_{\omega_x} \subseteq \Gamma
        \end{equation}
        for the \emph{support of $(\R,\omega)$}.
        \item Given $\omega \in \framestate$ and $\pi=(\pi_x : \Gamma^x \to M)$, we write
        \begin{equation}
            \label{eqn:Loc(R,w)}
            \text{Loc}_\pi(\R,\omega) := \esssupp_{x \in M} \mu^{\F_{\R,\pi}^x}_{\omega_x} \subseteq M
        \end{equation}
        for the \emph{(spacetime)\footnote{Note that these definitions also work when $M$ is not a spacetime manifold, but we will later use it in this sense.} localization of $(\R,\omega)$ given $\pi$.}
        \item We write
        \begin{equation}
            \supp \R := \esssupp_{x \in M} \E_\R^x \subseteq \Gamma
        \end{equation}
        for the \emph{support of $\R$.}
        \item Given $\pi = (\pi_x : \Gamma^x \to M)$, we write
        \begin{equation}
        \label{eqn:Loc(R)}
            \text{Loc}_\pi(\R) := \esssupp_{x \in M} \F_{\R,\pi}^x \subseteq M
        \end{equation}
        for the \emph{(spacetime) localization of $\R$ given $\pi$.}
    \end{enumerate}
\end{definition}

\begin{remark}
    Unlike the case of QRFs on homogeneous spaces, we need not have $\text{Loc}_\pi(\R) = M$ in general. This is because the groupoid covariance, in general, joins different $\E_\R^x$'s and $\E_\R^y$'s, rather than joining different Borel subsets in the range of each $\E_\R^x$. Note however that the isotropy part of the covariance, i.e. $\gamma \in \Gamma^x_x$ giving $U_\gamma \E_\R^x(\Delta) U_\gamma^\dagger = \E_\R^x(L_\gamma \Delta)$, does join different Borel subsets $\Delta$ and $L_\gamma \Delta$ for a fixed $x$; however, this is typically not enough to ensure the full support on the space if the isotropy group is not transitive on $\Gamma^x$. Moreover, a suitable subgroup $\mathscr{B}$ of the group of bisections may or may not help with this, depending on whether the $(\F_{\R,\pi}^x)$ are $\mathscr{B}$-compatible or not.
\end{remark}

\begin{remark}
    It can now be difficult to place a canonical notion of ``localization" to a QRF, given it can be dependent on a choice $\pi$: it is not always true that $\text{Loc}_{\pi_1}(\R) = \text{Loc}_{\pi_2}(\R)$ for two families of projections $\pi_1$ and $\pi_2$. There does exist an intrinsic choice of $\pi$ for $\Gamma$ given by $\pi=s|_{\Gamma^x}$, though one would expect the localization of the QRF to change covariantly under (a subgroup $\mathscr{B}$, e.g. in the case of the Poincaré groupoid over spacetime, the isometry subgroup $\mathrm{Bis}_{\mathrm{Isom}}(\mathrm{Poin}(\M,g))$ of) bisections rather than stay invariant. Thus, perhaps it is more natural to choose $\mathscr{B}$-compatible $\pi$'s for suitable $\mathscr{B}$; again, this is typically non-unique. Note however that $\supp \R$ is \emph{not} dependent on any choice of $\pi$.
\end{remark}

\begin{remark}
    Different choices of $\pi$ have different physical meanings and answer different questions related to localization. For example, take $\pi_t = t_{\Gamma^x}$, then $\supp \F_{\R,\pi_t}^x = \{x\}$ and $\text{Loc}_{\pi_t}(\R,\omega) = M$ for all $\omega \in \framestate$ and $\text{Loc}_{\pi_t}(\R) = M$, i.e. the notions of localization from this $\pi_t$ answer a question of the form ``where are the ontic frame observables localized in spacetime" (and, indeed, they are localized $\mu$-almost everywhere). Other choices, notably those $\mathrm{Bis}_{\mathrm{Isom}}(\mathrm{Poin}(\M,g))$-compatible projections, answer more epistemic questions of the type ``where are the frame observables non-trivially realized in spacetime", in which case the answer will typically not be all of $M$.
\end{remark}

\begin{remark}
    The fact that we can have $\text{Loc}_\pi(\R) \neq M$ is a quantum mechanical counterpart of the idea that in General Relativity, tetrads need not have a global description and are localised in space and time: only in Minkowski spacetime do we tend to have globally defined reference frames. That is, we can write (as an analogy)
    \begin{quote}
        There exists a global inertial reference frame $\leftrightsquigarrow$ $\text{Loc}_\pi(\R)=M$ for some ``suitable" $\pi$.
    \end{quote}
\end{remark}

It would certainly be interesting to examine the different meanings of different possible $\pi$'s, and whether natural choices arise. We keep this for future work.

\begin{lemma}
\label{lem:supp E fiber closed}
Let $\R$ be a principal groupoid QRF for $\Gamma\rightrightarrows M$. For every
$x\in M$ such that $(\Gamma^x)$ is second-countable, write
\begin{equation}
    O^x_{\max}:=\bigcup\{O\subseteq\Gamma^x\ \text{open}\mid \E_\R^x(O)=0\}\,.
\end{equation}
Then
\begin{equation}
    \supp\E_\R^x = \Gamma^x\smallsetminus O^x_{\max}\,,
\end{equation}
which is closed. Moreover, there is a unit vector
$\ket{\psi}\in\hirx$ with $\supp\mu^{\E_\R^x}_{\ketbra{\psi}}=\supp\E_\R^x$.
\end{lemma}

\begin{proof}
    See App. \ref{app:proof supp E fiber closed}.
\end{proof}

\begin{proposition}
    \label{prop:norm 1 property povm}
    Let $\R$ be a principal groupoid QRF for a continuous groupoid $\Gamma \rightrightarrows M$ and $x \in M$. Suppose $(\Gamma^x)$ is metrizable\footnote{This is automatic for Lie groupoids: every target fiber $\Gamma^x$ of a Lie groupoid is a smooth (hence second-countable, metrizable, locally compact Hausdorff) manifold.}. The following are equivalent:
    \begin{enumerate}
        \item $\E_\R^x$ has the norm-1 property.
        \item For every $X \in \Bor(\Gamma^x)$ with $\E_\R^x(X) \neq 0$ there exists unit vectors $(\ket{\phi^x_n}) \subset \hirx$ with 
        \begin{equation}
            \lim_{n \to \infty} \expval{\E_\R^x(X)}{\phi_n^x} = 1 \, .
        \end{equation}
        \item ($\epsilon$-decidability property) For every $X \in \Bor(\Gamma^x)$ such that $\E_\R^x(X) \neq 0$ and for any $\epsilon > 0$ there exists a unit vector $\ket{\phi_\epsilon^x} \in \hirx$ for which
        \begin{equation}
            \expval{\E_\R^x(X)}{\phi_\epsilon^x} > 1-\epsilon
        \end{equation}
    \end{enumerate}
\end{proposition}

\begin{proof}
    This is Prop. \ref{prop:norm 1 equivalences} applied to the normalized POVM $\E_\R^x$ on the separable Hilbert space $\hirx$ over $\Gamma^x$.
\end{proof}

\begin{remark}
    Note that the second property has to do with pure states at every $x \in M$, so the right notion of states to work with at the level of the frame for localization purposes is $\framestate$, rather than $\normalframestate$.
\end{remark}

The notion of localizability is however more intricate than in the homogeneous space case. Indeed, if there is an $\alpha \in \Gamma$ such that $\alpha \notin \supp \R$, then one typically would not expect observables relativised with respect to any $\E_\R^{t(\alpha)}$ to be, even approximately, the absolute observable transported by $\alpha$. In spacetime (see Sec.~ \ref{sec:RQFT curved}), we can understand this as saying that we shouldn't expect to be able to approximate an absolute description of an observable localised in another galaxy if the QRF we have access to is only localised in ours. However, we now show that when both localisations concur, such a correspondence between relativised and absolute observables does exist.

\begin{proposition}
    \label{prop:norm 1 groupoid sequence}
    Let $\R$ be a localizable principal groupoid QRF for $\Gamma \rightrightarrows M$. For $\mu$-almost every $x \in M$ with $\beta_x \in \supp \E_\R^x$ such that $(\Gamma^x)$ is metrizable, there is a sequence of pure states $(\omega_n^x) \subset \framestatex$ with $\mu^{\E_\R^x}_{\omega_n^x} \to \delta_{\beta_x}$ weakly. 
\end{proposition}

\begin{proof}
    See App. \ref{app:proof norm 1 groupoid sequence}.
\end{proof}

\begin{remark}
    Note that in the group analogue of this proposition found in \cite{carette_operational_2025}, the assumption that $x \in \supp \E$ was omitted, and it was instead assumed simply that $x \in \Sigma$. This is not problematic provided either $x \in \supp \E$ or $\E$ is $G$-covariant and $\Sigma$ is a homogeneous $G$-space (making $\supp \E = \Sigma$).
\end{remark}

We now show that $\beta_x$ must indeed lie in $\supp \E_\R^x$ for there to be a localizing sequence at $\beta_x$.

\begin{proposition}
    \label{prop:beta not in supp E}
    Let $\R$ be a principal groupoid QRF for $\Gamma \rightrightarrows M$. For $\mu$-almost every $x \in M$ with $\beta_x \in \Gamma^x$ such that $\beta_x \notin \supp \E_\R^x$, if $(\Gamma^x)$ is metrizable then $\nexists (\omega_n^x) \subset \framestatex$ such that $\mu^{\E_\R^x}_{\omega_n^x} \to \delta_{\beta_x}$ weakly.
\end{proposition}

\begin{proof}
    If $\beta_x \notin \supp \E_\R^x$ (a closed set by Lemma \ref{lem:supp E fiber closed}), choose an open $B \ni \beta_x$ with $B \cap \supp \E^x_\R = \varnothing$; then $\E_\R^x(B) = 0$, so $\mu^{\E_\R^x}_{\omega^x}(B) = 0$ for every state $\omega_x \in \framestatex$. The open-set criterion of portmanteau at the open $B \ni \beta_x$ gives $1 = \delta_{\beta_x}(B) \leq \liminf_{n \to \infty} \mu^{\E_\R^x}_{\omega^x_n}(B) = 0$ which is a contradiction. Hence no sequence localizes at $\beta_x$.
\end{proof}

In particular, a localizing sequence centred at the unit $1_x$ exists iff $1_x \in \supp \E_\R^x$. From these results, we will recover in Sec.~ \ref{subsec:groupoid-relativization} an absolutist description in the localizing limit of a relational description.

\subsection{Relativization for continuous groupoids}
\label{subsec:groupoid-relativization}

Let $\S=(\scrhis,V)$ be a quantum system over the same base $M$, where $V$ is a fiber-wise ultraweakly continuous representation of $\Gamma$. We now define the groupoid relativization map associated with the frame $\R=(\scrhir,U,\E_\R)$.

\begin{proposition}
\label{def:groupoid-relativization-section4}
Let $A\in \A(\scrhis)$ and $\R$ be a $\mu$-admissible principal groupoid QRF. Then for $\mu$-almost every $x \in M$,\footnote{The right-hand side is an operator-valued integral of the type studied in \cite{glowacki_w-algebraic_2026}: it is well-defined as it is the operator-valued integration of a uniformly bounded ultraweakly measurable function quotiented by $\E_\R^x$-null sets: $(\beta \mapsto V_\beta A_{s(\beta)} V_\beta^\dagger) \in L^\infty_{\E_\R^x}(\Gamma^x,\bhsx)$.}
\begin{align}
    \euro^{\R}|_x : \A(\scrhis) &\to \bhsrx \\
    \euro^{\R}(A)_x &:= \int_{\Gamma^x} 
\bigl(
V_\beta\, A_{s(\beta)}\, V_\beta^\dagger
\bigr)\otimes d\E^x_\R(\beta)
\end{align}
defines an essentially bounded measurable field $(\euro^{\R}(A)_x) \in \A(\scrhis\boxtimes\scrhir)$. Likewise, if $\omega \in \framestate$, we have that for $\mu$-almost every $x \in M$,
\begin{align}
    \euro^{\R}_{\omega}|_x : \A(\scrhis) &\to \bhsx \\
    \euro^{\R}_{\omega}(A)_x &= \int_{\Gamma^x} 
\bigl(
V_\beta\, A_{s(\beta)}\, V_\beta^\dagger
\bigr) d\mu^{\E^x_\R}_{\omega_x}(\beta)
\end{align}
defines an essentially bounded measurable field $(\euro^{\R}_{\omega}(A)_x) \in \A(\scrhis)$.
\end{proposition}

\begin{proof}
    See App.~\ref{app:proof groupoid-relativization-section4}.
\end{proof}

\begin{remark}
    Note that if we were to use $\omega \in \normalframestate$ instead of $\omega \in \framestate$, it is not obvious that $(\euro^\R_{\omega_x}(A_x))$ would define a measurable field of operators (since we would not necessarily have $\esssup_{x \in M} \norm{\omega_x}_{\bhsx} < \infty$).
\end{remark}

\begin{remark}
\label{rem:groupoid-relativization-natural}
The formula above is the groupoid analogue of the usual relativization map for groups and torsors on each fiber of $\A(\scrhis)$. One integrates the transported system observable against the frame observable. The only novelty is that transport takes place along the arrows of the target fiber $\Gamma^x$, and the resulting operator lives in the fiber $\hisx\otimes \hirx$ (respectively $\hix$) of the system-frame composite (respectively system).
\end{remark}

\begin{remark}
    In \cite{bartlett_reference_2007}, the Canadian authors introduced a relativization map for compact groups and ideal frames using a \$ (probably a Canadian dollar) sign. The extension to non-ideal frames and to more general group structures led to replacing the \$ map by a $\yen$ map, since one of the authors of Ref.~\cite{miyadera_approximating_2016} is Japanese. Since we now further extend this relativization procedure beyond group structures on homogeneous spaces to groupoids and principal bundles, we propose the replacement of the $\yen$ map by the $\euro$ map (given the authors of the present paper are European).\footnote{It came to our attention that a recent paper \cite{ludescher_quantum_2026} also introduced a $\euro$ map, although for the relativization of quantum instruments for compact Lie groups in the context of the perspective-neutral approach to quantum reference frames \cite{de_la_hamette_perspective-neutral_2021}. In the context of the operational approach to quantum reference frames (i.e. the approach used in this paper), the relativization of quantum instruments has been discussed using the $\yen$ map \cite{jorquera_riera_relational_2026} rather than the $\euro$ map.}
\end{remark}

Since $(\euro^\R(A)_x)_{x \in M} \in \A(\scrhis \boxtimes \scrhir)$ and $(\euro^\R_{\omega}(A_x))_{x \in M} \in \A(\scrhis)$ are essentially bounded measurable fields, we can now understand direct integrals of those.

\begin{definition}
    Let $A=(A_x)_{x\in M}\in \A(\scrhis)$ and $\R$ be a $\mu$-admissible principal groupoid QRF.  The \emph{relativization map associated with} $\R$ is defined as
    \begin{align}
        \euro^\R : \A(\scrhis) &\to \A(\scrhis \boxtimes \scrhir) \\
    \euro^\R(A) &:= \int_M^\oplus \euro^\R(A)_x d\mu(x)
    \end{align}
    and the \emph{restricted relativization map}, given $\omega \in \framestate$, is
    \begin{align}
        \euro^\R_\omega : \A(\scrhis) &\to \A(\scrhis) \\
    \euro^\R_\omega(A) &:= \int_M^\oplus \euro^{\R}_{\omega}(A)_x d\mu(x)
    \end{align}
\end{definition}

\begin{definition}
    Let $\omega \in \framestate$. We define
    \begin{align}
        \Gamma_\omega : \A(\scrhis \boxtimes \scrhir) &\to \A(\scrhis) \\
        (\Gamma_\omega A)_x &:= \Gamma^x_{\omega_x}(A_x)
    \end{align}
    where $\Gamma_{\omega_x}^x : \bhsrx \to \bhsx$ is given as in Def. \ref{def:restriction-group}. Then $\euro^\R_\omega = \Gamma_\omega \circ \euro^\R$ and $\euro^\R_\omega|_x = \Gamma^x_{\omega_x} \circ \euro^\R|_x$.
\end{definition}

\begin{theorem}
\label{thm:groupoid-relativization-invariant}
Let $A\in \A(\scrhis)$, $\R$ be a $\mu$-admissible principal groupoid QRF and $\omega \in \framestate$. Then:
\begin{enumerate}
    \item $\euro^\R(A)|_x$, $\euro^\R_\omega(A)|_x$, $\euro^\R(A)$ and $\euro^\R_\omega(A)$ are independent of the choice of the representative of $A$.
    \item For $\mu$-almost every $x \in M$, $\euro^\R|_x$ is linear, normal, unital, completely positive, contractive, $*$-preserving, effect preserving.
    \item For $\mu$-almost every $x \in M$, $\Gamma^x_{\omega_x}$ is linear, normal, unital, completely positive, contractive, $*$-preserving and effect-preserving. 
    \item For $\mu$-almost every $x \in M$, $\euro^\R_\omega|_x$ is linear, normal, unital, completely positive, contractive, $*$-preserving, effect preserving.
    \item $\euro^\R$ is linear, normal, unital, completely positive, contractive, $*$-preserving, effect preserving.
    \item $\Gamma_\omega$ is linear, normal, unital, completely positive, contractive, $*$-preserving, effect preserving.
    \item $\euro^\R_\omega$ is linear, normal, unital, completely positive, contractive, $*$-preserving, effect preserving. 
    \item  $\euro^\R\bigl(\A(\scrhis)\bigr)
    \subseteq
    \A(\scrhis\boxtimes \scrhir)^\Gamma$.
    \item if $\R$ is sharp, then $\euro^\R|_x$ and $\euro^\R$ are multiplicative and hence normal $*$-homomorphisms.
    \item If $\R$ is localizable and $(\Gamma^x)$ is metrizable and for all $x \in M$, $1_x \in \supp \E_\R^x$ then $\euro^\R$ is injective (hence isometric) on the space of ultraweakly continuous operator fields.
\end{enumerate}
\end{theorem}

\begin{proof}
See App. \ref{app:proof properties euro}.
\end{proof}

\begin{remark}
    Note that we do not necessarily require the field of POVMs to be ultraweakly continuous here, though to prove results of the form $\euro^\R(\A(\scrhis)_{\mathrm{uw}}) \subset \A(\scrhis\boxtimes\scrhir)^\Gamma_\mathrm{uw}$ and $\euro^\R_\omega(\A(\scrhis)_{\mathrm{uw}}) \subset \A(\scrhis)_{\mathrm{uw}}$ this might become necessary. We leave these considerations to future work.
\end{remark}

\begin{proposition}
    \label{prop:invariant observable}
    Let $A \in \A(\scrhis)^\Gamma$ and $\R$ be a $\mu$-admissible principal groupoid QRF. Then for $\mu$-almost every $x \in M$,
    \begin{equation}
        \euro^\R(A)_x = A_x \otimes \id_{\bhrx}, \qquad \euro^\R(A) = A \boxtimes \id_{\A(\scrhir)} \equiv \int^\oplus_M (A_x \otimes \id_{\bhrx}) d\mu(x)
    \end{equation}
    and for all $\omega \in \framestate$,
    \begin{equation}
        \euro^\R_\omega(A)_x = A_x, \qquad \euro^\R_\omega(A) = A.
    \end{equation}
\end{proposition}

\begin{proof}
    We have that for $\mu$-almost every $x \in M$,
    \begin{multline}
        \euro^\R(A)_x = \int_{\Gamma^x} V_\beta A_{s(\beta)} V_\beta^\dagger \otimes d\E_\R^x(\beta) = \int_{\Gamma^x} A_x \otimes d\E_\R^x(\beta) = (A_x \otimes \id_{\bhrx}) \int_{\Gamma^x} \id_{\bhsx} \otimes d\E_\R^x(\beta) \\  = A_x \otimes \id_{\bhrx}
    \end{multline}
    since $t(\beta)=x$ and the last equality follows by the unitality of operator-valued integration \cite{glowacki_w-algebraic_2026}, and likewise for all $\omega \in \framestate$,
    \begin{equation}
        \euro^\R_\omega(A)_x = \int_{\Gamma^x} V_\beta A_{s(\beta)} V_\beta^\dagger d\mu^{\E_\R^x}_{\omega_x}(\beta) = \int_{\Gamma^x} A_x d\mu^{\E_\R^x}_{\omega_x}(\beta) = A_x \int_{\Gamma^x} d\mu^{\E_\R^x}_{\omega_x}(x) = A_x
    \end{equation}
     since $\mu^{\E_\R^x}_{\omega_x}(\Gamma^x) = 1$.
\end{proof}

Thus, invariant operators of the system are untouched by the relativization maps.

\begin{theorem}
    \label{thm:covariance euro}
    Let $A\in \A(\scrhis)$, $\R$ be a $\mu$-admissible principal groupoid QRF for $\Gamma \rightrightarrows M$ and $\omega \in \framestate$. Then for $\mu$-almost every $x \in M$ and arrow $\alpha : x \to z$ in $\Gamma$,
    \begin{equation}
        V_\alpha \euro^\R_\omega(A)_x V_\alpha^\dagger = \euro^\R_{\omega^\alpha}(A)_z
    \end{equation}
    where we write the transformed frame state as $\omega_z^\alpha := U_\alpha \omega_x U_\alpha^\dagger$ so that $\mu$-almost every $x \in M$, $\omega^\alpha_x = \begin{cases}
        \omega_x \text{ if } x \neq z \\
        \omega_z^\alpha \text{ if } x=z
    \end{cases}$.
\end{theorem}

\begin{proof}
    See App.~\ref{app:proof of covariance euro}.
\end{proof}

\begin{remark}
    Note that $V_\alpha$ does not act at the level of $\A(\scrhis)$, so we need a notion of unitary action at this direct integral level. This is given by bisections.
\end{remark}

\begin{corollary}
    \label{cor:invariance and covariance bisections euro}
    Let $\R$ be a $\mu$-admissible bisection-admissible principal (Lie) groupoid QRF for $\Gamma \rightrightarrows M$. Then $\euro^\R(\A(\scrhis)) \subseteq \A(\scrhis \boxtimes \scrhir)^{\mathrm{Bis}(\Gamma)}$, and for all $\omega \in \framestate$ and all $A \in \A(\scrhis)$ and all $b \in \mathrm{Bis}(\Gamma)$,
    \begin{equation}
        \mathbf{V}_b \euro^\R_{\omega}(A) \mathbf{V}_b^\dagger = \euro^\R_{\omega^b}(A)
    \end{equation}
    where $\omega^b = \left(\omega^b_x = U_{b(\varphi_b^{-1}(x))} \omega_{\varphi_b^{-1}(x)} U_{b(\varphi_b^{-1}(x))}^\dagger\right)$.
\end{corollary}

\begin{proof}
    The claim that $\euro^\R(\A(\scrhis)) \subseteq \A(\scrhis \boxtimes \scrhir)^{\mathrm{Bis}(\Gamma)}$ follows directly from Thm.~\ref{thm:groupoid-relativization-invariant} item 8.~and Lem.~\ref{lem:Gamma invariance implies bisection invariance}.

    Let $y=\varphi_b^{-1}(x)$. Conjugating any decomposable operator $B \in \A(\scrhis)$ by the bisection unitary gives $(\mathbf V_bB\mathbf V_b^\dagger)_x = V_{b(y)}B_yV_{b(y)}^\dagger$ for $\mu$-almost every $x \in M$. fiber-wise arrow covariance therefore gives, for any $A \in \A(\scrhis)$,
\begin{equation}
    (\mathbf V_b\euro_\omega^\R(A)\mathbf V_b^\dagger)_x =   V_{b(y)}\euro_{\omega}^\R(A)_yV_{b(y)}^\dagger \stackrel{\ref{thm:covariance euro}}{=} \euro_{\omega^{b}}^\R(A)_x
\end{equation}
which concludes the proof.
\end{proof}

We thus have a local statement (Thm.~\ref{thm:covariance euro}) and a global statement (Cor.~\ref{cor:invariance and covariance bisections euro}). The absolutist (localization) limit of relativised observables can also be expressed in the language of groupoid quantum reference frames, as follows.

\begin{theorem}[Fundamental theorem of operational groupoid quantum reference frames]
    \label{thm:euro localisation limit}
    Let $\R$ be a localizable $\mu$-admissible principal groupoid QRF for a continuous $\Gamma \rightrightarrows M$, where $(\Gamma^x)$ is metrizable, such that $\exists x \in M$ with $1_x \in \supp \E_\R^x$, and suppose that $(\omega_n^x) \in \mathscr{D}(\hirx)$ is a localizing sequence centred at $1_x$, with $\omega_n = (\omega_n^x) \in \framestate$.  Then for any $A \in \A(\scrhis)_{\mathrm{uw}}$, 
    \begin{equation}
        \label{eqn:localization euro fiber-wise}
        \lim_{n \to \infty} \euro_{\omega_n}^\R(A)_x = A_x \, ,
    \end{equation}
    where the limit is understood ultraweakly in $\bhsx$.
   Moreover, assume that there exists a sequence of measurable fiber-state
sections $\omega_n=(\omega_n^x)\in\framestate$ such that
\begin{equation}
  \mu_{\omega_{n,x}}^{\E_\R^x}\longrightarrow\delta_{1_x}
  \quad\text{weakly, for $\mu$-almost every }x.
\end{equation}
Then $\euro^\R_{\omega_n}(A)\to A$ ultraweakly in $\A(\scrhis)_{\mathrm{uw}}$.
\end{theorem}

\begin{proof}
    See App. \ref{app:proof yen localisation limit}.
\end{proof}

\begin{remark}
    Note that the use of continuous groupoids, as opposed to the more general measurable groupoids, is relevant precisely here: we make use of ultraweak continuity. The well-definedness of the integrals of relativization does not rely on such continuity requirements however. It may be that some physical contexts of interest, notably related to quantum gravity and Wigner's friend-type scenarios, require the use of non-continuous measurable groupoids, for which relativization maps are well-defined but do not reduce, even approximately, to any absolute description. 
\end{remark}

Note that we associate normal states $\rho \in \normalsystemstate$ of the system to operator fields of states $\omega \in \framestate$ of the frame; these are two different kinds of states. We have fiber-wise frame preparations giving us probability measures of frame configurations: these are \emph{preparation states} of the frame. Normal states are giving us probabilities of (relativised) observables: these are \emph{evaluation states}. This is mathematically consistent for the following reason.

\begin{proposition}
    \label{prop:joint normal state}
    Let $\rho \in \normalsystemstate$ and $\omega \in \framestate$. Then the fiber-wise tensor product
    \begin{equation}
        \rho \boxtimes \omega := \int_M^\oplus (\rho_x \otimes \omega_x) \, d\mu(x)
    \end{equation}
    is a normal state of $\A(\scrhis \boxtimes \scrhir)$: $\rho \boxtimes \omega \in \normalsystemframestate \subset \A(\scrhis \boxtimes \scrhir)_*$.
\end{proposition}

\begin{proof}
We first check measurability. Since $\rho\in\normalsystemstate$, the field $x \mapsto \rho_x \in \thsx$ is measurable. Since $\omega\in\framestate$, the field $x\mapsto\omega_x \in \framestatex$ is a measurable field of density operators on $\hirx$. Hence $x\mapsto \rho_x\otimes\omega_x \in \thsrx$ is measurable. Next, for $\mu$-almost every $x$, $\rho_x\geq0$, $\omega_x\geq0$ and therefore $\rho_x\otimes\omega_x\geq0$ as a trace-class operator in $\thsrx$. Moreover, $\norm{\rho_x\otimes\omega_x}_1 = \norm{\rho_x}_1\norm{\omega_x}_1$. Since $\omega_x$ is a density operator, $\norm{\omega_x}_1=\Tr_{\hirx}[\omega_x]=1$. Thus,
\begin{equation}
    \int_M \norm{\rho_x\otimes\omega_x}_1\,d\mu(x) = \int_M
\norm{\rho_x}_1\,d\mu(x) = \norm{\rho}_1 = 1.
\end{equation}
Therefore, $\rho\boxtimes\omega \in \A(\scrhis\boxtimes\scrhir)_*$. It remains to be shown that this predual element is a state. Let $B=(B_x)_{x\in M}\in\A(\scrhis\boxtimes\scrhir)$ be positive. Then $B_x\geq0$ for $\mu$-almost every $x$, and hence $\Tr_{\hisx\otimes\hirx}[(\rho_x\otimes\omega_x)B_x] \geq 0$ for $\mu$-almost every $x$. Therefore
\begin{equation}
(\rho\boxtimes\omega)(B) = \int_M \Tr_{\hisx\otimes\hirx}[(\rho_x\otimes\omega_x)B_x] \,d\mu(x) \geq 0.
\end{equation}
So $\rho\boxtimes\omega$ is a positive normal functional. Finally, evaluating on the identity field gives
\begin{multline}
    (\rho\boxtimes\omega)(\id_{\scrhis\boxtimes\scrhir}) =
\int_M
\Tr_{\hisx\otimes\hirx}[\rho_x\otimes\omega_x]\,d\mu(x)
\\ = \int_M \Tr_{\hisx}[\rho_x] \Tr_{\hirx}[\omega_x]\, d\mu(x) = \int_M \Tr_{\hisx}[\rho_x]\,d\mu(x) = \Tr_{\his}[\rho] =1.
\end{multline}
Thus $\rho\boxtimes\omega$ is a positive normal functional of norm $1$ on $\A(\scrhis\boxtimes\scrhir)$, i.e. $\rho\boxtimes\omega\in\normalsystemframestate$.
\end{proof}

Note however that not every product decomposable state in $\normalsystemframestate$ can be written as $\rho \boxtimes \omega$ where $\rho \in \normalsystemstate$ and $\omega \in \framestate$ (e.g. given $\eta \in \systemstate$ and $\kappa \in \normalframestate$ also gives $\eta \boxtimes \kappa \in \normalsystemframestate$). Here, we merely argue that restricting to the consideration of normal states of the joint algebra given by $\rho \boxtimes \omega$ where $\rho \in \normalsystemstate$ and $\omega \in \framestate$ is physically meaningful. We argue that computing probabilities of joint observables in $\euro^\R(\A(\scrhis))$ should be done using such normal states.  

\begin{remark}
    It is here important to distinguish the role of agents and of quantum reference frames. It should \emph{not} be assumed that agents \emph{are} quantum reference frames: agents may (or perhaps may not) always carry a quantum reference frame. This is a statement that depends on the choice of how to model agents. The operational quantum reference frames formalism is \emph{not} a formalism aimed at doing so, though it might help our understanding of the resources accessible to agents. Rather, one should take a perspectival view here: agents are assumed to be in the picture and giving a view of both the quantum system $\S$ and the quantum reference frame $\R$; it is assumed that the agent somehow (again, this is dependent on the model of agents) can prepare the quantum reference frame's state. The role of quantum reference frames is then to help agents extract information about the system $\S$, very much in the spirit of measurement schemes \cite{loveridge_relational_2020}. In this sense, the consideration of states $\rho \boxtimes \omega$ for $\rho \in \normalsystemstate$ and $\omega \in \framestate$ is natural.
\end{remark}

\subsection{An example: the canonical groupoid quantum reference frame}
\label{subsec:canonical groupoid qrf}

Definition \ref{def:groupoid-qrf-section4} is a definition, not an existence statement, and it is natural to ask whether generic curved spacetimes, which motivated the whole construction, actually carry groupoid quantum reference frames. We show in this subsection that they do, and that the relevant example is canonical: it is built out of the Haar system alone, requires no further choices, and has all of the good properties outlined in Definition \ref{def:groupoid-qrf-sharp-localizable}. It is the groupoid counterpart of the ideal frame $\hir = L^2(G)$ carried by the left regular representation together with the canonical projection-valued measure, which plays the analogous role in the group-based theory \cite{loveridge_symmetry_2018,carette_operational_2025}.

Throughout this subsection $(\Gamma\rightrightarrows M,\nu,\mu)$ is a second countable Hausdorff continuous groupoid that can be equipped with a continuous left Haar system $\nu=\{\nu^x\}_{x\in M}$ and a $\sigma$-finite quasi-invariant Radon measure $\mu$ on $M$, as in Sec.~\ref{subsec:continuous-groupoid-preliminaries}. Recall in particular that $\supp\nu^x=\Gamma^x$ for every $x\in M$.

\begin{definition}
\label{def:canonical-qrf}
The \emph{canonical quantum reference frame} of $(\Gamma\rightrightarrows M,\nu,\mu)$
is the triple
\begin{equation}
    \R_{\mathrm{can}} := \bigl(U^{\mathrm{can}},\E_{\R_{\mathrm{can}}},\scrhir^{\,\mathrm{can}}\bigr)
\end{equation}
defined as follows.
\begin{enumerate}
    \item The field of Hilbert spaces is the field of fibrewise square-integrable functions,
    \begin{equation}
        \label{eqn:canonical-hilbert-field}
        \hir^{x,\mathrm{can}} :=       L^2\bigl(\Gamma^x,\nu^x\bigr), \qquad        \scrhir^{\,\mathrm{can}} :=      \bigsqcup_{x\in M} L^2\bigl(\Gamma^x,\nu^x\bigr)\to M,
    \end{equation}
    whose continuous structure is the one generated by the restrictions $f|_{\Gamma^x}$ of functions $f\in C_c(\Gamma)$  \cite{renault_locally_1980,landsman_mathematical_1998}.
    \item The representation is the \emph{left regular representation} of $\Gamma$: for $\alpha:x\to y$,
    \begin{equation}
        \label{eqn:canonical-representation}
        U^{\mathrm{can}}_\alpha :        L^2\bigl(\Gamma^{x},\nu^{x}\bigr)        \longrightarrow     L^2\bigl(\Gamma^{y},\nu^{y}\bigr), \qquad
        \bigl(U^{\mathrm{can}}_\alpha f\bigr)(\gamma)  := f\bigl(\alpha^{-1}\circ\gamma\bigr),        \quad \gamma\in\Gamma^{y}.
    \end{equation}
    \item The frame observable is the family of \emph{canonical projection-valued measures}, acting by multiplication by indicator functions:
    \begin{equation}
        \label{eqn:canonical-povm}
        \E^x_{\R_{\mathrm{can}}}(\Delta) :=        M_{\mathbf 1_\Delta}, \qquad \bigl(M_{\mathbf 1_\Delta}f\bigr)(\beta) = \mathbf{1}_\Delta(\beta)\,f(\beta), \qquad   \Delta\in\Bor(\Gamma^x).
    \end{equation}
\end{enumerate}
\end{definition}

The formulas \eqref{eqn:canonical-representation} and \eqref{eqn:canonical-povm} should be read as saying that the canonical frame \emph{is} the totality of frame comparisons available at each point: its fibre Hilbert space is spanned by the arrows of $\Gamma^x$, and its frame observable is the sharp question ``which arrow?''. All of the structure is supplied by the
Haar system.

\begin{theorem}
\label{thm:canonical-qrf}
Let $(\Gamma\rightrightarrows M,\nu,\mu)$ be as above and $\nu$ be source-compatible for $\mu$. Then $\R_{\mathrm{can}}$ of Def.~\ref{def:canonical-qrf} is a principal groupoid quantum reference frame for $\Gamma$ in the sense of Def.~\ref{def:groupoid-qrf-section4}. Moreover, $\R_{\mathrm{can}}$ is \emph{sharp} (hence \emph{ideal} and localizable) and \emph{$\mu$-admissible} in the sense of Def.~\ref{def:groupoid-qrf-sharp-localizable}. If in addition $\Gamma$ is a Lie groupoid then $\R_{\mathrm{can}}$ is \emph{Lie}, and if $\Gamma$ is bisection-admissible then $\R_{\mathrm{can}}$ is \emph{bisection-admissible}. 
\end{theorem}

\begin{proof}
    See App.~\ref{app:proof canonical qrf}.
\end{proof}

The canonical frame is, as one would expect of an idealized frame, maximally spread out: it sees every arrow and localizes nowhere in particular.

\begin{proposition}
\label{prop:canonical-supports}
Let $\R_{\mathrm{can}}$ be as in Def.~\ref{def:canonical-qrf}. Then
\begin{equation}
    \supp\E^x_{\R_{\mathrm{can}}}=\Gamma^x    \quad\text{for } \mu\text{-almost every } x\in M, \qquad    \supp\R_{\mathrm{can}}=\Gamma .
\end{equation}
In particular $1_x\in\supp\E^x_{\R_{\mathrm{can}}}$ for $\mu$-almost every $x\in M$. If moreover $\Gamma$ is transitive, then for the source projection $\pi=s|_{\Gamma^{\bullet}}$ one has
\begin{equation}
    \supp\F^x_{\R_{\mathrm{can}},s}=M \quad\text{for } \mu\text{-almost every }x\in M, \qquad    \text{\emph{Loc}}_s(\R_{\mathrm{can}})=M .
\end{equation}
\end{proposition}

\begin{proof}
By Lemma \ref{lem:supp E fiber closed}, $\supp\E^x_{\R_{\mathrm{can}}}=\Gamma^x\smallsetminus O^x_{\max}$ where $O^x_{\max}$ is the union of the open sets on which the frame observable vanishes. If $O\subseteq\Gamma^x$ is open and nonempty then $\nu^x(O)>0$, because $\supp\nu^x=\Gamma^x$, and therefore $\E^x_{\R_{\mathrm{can}}}(O)=M_{\mathbf 1_O}\neq0$. Hence $O^x_{\max}=\varnothing$ and $\supp\E^x_{\R_{\mathrm{can}}}=\Gamma^x$. Taking the union over $x\in M$ gives $\supp\R_{\mathrm{can}}=\Gamma$, and $1_x\in\Gamma^x$ always.

Now suppose $\Gamma$ is transitive and let $N\subseteq M$ be open and nonempty. Transitivity gives $s(\Gamma^x)=M$, so $s^{-1}(N)\cap\Gamma^x$ is nonempty, and it is open because $s$ is continuous. By
\eqref{eqn:canonical-source-marginal} and the preceding paragraph, $\F^x_{\R_{\mathrm{can}},s}(N)\neq0$. Hence no nonempty open subset of $M$ is $\F^x_{\R_{\mathrm{can}},s}$-null, so $\supp\F^x_{\R_{\mathrm{can}},s}=M$, and $\text{Loc}_s(\R_{\mathrm{can}})=\bigcup_{x\in M}\supp\F^x_{\R_{\mathrm{can}},s}=M$.
\end{proof}

\begin{remark}
\label{rem:canonical-global-frame}
Prop.~\ref{prop:canonical-supports} says that the canonical frame realizes, on an \emph{arbitrary} curved spacetime, the condition $\text{Loc}_\pi(\R)=M$ which was interpreted in Sec.~\ref{subsec:localization-limit-groupoids} as the quantum counterpart of the existence of a global inertial reference frame. There is no tension with the discussion there: $\mathrm{Loc}_s(\R_\mathrm{can}) = M$ says only ``the canonical frame observable has access to arrows originating at every spacetime point". It is when $\pi$ is operationally meaningful (e.g. $\mathscr{B}$-compatible for a subgroup $\mathscr{B}$ of bisections) that one interprets such a statement as defining the operational localization of the QRF in spacetime.
\end{remark}

For the canonical frame the localizing sequences whose existence is postulated in Thm.~\ref{thm:euro localisation limit} can be written down explicitly, and crucially they can be chosen to depend measurably on the base point. This is the hypothesis needed for the global part of that theorem and for the injectivity statement of Thm.~\ref{thm:groupoid-relativization-invariant}.

\begin{proposition}
\label{prop:canonical-localizing-sequence}
Let $\Gamma\rightrightarrows M$ be a Lie groupoid with continuous Haar system $\nu$, and let $\R_{\mathrm{can}}$ be its canonical quantum reference frame. Then there exists a sequence $\omega_n=(\omega^x_n)_{x\in M}\in\framestate$ of measurable fields of pure fibre states such that
\begin{equation}
    \mu^{\E^x_{\R_{\mathrm{can}}}}_{\omega^x_n}    \longrightarrow \delta_{1_x} \qquad\text{weakly, for } \mu\text{-almost every }x\in M .
\end{equation}
\end{proposition}

\begin{proof}
    See App.~\ref{app:proof canonical localizing sequence}.
\end{proof}

Combining the above with the results of Sec.~\ref{subsec:groupoid-relativization} gives a clean structural statement about the canonical relativization map.

\begin{corollary}
\label{cor:canonical-relativization}
Let $(\Gamma\rightrightarrows M,\mu,\nu)$ be a Lie groupoid where $\nu$ is source-compatible for $\mu$, and let $\R_{\mathrm{can}}$ be its canonical quantum reference frame. Then for every system $\S=(\scrhis,V)$ over $M$ the relativization map
\begin{equation}
    \euro^{\R_{\mathrm{can}}} : \A(\his)    \longrightarrow    \A\bigl(\his\boxtimes\hir\bigr)^\Gamma
\end{equation}
is a normal unital $*$-homomorphism. Consequently $\A(\his)$ embeds as a von Neumann subalgebra of the algebra of $\Gamma$-invariant observables of the joint system-frame theory.
\end{corollary}

\begin{proof}
$\R_{\mathrm{can}}$ is $\mu$-admissible and principal by
Thm.~\ref{thm:canonical-qrf}, so Thm.~\ref{thm:groupoid-relativization-invariant} applies: $\euro^{\R_{\mathrm{can}}}$ is linear, normal, unital, completely positive and takes values in $\A(\his\boxtimes\hir)^\Gamma$. Since $\R_{\mathrm{can}}$ is sharp, item $9$ of that theorem gives multiplicativity, so $\euro^{\R_{\mathrm{can}}}$ is a normal unital $*$-homomorphism. 
\end{proof}

\begin{corollary}
\label{cor:canonical-general-relativistic-qrf}
Every oriented and time-oriented Lorentzian spacetime $(\M,g)$ admits an ideal, general relativistic quantum reference frame, namely the canonical quantum reference frame of its Poincar\'e groupoid $\mathrm{Poin}(\M,g)\rightrightarrows\M$ with respect to the volume measure $\mu_g$.
\end{corollary}

\begin{proof}
By Prop.~\ref{prop:poincare-transitive}, $\mathrm{Poin}(\M,g)$ is a transitive Lie groupoid. Prop.~\ref{prop:transitive groupoid implies nu compatible} and Thm.~\ref{thm:canonical-qrf} then produce a $\mu_g$-admissible principal, sharp groupoid QRF for $\mathrm{Poin}(\M,g)$, which is by definition a general relativistic QRF.
\end{proof}

Corollary \ref{cor:canonical-general-relativistic-qrf} is the existence statement that makes the constructions of Sec.~\ref{sec:RQFT curved} non-vacuous: relational quantum fields, relational covariance and the localization limit are available on \emph{every} Lorentzian spacetime, with no assumption whatsoever on its isometry group. This should be contrasted with the group-based formalism, which requires a transitive global symmetry action and therefore has essentially no instances outside the maximally symmetric case.

\begin{example}[Minkowski spacetime]
\label{ex:canonical-minkowski}
Let $(\mink,\eta)$ be Minkowski spacetime. By Prop.~\ref{prop:minkowski-action-groupoid}, $\mathrm{Poin}(\mink,\eta)\cong\Poincup\ltimes\mink$, and the target fibre over $x\in\mink$ is
\begin{equation}
    \mathrm{Poin}(\mink,\eta)^x = \bigl\{(g,g^{-1}\cdot x)\;\big|\;g\in\Poincup\bigr\} \cong \Poincup .
\end{equation}
Under this identification the canonical frame becomes the constant field
\begin{equation}
    \hir^{x,\mathrm{can}}\cong L^2\bigl(\Poincup\bigr), \qquad U^{\mathrm{can}}\;\text{the left regular representation of }\Poincup, \qquad    \E^x_{\R_{\mathrm{can}}}\;\text{the canonical PVM on }\Poincup,
\end{equation}
that is, exactly the ideal relativistic quantum reference frame of Sec.~\ref{sec:RQFT Minkowski}. The canonical groupoid construction therefore reduces, in the flat case, to the familiar ideal frame of the group-based theory, as it must.
\end{example}

\begin{remark}
\label{rem:canonical-vs-absolute-continuity}
The canonical frame is sharp, hence certainly not absolutely continuous with respect to any non-atomic measure on $\Gamma^x$, in accordance with Thm.~\ref{thm:no go absolute continuity}: localizable frames cannot have operator-norm bounded densities. This is the groupoid manifestation of the tension, discussed in Sec.~\ref{subsec:localization and continuity}, between sharp localizability and regularity of the frame observable. 
\end{remark}

\begin{remark}
\label{rem:canonical-caveats}
The canonical frame is an idealization: its fibre Hilbert space $L^2(\Gamma^x,\nu^x)$ is ``as large as the groupoid", and it should be thought of as a reference object rather than as a model of any physically prepared frame, much as $L^2(G)$ with the canonical PVM plays that role in the group case. 

The construction uses no dynamical input: the Haar system is a purely kinematical datum. Selecting physically preferred frames, for instance by restricting to a holonomy subgroupoid determined by the Levi-Civita connection, in the spirit of Remark \ref{rem:poincare-no-connection}, would produce smaller and more realistic frames.
\end{remark}


\subsection{Reduction to quantum reference frames for locally compact groups}
\label{subsec:reduction to group qrfs}

We now prove that the groupoid formalism reduces to the usual operational theory of QRFs when the Lie groupoid is an action groupoid. This is a key consistency check of the whole construction. We distinguish two cases: relativization on a locally compact second countable Hausdorff group $G$, and relativization on torsors of locally compact second countable Hausdorff groups. These two cases are often identified as they share many analytic similarities, but the reduction from the groupoid case is quite different in those two cases.

\subsubsection{Groupoid quantum reference frames on locally compact groups} 

We first start with the locally compact group case. Groupoids reduce to groups when the base space $M$ contains a single element $M=\{*\}$.

\begin{theorem}
    \label{thm:groupoid rel equiv group rel}
    Let $G$ be a locally compact second countable Hausdorff topological group and $\Gamma : G \rightrightarrows \{*\}$ be endowed with the Dirac measure $\mu=\delta_*$ on $M=\{*\}$. Then
    \begin{enumerate}
        \item $\R=(U,\E_\R=(\E_\R^x)_{x \in \{*\}},\scrhir=\{*\} \times \hir \to \{*\})$ is a $\delta_*$-admissible groupoid QRF for $\Gamma$ if and only if $(U_\R \equiv U,\E_\R\equiv \E_\R^*,\hir)$ is a group QRF for $G$.
        \item $\S = (V,\scrhis = \{*\} \times \his \to \{*\})$ is a $\Gamma$-system if and only if $(U_\S \equiv V, \his)$ is a $G$-system.
        \item $\homframestate = \normalframestate = \framestate$, $\homsystemstate = \normalsystemstate = \systemstate$.
        \item $\A(\scrhis) = \bhs$, $\A(\scrhis\boxtimes\scrhir) = \bhsr$ and $\A(\scrhis\boxtimes\scrhir)^\Gamma = \bhsr^G$.
        \item For all $A \in \bhs$ and all $\omega \in \homframestate$,
        \begin{equation}
            \euro^\R(A) = \yen^\R(A), \qquad \euro^\R_\omega(A) = \yen^\R_\omega(A).
        \end{equation}
    \end{enumerate}
\end{theorem}

\begin{proof}
    These statements are immediate from the fact that $V \equiv U_\S$ and $U \equiv U_\R$ are ultraweakly continuous unitary representations of $\Gamma=G$: fiber-wise ultraweak continuity becomes equivalent to ultraweak continuity when $M=\{*\}$. Moreover,  $\E_\R = \E_\R^*$ is a single POVM which is covariant with respect to $\Gamma=G$. No direct integral is performed, and consequently $\homframestate  = \normalframestate = \framestate$, $\homsystemstate = \normalsystemstate = \systemstate$,  $\A(\scrhis) =\bhs$ and $\A(\scrhir) = \bhr$. $\mu$-admissibility is automatic for $\mu=\delta_*$ on $M=\{*\}$. Furthermore, $\Gamma = G$ so trivially $\A(\scrhis \boxtimes \scrhir)^\Gamma = \bhsr^G$. Finally, the integrals of $\yen^\R(A)$ (resp. $\yen^R_\omega(A)$) are now the same as those of $\euro^\R(A)$ (resp. $\euro^\R_\omega(A)$), which completes the proof.
\end{proof}

Thus, in this locally compact group case, we recover exactly the standard relativization map: this goes both ways. We can indeed start from a groupoid QRF and recover a group QRF when the groupoid is a group with $M=\{*\}$, with the groupoid relativization map $\euro^\R$ (respectively $\euro^\R_\omega$) reducing to the group relativization map $\yen^\R$ (respectively $\yen^\R_\omega$). We can otherwise start from a group QRF and recover a groupoid QRF for $\Gamma = G \rightrightarrows \{*\}$, with the group relativization map $\yen^\R$ (respectively $\yen^\R_\omega$) becoming the groupoid relativization map $\euro^\R$ (respectively $\euro^\R_\omega$). As we now see, the situation is slightly different for the case of torsor QRFs.

\subsubsection{Groupoid quantum reference frames on torsors} 
\label{subsubsec:groupoid QRFs on torsors}

We now consider the case where $\Sigma \cong G$ is a topological left $G$-torsor for a locally compact second-countable Hausdorff topological group $G$. Any such $G$-torsor can be seen as a principal $H$-bundle $q : \Sigma \to M = \Sigma/H$, where $H \leq G$ is an admissible\footnote{A subgroup $H$ of $G$ is called \emph{admissible} if the quotient map $G \to G/H$ is a principal $H$-bundle. For example, any subgroup of a discrete group is admissible, and any closed subgroup of a Lie group is admissible \cite{mitchell_notes_2006}.} closed subgroup. Choosing $H$ when several choices are available amounts to a physical choice. This choice of $H$ is therefore non-canonical, but given a torsor there is always an $H$ for which the reduction is meaningful. 

Choose an origin $o \in \Sigma$. Every $p \in \Sigma$ has a unique coordinate $g_p \in G$ such that $p=g_p \cdot o$. Define $\theta_o : G \to \Sigma$ as the $G$-equivariant homeomorphism $\theta_o(g) = g \cdot o$. Define the right $H$-action on $\Sigma$ by $p \cdot h = g_p h \cdot o$. This action is free, and the quotient $q : \Sigma \to M := \Sigma/H$ is a principal $H$-bundle. Note that if $H=G$ then $M=\{*\}$, while if $H=\{e\}$ then $M=\Sigma$. 

The Atiyah groupoid is $\Gamma = \At(\Sigma \to M) = (\Sigma \times \Sigma)/H \rightrightarrows M$. An arrow $[p_t,p_s]$ has source $s([p_t,p_s])=q(p_s)$ and target $t([p_t,p_s]) = q(p_t)$. There is a groupoid isomorphism $\At(\Sigma \to \Sigma/H) \cong G \ltimes G/H$ given, after choosing $o$, by
\begin{equation}
    [g_1 \cdot o, g_2 \cdot o] \mapsto (g_1 g_2^{-1},g_2 H).
\end{equation}
This is an action groupoid over $M=\Sigma/H$. The target fibers are indexed by $x \in M = \Sigma/H$. For each $x \in M$, define
\begin{align}
    \Pi_x^o : \Gamma^x & \to \Sigma \\
    \beta = [p_t,p_s] &\mapsto g_{p_t}g^{-1}_{p_s} \cdot o, \qquad q(p_t) = x.
\end{align}
The map $\Pi^o_x$ is a Borel isomorphism: under $\At(\Sigma \to G/H) \cong G \ltimes G/H$, one has $\At(\Sigma \to G/H)^x = \{(g,g^{-1}x) \mid g \in G\}$. Then $\Pi^o_x(g,g^{-1}x) = g \cdot o$ is a homeomorphism with inverse $g \cdot o \mapsto (g,g^{-1}x)$.  Even though the direct integral is over $M=\Sigma/H$, each target fiber $\Gamma^x$ is identified with the full torsor $\Sigma$, not with $\Sigma/H$. Thus, a torsor POVM on $\Sigma$ can be pulled to a POVM on every $\Gamma^x$.

Moreover, for an arrow $\alpha = [p_t,p_s]$, define $\kappa(\alpha) := g_{p_t}g_{p_s}^{-1} \in G$ and set $U_\alpha := U_{\R}(\kappa(\alpha))$. $\kappa([p_t,p_s])$ is the unique $g \in G$ such that $p_t = g \cdot p_s$. Note that replacing $[p_t,p_s]$ by $[p_t \cdot h,p_s \cdot h]$ for $h \in H$ leaves $g_{p_t}g_{p_s}^{-1}$ unchanged. This relates unitary group representations to the associated action groupoid unitary representation.

Furthermore, we can embed the system observables as constant operator fields over $M$: let
\begin{align}
    \iota_\S : \bhs &\to \A(\scrhis) \\
    \iota_\S(A)_x &= A \qquad \forall x \in M,
\end{align}
and likewise we can raise the frame observables as constant operator fields:
\begin{align}
    \iota_\R : \bhr &\to \A(\scrhir) \\
    \iota_\R(C)_x &= C \qquad \forall x \in M,
\end{align}
and, for joint observables,
\begin{align}
    \iota_{\S\R} : \mathcal{B}(\his \otimes \hi_{\R}) &\to \A(\scrhis\boxtimes\scrhir) \\
    \iota_{\S\R}(B)_x &= B \qquad \forall x \in M.
\end{align}
This allows us to recover groupoid QRFs from torsor QRFs as follows.

\begin{theorem}
    \label{thm:torsor reduction}
    Let $G$ be a locally compact second countable Hausdorff group and $\R = (U_\R,\E_\R,\hir)$ be a $G$-covariant QRF over $\Sigma \cong G$ for $o \in \Sigma$. Let $H\leq G$ be an admissible closed subgroup, and let $q:\Sigma\to M:=\Sigma/H$ be the principal $H$-bundle defined above, and let $\Gamma:=\At(\Sigma\to M)$. Let $m_G$ be a left Haar measure on $G$, let $\nu_\Sigma:=(\theta_o)_*m_G$ be the corresponding invariant measure on the torsor $\Sigma$, and let $\mu_M$ be a sigma-finite quasi-invariant Radon measure on $M\cong G/H$. Then $\underline{\R} = (U,\E_{\underline{\R}},\scrhi_{\underline{\R}})$ is a $\mu_M$-admissible\footnote{The direct integrals defining the groupoid theory are taken with respect to $\mu_M$.} principal groupoid QRF, where
    \begin{enumerate}
        \item $\scrhi_{\underline{\R}} = M \times \hir \to M$,
        \item $\forall \alpha \in \Gamma$ , $U_\alpha = U_\R(\kappa(\alpha))$,
        \item $\forall x \in M$, $\E_{\underline{\R}}^x(\Delta) := \E_\R(\Pi^o_x(\Delta))$ for all $\Delta \in \Bor(\Gamma^x)$.
    \end{enumerate}
    Moreover, let $\S=(V = (V_\alpha = U_\S(\kappa(\alpha))),\scrhis = M \times \his \to M)$ be a $\Gamma$-system. Then for all $\omega \in \homframestate$
    \begin{equation}
        \euro^{\underline{\R}} \circ \iota_\S = \iota_{\S\R} \circ \yen^\R, \qquad \euro^{\underline{\R}}_{\underline{\omega}} \circ \iota_\S = \iota_\S \circ \yen^\R_\omega
    \end{equation}
    where $\underline{\omega} = \iota_\R (\omega) \in \mathscr{S}(\scrhi_{\underline{\R}})$. That is, the following diagram commutes:
\[\begin{tikzcd}[sep=huge]
	{\mathcal{B}(\mathcal{H}_\mathcal{S})} & {\A(\scrhis)} \\
	{\mathcal{B}(\mathcal{H}_\mathcal{S} \otimes \mathcal{H}_\mathcal{R})} & {\mathcal{A}(\scrhis\boxtimes\scrhir)} \\
	{\mathcal{B}(\mathcal{H}_\mathcal{S})} & {\A(\scrhis)}
	\arrow["{\iota_\mathcal{S}}", from=1-1, to=1-2]
	\arrow["{\yen^{\mathcal{R}}}"', from=1-1, to=2-1]
	\arrow["{\yen^{\mathcal{R}}_{\omega}}"', curve={height=40pt}, from=1-1, to=3-1]
	\arrow["{\euro^{\underline{\R}}}", from=1-2, to=2-2]
	\arrow["{\euro^{\underline{\R}}_{\underline{\omega}}}", curve={height=-40pt}, from=1-2, to=3-2]
	\arrow["{\iota_{\mathcal{S}\mathcal{R}}}", from=2-1, to=2-2]
	\arrow["{\Gamma_\omega}"', from=2-1, to=3-1]
	\arrow["{\Gamma_{\underline{\omega}}}", from=2-2, to=3-2]
	\arrow["{\iota_\mathcal{S}}"', from=3-1, to=3-2]
\end{tikzcd}\]
\end{theorem}

\begin{proof}
    See App.~\ref{app:proof torsor reduction}.
\end{proof}

Note that for the trivial bundle one has the canonical von Neumann algebra identifications
\begin{align}
 \A(M \times \his)
   &\cong L^\infty(M,\mu_M)\,\bar\otimes\,\B(\his),\\
 \A((M \times \his) \boxtimes (M \times \hir))
   &\cong L^\infty(M,\mu_M)\,\bar\otimes\,
          \B(\his\otimes\hir),
\end{align}
where $\bar{\otimes}$ is the spatial tensor product, with corresponding direct-integral Hilbert spaces
\begin{align}
    \int_M^\oplus\hi_\S\,d\mu_M(x) = L^2(M,\mu_M;\his)
 &\cong L^2(M,\mu_M)\otimes\hi_\S, \\
 \hisr = \int^\oplus_M (\his \otimes \hir) \, d\mu_M(x) = L^2(M,\mu_M; \his\otimes\hir) &\cong L^2(M,\mu_M) \otimes (\his \otimes \hir).
\end{align}
The isomorphisms follow since, on a direct-integral Hilbert space $L^2(M,\mu_M) \otimes \hi$, there is a natural unitary\footnote{The (unique extension of this) map is unitary since \begin{equation}
    \braket{J(f \otimes \xi)}{J(g \otimes \eta)}_{L^2(M,\mu_M;\hi)} = \int_M \overline{f(x)}g(x) \braket{\xi}{\eta}_\hi \, d\mu_M(x) = \braket{f}{g}_{L^2(M,\mu_M)} \braket{\xi}{\eta}_\hi
\end{equation} i.e. $J$ is isometric and its range is dense so it extends uniquely to a unitary.}
\begin{align}
    \label{eqn:J unitary}
    J : L^2(M,\mu_M) \otimes \hi &\to L^2(M,\mu_M;\hi) \\
    J(f \otimes \xi)(x) &= f(x) \xi.
\end{align}
For $A \in \bh$ and the associated constant operator field $\underline{A} := \int^\oplus_M A \, d\mu_M(x) \in \A(M \times \hi)$, we have
\begin{equation}
    \underline{A} J (f \otimes \xi)(x) = A(f(x) \xi) = f(x) A \xi = J(f \otimes A\xi)(x) \Longrightarrow J^* \underline{A} J = \mathbb{1}_{L^2(M,\mu_M)} \otimes A.
\end{equation}
Consider the $*$-isomorphism
\begin{align}
\label{eqn:normal * hom}
    \Theta : L^\infty(M,\mu_M) \bar\otimes \bh &\to \A(M \times \hi) \\
    \Theta(f \otimes A) &= \int^\oplus_M f(x) A \, d\mu_M(x).
\end{align}
Writing
\begin{align}
    \Theta_{\S\R} : L^\infty(M,\mu_M) \bar\otimes \bhsr &\to \A(M \times (\his \otimes \hir)) \\
    \Theta_\S : L^\infty(M,\mu_M) \bar\otimes \bhs &\to \A(M \times \his)
\end{align}
defined as in Eqn.~\eqref{eqn:normal * hom}, note that the constant-field embeddings are 
\begin{equation}
 \iota_\S(A)=\Theta_\S(\mathbf 1_{L^\infty(M)}\otimes A),
 \qquad
\iota_{\S\R}(B)=\Theta_{\S\R}(\mathbf 1_{L^\infty(M)}\otimes B),
\end{equation}
for $A \in \bhs$ and $B \in \bhsr$, where $\mathbf 1_{L^\infty(M)}$ is the unit of the commutative von Neumann algebra $L^\infty(M,\mu_M)$. To relate this $\mathbf{1}_{L^\infty(M,\mu_M)}$ to $\mathbb{1}_{\mathcal{B}(L^2(M,\mu_M))}$, note that
\begin{equation}
    J^* \Theta(f\otimes A) J = M_f \otimes A \in \mathcal{B}(L^2(M,\mu_M) \otimes \hi)
\end{equation}
where $(M_f \psi)(x) = f(x) \psi(x)$ for $\psi \in L^2(M,\mu_M)$ is the multiplication operator. Thus,
\begin{equation}
    \mathbf{1}_{L^\infty(M,\mu_M)} \otimes A \stackrel{\Theta}{\longmapsto} \underline{A} \stackrel{\mathrm{Ad}_{J^*}}{\longmapsto} \mathbb{1}_{L^2(M,\mu_M)} \otimes A.
\end{equation}
Writing \begin{align}
    J_{\S\R} : L^2(M,\mu_M) \otimes (\his \otimes \hir) &\to L^2(M,\mu_M; \his \otimes \hir) \\
    J_{\S} : L^2(M,\mu_M) \otimes \his &\to L^2(M,\mu_M;\his)
\end{align}
defined as in Eqn.~\eqref{eqn:J unitary}, we thus have that $\forall A \in \bhs$ and $\forall \omega \in \homframestate$,
\begin{align}
    \euro^{\underline{\R}}(\iota_\S(A)) = \iota_{\S\R}(\yen^\R(A)) &= \int^\oplus_M \yen^R(A) \, d\mu_M(x), \\ &= J_{\S\R} \left(\mathbb{1}_{L^2(M,\mu_M)} \otimes \yen^\R(A)\right) J_{\S\R}^* \\ &= \Theta_{\S\R}\left(\mathbf{1}_{L^\infty(M,\mu_M)} \otimes \yen^\R(A)\right)
    \\
    \euro^{\underline{\R}}_{\underline{\omega}}(\iota_\S(A)) = \iota_{\S}(\yen^\R_\omega(A)) &= \int^\oplus_M \yen^R_\omega(A) \, d\mu_M(x), \\ &= J_{\S} \left(\mathbb{1}_{L^2(M,\mu_M)} \otimes \yen^\R_\omega(A)\right) J_{\S}^* \\ &= \Theta_{\S}\left(\mathbf{1}_{L^\infty(M,\mu_M)} \otimes \yen^\R_\omega(A)\right)  
\end{align}
Note that the identity factor is not an additional physical observable: it records that the same group-based operator is embedded in every base fiber of the groupoid theory. That is, ordinary torsor relativization is the constant operator field sector of groupoid relativization. Importantly, this is an embedding of the torsor theory into the groupoid theory for action groupoids, rather than equivalence between the two. This is different to the group case, where both pictures are equivalent when $\Gamma$ is a locally compact second countable Hausdorff topological group $G$.

We can indeed now have more general groupoid QRFs in this action groupoid case. For example, we can have $G$-covariance across the Hilbert bundle linking different fibers $\hirx \to \hir^y$, rather than on a fixed fiber $\hirx$. In the special case where the bundle $\scrhir$ is built from constant fibers $\hir$, and the representation $U_{(g,x)} : \hirx \to \hir^y$ fully reconstructs the representation $U_\R : G \to \U(\hir)$, we can recover a torsor QRF from the action groupoid QRF as in Thm.~\ref{thm:torsor reduction}; more generally, we do not.

Likewise, at the level of the relativization map, we only recover the usual $\yen$ maps for constant operator fields $A_x = A$ for $\mu$-almost every $x \in M$, in which case $(\euro^{\underline{\R}}(\iota_\S(A)))_x = \yen^\R(A)$ independently of $x$. More generally, even if the action groupoid QRF $\underline{\R}$ is built from torsor QRF $\R$ as in Thm.~\ref{thm:torsor reduction} and $V_\alpha = U_\S(\kappa(\alpha))$, we have that for $A \in \A(\scrhis)$ and $\omega \in \homframestate$ with $\underline{\omega} := \iota_\R(\omega)$,
\begin{align}
    (\euro^{\underline{\R}}(A))_x &= \int_\Sigma U_\S(g_{\eta,o})  A_{g_{\eta,o}^{-1} \cdot x} U_\S(g_{\eta,o})^\dagger \, \otimes d\E_\R(\eta), \\
    (\euro^{\underline{\R}}_{\underline{\omega}}(A))_x &= \int_\Sigma U_\S(g_{\eta,o})  A_{g_{\eta,o}^{-1} \cdot x} U_\S(g_{\eta,o})^\dagger \,  d\mu^{\E_\R}_{\omega}(\eta) .
\end{align}
This is now $x$-dependent. Thus, the groupoidal approach to quantum reference frames extends the torsor formalism even in the case of action groupoids. Only for certain classes of action groupoid QRFs and for constant operator fields of the system do we retrieve the standard relativization map given by $\yen$.

\section{Foundations of relational quantum field theory on Minkowski spacetime}
\label{sec:RQFT Minkowski}

In the rest of this paper we apply the groupoid QRF formalism to relational quantum field theory. We proceed in three stages. First, in this section, we specialize to Minkowski spacetime, where we review the framework introduced in \cite{fedida_foundations_2025} and link it to Wightman quantum field theory and Algebraic quantum field theory. Next, in Sec.~\ref{sec:RQFT curved}, we will pass to generic curved spacetimes, where the relational field is defined intrinsically on the Poincar\'e groupoid without any global trivialization of the frame bundle. Finally, in Sec.~\ref{sec:towards RQGFT}, we will indicate how the same ideas extend to internal gauge symmetry via Atiyah groupoids.

We start by examining the notion of quantum reference frames on the space of oriented inertial reference frames on Minkowski spacetime $O_+^\uparrow(\mink,\eta)$. 

\begin{definition}[\cite{fedida_foundations_2025}]
    We call a $\Poincup$-covariant QRF based on $O_+^\uparrow(\mink,\eta)$ a \emph{relativistic QRF}.
\end{definition}

A relativistic QRF lives on $O_+^\uparrow(\mink,\eta)$, a trivial Lorentz bundle over Minkowski (see Fig. \ref{fig:Visualisation bundle}). For a frame state $\omega\in \homframestate$, we denote by
\begin{equation}
\mu_\omega^{\E_\R}(W):=\Tr\!\bigl[\omega\,\E_\R(W)\bigr], \qquad W \in \Bor(O_+^\uparrow(\mink,\eta))
\end{equation}
the corresponding Born probability measure on $O_+^\uparrow(\mink,\eta)$. Note that tetrads are points $(x,\lambda) \in O_+^\uparrow(\mink,\eta)$ and are localised in $\pi(x,\lambda) \in \mink$. We can then likewise discuss the projection of the support of $\mu^{\E_\R}_{\omega}$ onto $\mink$ to discuss the localization of a QRF $\R$ prepared in a state $\omega$.

\begin{definition}
    Let $\R$ be a relativistic QRF and $\omega \in \homframestate$. Let
\begin{equation}
\F_\R:\Bor(\mink)\to \mathcal E(\hir),
\qquad
\F_\R(\mathcal{V}):=\E_\R(\mathcal{V}\times \Lup),
\end{equation}
be the spacetime marginal of the frame observable, and write
\begin{equation}
\mu_\omega^{\F_\R}(U):=\Tr[\omega\,\F_\R(U)]
\end{equation}
for the induced probability measure on spacetime. We call
    \begin{equation}
        \text{Loc}(\R,\omega) := \supp \mu^{\F_\R}_{\omega}
    \end{equation}
    the \emph{spacetime localization of $(\R,\omega)$}.
\end{definition}

It is easily seen that $\F_\R$ is Poincaré covariant \cite{fedida_foundations_2025}:
\begin{equation}
    U_\R(a,\Lambda) \F_\R(\mathcal{V}) U_\R(a,\Lambda)^\dagger = \F_\R((a,\Lambda) \cdot \mathcal{V}) \qquad \forall (a,\Lambda) \in \Poincup, \mathcal{V} \in \Bor(\mink).
\end{equation}

\begin{remark}
    Note the similarity between this notion of $\text{Loc}(\R,\omega)$ and that given in \eqref{eqn:Loc(R,w)}. However, the notion of localization given here is canonical: this is because there exists a global trivialization of the principal bundle, and thus a canonical $\pi$. Further note that the covariance of $\F_\R$ indicates that this $\pi$ is $\Poincup$-compatible. Note moreover that, for a relativistic QRF, $\text{Loc}(\R) = \bigcup_{\omega \in \homframestate} \text{Loc}(\R,\omega) = \mink$ by the Poincaré covariance of $\F_\R$. This is the QRF analogue of the fact that there exists a global inertial reference frame on Minkowski. 
\end{remark}

\begin{figure*}[t!]
    \centering
    \begin{subfigure}[t]{0.4\textwidth}
    \begin{tikzpicture}[x=0.75pt,y=0.75pt,yscale=-1,xscale=1]

\draw   (209.56,47.56) .. controls (219.21,39.02) and (439.36,31.69) .. (431.16,47.06) .. controls (422.96,62.43) and (430.68,190.51) .. (430.68,212.71) .. controls (430.68,234.91) and (198.64,236.62) .. (188.99,211) .. controls (179.34,185.39) and (199.91,56.1) .. (209.56,47.56) -- cycle ;
\draw    (458.79,42.79) -- (458.79,221.25) ;
\draw    (188.03,263.09) .. controls (207.5,244.66) and (415.24,281.87) .. (434.53,256.26) ;
\draw  [fill={rgb, 255:red, 0; green, 0; blue, 0 }  ,fill opacity=1 ] (301.22,134.21) .. controls (301.22,132.56) and (301.98,131.22) .. (302.91,131.22) .. controls (303.85,131.22) and (304.6,132.56) .. (304.6,134.21) .. controls (304.6,135.87) and (303.85,137.21) .. (302.91,137.21) .. controls (301.98,137.21) and (301.22,135.87) .. (301.22,134.21) -- cycle ;
\draw   (302.39,110.16) .. controls (302.39,100.97) and (308.84,93.52) .. (316.8,93.52) .. controls (324.75,93.52) and (331.2,100.97) .. (331.2,110.16) .. controls (331.2,119.35) and (324.75,126.8) .. (316.8,126.8) .. controls (308.84,126.8) and (302.39,119.35) .. (302.39,110.16) -- cycle ;
\draw    (328.1,100.33) -- (316.8,110.16) ;
\draw [line width=3.75]    (280.06,122.53) -- (300.01,122.53) ;
\draw [line width=3.75]    (280.06,124.24) -- (280.06,98.63) ;
\draw [line width=3.75]    (280.06,122.53) -- (294.1,107.16) ;

\draw (439.07,245.66) node [anchor=north west][inner sep=0.75pt]    {$\mink$};
\draw (445.46,20) node [anchor=north west][inner sep=0.75pt]    {$\mathcal{SO}_+^\uparrow(1,d-1)$};
\draw (145.59,62.94) node [anchor=north west][inner sep=0.75pt]    {$O_+^\uparrow(\mink,\eta)$};

\end{tikzpicture}
    \caption{A classical inertial reference frame, which can be thought of as a clock and rods, sharply localised in spacetime and with a definite Lorentz orientation.}
    \label{fig:Visualisation pointwise}
    \end{subfigure}
    \qquad \qquad \quad
    \begin{subfigure}[t]{0.4\textwidth}
        \begin{tikzpicture}[x=0.75pt,y=0.75pt,yscale=-0.95,xscale=0.8]
\draw   (202.6,45.39) .. controls (213.62,36.76) and (465.17,29.36) .. (455.8,44.89) .. controls (446.43,60.42) and (455.25,189.83) .. (455.25,212.26) .. controls (455.25,234.69) and (190.12,236.42) .. (179.09,210.54) .. controls (168.07,184.65) and (191.57,54.02) .. (202.6,45.39) -- cycle ;
\draw    (491.79,40.58) -- (491.79,220.89) ;
\draw    (177.99,263.16) .. controls (200.24,244.54) and (437.61,282.14) .. (459.66,256.26) ;
\draw   (311.75,74.28) .. controls (311.75,65) and (318.82,57.47) .. (327.55,57.47) .. controls (336.28,57.47) and (343.36,65) .. (343.36,74.28) .. controls (343.36,83.57) and (336.28,91.1) .. (327.55,91.1) .. controls (318.82,91.1) and (311.75,83.57) .. (311.75,74.28) -- cycle ;
\draw    (339.95,64.35) -- (327.55,74.28) ;
\draw [line width=3.75]    (287.36,86.79) -- (309.97,86.79) ;
\draw [line width=3.75]    (287.36,88.51) -- (287.36,62.63) ;
\draw [line width=3.75]    (287.36,86.79) -- (303.27,71.26) ;
\draw  [pattern=_nikqpx0v0,pattern size=6pt,pattern thickness=0.75pt,pattern radius=0pt, pattern color={rgb, 255:red, 0; green, 0; blue, 0}][dash pattern={on 4.5pt off 4.5pt}] (259.02,109.42) .. controls (270.04,100.8) and (391.31,95.62) .. (410.6,111.15) .. controls (429.89,126.68) and (378.08,174.99) .. (328.47,141.34) .. controls (278.86,107.7) and (226.5,209.5) .. (215.47,183.62) .. controls (204.45,157.74) and (248,118.05) .. (259.02,109.42) -- cycle ;
\draw  [fill={rgb, 255:red, 0; green, 0; blue, 0 }  ,fill opacity=1 ] (312.29,105.86) .. controls (312.29,104.24) and (313.12,102.93) .. (314.16,102.93) .. controls (315.19,102.93) and (316.02,104.24) .. (316.02,105.86) .. controls (316.02,107.47) and (315.19,108.78) .. (314.16,108.78) .. controls (313.12,108.78) and (312.29,107.47) .. (312.29,105.86) -- cycle ;
\draw   (389.17,186.29) .. controls (380.65,184.94) and (374.75,176.4) .. (376,167.21) .. controls (377.25,158.02) and (385.17,151.67) .. (393.69,153.02) .. controls (402.21,154.37) and (408.11,162.91) .. (406.86,172.1) .. controls (405.61,181.29) and (397.69,187.64) .. (389.17,186.29) -- cycle ;
\draw    (384.09,155.16) -- (391.43,169.66) ;
\draw [line width=3.75]    (363.37,182.37) -- (374.05,158.66) ;
\draw [line width=3.75]    (364.61,183.49) -- (346.03,166.68) ;
\draw [line width=3.75]    (363.37,182.37) -- (359.73,155.6) ;
\draw  [fill={rgb, 255:red, 0; green, 0; blue, 0 }  ,fill opacity=1 ] (374.48,143.23) .. controls (374.48,141.61) and (375.32,140.31) .. (376.35,140.31) .. controls (377.38,140.31) and (378.22,141.61) .. (378.22,143.23) .. controls (378.22,144.85) and (377.38,146.15) .. (376.35,146.15) .. controls (375.32,146.15) and (374.48,144.85) .. (374.48,143.23) -- cycle ;
\draw   (263.66,167.08) .. controls (271.85,163.02) and (281.72,166.51) .. (285.7,174.86) .. controls (289.68,183.21) and (286.26,193.27) .. (278.07,197.33) .. controls (269.87,201.39) and (260,197.9) .. (256.02,189.55) .. controls (252.04,181.19) and (255.46,171.13) .. (263.66,167.08) -- cycle ;
\draw    (285.28,189.73) -- (270.86,182.2) ;
\draw [line width=3.75]    (234.47,194.05) -- (247.53,211.64) ;
\draw [line width=3.75]    (233.36,195.18) -- (250.07,178.21) ;
\draw [line width=3.75]    (234.47,194.05) -- (253.68,196.24) ;
\draw  [fill={rgb, 255:red, 0; green, 0; blue, 0 }  ,fill opacity=1 ] (238.08,171.6) .. controls (238.08,169.99) and (238.92,168.68) .. (239.95,168.68) .. controls (240.98,168.68) and (241.82,169.99) .. (241.82,171.6) .. controls (241.82,173.22) and (240.98,174.53) .. (239.95,174.53) .. controls (238.92,174.53) and (238.08,173.22) .. (238.08,171.6) -- cycle ;
\draw    (233.11,91.72) -- (267.88,125.02) ;
\draw [shift={(270.04,127.09)}, rotate = 223.77] [fill={rgb, 255:red, 0; green, 0; blue, 0 }  ][line width=0.08]  [draw opacity=0] (8.93,-4.29) -- (0,0) -- (8.93,4.29) -- cycle    ;
\draw  [dash pattern={on 0.84pt off 2.51pt}]  (213.71,172.51) -- (213.71,257.17) ;
\draw  [dash pattern={on 0.84pt off 2.51pt}]  (415.04,117.17) -- (415.71,263.84) ;
\draw  [pattern=_ng5lqn19c,pattern size=6pt,pattern thickness=0.75pt,pattern radius=0pt, pattern color={rgb, 255:red, 0; green, 0; blue, 0}] (258.37,256.64) .. controls (280.37,256.64) and (317.71,259.31) .. (340.37,260.64) .. controls (363.04,261.97) and (399.04,253.31) .. (415.71,263.84) .. controls (432.37,274.37) and (187.04,261.71) .. (213.71,257.17) .. controls (240.37,252.64) and (236.37,256.64) .. (258.37,256.64) -- cycle ;
\draw    (188.86,235.48) -- (250.23,258.26) ;
\draw [shift={(253.04,259.31)}, rotate = 200.37] [fill={rgb, 255:red, 0; green, 0; blue, 0 }  ][line width=0.08]  [draw opacity=0] (8.93,-4.29) -- (0,0) -- (8.93,4.29) -- cycle    ;

\draw (466.12,245.71) node [anchor=north west][inner sep=0.75pt]    {$\mink$};
\draw (470,20) node [anchor=north west][inner sep=0.75pt]    {$\mathcal{SO}_+^\uparrow(1,d-1)$};
\draw (120.53,61.09) node [anchor=north west][inner sep=0.75pt]    {$O_+^\uparrow(\mink,\eta)$};
\draw (199.28,70.3) node [anchor=north west][inner sep=0.75pt]    {$\supp \mu_{\omega}^{\E_\R}$};
\draw (123.95,218.2) node [anchor=north west][inner sep=0.75pt]    {$\text{Loc}(\R,\omega)$};

\end{tikzpicture}

    \caption{Born probability distribution over possible clocks and rods associated to a quantum reference frame prepared in a state $\omega$ \cite{fedida_foundations_2025}. The \enquote{fuzziness} in the space of inertial reference frames projects to the level of the localization in spacetime.}
    \label{fig:Visualisation probability distributions}
    \end{subfigure}
    \caption{Space of inertial frames for which every point represent a different viewpoint from which physical systems can be described. A classical inertial reference frame can be thought of as a Dirac delta distribution over this space, while an oriented relativistic quantum reference frame gives a Born probability distribution over it. The measure projected to $\mink$ is a Born probability measure of the marginal POVM $\F_\R$.}
    \label{fig:Visualisation bundle}
\end{figure*}
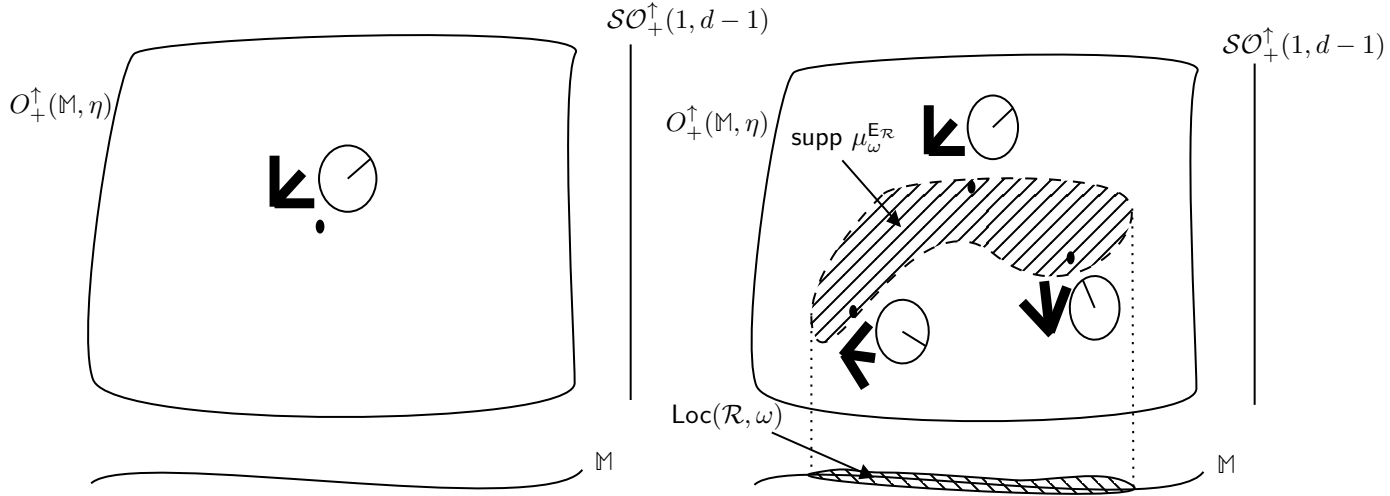

Let $U_\S:\Poincup\to \mathcal U(\his)$ be an ultraweakly continuous unitary representation of the proper orthochronous Poincaré group describing the kinematics of the system $\S$, and let $\R$
be a relativistic QRF. No irreducibility or spin assumption on $U_\S$ is needed for the construction
below. The adjective ``scalar'' refers here to the absence of an additional
finite-dimensional Lorentz index: the Lorentz orientation is retained as part
of the frame variable $\lambda$. Wigner's classification concerns irreducible
positive-energy one-particle representations. If $U_\S$ is such a
representation with mass $m>0$ and trivial little-group representation, one
may call $\S$ a scalar one-particle system of mass $m$. General field-theoretic
and Fock representations are typically reducible, and all results below apply
to them as well.

A relativistic QRF should not, in general, be assumed irreducible or assigned a
sharp mass and spin. A covariant POVM on the full frame torsor is a system of
covariance whose dilation is governed by imprimitivity theory and is generally
not a single Wigner sector. One may impose a spectral-support condition on
$U_\R$ if a definite particle content of the frame is desired, but this is an
additional physical hypothesis and is not used here.

Fix a bounded operator $\phi\in \bhs$.
We define its associated \emph{Lorentz-oriented absolute field} by
\begin{align}
\label{eq:absolute-oriented-field}
\hat{\phi}_{\cdot}(\cdot) : \mink \times \mathcal{SO}_+^\uparrow(1,d-1) &\to \bhs \\
\hat\phi_\lambda(x)
&:=
U_\S(x,\lambda)\,\phi\,U_\S(x,\lambda)^\dagger.
\end{align}

\begin{proposition}[\cite{fedida_foundations_2025}]
\label{prop:absolute-field-covariance}
For every $(a,\Lambda)\in \Poincup$ one has
\begin{equation}
\label{eq:absolute-field-transform}
U_\S(a,\Lambda) \hat\phi_\lambda(x) U_\S(a,\Lambda)^\dagger
=
\hat\phi_{\Lambda\lambda}(\Lambda x+a).
\end{equation}
\end{proposition}

\begin{definition}
    Let $\R$ be a relativistic QRF and $\phi \in \bhs$. We call the map
    \begin{align*}
        \hat{\Phi}^\R : \homframestate &\to \bhs \\
        \hat{\Phi}^\R(\omega) &= \yen^\R_\omega(\phi) = \iint_{O_+^\uparrow(\mink,\eta)} \hat{\phi}_{\lambda}(x) \, d\mu^{\E_\R}_{\omega}(x,\lambda)
    \end{align*}
    a \emph{scalar relational quantum field}.
\end{definition}

The absolute quantum field $(x,\lambda) \mapsto \hat{\phi}_{\lambda}(x)$ is not directly observable: what is operationally meaningful is its average with respect to the quantum uncertainty of the frame. In parallel, (scalar) relational quantum fields answer the question: for $\phi \in \bhs$ and $\R$ fixed, if given a specific state $\omega \in \homframestate$ of the frame, what is the operator (the observable if self-adjoint) one ``sees" using that QRF? This is precisely $\hat{\Phi}^\R(\omega)$.

\begin{remark}
    In the above, we omitted the mention of a specific choice of tetrad in $O_+^\uparrow(\mink,\eta)$ for the relativistic QRF and the relativization map (and the relational quantum field). In \cite{fedida_foundations_2025}, a physical motivation was given to argue for this. For now, we will not worry about such a choice, and will omit it for conciseness.
\end{remark}

How does the notion of relational quantum field relate to the ``standard" concepts of relational quantum fields we are used to in high-energy physics? We answer this question using tools from probability theory. Write down a standard disintegration of $\mu^{\E_\R}_{\omega}$ with respect to the marginal measure (also a Born probability measure \cite{fedida_foundations_2025}) $\mu^{\F_\R}_{\omega}$ over spacetime, such that there exists (by Lem.~\ref{lem:disintegration exists Lie QRF}) a family of conditional probability measures $\{\nu^{\R}_{\omega}(\cdot \mid x)\}_{x \in \mink}$ on $\mathcal{SO}_+^\uparrow(1,d-1)$ such that
\begin{equation}
    d\mu^{\E_{\R_0}}_{\omega}((x,\lambda) \cdot e_0) = d\nu^{\R_0}_{\omega}(\lambda \mid x; e_0) d\mu^{\F_{\R_0}}_{\omega}(x;e_0).
\end{equation}

As above, from now on we will omit writing down $e_0 \in O_+^\uparrow(\mink,\eta)$.

\begin{definition}
    Let $\R$ be a relativistic QRF and $\phi \in \bhs$. The \emph{(scalar) relational local quantum field} associated with $(\R,\omega)$ is the operator-valued function
    \begin{align}
        \hat{\phi}^\R_{\omega} : \supp \mu^{\F_\R}_{\omega} &\to \bhs \\
        \hat{\phi}^\R_{\omega}(x) &:= \int_{\mathcal{SO}_+^\uparrow(1,d-1)} \hat{\phi}_{\lambda}(x) \, d\nu^{\R}_{\omega}(\lambda \mid x).
    \end{align}
\end{definition}

\begin{remark}
    Note that these are, rigorously speaking, equivalence classes in $L^\infty(\supp \mu^{\F_\R}_\omega,\bhs)$ rather than pointwise-defined fields of operators. Any pointwise-defined version is a chosen representative.
\end{remark}

Colloquially, this is saying: transport an operator $\phi$ to the point $x \in \mink$, then average over the Lorentz fibers with respect to the conditional measure of the frame at $x$; this is $\hat{\phi}^\R_\omega(x)$. With this notation, the relational quantum fields become the spacetime smearing of the relational local quantum fields:
\begin{equation}
\label{eqn:rlo-as-smearing}
    \hat{\Phi}^\R(\omega) = \int_\mink \hat{\phi}^\R_\omega(x) \, d\mu^{\F_\R}_{\omega}(x) = \int_\mink \hat{\phi}^\R_\omega(x) f^{\F_\R}_{\omega}(x) d^dx,
\end{equation}
where the last equality holds by Thm.~\ref{thm:covariance implies mu continuity}.

\begin{remark}
\label{rem:rlo-vs-kernel}
Equation \eqref{eqn:rlo-as-smearing} is the flat-spacetime prototype of the curved-spacetime construction that will be
developed in Sec.~\ref{sec:RQFT curved}. The relational local quantum field $\hat\phi_\omega^\R(x)$ is not introduced independently; it appears as the disintegrated kernel of the relational quantum field.
\end{remark}

\subsection{Relational covariance} The covariance of the frame observable implies a corresponding covariance law for the relational (local) quantum fields. 

\begin{proposition}[\cite{fedida_foundations_2025}]
\label{prop:relational-covariance-rlo}
For every $(a,\Lambda)\in \Poincup$ and every $\omega\in\homframestate$,
\begin{equation}
\label{eq:rlo-covariance}
U_\S(a,\Lambda) \hat{\Phi}^{\R}(\omega) U_\S(a,\Lambda)^\dagger
=
\hat{\Phi}^{\R}(U_\R(a,\Lambda) \omega U_\R(a,\Lambda)^\dagger)
\end{equation}
and, for $\mu^{\F_\R}_{\omega}$-almost every $x \in \mink$,
\begin{equation}
\label{eqn:kernel-covariance}
    U_\S(a,\Lambda) \hat{\phi}^{\R}_\omega(x) U_\S(a,\Lambda)^\dagger = \hat{\phi}^\R_{U_\R(a,\Lambda) \omega U_\R(a,\Lambda)^\dagger}(\Lambda x + a).
\end{equation}
\end{proposition}

Equation \eqref{eqn:kernel-covariance} is the characteristic relational covariance law: the system transforms actively, but the frame state transforms at the same time. In general, this is not the naive pointwise scalar covariance law, which is precisely what allows us to avoid Wizimirski's no-go theorem \cite{wizimirski_existence_1966,fedida_foundations_2025,fedida_no-go_2026}.

\subsection{Localization limit} The formalism reduces to the ordinary non-relational picture when the frame is prepared in increasingly localized
states.

\begin{proposition}[\cite{fedida_foundations_2025}]
    Let $\R$ be a localizable relativistic QRF. Then for any localizing sequence $(\omega_n^{(x,\lambda)})_{n \in \Nn}$ of frame states centred around $(x,\lambda) \in O_+^\uparrow(\mink,\eta)$,
    \begin{equation}
        \lim_{n \to \infty} \hat{\Phi}^\R(\omega^{(x,\lambda)}_n) = \hat{\phi}_\lambda(x)
    \end{equation}
    ultraweakly.
\end{proposition}

Thus the absolute Lorentz-oriented quantum field appears as the sharp-frame limit of the relational
construction, in exact analogy with the localization results of Sec.~\ref{sec:operational-qrf-groups}.

\subsection{Relational causality} The notion of causality can also be expressed in relational quantum field theory, with interesting features arising. Notably, the causality properties of relational quantum fields can now depend on the QRF, and one may restrict causality conditions to operationally meaningful preparations of frames and to ``sufficiently" spacelike-separated supports. We review some of the discussion first presented in \cite{fedida_foundations_2025}, and will generalise it in Sec.~\ref{sec:RQFT curved} to curved spacetimes.

We start by highlighting that the concept of causality is most naturally understood epistemically through the no-superluminal signalling principle: it is a statement about what can be observed, not what occurs underneath the carpet. For that reason, we may want to make notions of causality depend on available resources: notably, on operationally meaningful preparations of the frame (e.g. $\omega \in \homframestate$ such that $\mu^{\E_\R}_{\omega}$ is compact), and on operationally meaningful notions of spacelike separations. For example, one can consider notions of spacelike separations for which an agent can indeed confirm the spacelike-ness of two regions. This is one way to implement smeared light cones; another is through considerations of quantum gravity. Regardless, this is an interesting notion to work with.

\begin{definition}
    Let $\sigma \geq 0$. If $\U, \V \subset \mathcal{M}$, we say that $\U$ and $\V$ are \emph{$\sigma$-spacelike separated} and write $\U \perp_{\sigma} \V$ if
    \begin{equation}
        \sup \{(x-y)^2 \mid x \in \U, \, y \in \V\} < -\sigma \, .
    \end{equation} Moreover, let $\R$ be a relativistic QRF and $\omega_1,\omega_2 \in \homframestate$. We say that $\omega_1$ and $\omega_2$ are $(\R,\sigma)$-spacelike separated and write $\omega_1 \perp^\R_\sigma \omega_2$ if $\text{Loc}(\R,\omega_1) \perp_\sigma \text{Loc}(\R,\omega_2)$.
\end{definition}

In effect, the distinguished family of states $\operationalstate$ of the frame $\R$ naturally defines a \enquote{cut-off scale} for spacelike separation as follows.

\begin{definition}[\cite{fedida_foundations_2025}]
    \label{def:spacelike resolution of R}
    Let $\R$ be a relativistic QRF and $\operationalstate \subseteq \homframestate$ be convex. The quantity
    \begin{equation}
        \sigma(\operationalstate) := \inf_{\omega \in \operationalstate}  \sup \{\abs{(x-y)^2} \mid x \perp y \in \supp \mu^{\F_\R}_{\omega}\}
    \end{equation}
    is the \emph{$\operationalstate$-spacelike resolution of $\R$}.
\end{definition}

This is shown pictorially in Fig.~\ref{fig:approximate causality flat spacetime}. From there, we can discuss relational aspects of causality for relational quantum fields.

\begin{figure}
    \centering
    \begin{subfigure}[t]{0.3\textwidth}
        \begin{tikzpicture}[x=0.75pt,y=0.75pt,yscale=-1,xscale=1]

\draw   (353.6,97.1) .. controls (374.6,89.1) and (383.6,99.1) .. (389.6,104.1) .. controls (395.6,109.1) and (418.93,122.43) .. (423.6,153.1) .. controls (428.27,183.77) and (367.6,221.1) .. (355.67,223.93) .. controls (343.73,226.77) and (314.27,183.77) .. (317.6,153.1) .. controls (320.93,122.43) and (332.6,105.1) .. (353.6,97.1) -- cycle ;
\draw    (260.67,175.67) -- (260.03,137.27) ;
\draw [shift={(260,135.27)}, rotate = 89.05] [color={rgb, 255:red, 0; green, 0; blue, 0 }  ][line width=0.75]    (10.93,-3.29) .. controls (6.95,-1.4) and (3.31,-0.3) .. (0,0) .. controls (3.31,0.3) and (6.95,1.4) .. (10.93,3.29)   ;
\draw    (260.67,175.67) -- (295.33,175.29) ;
\draw [shift={(297.33,175.27)}, rotate = 179.37] [color={rgb, 255:red, 0; green, 0; blue, 0 }  ][line width=0.75]    (10.93,-3.29) .. controls (6.95,-1.4) and (3.31,-0.3) .. (0,0) .. controls (3.31,0.3) and (6.95,1.4) .. (10.93,3.29)   ;
\draw    (320.6,153.1) -- (420.6,153.1) ;
\draw [shift={(423.6,153.1)}, rotate = 180] [fill={rgb, 255:red, 0; green, 0; blue, 0 }  ][line width=0.08]  [draw opacity=0] (8.93,-4.29) -- (0,0) -- (8.93,4.29) -- cycle    ;
\draw [shift={(317.6,153.1)}, rotate = 0] [fill={rgb, 255:red, 0; green, 0; blue, 0 }  ][line width=0.08]  [draw opacity=0] (8.93,-4.29) -- (0,0) -- (8.93,4.29) -- cycle    ;

\draw (256,123.4) node [anchor=north west][inner sep=0.75pt]  [font=\scriptsize]  {$t$};
\draw (300,168.07) node [anchor=north west][inner sep=0.75pt]  [font=\scriptsize]  {$\vec{x}$};
\draw (356.67,158.73) node [anchor=north west][inner sep=0.75pt]  [font=\footnotesize]  {$\sigma ( \operationalstate)$};
\draw (330,111.4) node [anchor=north west][inner sep=0.75pt]    {$\supp \mu _{\omega_\text{min}}^{\F_{\R}}$};
\end{tikzpicture}

    \caption{Spacetime support of the state preparation $\omega_\text{min} \in \operationalstate$ which provides the ``smallest`` spacelike resolution $\sigma(\operationalstate)$ for the frame $\R$ out of all ``operationally meaningful'' states $\omega \in \operationalstate$.}
    \end{subfigure}
    \qquad
    \begin{subfigure}[t]{0.6\textwidth}
        \begin{tikzpicture}[x=0.75pt,y=0.75pt,yscale=-0.85,xscale=0.85]

\draw   (118,116.27) .. controls (138,106.27) and (144,112.93) .. (156.67,118.93) .. controls (169.33,124.93) and (173.33,135.6) .. (180,160.93) .. controls (186.67,186.27) and (153.33,208.93) .. (126,204.93) .. controls (98.67,200.93) and (98,126.27) .. (118,116.27) -- cycle ;
\draw  [dash pattern={on 0.84pt off 2.51pt}]  (172,189.67) -- (245.33,88.93) ;
\draw   (290.67,118) .. controls (310.67,108) and (352,112.4) .. (380.67,118) .. controls (409.33,123.6) and (418,152.93) .. (380.67,178) .. controls (343.33,203.07) and (310.67,208) .. (290.67,178) .. controls (270.67,148) and (270.67,128) .. (290.67,118) -- cycle ;
\draw  [dash pattern={on 0.84pt off 2.51pt}]  (280,126) -- (225.33,208.27) ;
\draw    (189.6,170.62) -- (247,170.92) ;
\draw [shift={(250,170.93)}, rotate = 180.3] [fill={rgb, 255:red, 0; green, 0; blue, 0 }  ][line width=0.08]  [draw opacity=0] (8.93,-4.29) -- (0,0) -- (8.93,4.29) -- cycle    ;
\draw [shift={(186.6,170.6)}, rotate = 0.3] [fill={rgb, 255:red, 0; green, 0; blue, 0 }  ][line width=0.08]  [draw opacity=0] (8.93,-4.29) -- (0,0) -- (8.93,4.29) -- cycle    ;
\draw    (224.6,145.64) -- (263.67,146.22) ;
\draw [shift={(266.67,146.27)}, rotate = 180.85] [fill={rgb, 255:red, 0; green, 0; blue, 0 }  ][line width=0.08]  [draw opacity=0] (8.93,-4.29) -- (0,0) -- (8.93,4.29) -- cycle    ;
\draw [shift={(221.6,145.6)}, rotate = 0.85] [fill={rgb, 255:red, 0; green, 0; blue, 0 }  ][line width=0.08]  [draw opacity=0] (8.93,-4.29) -- (0,0) -- (8.93,4.29) -- cycle    ;
\draw   (486,84.4) .. controls (506,74.4) and (537.33,89.2) .. (554.67,97.87) .. controls (572,106.53) and (603.33,115.87) .. (576,144.4) .. controls (548.67,172.93) and (506,174.4) .. (486,144.4) .. controls (466,114.4) and (466,94.4) .. (486,84.4) -- cycle ;
\draw  [dash pattern={on 0.84pt off 2.51pt}]  (396.67,164.33) -- (470,63.6) ;
\draw  [dash pattern={on 0.84pt off 2.51pt}]  (402.67,193) -- (476,92.27) ;
\draw    (416.6,142.16) -- (436.33,142.57) ;
\draw [shift={(439.33,142.63)}, rotate = 181.19] [fill={rgb, 255:red, 0; green, 0; blue, 0 }  ][line width=0.08]  [draw opacity=0] (8.93,-4.29) -- (0,0) -- (8.93,4.29) -- cycle    ;
\draw [shift={(413.6,142.1)}, rotate = 1.19] [fill={rgb, 255:red, 0; green, 0; blue, 0 }  ][line width=0.08]  [draw opacity=0] (8.93,-4.29) -- (0,0) -- (8.93,4.29) -- cycle    ;
\draw    (458.47,85.98) -- (497.53,86.56) ;
\draw [shift={(500.53,86.6)}, rotate = 180.85] [fill={rgb, 255:red, 0; green, 0; blue, 0 }  ][line width=0.08]  [draw opacity=0] (8.93,-4.29) -- (0,0) -- (8.93,4.29) -- cycle    ;
\draw [shift={(455.47,85.93)}, rotate = 0.85] [fill={rgb, 255:red, 0; green, 0; blue, 0 }  ][line width=0.08]  [draw opacity=0] (8.93,-4.29) -- (0,0) -- (8.93,4.29) -- cycle    ;

\draw (207.67,175.4) node [anchor=north west][inner sep=0.75pt]  [font=\scriptsize]  {$\sigma _{12}$};
\draw (234.33,129.07) node [anchor=north west][inner sep=0.75pt]  [font=\scriptsize]  {$\sigma ( \operationalstate)$};
\draw (105.33,148.73) node [anchor=north west][inner sep=0.75pt]  [font=\small]  {$\text{Loc}(\R,\omega_1)$};
\draw (305.33,136.73) node [anchor=north west][inner sep=0.75pt]  [font=\small]  {$\text{Loc}(\R,\omega_2)$};
\draw (498.67,113.4) node [anchor=north west][inner sep=0.75pt]  [font=\small]  {$\text{Loc}(\R,\omega_3)$};
\draw (415.6,145.5) node [anchor=north west][inner sep=0.75pt]  [font=\scriptsize]  {$\sigma _{23}$};
\draw (467.33,69.07) node [anchor=north west][inner sep=0.75pt]  [font=\scriptsize]  {$\sigma ( \operationalstate)$};

\end{tikzpicture}

    \caption{Examples of $\sigma$-spacelike separations for different state preparations of the frame $\R$ with $\operationalstate$-spacelike resolution $\sigma(\operationalstate)$. We see that $\sigma_{12} > \sigma(\operationalstate)$ i.e. the preparations $\omega_1$ and $\omega_2$ are ``operationally spacelike separated'': the frame can resolve this spacelike separation. On the other hand, $\sigma_{23} < \sigma(\operationalstate)$ i.e. the frame's spacelike resolution is too small to determine whether $\omega_2$ and $\omega_3$ should be understood as yielding causally disconnected supports: relational quantum fields generated by these preparations may not commute.}
    \end{subfigure}
    \caption{Pictorial representations of the frame resolutions and of $\sigma(\operationalstate)$-spacelike separation of regions in flat spacetime \cite{fedida_foundations_2025}.}
    \label{fig:approximate causality flat spacetime}
\end{figure}

\begin{definition}
    Let $\mathcal{O}_\S \subset \bhs$ be closed under adjoints, $\sigma \geq 0$, $\R$ be a relativistic QRF and $\operationalstate \subset \homframestate$ be convex. We say that $\mathcal{O}_\S$ is
    \begin{enumerate}
        \item \emph{absolutely $\sigma$-microcausal} if for every $\phi_1,\phi_2 \in \mathcal{O}_\S$ and $(x,\lambda_1),(y,\lambda_2) \in O_+^\uparrow(\mink,\eta)$,
        \begin{equation}
            x \perp_\sigma y \Longrightarrow \comm{\hat{\phi}_{1,\lambda_1}(x)}{\hat{\phi}_{2,\lambda_2}(y)} = \comm{\hat{\phi}_{1,\lambda_1}(x)^\dagger}{\hat{\phi}_{2,\lambda_2}(y)} = 0.
        \end{equation}
        \item \emph{strongly $(\operationalstate,\sigma)$-microcausal} if for all $\phi_1,\phi_2 \in \mathcal{O}_\S$ and every pair $\omega_1 $ and $\omega_2$ of frame states in $\operationalstate$ and $\forall x \in \text{Loc}(\R,\omega_1)$ and $\forall y \in \text{Loc}(\R,\omega_2)$,
        \begin{equation}
            x \perp_\sigma y \Longrightarrow\comm{\hat{\phi}^\R_{1,\omega_1}(x)}{\hat{\phi}^\R_{2,\omega_2}(y)} = \comm{\hat{\phi}^\R_{1,\omega_1}(x)^\dagger}{\hat{\phi}^\R_{2,\omega_2}(y)} = 0,
        \end{equation}
        \item \emph{weakly $(\operationalstate,\sigma)$-microcausal} if for all $\phi_1,\phi_2 \in \mathcal{O}_\S$ and every pair $\omega_1$ and $\omega_2$ of frame states in $\operationalstate$,
        \begin{multline}
            \text{Loc}(\R,\omega_1) \perp_\sigma \text{Loc}(\R,\omega_2) \Longrightarrow \comm{\hat{\phi}^\R_{1,\omega_1}(x)}{\hat{\phi}^\R_{2,\omega_2}(y)} = \comm{\hat{\phi}^\R_{1,\omega_1}(x)^\dagger}{\hat{\phi}^\R_{2,\omega_2}(y)} = 0 \\ \forall x \in \text{Loc}(\R,\omega_1), y \in \text{Loc}(\R,\omega_2).
        \end{multline}
        \item $(\operationalstate,\sigma)$-causal if for all $\phi_1,\phi_2 \in \mathcal{O}_\S$ and every pair $\omega_1$ and $\omega_2$ of frame states in $\operationalstate$,
        \begin{equation}
            \text{Loc}(\R,\omega_1) \perp_\sigma \text{Loc}(\R,\omega_2) \Longrightarrow \comm{\hat{\Phi}_1^\R(\omega_1)}{\hat{\Phi}_2^\R(\omega_2)} = \comm{\hat{\Phi}_1^\R(\omega_1)^\dagger}{\hat{\Phi}_2^\R(\omega_2)} = 0
        \end{equation}
    \end{enumerate}
\end{definition}

\begin{remark}
    Note that a simple generalisation to the following is to consider relational causality for two choices of QRFs $\R_1$ and $\R_2$. This might be meaningful if two agents have access to two different kinds of clocks and rulers. We however keep such a discussion for future work.
\end{remark}

\begin{theorem}[\cite{fedida_foundations_2025}]
	\label{thm:causality rqft flat}
	Let $\R$ be a relativistic QRF, $\operationalstate \subseteq \homframestate$ and $\sigma \geq 0$. If $\mathcal{O}_\S \subseteq \bhs$ is
	\begin{enumerate}
		\item strongly $(\operationalstate,\sigma)$-microcausal, it is weakly $(\operationalstate,\sigma)$-microcausal,
		\item weakly $(\operationalstate,\sigma)$-microcausal, it is $(\operationalstate,\sigma)$-causal.
\end{enumerate}	
\end{theorem}

Microcausality is a more ontological implementation than Einstein causality of the no-superluminal signalling principle. In turn, strong microcausality (causality holding pointwise, regardless of frame preparations) is ``more ontic" than weak microcausality (causality holding pointwise provided frame preparations are spacelike separated). Arguably, absolute microcausality is the ``most ontic" of all those conditions, but is too strong to assume in general (for $\sigma=0$) by Wightman's theorem \cite{wightman_theorie_1964}. We see that, in relational quantum field theory, causality is understood relationally. Some operators may lead to causal relational quantum fields in some QRFs but not others. This is an interesting specificity of this formalism.

\begin{example}[\cite{fedida_foundations_2025}]
    \label{ex:example causal rqf}
    Let $\R$ be a relativistic QRF, $K$ be any compact subset of $O_+^\uparrow(\mink,\eta)$, $W_{\hat{\phi}} : \mathscr{S}(\mink,\mathbb{C}) \to \bhs$ be the Weyl operator for a Klein-Gordon Wightmannian quantum field $\hat{\phi}$. Define
    \begin{equation}
        \operationalstate^K := \{\omega \in \homframestate \mid \supp \mu^{\E_\R}_{\omega} \subset K\}.
    \end{equation}
    Then $\forall \sigma > 0$, $\exists f_{\sigma,K} \in C^\infty_c(\mink)$ such that $W_{\hat{\phi}}(f_{\sigma,K}) \in \bhs$ is weakly $(\operationalstate^K,\sigma)$-microcausal.
\end{example}

\subsection{Relational and Wightman quantum field theory}
    It may, at first view, be worrying to be able to discuss a notion of ``quantum field at $x \in \mink$" through the relational local quantum field $\hat{\phi}^\R_\omega(x)$. Indeed, ultraviolet divergences typically prevent one from doing so. In Wightman quantum field theory, where quantum fields are (unbounded) operator-valued \emph{distributions} $\hat{\Phi}_W : \mathscr{S}(\mink,\mathbb{C}) \to \mathcal{L}(\D,\hi)$, with $\mathscr{S}(\mink,\mathbb{C})$ the space of Schwartz functions, these ultraviolet divergences prevent one from assuming these distributions are regular, i.e. of the form $\hat{\Phi}_W(f) = \int_\mink \hat{\phi}(x) f(x)d^dx$ for a \emph{function-independent} kernel $\hat{\phi}:\mink \to \mathcal{L}(\D,\hi)$. However, such quantum fields can, in fact, be written as smearings of pointwise quantum fields, though these need to be function-dependent as follows. 

    For any nonzero $f \in \mathscr{S}(\mink,\mathbb{C})$, let $g_f : \supp f \to \mathbb{C}$ be such that $\int_{\mink} g_f(x) f(x) d^dx$ exists and is nonzero. Then, we can define
    \begin{align}
        \hat{\phi}^{g_f}_f : \supp f &\to \mathcal{L}(\D,\hi) \\
        \hat{\phi}^{g_f}_f(x) &:= \frac{g_f(x)}{\int_\mink g_f(x) f(x)d^dx} \hat{\Phi}_W(f).
    \end{align}
    Then, trivially, $\hat{\Phi}_W(f) = \int_\mink \hat{\phi}^{g_f}_f(x) f(x) d^dx$. Note however that not all function-dependent kernels can be written explicitly as above; for example, $\hat{\phi}^\R_\omega$ cannot. Nonetheless, it provides some understanding of where these kernels come from.
    
    Note also that, from Wightmannian covariance $U_\S(a,\Lambda) \hat{\Phi}_W(f) U_\S(a,\Lambda)^\dagger = \hat{\Phi}_W(\{a,\Lambda\} \cdot f)$, where $(\{a,\Lambda\} \cdot f)(x) := f(\Lambda^{-1}(x-a))$, we have that
    \begin{equation}
        U_\S(a,\Lambda) \hat{\phi}^{g_f}_f(x) U_\S(a,\Lambda)^\dagger = \hat{\phi}^{g_{\{a,\Lambda\} \cdot f}}_{\{a,\Lambda\} \cdot f}(\Lambda x +a).
    \end{equation}
    This construction also allows us to compare the different notions of causality: strong microcausality $x \perp y \Rightarrow \comm{\hat{\phi}^{g_{f_1}}_{f_1}(x)}{\hat{\phi}^{g_{f_2}}_{f_2}(y)} = 0$ is obviously too strong, while weak microcausality $\supp f_1 \perp \supp f_2 \Rightarrow \comm{\hat{\phi}^{g_{f_1}}_{f_1}(x)}{\hat{\phi}^{g_{f_2}}_{f_2}(y)} = 0$ is now equivalent (in this specific case where the kernels are explicitly defined with respect to $\hat{\Phi}_W$) to Wightmannian Einstein causality $\supp f_1 \perp \supp f_2 \Rightarrow \comm{\hat{\Phi}_W(f_1)}{\hat{\Phi}_W(f_2)} = 0$. We can thus start motivating relational causality conditions from such analogies.

    \begin{remark}
        One can see quantum reference frames and relational quantum field theory as providing an operational meaning to the (arguably obscure) notion of test function in Wightman quantum field theory. These smearing functions are nothing but the frame smearing functions from the perspective of a quantum reference frame prepared in a certain state. In turn, given sufficiently ``nice" properties of the frame observable $\E_\R$ and/or of the operationally accessible states $\operationalstate$, these frame smearing functions can be more regular and smooth than just $L^1$; perhaps, in suitable regimes, they are indeed Schwartz, but they need not be in general.
    \end{remark}

    Overall, the following comparison holds. 

    \begin{center}
\begin{tabular}{ccl}
Relational quantum field theory & $\leftrightsquigarrow$ & Wightman quantum field theory \\
$f^{\F_\R}_{\omega} \in L^1(\mink,\mathbb{C})$ & $\leftrightsquigarrow$ & $f \in \mathscr{S}(\mink,\mathbb{C})$,\\
$\int_{\mathcal{SO}_+^\uparrow(1,d-1)} d\nu^{\R}_{\omega}(\lambda \mid x) : \supp f^{\F_\R}_{\omega} \to \mathbb{C}$ & $\leftrightsquigarrow$ & $g_f : \supp f \to \mathbb{C}$, \\
$\hat{\phi}^\R_\omega : \supp f^{\F_\R}_{\omega} \to \bhs$ & $\leftrightsquigarrow$ & $\hat{\phi}^{g_f}_{f} : \supp f \to \mathcal{L}(\D_\S,\his)$,\\ 
$\hat{\Phi}^\R : \homframestate \to \bhs$ & $\leftrightsquigarrow$ & $\hat{\Phi}_W : \mathscr{S}(\mink,\mathbb{C}) \to \mathcal{L}(\D_\S,\his)$,\\ 
$U_\S(a,\Lambda) \hat{\Phi}^\R(\omega) U_\S(a,\Lambda)^\dagger = \hat{\Phi}^\R(U_\R(a,\Lambda) \omega U_\R(a,\Lambda)^\dagger)$ & $\leftrightsquigarrow$ & $U_\S(a,\Lambda) \hat{\Phi}_W(f) U_\S(a,\Lambda)^\dagger = \hat{\Phi}_W(\{a,\Lambda\} \cdot f)$,\\
$U_\S(a,\Lambda) \hat{\phi}^\R_\omega(x) U_\S(a,\Lambda)^\dagger = \hat{\phi}^\R_{U_\R(a,\Lambda) \omega U_\R(a,\Lambda)^\dagger}(\Lambda x + a)$ & $\leftrightsquigarrow$ & $U_\S(a,\Lambda) \hat{\phi}^{g_f}_f(x) U_\S(a,\Lambda)^\dagger = \hat{\phi}^{g_{\{a,\Lambda\} \cdot f}}_{\{a,\Lambda\} \cdot f}(\Lambda x+a) $,\\
$(\operationalstate,\sigma)$-causality  & $\leftrightsquigarrow$ & Causality for $f,g \in C^\infty_c$.
\end{tabular}
\end{center}
Note that, in the Wightmannian framework, the analytic properties of the Schwartz smearing functions (which help with the analytic continuation of the relevant vacuum expectation values) make $\sigma>0$ causality become equivalent to $\sigma=0$ causality. This is known as the global nature of local commutativity \cite{streater_pct_1989}. The rigidity coming from the analyticity of these smearing functions is weakened in relational quantum field theory, so $\sigma>0$ causality may not be equivalent to $\sigma=0$ causality in general. This is an interesting prospect to find new physically meaningful types of causal quantum fields.

\begin{remark}
\label{rem:boundedness-rqft}
At this stage the construction is performed for bounded seed operators $\phi\in \bhs$, so the resulting
relational quantum fields and observables are bounded. This is sufficient for the foundation developed here and already exhibits the basic mechanisms of locality, covariance and localization. The extension to unbounded operators belongs to a later stage of the programme.
\end{remark}

\begin{remark}
    Let $(f_n) \subset \mathscr{S}(\mink,\mathbb{C})$ be a sequence of Schwartz functions converging (distributionally) to $\delta_x$. In Wightman quantum field theory, the limit $\hat{\Phi}_W(f_n) \to \hat{\phi}(x)$ does not exist (in the space of unbounded operator-valued distributions): the right-hand side is not an unbounded operator on a Hilbert space \cite{halvorson_algebraic_2007}. Further note that $\mathscr{S}(\mink,\mathbb{C}) \subset L^1(\mink,\mathbb{C}) \cap L^\infty(\mink,\mathbb{C})$. That is, from the lens of relational quantum field theory and Sec.~\ref{subsec:localization and continuity}, the fact that one cannot localize Wightmannian quantum fields is to be understood as the fact that they are described using a QRF whose frame observable is absolutely continuous with respect to the Lebesgue measure on $\mink$ and thus not localizable. Whether it is possible, in relational quantum field theory, to retrieve an analogous version of Wightmannian quantum fields defined instead over $L^1(\mink,\mathbb{C})$ and described using a localizable QRF is an interesting open question.
\end{remark}

\subsection{Relational and Algebraic quantum field theory} One can also generate algebras of relational observables as follows (see \cite{fedida_foundations_2025} for a more thorough exposition and proof of the following).

\begin{definition}
    \label{def:relational local algebras minkowski}
    Let $\R$ be a relativistic QRF, $\operationalstate \subseteq \homframestate$ be convex, $\sigma \geq 0$ and $\mathcal{O}_\S \subseteq \bhs$ be closed under adjoints and $(\operationalstate,\sigma)$-causal. Given a subset $\U \subseteq \mink$, we call
    \begin{equation}
        \A^{(\operationalstate,\sigma)}_{\mathcal{O}_\S}(\U) := \{\hat{\Phi}^\R(\omega) \mid \phi \in \mathcal{O}_\S, \omega \in \operationalstate \text{ s.t. } \text{Loc}(\R,\omega) \subset \U\}''
    \end{equation}
    a \emph{relational local algebra}.
\end{definition}

These relational local algebras satisfy the core axioms of algebraic quantum field theory \cite{haag_algebraic_1964}.

\begin{theorem}[\cite{fedida_foundations_2025}]
    Let $\R$ be a relativistic QRF, $\operationalstate \subseteq \homframestate$ be convex, $\sigma \geq 0$ and $\mathcal{O}_\S \subseteq \bhs$ be $(\operationalstate,\sigma)$-causal. Then
    \begin{enumerate}
        \item (Isotony) For all $\U \subseteq \V \subseteq \mink$, $\A^{(\operationalstate,\sigma)}_{\mathcal{O}_\S}(\U) \subseteq \A^{(\operationalstate,\sigma)}_{\mathcal{O}_\S}(\V)$.
        \item (Covariance) For all $(a,\Lambda) \in \Poincup$ and $\U \subseteq \mink$,
        \begin{equation}
            U_\S(a,\Lambda) \A^{(\operationalstate,\sigma)}_{\mathcal{O}_\S}(\U) U_\S(a,\Lambda)^\dagger = \A^{(U_\R(a,\Lambda)\operationalstate U_\R(a,\Lambda)^\dagger,\sigma)}_{\mathcal{O}_\S}((a,\Lambda) \cdot \U).
        \end{equation}
        \item (Causality) For all $\U \perp_\sigma \V \subset \mink$, $\comm{\A^{(\operationalstate,\sigma)}_{\mathcal{O}_\S}(\U) }{\A^{(\operationalstate,\sigma)}_{\mathcal{O}_\S}(\V)} = \{0\}$. 
    \end{enumerate}
\end{theorem}

Thus, just as Wightman quantum field theory can generate algebras of local observables and be seen as a basis to understand algebraic quantum field theory, so can relational quantum field theory.

\begin{remark}
    This can be seen as an algebraic approach to relational quantum field theory -- we coin this ARQFT for short. This shares some similarities with recent work \cite{fewster_quantum_2024,fewster_semi-local_2025} where it was argued that supplementing algebraic quantum field theory with quantum reference frames can lead to type reduction of the local algebras. These recent works form the basis of a top-down approach (which we can call relational algebraic quantum field theory, or RAQFT for short), which contrasts with our bottom-up approach in that in RAQFT, a background spacetime structure is required to introduce local nets of algebras in the first place. At the moment, (A)RQFT also works on (frame bundles over) spacetime, but we can plausibly imagine extensions to cases where notions of spacetime disintegrate but quantum reference frames (and thus relational quantum fields) are still meaningful. This may be relevant in the context of quantum gravity: one might not be able to talk about a background spacetime structure (i.e. we might not be able to assign nets of algebras to regions of spacetime), but one can still discuss relations between physical systems. It would be interesting to understand the links between both approaches, whether they concur on a similar domain of applicability, and whether RQFT is indeed better suited to discuss quantum notions of spacetimes.
\end{remark}

\section{Foundations of relational quantum field theory on curved spacetime}
\label{sec:RQFT curved}

We now formulate the curved-spacetime counterpart of the Minkowski construction of Sec.~\ref{sec:RQFT Minkowski}. The point is to retain the operational and relational content of the flat-spacetime formalism while replacing the global Poincar\'e symmetry by the Poincar\'e groupoid of the spacetime. This achieves the passage from a global homogeneous-space description of frames to a genuinely local, bundle/groupoid-based one.

Let $(\mathcal{M},g)$ be an oriented and time-oriented $d$-dimensional Lorentzian spacetime. Let $O_+^\uparrow(\mathcal{M},g)\to \M$ be its bundle of oriented orthochronous orthonormal frames, and let $\Gamma:=\mathrm{Poin}(\mathcal{M},g) = O_+^\uparrow(\mathcal{M},g)\times_{\Lup}O_+^\uparrow(\mathcal{M},g) \rightrightarrows \M$, be the associated Poincar\'e groupoid. We use the volume measure $\mu_g$ of the spacetime $(\mathcal{M},g)$. The quantum system of interest is described by a continuous Hilbert bundle $\scrhis=\bigsqcup_{x\in \M}\hisx \to \M$ carrying a fiber-wise ultraweakly continuous unitary representation $V:\mathrm{Poin}(\M,g)\curvearrowright \scrhis$. The groupoid relativization theory requires no particle-sector assumption on
$V$. Moreover, mass is not classified by irreducible representations of the
transitive Poincare groupoid: such representations are controlled by the
Lorentz isotropy representation, whereas the mass parameter belongs to the
orbit structure of the Wigner groupoid (or to the dynamics, for example a
Klein--Gordon operator). Thus a ``scalar system of mass $m$'' on a curved
spacetime should be introduced only after adding Wigner-groupoid data or a
specific field equation. No such extra hypothesis is needed in the present
kinematical construction, and a general relativistic QRF is not assigned a
mass or spin.

\begin{definition}
    We call a $\mu_g$-admissible principal QRF for the groupoid $\mathrm{Poin}(\M,g)$ a \emph{general relativistic QRF}. 
\end{definition}

\begin{lemma}
    \label{lem:general relativistic QRF is bisection admissible}
    A general relativistic QRF is bisection-admissible.
\end{lemma}

\begin{proof}
    This is immediate from the fact that diffeomorphisms of a manifold preserve the null sets of the volume measure.
\end{proof}

For every $x \in \M$, we can, as in the flat spacetime case, discuss the projection of the support of $\mu^{\E_\R^x}_{\omega_x}$ onto $\M$ to discuss the localization of a QRF $\R$ prepared in a state $\omega = (\omega_x) \in \framestate$.

\begin{definition}
    Let $\R$ be a general relativistic QRF, $\omega \in \framestate$ and $\pi = (\pi_x : \mathrm{Poin}(\M,g)^x \to \M)$ a family of measurable projections. We call
    \begin{equation}
        \text{Loc}_\pi(\R,\omega) := \esssupp_{x \in \M} \mu^{\F_{\R,\pi}^x}_{\omega_x}
    \end{equation}
    the \emph{spacetime localization of $(\R,\omega)$ given $\pi$}, and
    \begin{equation}
        \text{Loc}_\pi(\R) := \esssupp_{x \in \M} \F_{\R,\pi}^x
    \end{equation}
    the \emph{spacetime localization of $\R$ given $\pi$}, where $(\F_{\R,\pi}^x)$ is the $\pi$-marginal frame POVM of the QRF.
\end{definition}

This precisely matches the notion of localization in Eqns.~\eqref{eqn:Loc(R,w)} and \eqref{eqn:Loc(R)}, with a pictorial representation of $\text{Loc}_\pi(\R,\omega)$ given in Fig.~\ref{fig:localization gr QRF}.

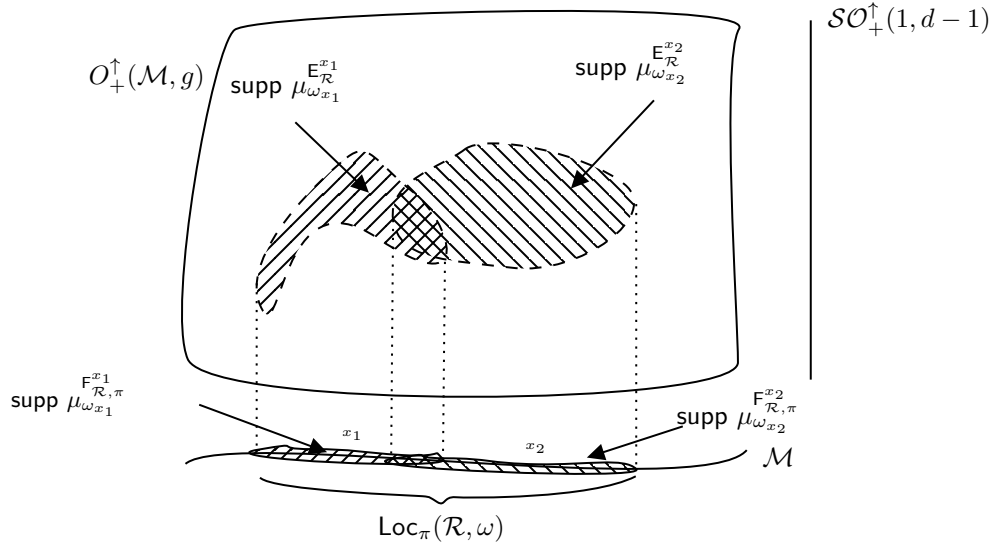
\begin{figure*}
    \centering
    \begin{tikzpicture}[x=0.75pt,y=0.75pt,yscale=-1,xscale=1]

\draw   (229.6,37.08) .. controls (240.62,28.46) and (492.17,21.05) .. (482.8,36.58) .. controls (473.43,52.11) and (482.25,181.52) .. (482.25,203.95) .. controls (482.25,226.38) and (217.12,228.11) .. (206.09,202.23) .. controls (195.07,176.34) and (218.57,45.71) .. (229.6,37.08) -- cycle ;
\draw    (518.79,32.27) -- (518.79,212.58) ;
\draw    (204.99,254.85) .. controls (227.24,236.23) and (464.61,273.83) .. (486.66,247.95) ;
\draw  [pattern=_h3f2g35xw,pattern size=6pt,pattern thickness=0.75pt,pattern radius=0pt, pattern color={rgb, 255:red, 0; green, 0; blue, 0}][dash pattern={on 4.5pt off 4.5pt}] (286.02,101.11) .. controls (297.04,92.49) and (299.91,102.93) .. (319.2,118.46) .. controls (338.49,133.99) and (350.81,175.11) .. (301.2,141.46) .. controls (251.59,107.81) and (253.5,201.19) .. (242.47,175.31) .. controls (231.45,149.43) and (275,109.74) .. (286.02,101.11) -- cycle ;
\draw    (260.11,83.41) -- (294.88,116.71) ;
\draw [shift={(297.04,118.78)}, rotate = 223.77] [fill={rgb, 255:red, 0; green, 0; blue, 0 }  ][line width=0.08]  [draw opacity=0] (8.93,-4.29) -- (0,0) -- (8.93,4.29) -- cycle    ;
\draw  [dash pattern={on 0.84pt off 2.51pt}]  (240.71,164.2) -- (240.71,248.86) ;
\draw  [dash pattern={on 0.84pt off 2.51pt}]  (335.2,153.46) -- (334.2,253.46) ;
\draw  [pattern=_cebealfc3,pattern size=6pt,pattern thickness=0.75pt,pattern radius=0pt, pattern color={rgb, 255:red, 0; green, 0; blue, 0}] (268.2,247.46) .. controls (290.2,247.46) and (291.53,249.13) .. (314.2,250.46) .. controls (336.87,251.79) and (328.2,246.46) .. (334.2,252.46) .. controls (340.2,258.46) and (214.04,252.4) .. (240.71,247.86) .. controls (267.37,243.33) and (246.2,247.46) .. (268.2,247.46) -- cycle ;
\draw    (211.86,225.17) -- (273.23,247.95) ;
\draw [shift={(276.04,249)}, rotate = 200.37] [fill={rgb, 255:red, 0; green, 0; blue, 0 }  ][line width=0.08]  [draw opacity=0] (8.93,-4.29) -- (0,0) -- (8.93,4.29) -- cycle    ;
\draw  [pattern=_cl4zvaxox,pattern size=6pt,pattern thickness=0.75pt,pattern radius=0pt, pattern color={rgb, 255:red, 0; green, 0; blue, 0}][dash pattern={on 4.5pt off 4.5pt}] (345.02,97.11) .. controls (356.04,88.49) and (405.91,98.93) .. (425.2,114.46) .. controls (444.49,129.99) and (403.2,160.46) .. (369.2,156.46) .. controls (335.2,152.46) and (312.2,152.46) .. (309.2,132.46) .. controls (306.2,112.46) and (334,105.74) .. (345.02,97.11) -- cycle ;
\draw  [dash pattern={on 0.84pt off 2.51pt}]  (431.2,124.46) -- (431.2,257.46) ;
\draw  [dash pattern={on 0.84pt off 2.51pt}]  (309.2,132.46) -- (308.2,252.46) ;
\draw  [pattern=_yh1bf9qsy,pattern size=6pt,pattern thickness=0.75pt,pattern radius=0pt, pattern color={rgb, 255:red, 0; green, 0; blue, 0}] (335.69,252.06) .. controls (357.69,252.06) and (356.44,254.66) .. (381.69,255.06) .. controls (406.94,255.46) and (425.2,251.46) .. (431.2,257.46) .. controls (437.2,263.46) and (281.53,256.99) .. (308.2,252.46) .. controls (334.87,247.93) and (313.69,252.06) .. (335.69,252.06) -- cycle ;
\draw    (440.11,72.41) -- (400.46,114.52) ;
\draw [shift={(398.4,116.7)}, rotate = 313.28] [fill={rgb, 255:red, 0; green, 0; blue, 0 }  ][line width=0.08]  [draw opacity=0] (8.93,-4.29) -- (0,0) -- (8.93,4.29) -- cycle    ;
\draw    (454.57,236.04) -- (409.25,250.77) ;
\draw [shift={(406.4,251.7)}, rotate = 341.98] [fill={rgb, 255:red, 0; green, 0; blue, 0 }  ][line width=0.08]  [draw opacity=0] (8.93,-4.29) -- (0,0) -- (8.93,4.29) -- cycle    ;
\draw    (242.4,261.76) .. controls (248.8,268.34) and (319.56,270.72) .. (323.96,271.09) .. controls (328.36,271.46) and (327.47,271.7) .. (331.46,274.09) .. controls (335.45,276.48) and (334.29,271.15) .. (340.46,271.59) .. controls (346.63,272.03) and (408.46,272.59) .. (430.96,261.09) ;

\draw (493.12,245.4) node [anchor=north west][inner sep=0.75pt]    {$\mathcal{M}$};
\draw (525.94,21.37) node [anchor=north west][inner sep=0.75pt]    {$\mathcal{SO}_+^\uparrow(1,d-1)$};
\draw (155,49.78) node [anchor=north west][inner sep=0.75pt]    {$O_+^\uparrow(\M,g)$};
\draw (226.28,50.99) node [anchor=north west][inner sep=0.75pt]    {$\supp \mu _{\omega_{x_{1}}}^{\E_{\R}^{x_{1}}}$};
\draw (115.95,207.89) node [anchor=north west][inner sep=0.75pt]  [font=\small]  {$\supp \mu _{\omega_{x_{1}}}^{\F_{\R,\pi}^{x_{1}}}$};
\draw (282,236.4) node [anchor=north west][inner sep=0.75pt]  [font=\tiny]  {$x_{1}$};
\draw (374,243.14) node [anchor=north west][inner sep=0.75pt]  [font=\tiny]  {$x_{2}$};
\draw (399.28,41.99) node [anchor=north west][inner sep=0.75pt]    {$\supp \mu _{\omega_{x_{2}}}^{\E_{\R}^{x_{2}}}$};
\draw (450,215.89) node [anchor=north west][inner sep=0.75pt]    {$\supp \mu _{\omega_{x_2}}^{\F_{\R,\pi}^{x_{2}}}$};
\draw (300.06,279.28) node [anchor=north west][inner sep=0.75pt]   [align=left] {Loc$_\pi(\R,\omega)$};
\end{tikzpicture}
    \caption{Spacetime localization of the QRF (given a measurable projection $\pi$) prepared in state $\omega = (\omega_x) \in \framestate$ as the union of the spacetime localization of the preparations at different $x \in \mathcal{M}$ of the QRF. Here, only two frame $\omega_{x_1}$ and $\omega_{x_2}$ are displayed, and for a single $\omega = (\omega_x)$, for clarity of the representation.}
    \label{fig:localization gr QRF}
\end{figure*}

Any $\mathscr{B}$-compatible projection yields a covariant notion of frame localization, as we now show.

\begin{proposition}
    \label{prop:covariance localization} Let $\R$ be a general relativistic QRF, $\mathscr{B} \leq  \mathrm{Bis}(\mathrm{Poin}(\M,g))$ and $\pi = (\pi_x : \mathrm{Poin}(\M,g)^x \to \M)$ be a $\mathscr{B}$-compatible measurable projection. Then $\forall \omega \in \framestate$ and $\forall b \in \mathscr{B}$,
    \begin{equation}
        \text{Loc}_\pi(\R,\omega^b) = \varphi_b(\text{Loc}_\pi(\R,\omega))
    \end{equation}
    where $\omega^b = \left(\omega^b_x = U_{b(\varphi_b^{-1}(x))} \omega_{\varphi_b^{-1}(x)} U_{b(\varphi_b^{-1}(x))}^\dagger\right)$.
\end{proposition}

\begin{proof}
    See App.~\ref{app:proof covariance localization}.
\end{proof}

Thus, $\mathscr{B}$-compatible measurable projections really do encode a covariant notion of localization in spacetime, which is physically meaningful when e.g. $\mathscr{B} = \mathrm{Bis}_{\mathrm{Isom}}(\mathrm{Poin}(\M,g))$.

\subsection{Relational quantum fields} In the present work, we want to understand how relational quantum fields can emerge from the relativization of \emph{operators} (\textit{a fortiori} observables). A similar picture can be drawn, though we now relativize operator fields $\phi = (\phi_x)_{x \in \mathcal{M}} \in \A(\scrhis)$, such that
\begin{align}
    \label{eqn:absolute quantum field def}
    \hat{\phi} : \mathcal{M} &\to \bigsqcup_{x \in M} \bhsx \\
    \hat{\phi}(x) &:= \phi_x
\end{align}
is, in fact, a quantum field defined over spacetime points. Regardless, we find it more elegant to stick to the picture of relativizing operators and \emph{recovering} a notion of relational observables arising as relational quantum fields. Indeed, it highlights that the consideration of relational (general) relativistic quantum physics \emph{necessarily} implies the emergence of quantum fields, rather than having to postulate them. We now go on to generalise the notion of relational quantum fields to curved spacetimes. We do this in several steps.

\begin{definition}
    Let $\phi \in \A(\scrhis)$ and $\R$ be a general relativistic QRF. For $\mu_g$-almost every $x \in \mathcal{M}$, the \emph{relational semilocal quantum field at $x$} is the map
\begin{align}
    \hat{\Phi}^\R_x : \framestate &\to \bhsx \\
    \label{eqn:relational semilocal observable}
    \hat{\Phi}^\R_x(\omega) &:= \euro^{\R}_{\omega}(\phi)_x = \iint_{\mathrm{Poin}(\M,g)^x} V_\beta \phi_{s(\beta)} V_\beta^\dagger \, d\mu^{\E_\R^x}_{\omega_x}(\beta).
\end{align}
    We call the right-hand side of Eqn.\eqref{eqn:relational semilocal observable} a \emph{relational semilocal observable}.
\end{definition}

The name \emph{semilocal} stems from the fact that these operators are defined at a point, but also depend on the points in their vicinity, as we now highlight. By disintegrating the measure $\mu^{\E_\R^x}_{\omega_x}$ with respect to its source marginal $\mu^{\F_{\R,\pi}^x}_{\omega_x}$ (Lem.~\ref{lem:disintegration exists Lie QRF}), we get a family of conditional measures $\{\nu^{\R,\pi}_{\omega_x}(\cdot \mid y)\}_{y \in \supp \mu^{\F_{\R,\pi}^x}_{\omega_x}}$.

\begin{definition}
    Let $\R$ be a general relativistic QRF, $\omega \in \framestate$, $\pi = (\pi_x : \mathrm{Poin}(\M,g)^x \to \M)$ be a family of measurable projections. For $\mu_g$-almost every $x \in \M$, we call
    \begin{align}
    \hat{\phi}^{\R,\pi}_{x,\omega} : \supp \mu^{\F_{\R,\pi}^x}_{\omega_x} &\to \bhsx \\
    \hat{\phi}^{\R,\pi}_{x,\omega}(y) &:= \int_{\pi_x^{-1}(\{y\})} V_\beta \phi_{s(\beta)} V_\beta^\dagger d\nu^{\R,\pi}_{\omega_x}(\beta \mid y)
\end{align}
    a \emph{relational local quantum field at $x$ given $\pi$}.\footnote{Strictly speaking these are representatives of $\mu^{\F_{\R,\pi}^x}_{\omega_x}$-almost everywhere defined kernels given in Eqn.~\eqref{eqn:quantum field disintegration curved}.}
\end{definition}

\begin{remark}
    Note that relational local quantum fields now depend on the choice of measurable projection. Further note that this projection may, but need not, be $\mathscr{B}$-compatible for a nontrivial $\mathscr{B} \leq \mathrm{Bis}(\mathrm{Poin}(\M,g))$.
\end{remark}

For example, if $\pi_x = s|_{\mathrm{Poin}(\M,g)^x}$, then $\pi_x^{-1}(\{y\}) = \mathrm{Poin}(\M,g)^x_y$ and so the ``quantum field at $x$" is some average of the absolute quantum field over points ``around x" weighted by the frame's probability measure. We obtain the local expression
\begin{equation}
    \hat{\Phi}^\R_x(\omega) = \int_{\mathcal{M}} \hat{\phi}^{\R,\pi}_{x,\omega}(y) \, d\mu^{\F_{\R,\pi}^x}_{\omega_x}(y),
\end{equation}
since
\begin{equation}
    \label{eqn:quantum field disintegration curved}
    \hat{\Phi}^\R_x(\omega)=\iint_{\mathrm{Poin}(\M,g)^x} V_\beta \phi_{s(\beta)} V_\beta^\dagger \, d\mu^{\E_\R^x}_{\omega_x}(\beta) = \int_{\mathcal{M}} \underbrace{\left(\int_{\pi_x^{-1}(\{y\})} V_\beta \phi_{s(\beta)} V_\beta^\dagger \, d\nu^{\R,\pi}_{\omega_x}(\beta \mid y)\right)}_{\hat{\phi}^{\R,\pi}_{x,\omega}(y)} \, d\mu^{\F_{\R,\pi}^x}_{\omega_x}(y).
\end{equation}

Thus, for each target point $x$, the relational field is obtained by averaging the absolute field over all source points $y$ and all local frame comparisons $\beta:y\to x$ weighted by the quantum statistics of the frame. 

\begin{theorem}
    Let $\R$ be a general relativistic QRF and $\phi \in \A(\scrhis)$. Then for all $\omega \in \framestate$ and $\mu_g$-almost every $x \in M$, $\exists f^{\F_{\R,s}^x}_{\omega_x} \equiv \frac{d\mu^{\F_{\R,s}^x}_{\omega_x}}{d\mu_g} \in L^1(M,\mu_g)$, such that
    \begin{equation}
        \hat{\Phi}^\R_x(\omega) = \int_\M \hat{\phi}^{\R,s}_{x,\omega}(y) f^{\F_{\R,s}^x}_{\omega_x}(y) d\mu_g(y).
    \end{equation}
\end{theorem}

\begin{proof}
    For a general relativistic QRF, $\mu_g$-admissibility already says $\F_{\R,s}^x \ll \mu_g$, thus for every $\omega_x \in \framestatex$ we have $\mu^{\F_{\R,s}^x}_{\omega_x} \ll \mu_g$ and there thus exists $f^{\F_{\R,s}^x}_{\omega_x} \equiv \frac{d\mu^{\F_{\R,s}^x}_{\omega_x}}{d\mu_g} \in L^1(M,\mu_g)$.
\end{proof}

\begin{remark}
    Note that if $\pi=t$ then $\F_{\R,t}^x(N) = \text{id}_N(x) \id_{\hirx}$ so $\mu^{\F_{\R,t}^x}_{\omega_x} = \delta_x$ hence $\hat{\Phi}^\R_x(\omega) = \hat{\phi}^{\R,t}_{x,\omega}(x)$. In such a case, there does not exist an $f^{\F_{\R,t}^x}_{\omega_x}$ such that $\hat{\Phi}^\R_x(\omega) = \int_\M \hat{\phi}^{\R,t}_{x,\omega}(y) f^{\F_{\R,t}^x}_{\omega_x}(y) d\mu_g(y)$.
\end{remark}

\begin{remark}
    It may be that Thm.~\ref{thm:covariance implies mu continuity} can be generalised to curved spacetime for certain $\mathscr{B}$-compatible $\pi$ (e.g. when $\mathscr{B} = \mathrm{Bis}_{\mathrm{Isom}}(\mathrm{Poin}(\M,g))$) so that we can sometimes write $d\mu^{\F_{\R,\pi}^x}_{\omega_x}(y) = f^{\F_{\R,\pi}^x}_{\omega_x}(y) d\mu_g(y)$ with $f^{\F_{\R,\pi}^x}_{\omega_x} \in L^1(\mathcal{M},\mu_g)$. In such cases, we can write the relational semilocal quantum fields as
    \begin{equation}
        \hat{\Phi}^\R_x(\omega) = \int_{\mathcal{M}} \hat{\phi}^{\R,\pi}_{x,\omega}(y) f^{\F_{\R,\pi}^x}_{\omega_x}(y) d\mu_g(y)
    \end{equation}
    as a direct generalisation of Eqn.~\eqref{eqn:rlo-as-smearing}. We keep such a generalisation of Thm.~\ref{thm:covariance implies mu continuity} to curved spacetimes to future work.
\end{remark}

From there, we can define relational quantum fields in curved spacetimes.

\begin{definition}
    Let $\R$ be a general relativistic QRF, $\mu_g$ be the spacetime volume measure determined by the metric $g$, and $\phi \in \A(\scrhis)$. The \emph{relational quantum field} is the map
    \begin{align}
        \hat{\Phi}^\R : \framestate &\to \A(\scrhis) \\
        \hat{\Phi}^\R(\omega) &:= \euro^\R_\omega(\phi)=
        \label{eqn:relational observable} \int_\M^\oplus \hat{\Phi}^\R_x(\omega) \, d\mu_g(x).
    \end{align}
    We call the right-hand side of Eqn.~\eqref{eqn:relational observable} a \emph{relational  observable}. 
\end{definition}

\begin{example}
    For $\pi=t$, we have
    \begin{equation}
        \hat{\Phi}^\R(\omega) = \int_\M^\oplus \hat{\phi}^{\R,t}_{x,\omega}(x) \, d\mu_g(x).
    \end{equation}
    For more general projections $\pi$, there are now two notions of integration over spacetime: one at the level of the relational observable, the other at the level of the relational semilocal observable.
\end{example}

\emph{Relational groupoid covariance.} In the absence of a global symmetry group, covariance is expressed relative to arrows of the Poincar\'e groupoid.

\begin{theorem}
    \label{thm:covariance semilocal}
    Let $\R$ be a general relativistic QRF. For $\mu_g$-almost every $x,z \in \mathcal{M}$, let $\alpha : x \to z$ be an arrow of $\mathrm{Poin}(\M,g)$ and $\omega = (\omega_x) \in \framestate$. Write the transformed frame state as $\omega_z^\alpha := U_\alpha \omega_x U_\alpha^\dagger$ so that for $\mu_g$-almost every $x \in \mathcal{M}$, $\omega^\alpha_x = \begin{cases}
        \omega_x \text{ if } x \neq z \\
        \omega_z^\alpha \text{ if } x=z
    \end{cases}$. Then for all $\omega \in \framestate$,
    \begin{equation}
        V_\alpha \hat{\Phi}^\R_x(\omega)V_\alpha^\dagger = \hat{\Phi}^\R_z(\omega^\alpha).
    \end{equation}
\end{theorem}

\begin{proof}
    We have
    \begin{equation}
        V_\alpha \hat{\Phi}^\R_x(\omega) V_\alpha^\dagger = V_\alpha \euro^\R_{\omega}(\phi)_x V_\alpha^\dagger \stackrel{\ref{thm:covariance euro}}{=} \euro^\R_{\omega^{\alpha}}(\phi)_z = \hat{\Phi}^\R_z(\omega^\alpha).
    \end{equation}
\end{proof}

Likewise, we recover a notion of covariance at the level of the whole relational quantum field.

\begin{theorem}
    \label{thm:covariance rqf}
    Let $\R$ be a general relativistic QRF, $\phi \in \A(\scrhis)$ and $\omega = (\omega_x) \in \framestate$. Then for all $b \in \mathrm{Bis}(\mathrm{Poin}(\M,g))$,
    \begin{equation}
        \mathbf{V}_{b} \hat{\Phi}^\R(\omega) \mathbf{V}_{b}^\dagger = \hat{\Phi}^\R(\omega^b)
    \end{equation}
    where $\omega^b = \left(\omega^b_x = U_{b(\varphi_b^{-1}(x))} \omega_{\varphi_b^{-1}(x)} U_{b(\varphi_b^{-1}(x))}^\dagger\right)$.
\end{theorem}

\begin{proof}
    This follows immediately from above and Lem.~\ref{lem:general relativistic QRF is bisection admissible} and Cor.~\ref{cor:invariance and covariance bisections euro}.
\end{proof}

Finally, we can discuss relational covariance at the level of the relational local quantum fields, as follows.

\begin{theorem}
    \label{thm:covariance local rqf}
    Let $\R$ be a general relativistic QRF, $\pi$ be a family of measurable projections. For $\mu_g$-almost every $x,z \in \mathcal{M}$, let $\alpha : x \to z$ be an arrow of $\mathrm{Poin}(\M,g)$ and $\omega = (\omega_x) \in \framestate$. Let $\tau_\alpha : \M \to \M$ be a Borel isomorphism such that $\pi_z \circ L_\alpha = \tau_\alpha \circ \pi_x$, where $L_\alpha : \mathrm{Poin}(\M,g)^{s(\alpha)} \to \mathrm{Poin}(\M,g)^{t(\alpha)}$ is such that $L_\alpha(\beta) = \alpha \circ \beta$. Write transformed frame states as $\omega_z^\alpha := U_\alpha \omega_x U_\alpha^\dagger$ so that for $\mu_g$-almost every $x \in \mathcal{M}$, $\omega^\alpha_x = \begin{cases}
        \omega_x \text{ if } x \neq z \\
        \omega_z^\alpha \text{ if } x=z
    \end{cases}$. Then for all $\omega \in \framestate$ and for $\mu^{\F_{\R,\pi}^x}_{\omega_x}$-almost every $y \in \M$, 
    \begin{equation}
    \label{eqn:conditional-probability-covariance}
        \nu_{\omega_z^\alpha}^{\R,\pi}
        (\,\cdot\mid\tau_\alpha(y))
        =
        (L_\alpha)_*
        \nu_{\omega_x}^{\R,\pi}(\,\cdot\mid y).
    \end{equation}
     Consequently, for all $\phi \in \A(\scrhis)$,
    \begin{equation}
        V_\alpha \hat{\phi}^{\R,\pi}_{x, \omega}(y) V_{\alpha}^\dagger = \hat{\phi}^{\R,\pi}_{z,\omega^\alpha}(\tau_\alpha(y)).
    \end{equation}
    
\end{theorem}

\begin{proof}
    See App.~\ref{app:proof covariance local rqf}
\end{proof}

\begin{remark}
    For a $\mathscr{B}$-compatible $\pi$, $\tau_{b(x)} = \varphi_b$ for any $b \in \mathscr{B}$, while for the source projection, $\tau_\alpha = \text{id}_M$.
\end{remark}

\begin{remark}
    In Minkowski spacetime, we see that this recovers the covariance of Prop.~\ref{prop:relational-covariance-rlo}: the groupoid action reduces to the action of bisections which is just the action of the proper orthochronous Poincaré group, and since the projection is Poincaré-covariant we have $y \mapsto \Lambda y + a$. Moreover, all different fibers are identified, so $z=x$. Finally, $\omega^\alpha = U_\R(a,\Lambda) \omega U_\R(a,\Lambda)^\dagger$, which proves the identification.
\end{remark}

Thus, relational covariance can be understood at three levels: for relational local quantum fields, for relational semilocal quantum fields, and for relational quantum fields.

\subsection{Localization limit} The relational description reduces to the ordinary one when the frame becomes sharply
localized at the identity arrow.

\begin{theorem}
    \label{thm:localization rqf}
    Let $\R$ be a localizable general relativistic QRF. Let $x \in \M$ and suppose that $1_x \in \supp \E_\R^x$. Let $(\omega_n^x) \subset \framestatex$ be a localizing sequence centred around $1_x$, with $\omega_n = (\omega_n^x) \in \framestate$. Then for any $\phi \in \A(\scrhis)_{\mathrm{uw}}$,
    \begin{equation}
        \lim_{n \to \infty} \hat{\Phi}^\R_x(\omega_n) = \hat{\phi}(x)
    \end{equation}
    ultraweakly in $\bhsx$. Moreover, if for $\mu_g$-almost every $x \in \M$, $1_x \in \supp \E_\R^x$, we have that
    \begin{equation}
        \lim_{n \to \infty} \hat{\Phi}^\R(\omega_n) = \phi
    \end{equation}
    ultraweakly in $\A(\scrhis)$.
\end{theorem}

\begin{proof}
    This follows immediately from Thm.~\ref{thm:euro localisation limit} and the fact that $\hat{\Phi}^\R(\omega_n) = \euro^\R_{\omega_n}(\phi)$ and $\hat{\Phi}^\R_x(\omega_n) = \euro^\R_{\omega_n}(\phi)_x$.
\end{proof}

Thus, absolute quantum fields are recovered in the localization limit of the relational description, just as for relational quantum field theory in Minkowski spacetime.

\subsection{Relational causality}
\label{subsec:relational causality}

The groupoid framework also admits a natural curved-spacetime version of finite-precision
causality. 

\begin{definition}
    We say that two points $x,y\in\mathcal M$ are \emph{locally spacelike separated}, written $x \indep y$, if they lie in a common geodesically convex normal neighbourhood (CNN) and the unique geodesic segment connecting them in that neighbourhood has spacelike tangent. For $\U,\V \subset \mathcal{M}$, we write 
    \begin{equation}
        CNN(\U,\V) := \{(x,y) \in \U \times \V \mid x, y \text{ lie in a common geodesically CNN}\}.
    \end{equation}
\end{definition}

We use Synge's world function $\Omega(x,y) = \frac{1}{2} g_x(\exp_x^{-1}(y),\exp_x^{-1}(y))$ to characterise local spacelike separation in curved spacetimes. Recall that we work in mostly minus signature.

\begin{definition}
   Let $\sigma\geq0$ and $\U,\V \subset \mathcal{M}$. We say that $\U$ and $\V$ are $\sigma$-locally spacelike separated if $\U\subseteq\V'$ and
\begin{equation}
\label{eqn:condition sigma causality}
 \sup_{(x,y)\in CNN(\U,\V)}
 g_x(\exp_x^{-1}(y),\exp_x^{-1}(y))<-\sigma.
\end{equation}
If $CNN(\U,\V) = \varnothing$, we work with the convention that $\sup \varnothing = - \infty$ i.e. \ \eqref{eqn:condition sigma causality} is satisfied. 

We write $x \indepsigma y$ if $\{x\} \indepsigma \{y\}$. Moreover, let $\R = (\U_\R,\E_\R,\hir)$ be a general relativistic QRF and $\omega^{(1)}_x \in \framestatex,\omega^{(2)}_x \in \mathscr D(\hir^y)$. We say that $\omega_x^{(1)}$ and $\omega^{(2)}_y$ are $(\R,\sigma,\pi)$-locally spacelike separated, written $\omega^{(1)}_x \indepERsigma \omega^{(2)}_y$, if $\supp \mu^{\F^x_{\R,\pi}}_{\omega^{(1)}_x} \indepsigma \supp \mu^{\F^y_{\R,\pi}}_{\omega^{(2)}_y}$. We write $\omega^{(1)}_x \indepER \omega^{(2)}_y$ if $\omega^{(1)}_x \indepERsigma \omega^{(2)}_y$ for $\sigma=0$. For $\omega^{(1)} = (\omega^{(1)}_x)$ and $\omega^{(2)} = (\omega^{(2)}_x)$, we write $\omega^{(1)} \indepERsigma \omega^{(2)}$ if $\text{Loc}_\pi(\R,\omega^{(1)}) \indepsigma \text{Loc}_\pi(\R,\omega^{(2)})$.
\end{definition}

As in the flat spacetime case, the distinguished family of states $\operationalstate$ of the frame $\R$ naturally defines a \enquote{cut-off scale} for spacelike separation as follows.

\begin{definition}
    \label{def:spacelike resolution of R curved}
    Let $\R$ be a general relativistic QRF, $\pi = (\pi_x : \mathrm{Poin}(\M,g)^x \to \M)$ be a measurable projection, and $\operationalstate \subseteq \framestate$ be convex. The quantity
    \begin{equation}
        \sigma(\operationalstate,\pi) := \inf_{\omega \in \operationalstate} \inf_{z \in M}  \sup \{\abs{g_x(\exp_x^{-1}(y),\exp_x^{-1}(y))} \mid x \indep y \in \supp \mu^{\F^{z}_{\R,\pi}}_{\omega_z}\}
    \end{equation}
    is the \emph{$(\operationalstate,\pi)$-spacelike resolution of $\R$}.
\end{definition}
Any local spacelike separation $\abs{g_x(\exp_x^{-1}(y),\exp_x^{-1}(y))} \geq \sigma(\operationalstate,\pi)$ is, from the \enquote{blurry} point of view of the QRF, potentially causally separated. It is only if the frame can, in principle, perfectly distinguish spacelike-separated points that $\sigma(\operationalstate,\pi)=0$. Def. \ref{def:spacelike resolution of R curved} is the curved spacetime version of the same notion introduced in Def.~\ref{def:spacelike resolution of R} \cite{fedida_foundations_2025}. A pictorial representation of this idea is given in Fig.~\ref{fig:approximate causality}.

\begin{remark}
    Note that this notion is now dependent on the choice of $\pi$. This is not an issue in the flat spacetime case, since there is a canonical $\Poincup$-compatible projection which is simply the trivial projection $\pi : O_+^\uparrow(\mink,\eta) \to \mink$ of the trivial Lorentz frame bundle over Minkowski, which determines a canonical Poincaré-covariant spacetime marginal POVM. There is no such global trivialisation for general curved spacetime, so we now have to specify a choice of such $\pi$.
\end{remark}

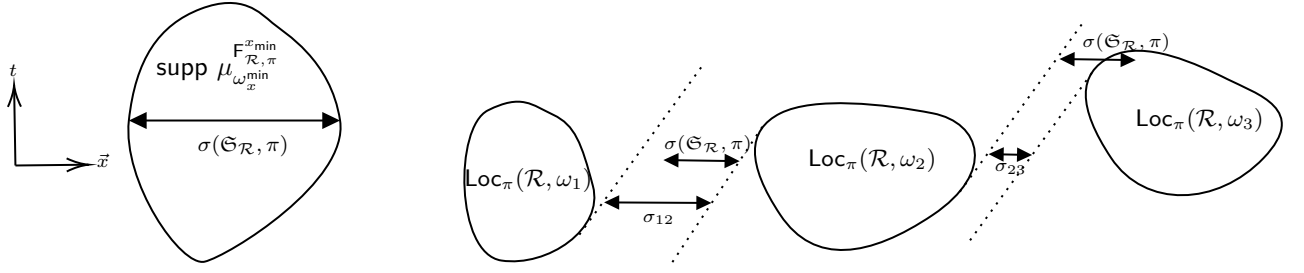
\begin{figure}
    \centering
    \begin{subfigure}[t]{0.3\textwidth}
        \begin{tikzpicture}[x=0.75pt,y=0.75pt,yscale=-1,xscale=1]

\draw   (353.6,97.1) .. controls (374.6,89.1) and (383.6,99.1) .. (389.6,104.1) .. controls (395.6,109.1) and (418.93,122.43) .. (423.6,153.1) .. controls (428.27,183.77) and (367.6,221.1) .. (355.67,223.93) .. controls (343.73,226.77) and (314.27,183.77) .. (317.6,153.1) .. controls (320.93,122.43) and (332.6,105.1) .. (353.6,97.1) -- cycle ;
\draw    (260.67,175.67) -- (260.03,137.27) ;
\draw [shift={(260,135.27)}, rotate = 89.05] [color={rgb, 255:red, 0; green, 0; blue, 0 }  ][line width=0.75]    (10.93,-3.29) .. controls (6.95,-1.4) and (3.31,-0.3) .. (0,0) .. controls (3.31,0.3) and (6.95,1.4) .. (10.93,3.29)   ;
\draw    (260.67,175.67) -- (295.33,175.29) ;
\draw [shift={(297.33,175.27)}, rotate = 179.37] [color={rgb, 255:red, 0; green, 0; blue, 0 }  ][line width=0.75]    (10.93,-3.29) .. controls (6.95,-1.4) and (3.31,-0.3) .. (0,0) .. controls (3.31,0.3) and (6.95,1.4) .. (10.93,3.29)   ;
\draw    (320.6,153.1) -- (420.6,153.1) ;
\draw [shift={(423.6,153.1)}, rotate = 180] [fill={rgb, 255:red, 0; green, 0; blue, 0 }  ][line width=0.08]  [draw opacity=0] (8.93,-4.29) -- (0,0) -- (8.93,4.29) -- cycle    ;
\draw [shift={(317.6,153.1)}, rotate = 0] [fill={rgb, 255:red, 0; green, 0; blue, 0 }  ][line width=0.08]  [draw opacity=0] (8.93,-4.29) -- (0,0) -- (8.93,4.29) -- cycle    ;

\draw (256,123.4) node [anchor=north west][inner sep=0.75pt]  [font=\scriptsize]  {$t$};
\draw (300,168.07) node [anchor=north west][inner sep=0.75pt]  [font=\scriptsize]  {$\vec{x}$};
\draw (350.67,158.73) node [anchor=north west][inner sep=0.75pt]  [font=\footnotesize]  {$\sigma ( \operationalstate,\pi)$};
\draw (330,111.4) node [anchor=north west][inner sep=0.75pt]    {$\supp \mu _{\omega_x^\text{min}}^{\F^{x_\text{min}}_{\R,\pi}}$};
\end{tikzpicture}

    \caption{Spacetime support (given the measurable projection $\pi$) of the state preparation $\omega_x^\text{min} \in \operationalstate$ which provides the ``smallest`` spacelike resolution $\sigma(\operationalstate,\pi)$ for the frame $\R$ out of all ``operationally meaningful'' states $\omega \in \operationalstate$.}
    \end{subfigure}
    \qquad
    \begin{subfigure}[t]{0.6\textwidth}
        \begin{tikzpicture}[x=0.75pt,y=0.75pt,yscale=-0.85,xscale=0.85]

\draw   (118,116.27) .. controls (138,106.27) and (144,112.93) .. (156.67,118.93) .. controls (169.33,124.93) and (173.33,135.6) .. (180,160.93) .. controls (186.67,186.27) and (153.33,208.93) .. (126,204.93) .. controls (98.67,200.93) and (98,126.27) .. (118,116.27) -- cycle ;
\draw  [dash pattern={on 0.84pt off 2.51pt}]  (172,189.67) -- (245.33,88.93) ;
\draw   (290.67,118) .. controls (310.67,108) and (352,112.4) .. (380.67,118) .. controls (409.33,123.6) and (418,152.93) .. (380.67,178) .. controls (343.33,203.07) and (310.67,208) .. (290.67,178) .. controls (270.67,148) and (270.67,128) .. (290.67,118) -- cycle ;
\draw  [dash pattern={on 0.84pt off 2.51pt}]  (280,126) -- (225.33,208.27) ;
\draw    (189.6,170.62) -- (247,170.92) ;
\draw [shift={(250,170.93)}, rotate = 180.3] [fill={rgb, 255:red, 0; green, 0; blue, 0 }  ][line width=0.08]  [draw opacity=0] (8.93,-4.29) -- (0,0) -- (8.93,4.29) -- cycle    ;
\draw [shift={(186.6,170.6)}, rotate = 0.3] [fill={rgb, 255:red, 0; green, 0; blue, 0 }  ][line width=0.08]  [draw opacity=0] (8.93,-4.29) -- (0,0) -- (8.93,4.29) -- cycle    ;
\draw    (224.6,145.64) -- (263.67,146.22) ;
\draw [shift={(266.67,146.27)}, rotate = 180.85] [fill={rgb, 255:red, 0; green, 0; blue, 0 }  ][line width=0.08]  [draw opacity=0] (8.93,-4.29) -- (0,0) -- (8.93,4.29) -- cycle    ;
\draw [shift={(221.6,145.6)}, rotate = 0.85] [fill={rgb, 255:red, 0; green, 0; blue, 0 }  ][line width=0.08]  [draw opacity=0] (8.93,-4.29) -- (0,0) -- (8.93,4.29) -- cycle    ;
\draw   (486,84.4) .. controls (506,74.4) and (537.33,89.2) .. (554.67,97.87) .. controls (572,106.53) and (603.33,115.87) .. (576,144.4) .. controls (548.67,172.93) and (506,174.4) .. (486,144.4) .. controls (466,114.4) and (466,94.4) .. (486,84.4) -- cycle ;
\draw  [dash pattern={on 0.84pt off 2.51pt}]  (396.67,164.33) -- (470,63.6) ;
\draw  [dash pattern={on 0.84pt off 2.51pt}]  (402.67,193) -- (476,92.27) ;
\draw    (416.6,142.16) -- (436.33,142.57) ;
\draw [shift={(439.33,142.63)}, rotate = 181.19] [fill={rgb, 255:red, 0; green, 0; blue, 0 }  ][line width=0.08]  [draw opacity=0] (8.93,-4.29) -- (0,0) -- (8.93,4.29) -- cycle    ;
\draw [shift={(413.6,142.1)}, rotate = 1.19] [fill={rgb, 255:red, 0; green, 0; blue, 0 }  ][line width=0.08]  [draw opacity=0] (8.93,-4.29) -- (0,0) -- (8.93,4.29) -- cycle    ;
\draw    (458.47,85.98) -- (497.53,86.56) ;
\draw [shift={(500.53,86.6)}, rotate = 180.85] [fill={rgb, 255:red, 0; green, 0; blue, 0 }  ][line width=0.08]  [draw opacity=0] (8.93,-4.29) -- (0,0) -- (8.93,4.29) -- cycle    ;
\draw [shift={(455.47,85.93)}, rotate = 0.85] [fill={rgb, 255:red, 0; green, 0; blue, 0 }  ][line width=0.08]  [draw opacity=0] (8.93,-4.29) -- (0,0) -- (8.93,4.29) -- cycle    ;

\draw (207.67,175.4) node [anchor=north west][inner sep=0.75pt]  [font=\scriptsize]  {$\sigma _{12}$};
\draw (220.33,129.07) node [anchor=north west][inner sep=0.75pt]  [font=\scriptsize]  {$\sigma ( \operationalstate,\pi)$};
\draw (102.33,148.73) node [anchor=north west][inner sep=0.75pt]  [font=\small]  {$\text{Loc}_\pi(\R,\omega_1)$};
\draw (305.33,136.73) node [anchor=north west][inner sep=0.75pt]  [font=\small]  {$\text{Loc}_\pi(\R,\omega_2)$};
\draw (498.67,113.4) node [anchor=north west][inner sep=0.75pt]  [font=\small]  {$\text{Loc}_\pi(\R,\omega_3)$};
\draw (415.6,145.5) node [anchor=north west][inner sep=0.75pt]  [font=\scriptsize]  {$\sigma _{23}$};
\draw (467.33,69.07) node [anchor=north west][inner sep=0.75pt]  [font=\scriptsize]  {$\sigma ( \operationalstate,\pi)$};

\end{tikzpicture}

    \caption{Examples of $\sigma$-spacelike separations for different state preparations of the frame $\R$ with $(\operationalstate,\pi)$-spacelike resolution $\sigma(\operationalstate,\pi)$. We see that $\sigma_{12} > \sigma(\operationalstate,\pi)$ i.e. the preparations $\omega_1$ and $\omega_2$ are ``operationally spacelike separated'': the frame can resolve this spacelike separation. On the other hand, $\sigma_{23} < \sigma(\operationalstate,\pi)$ i.e. the frame's spacelike resolution is too small to determine whether $\omega_2$ and $\omega_3$ should be understood as yielding causally disconnected supports: relational quantum fields generated by these preparations may not commute.}
    \end{subfigure}
    \caption{Pictorial representations of the frame resolutions and of $\sigma(\operationalstate,\pi)$-spacelike separation of regions on curved spacetime.}
    \label{fig:approximate causality}
\end{figure}

\begin{definition}
    Let $\mathcal{O}_\S \subset \A(\scrhis)$ be a set of measurable operator fields closed under taking adjoints, $\sigma \geq 0$, $\pi = (\pi_x : \mathrm{Poin}(\M,g)^x \to \M)$ be a measurable projection, $\R$ be a general relativistic QRF and $\operationalstate \subset \framestate$ be convex. We say that $\mathcal{O}_\S$ is
    \begin{enumerate}
        \item \emph{absolutely $(\sigma,\pi)$-microcausal} if for all $\phi,\psi \in \mathcal{O}_\S$ and $\mu_g$-almost every $z\in \mathcal{M}$ and $\forall \beta,\gamma \in \mathrm{Poin}(\M,g)^z$,
        \begin{equation}
            \pi_z(\beta) \indepsigma \pi_z(\gamma) \Longrightarrow \comm{V_\beta \phi_{s(\beta)} V_\beta^\dagger}{V_\gamma \psi_{s(\gamma)} V_\gamma^\dagger} = \comm{V_\beta \phi_{s(\beta)}^\dagger V_\beta^\dagger}{V_\gamma \psi_{s(\gamma)} V_\gamma^\dagger} = 0.
        \end{equation}
        \item \emph{strongly $(\operationalstate,\sigma,\pi)$-microcausal} if for all $\phi,\psi \in \mathcal{O}_\S$ and $\mu_g$-almost every $z\in \mathcal{M}$ and every pair $\omega^{(1)} = (\omega^{(1)}_x)$ and $\omega^{(2)} = (\omega^{(2)}_x)$ of frame states in $\operationalstate$ and $\forall x \in \supp \mu^{\F_{\R,\pi}^z}_{\omega^{(1)}_z}$ and $\forall y \in \mu^{\F_{\R,\pi}^z}_{\omega^{(2)}_z}$,
        \begin{equation}
            x \indepsigma y \Longrightarrow\comm{\hat{\phi}^{\R,\pi}_{z,\omega^{(1)}}(x)}{\hat{\psi}^{\R,\pi}_{z,\omega^{(2)}}(y)} = \comm{\hat{\phi}^{\R,\pi}_{z,\omega^{(1)}}(x)^\dagger}{\hat{\psi}^{\R,\pi}_{z,\omega^{(2)}}(y)} = 0,
        \end{equation}
        \item \emph{weakly $(\operationalstate,\sigma,\pi)$-microcausal} if for all $\phi,\psi \in \mathcal{O}_\S$ and every pair $\omega^{(1)} = (\omega^{(1)}_x)$ and $\omega^{(2)} = (\omega^{(2)}_x)$ of frame states in $\operationalstate$ and $\mu_g$-almost every $z \in M$,
        \begin{multline}
            \omega^{(1)}_z \indepERsigma \omega^{(2)}_z \Longrightarrow \comm{\hat{\phi}^{\R,\pi}_{z,\omega^{(1)}}(x)}{\hat{\psi}^{\R,\pi}_{z,\omega^{(2)}}(y)} = \comm{\hat{\phi}^{\R,\pi}_{z,\omega^{(1)}}(x)^\dagger}{\hat{\psi}^{\R,\pi}_{z,\omega^{(2)}}(y)} = 0 \\ \forall x \in \supp \mu^{\F_{\R,\pi}^{z}}_{\omega_z^{(1)}}, y \in \supp \mu^{\F_{\R,\pi}^{z}}_{\omega_z^{(2)}}.
        \end{multline}
        \item \emph{strongly semi $(\operationalstate,\sigma)$-causal} if for all $\phi,\psi \in \mathcal{O}_\S$ and every arrow $\alpha:y\to x$ for $\mu_g$-almost every $x, y \in \M$ such that $x \indepsigma y$ and every pair $\omega^{(1)} = (\omega^{(1)}_z)$ and $\omega^{(2)} = (\omega^{(2)}_z)$ of frame states in $\operationalstate$,
        \begin{equation}
            \comm{\hat{\Phi}^\R_x(\omega^{(1)})}{V_\alpha \hat{\Psi}^\R_y(\omega^{(2)}) V_\alpha^\dagger} = \comm{\hat{\Phi}^\R_x(\omega^{(1)})^\dagger}{V_\alpha \hat{\Psi}^\R_y(\omega^{(2)}) V_\alpha^\dagger} = 0.
        \end{equation}
        \item \emph{weakly semi $(\operationalstate,\sigma,\pi)$-causal} if for all $\phi,\psi \in \mathcal{O}_\S$ and every arrow $\alpha:y\to x$ for $\mu_g$-almost every $x, y \in \M$ such that $x \indepsigma y$ and every pair $\omega^{(1)} = (\omega^{(1)}_z)$ and $\omega^{(2)} = (\omega^{(2)}_z)$ of frame states in $\operationalstate$,
        \begin{equation}
            \omega^{(1)}_x \indepERsigma \omega^{(2)}_y \Longrightarrow \comm{\hat{\Phi}^\R_x(\omega^{(1)})}{V_\alpha \hat{\Psi}^\R_y(\omega^{(2)}) V_\alpha^\dagger} = \comm{\hat{\Phi}^\R_x(\omega^{(1)})^\dagger}{V_\alpha \hat{\Psi}^\R_y(\omega^{(2)}) V_\alpha^\dagger} = 0.
        \end{equation}
        \item $(\operationalstate,\sigma,\pi)$-causal if for all $\phi,\psi \in \mathcal{O}_\S$ and every pair $\omega^{(1)} = (\omega^{(1)}_x)$ and $\omega^{(2)} = (\omega^{(2)}_x)$ of frame states in $\operationalstate$,
        \begin{equation}
            \text{Loc}_\pi(\R,\omega^{(1)}) \indepsigma \text{Loc}_\pi(\R,\omega^{(2)}) \Longrightarrow \comm{\hat{\Phi}^\R(\omega^{(1)})}{\hat{\Psi}^\R(\omega^{(2)})} = \comm{\hat{\Phi}^\R(\omega^{(1)})^\dagger}{\hat{\Psi}^\R(\omega^{(2)})} = 0
        \end{equation}
    \end{enumerate}
\end{definition}

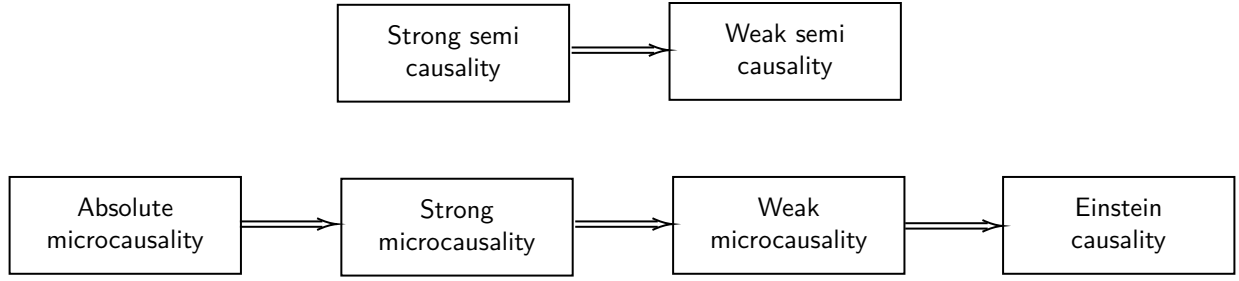
\begin{figure*}
    \centering

    \tikzset{every picture/.style={line width=0.75pt}} 

\begin{tikzpicture}[x=0.75pt,y=0.75pt,yscale=-0.9,xscale=0.9]

\draw    (322.6,202.8) -- (368.6,202.8)(322.6,205.8) -- (368.6,205.8) ;
\draw [shift={(376.6,204.3)}, rotate = 180] [color={rgb, 255:red, 0; green, 0; blue, 0 }  ][line width=0.75]    (10.93,-3.29) .. controls (6.95,-1.4) and (3.31,-0.3) .. (0,0) .. controls (3.31,0.3) and (6.95,1.4) .. (10.93,3.29)   ;
\draw    (506.6,203.8) -- (552.6,203.8)(506.6,206.8) -- (552.6,206.8) ;
\draw [shift={(560.6,205.3)}, rotate = 180] [color={rgb, 255:red, 0; green, 0; blue, 0 }  ][line width=0.75]    (10.93,-3.29) .. controls (6.95,-1.4) and (3.31,-0.3) .. (0,0) .. controls (3.31,0.3) and (6.95,1.4) .. (10.93,3.29)   ;
\draw    (135.6,202.8) -- (181.6,202.8)(135.6,205.8) -- (181.6,205.8) ;
\draw [shift={(189.6,204.3)}, rotate = 180] [color={rgb, 255:red, 0; green, 0; blue, 0 }  ][line width=0.75]    (10.93,-3.29) .. controls (6.95,-1.4) and (3.31,-0.3) .. (0,0) .. controls (3.31,0.3) and (6.95,1.4) .. (10.93,3.29)   ;
\draw    (320.6,105.8) -- (366.6,105.8)(320.6,108.8) -- (366.6,108.8) ;
\draw [shift={(374.6,107.3)}, rotate = 180] [color={rgb, 255:red, 0; green, 0; blue, 0 }  ][line width=0.75]    (10.93,-3.29) .. controls (6.95,-1.4) and (3.31,-0.3) .. (0,0) .. controls (3.31,0.3) and (6.95,1.4) .. (10.93,3.29)   ;

\draw    (377.18,178) -- (506.18,178) -- (506.18,232) -- (377.18,232) -- cycle  ;
\draw (441.68,205) node   [align=left] {\begin{minipage}[lt]{85.24pt}\setlength\topsep{0pt}
\begin{center}
Weak\\microcausality
\end{center}

\end{minipage}};
\draw    (561.18,178) -- (690.18,178) -- (690.18,232) -- (561.18,232) -- cycle  ;
\draw (625.68,205) node   [align=left] {\begin{minipage}[lt]{85.24pt}\setlength\topsep{0pt}
\begin{center}
Einstein\\causality\\
\end{center}

\end{minipage}};
\draw    (7.18,178) -- (136.18,178) -- (136.18,232) -- (7.18,232) -- cycle  ;
\draw (71.68,205) node   [align=left] {\begin{minipage}[lt]{85.24pt}\setlength\topsep{0pt}
\begin{center}
Absolute \\microcausality
\end{center}

\end{minipage}};
\draw    (192.18,179) -- (321.18,179) -- (321.18,233) -- (192.18,233) -- cycle  ;
\draw (256.68,206) node   [align=left] {\begin{minipage}[lt]{85.24pt}\setlength\topsep{0pt}
\begin{center}
Strong\\microcausality\\
\end{center}

\end{minipage}};
\draw    (375.18,81) -- (504.18,81) -- (504.18,135) -- (375.18,135) -- cycle  ;
\draw (439.68,108) node   [align=left] {\begin{minipage}[lt]{85.24pt}\setlength\topsep{0pt}
\begin{center}
Weak semi\\causality
\end{center}

\end{minipage}};
\draw    (190.18,82) -- (319.18,82) -- (319.18,136) -- (190.18,136) -- cycle  ;
\draw (254.68,109) node   [align=left] {\begin{minipage}[lt]{85.24pt}\setlength\topsep{0pt}
\begin{center}
Strong semi\\causality\\
\end{center}

\end{minipage}};

\end{tikzpicture}
    \caption{The relationship between various notions of absolute and relational causality conditions. Strong semi $(\operationalstate,\sigma)$-causality implies weak semi $(\operationalstate,\sigma,\pi)$-causality. Absolute $(\sigma,\pi)$-microcausality implies strong $(\operationalstate,\sigma,s)$-microcausality. Strong $(\operationalstate,\sigma,\pi)$-microcausality implies weak $(\operationalstate,\sigma,\pi)$-microcausality, which implies $(\operationalstate,\sigma,\pi)$-causality. }
    \label{fig:relations between causalities}
\end{figure*}

\begin{remark}
\label{rem:groupoid-microcausality}
In the group-based setting, where all operators act on a single Hilbert space, the absolute microcausality condition with $\sigma=0$ reduces to the familiar requirement $\comm{\hat{\phi}_{\lambda_1}(x_1)}{\hat{\phi}_{\lambda_2}(x_2)}=0$ for spacelike-separated points. On a curved spacetime, however, the field operators live in different fiber algebras, so their commutator is not directly defined. The groupoid microcausality condition transports them to a common fiber before imposing commutativity. Note however that, in the flat spacetime setting, Wightman's theorem \cite{wightman_fields_1965,fedida_no-go_2026} prevents exact ($\sigma=0$) absolute pointwise microcausality to hold under mild conditions (the spectrum condition and the existence of a pure vacuum).
\end{remark}

\begin{theorem}
    \label{thm:causality conditions}
    Let $\mathcal{O}_\S \subseteq \A(\scrhis)$ be a set of measurable operator fields closed under taking adjoints, $\sigma \geq 0$, $\R$ be a general relativistic QRF, $\pi = (\pi_x : \mathrm{Poin}(\M,g)^x \to \M)$ be a measurable projection, and $\operationalstate \subset \framestate$ be convex. Then
    \begin{enumerate}
        \item if $\mathcal{O}_\S$ is absolutely $(\sigma,\pi)$-microcausal, then it is strongly $(\operationalstate,\sigma,\pi)$-microcausal (for any such $\operationalstate$),
        \item if $\mathcal{O}_\S$ is strongly $(\operationalstate,\sigma,\pi)$-microcausal, then it is weakly $(\operationalstate,\sigma,\pi)$-microcausal,
        \item if $\mathcal{O}_\S$ is strongly semi $(\operationalstate,\sigma)$-causal, then it is weakly semi $(\operationalstate,\sigma, \pi)$-causal (for any such $\pi$),
        \item if $\mathcal{O}_\S$ is weakly $(\operationalstate,\sigma,\pi)$-microcausal, then it is $(\operationalstate,\sigma,\pi)$-causal,
    \end{enumerate}
\end{theorem}

\begin{proof}
    See App. \ref{app:proof causality}.
\end{proof}

We can therefore talk about causality in relational quantum field theory in curved spacetime. We keep as future work the task of constructing explicit nontrivial examples of causal relational quantum fields, as was done in Minkowski spacetime in \cite{fedida_foundations_2025} using similar notions of causality.

\begin{remark}
It is an interesting open question to understand whether given two $\mathscr{B}$-compatible measurable projections $\pi_1$ and $\pi_2$ for $\mathscr{B} \leq \mathrm{Bis}(\mathrm{Poin}(\M,g))$, $\mathcal{O}_\S \subset \A(\scrhis)$ is $(\operationalstate,\sigma,\pi_1)$-causal (resp. weakly microcausal, weakly semi-causal) if and only if $\mathcal{O}_\S$ is $(\operationalstate,\sigma,\pi_2)$-causal (resp. weakly microcausal, weakly semi-causal). That is, whether the notion of causality is stable under changes of notions of the frame's localization in spacetime. It would also be interesting to understand the links between semi-causality and the other notions of causality; in particular, we conjecture that if $\R$ is localizable and there exists localizing sequences centred around the identity in $\operationalstate$, then strong semi causality implies absolute microcausality, and Einstein causality becomes equivalent to weak microcausality. We keep such considerations for future work.
\end{remark}

\subsection{From curved to Minkowski spacetime}
\label{subsec:curved-to-flat-reduction}

Suppose now that $(\mathcal{M},g) = (\mink,\eta)$, i.e. $O_+^\uparrow(\mathcal{M},g) = O_+^\uparrow(\mink,\eta) \stackrel{q}{\longrightarrow} \mink$ is the torsor of oriented inertial reference frames on which the proper orthochronous Poincaré group $G=\Poincup$ acts freely and transitively by $(a,\Lambda)\cdot(q,\lambda) =(\Lambda q+a,\Lambda\lambda)$. Choose $o=(0,e) \in O_+^\uparrow(\mink,\eta)$. The right $SO_+^\uparrow(1,d-1)$-action induced by this choice is the usual principal action $(q,\lambda)\cdot \Lambda=(q,\lambda \Lambda)$, so $O_+^\uparrow(\mink,\eta)/SO_+^\uparrow(1,d-1)\cong\mink$. Consequently,
\begin{equation}
 \At(O_+^\uparrow(\mink,\eta)\to\mink)
 \cong \Poincup\ltimes\mink
 \cong \mathrm{Poin}(\mink,\eta).
\end{equation}
Write an arrow of $\mathrm{Poin}(\mink,\eta)$ as $(z,\Lambda,y)$, with source $y$ and target $z$, and corresponding Poincare transformation is $(z-\Lambda y,\Lambda)$. The target-fiber identification is
\begin{equation}
 \Pi_z^o(z,\Lambda,y)=(z-\Lambda y,\Lambda)\in O_+^\uparrow(\mink,\eta).
\end{equation}
Let $\R$ be a relativistic QRF on $O_+^\uparrow(\mink,\eta)$ and $\underline\R$ the associated general relativistic action groupoid QRF (see Sec.~\ref{subsec:reduction to group qrfs}). Its source marginal is
\begin{equation}
 \F_{\underline\R,s}^{z}(N)
 =\E_\R\!\left(\{(a,\Lambda)\mid
                 \Lambda^{-1}(z-a)\in N\}\right),
\end{equation}
which is generally different from the standard frame-bundle marginal $\F_\R(N)=\E_\R(N\times\Lup)$. The relevant projection on the groupoid target fibres is instead
\begin{equation}
\label{eq:pi-flat-standard}
    \pi_z^{\mathrm{std}} := q\circ\Pi_z^o, \qquad
    \pi_z^{\mathrm{std}}(z,\Lambda,y) = z-\Lambda y.
\end{equation}
It follows directly that
\begin{equation}
\label{eq:flat-marginals-identical}
    \F_{\underline{\R},\pi^{\mathrm{std}}}^{z}(N) =    \E_\R(N\times \Lup) = \F_\R(N)
\end{equation}
for $d^dx$-almost every $z\in\mink$ and every $N\in\Bor(\mink)$. The family $\pi^{\mathrm{std}}$ is compatible with the canonical Poincaré subgroup of bisections. Indeed, if
\begin{equation}
    b_{(a,\Lambda)}(z):=(\Lambda z+a,\Lambda,z),
\end{equation}
then $\varphi_{b_{(a,\Lambda)}}(z)=(a,\Lambda)\cdot z$ and
\begin{equation}
    \pi_{(a,\Lambda) \cdot z}^{\mathrm{std}}    \big(L_{b_{(a,\Lambda)}(z)}\beta\big) =   (a,\Lambda) \cdot\pi_z^{\mathrm{std}}(\beta).
\end{equation}
The adjective ``standard'' is relative to the chosen torsor origin $o$:
the bundle projection $q$ is intrinsic, whereas its transfer to every target
fibre through $\Pi_z^o$ uses that choice.

As was seen in Sec.~\ref{subsubsec:groupoid QRFs on torsors}, action groupoid QRFs and relativization on torsors reduce to torsor QRFs and standard torsor relativization when one considers constant fields of POVMs across fibers and constant operator fields of system observables. In such cases, where we build an action groupoid QRF $\underline{\R}$ from a torsor QRF $\R$ as in Thm.~\ref{thm:torsor reduction}, we indeed have
\begin{equation}
    \label{eqn:correspondance relativistic rel with gr rel}\hat{\underline{\Phi}}^{\underline{\R}}(\underline{\omega}) = \int^\oplus_\mink \hat{\Phi}^\R(\omega) \, d^dx
\end{equation}
where $\underline{\omega} = \iota_\R(\omega)$ and $\underline{\phi} = \iota_\S(\phi)$ with $\phi \in \bhs$, and since the volume measure on Minkowski is just $d^dx$. Equivalently, for $d^dx$-almost every $x \in \mink$,
\begin{equation}
    \label{eqn:correspondance relativistic rel with gr rel pointwise}\hat{\underline{\Phi}}^{\underline{\R}}(\underline{\omega})_x = \hat{\Phi}^\R(\omega).
\end{equation}
We can further disintegrate the flat Born measure as
\begin{equation}
 d\mu_\omega^{\E_\R}(a,\lambda)
   =d\nu_\omega^\R(\lambda\mid a)\,
      d\mu_\omega^{\F_\R}(a).
\end{equation}
For fixed $z$ and $x$, define
\begin{equation}
    j_{z,x}: \Lup\to
    (\pi_z^{\mathrm{std}})^{-1}(\{x\}),
    \qquad
    j_{z,x}(\Lambda)
    =
    \bigl(z,\Lambda,\Lambda^{-1}(z-x)\bigr).
\end{equation}
Compatible versions of the conditional measures for the groupoid disintegration are
\begin{equation}
    \nu^{\underline{\R},\pi^{\mathrm{std}}}_{\underline{\omega}_z}
    (\cdot\mid x)
    =
    (j_{z,x})_*
    \nu^\R_\omega(\cdot\mid x)
\end{equation}
for $\mu^{\F_\R}_\omega$-almost every $x$. Consequently, for a general operator field $A=(A_y) \in \A(\scrhis)$, its disintegrated kernel is
\begin{equation}
\label{eq:flat-local-kernel-general}
    \hat{A}^{\underline{\R},\pi^{\mathrm{std}}}
    _{z,\underline{\omega}}(x) = \int_{\Lup}    U_\S(x,\Lambda) A_{\Lambda^{-1}(z-x)}  U_\S(x,\Lambda)^\dagger d\nu^{\E_\R}_\omega(\Lambda\mid x).
\end{equation}
so that for a constant operator field $\underline{\phi} = \iota_\S(\phi)$, we have that for $d^dx$-almost every $z \in \mink$ and $\mu^{\F_\R}_{\omega}$-almost every $x \in \mink$,
\begin{equation}
    \hat{\underline{\phi}}^{\underline{\R},\pi^\text{std}}_{z,\underline{\omega}}(x) = \int_{SO_+^\uparrow(1,d-1)} U_\S(x,\lambda) \phi U_\S(x,\lambda)^\dagger \, d\nu^{\E_\R}_{\omega}(\lambda \mid x) = \hat{\phi}^\R_\omega(x).
\end{equation}
Thus, for such constant operator seeds $\underline{\phi} \in \A(\scrhis)$, the groupoidal relational local quantum field is independent of the target label $z$ and coincides with the flat spacetime relational local quantum field.

As seen in Sec.~\ref{subsubsec:groupoid QRFs on torsors}, we have
\begin{equation}
    \hat{\underline{\Phi}}^{\underline{\R}}(\underline{\omega}) = J_\S(\mathbb{1}_{L^2(\mink)} \otimes \hat{\Phi}^\R(\omega)) J_\S^\dagger
\end{equation}
where 
\begin{align}
    J_\S : L^2(\mink) \otimes \his &\to L^2(\mink;\his) \\
    J_\S(f \otimes \psi)(x) &= f(x) \psi.
\end{align}

We finally discuss the reduction of the causality conditions. Since Minkowski spacetime itself is a geodesically convex normal neighbourhood, the
curved-spacetime relation $x\indepsigma y$ becomes $x \perp_\sigma y \Leftrightarrow (x-y)^2 < - \sigma$. Let $\mathcal{O}_\S\subseteq\bhs$ be closed under adjoints, $\operationalstate \subseteq \homframestate$ be convex, and define
\begin{equation}
    \underline{\mathcal{O}}_\S :=    \iota_\S(\mathcal{O}_\S),  \qquad
    \underline{\operationalstate} := \left\{    \iota_\R(\omega) \mid \omega\in\operationalstate    \right\}.
\end{equation}
In the constant-field sector and with $\pi=\pi^{\mathrm{std}}$, the strong and weak microcausality conditions, as well as the Einstein causality and absolute microcausality conditions, reduce to the corresponding conditions of Sec.~\ref{sec:RQFT Minkowski}. 

\begin{proposition}
\label{prop:flat causality reduction}
Let $\R$ be a relativistic QRF and $\underline{\R}$ be the general relativistic QRF on $O_+^\uparrow(\mink,\eta)$ induced by $\R$. The following equivalences hold:
\begin{enumerate}
    \item $\underline{\mathcal{O}}_\S$ is absolutely
    $(\sigma,\pi^{\mathrm{std}})$-microcausal if and only if
    $\mathcal{O}_\S$ is absolutely $\sigma$-microcausal in the sense of
    Sec.~\ref{sec:RQFT Minkowski};
    \item $\underline{\mathcal{O}}_\S$ is strongly
    $(\underline{\operationalstate},\sigma,\pi^{\mathrm{std}})$-microcausal
    if and only if $\mathcal{O}_\S$ is strongly
    $(\operationalstate,\sigma)$-microcausal;
    \item $\underline{\mathcal{O}}_\S$ is weakly
    $(\underline{\operationalstate},\sigma,\pi^{\mathrm{std}})$-microcausal
    if and only if $\mathcal{O}_\S$ is weakly
    $(\operationalstate,\sigma)$-microcausal;
    \item $\underline{\mathcal{O}}_\S$ is
    $(\underline{\operationalstate},\sigma,\pi^{\mathrm{std}})$-causal
    if and only if $\mathcal{O}_\S$ is
    $(\operationalstate,\sigma)$-causal.
\end{enumerate}
\end{proposition}

\begin{proof}
For $\beta=j_{z,x}(\Lambda)$ and
$\gamma=j_{z,y}(\lambda)$, one has
\begin{align}
    \pi_z^{\mathrm{std}}(\beta)&=x, \qquad  \pi_z^{\mathrm{std}}(\gamma)=y, \\  
    V_\beta\underline{\phi}_{s(\beta)}V_\beta^\dagger &= U_\S(x,\Lambda)\phi U_\S(x,\Lambda)^\dagger = \hat{\phi}_\Lambda(x), \\    V_\gamma\underline{\psi}_{s(\gamma)}V_\gamma^\dagger &= U_\S(y,\lambda)\psi U_\S(y,\lambda)^\dagger = \hat{\psi}_\lambda(y).
\end{align}
Since $j_{z,\cdot}: O_+^\uparrow(\mink,\eta)\to\mathrm{Poin}(\mink,\eta)^z$ is a bijection, the groupoid absolute microcausality condition quantifies over exactly the same $(x,\Lambda)$ and $(y,\lambda)$ as the flat absolute condition. This proves the first equivalence. The equality of the local kernels, the marginals, and the localizations was proved above, from which the second and third equivalences straightforwardly follow.

For the last equivalence, we have that 
\begin{equation}
    J_\S^\dagger \comm{\hat{\underline{\Phi}}^{\underline{\R}}    (\underline{\omega_1})}{    \hat{\underline{\Psi}}^{\underline{\R}}    (\underline{\omega_2})} J_\S =  \comm{J_\S^\dagger\hat{\underline{\Phi}}^{\underline{\R}}    (\underline{\omega_1})J_\S}{ J_\S^\dagger   \hat{\underline{\Psi}}^{\underline{\R}}    (\underline{\omega_2})J_\S}  = \mathbb{1}_{L^2(\mink)} \otimes
    \comm{\hat{\Phi}^{\R}(\omega_1)}{    \hat{\Psi}^{\R}(\omega_2)},
\end{equation}
and likewise for the adjoint commutator. Together with Eqns.~\eqref{eqn:correspondance relativistic rel with gr rel} and \eqref{eqn:correspondance relativistic rel with gr rel pointwise}, this proves the equivalence of
the Einstein causality conditions.
\end{proof}

\begin{corollary}
    Let $\R$ be a relativistic QRF and $\underline{\R}$ be the general relativistic QRF on $O_+^\uparrow(\mink,\eta)$ induced by $\R$. Let $K$ be any compact subset of $O_+^\uparrow(\mink,\eta)$. Then $\forall \sigma > 0$, $\exists f_{\sigma,K} \in C^\infty_c(\mink)$ such that $\iota_\S(W_{\hat{\phi}}(f_{\sigma,K}))$ is weakly $(\underline{\operationalstate^K},\sigma,\pi^{\mathrm{std}})$-microcausal.
\end{corollary}

\begin{proof}
    This follows from Example \ref{ex:example causal rqf} and Prop.~\ref{prop:flat causality reduction}.
\end{proof}

The semi-causality conditions do not reduce to any of the four flat conditions in Prop.~\ref{prop:flat causality reduction}. The flat spacetime form of strong semi $(\underline{\operationalstate},\sigma)$-causality is the following: for every $\phi,\psi\in\mathcal{O}_\S$ and every $\omega_1,\omega_2\in\operationalstate$, one has, for $(d^d x,d^d y)$-almost every pair $(x,y)$ satisfying $(x-y)^2< -\sigma$ and for every $\Lambda\in \Lup$,
\begin{equation}
\label{eqn:flat strong semi causality}
    \comm{\hat{\Phi}^{\R}(\omega_1)}{U_\S(x-\Lambda y,\Lambda)\hat{\Psi}^{\R}(\omega_2)    U_\S(x-\Lambda y,\Lambda)^\dagger} =
    \comm{\hat{\Phi}^{\R}(\omega_1)^\dagger}{    U_\S(x-\Lambda y,\Lambda)\hat{\Psi}^{\R}(\omega_2)U_\S(x-\Lambda y,\Lambda)^\dagger}    = 0.
\end{equation}
For weak semi-causality, Eqn.~\eqref{eqn:flat strong semi causality} is required only when $\operatorname{Loc}(\R,\omega_1) \indepsigma    \operatorname{Loc}(\R,\omega_2)$. By relational covariance, the transformed observable in Eqn.~\eqref{eqn:flat strong semi causality} may equivalently be written with $\hat{\Psi}^{\R}\left(   U_\R(x-\Lambda y,\Lambda)\omega_2 U_\R(x-\Lambda y,\Lambda)^\dagger\right)$.

If one wishes to strengthen the definition of semi-causality by quantifying over every arrow with $\sigma$-spacelike endpoints (rather than just $(d^dx,d^dy)$-almost every endpoint), define
\begin{equation}
    \Poinc_\sigma := \left\{(a,\Lambda)\in \Poincup \mid \exists x\in\mink \text{ such that } \big((\Lambda-\id)x+a\big)^2< -\sigma \right\}.
\end{equation}
This is precisely the set of Poincaré transformations associated to the arrows whose source and target are $\sigma$-spacelike separated. Indeed, let 
\begin{equation}
    \Gamma_\sigma := \{\alpha \in \mathrm{Poin}(\mink,\eta) \mid s(\alpha) \perp_\sigma t(\alpha)\}.
\end{equation}
Since $y=\Lambda x + a$, $y-x = (\Lambda - \id)x + a$. But $x \perp_\sigma y \Leftrightarrow (y-x)^2 < -\sigma$ so $\alpha \in \Gamma_\sigma$ if and only if $\big((\Lambda - \id)x + a\big)^2 < - \sigma$. For this stronger pointwise definition, strong semi-causality is equivalent to
\begin{equation}
\label{eq:flat-strong-semi-pointwise}
    \left[\hat{\Phi}^{\R}(\omega_1),    U_\S(a,\Lambda)\hat{\Psi}^{\R}(\omega_2)U_\S(a,\Lambda)^\dagger\right]=\left[    \hat{\Phi}^{\R}(\omega_1)^\dagger,    U_\S(a,\Lambda)\hat{\Psi}^{\R}(\omega_2)U_\S(a,\Lambda)^\dagger    \right] = 0
\end{equation}
for every $(a,\Lambda)\in \Poinc_\sigma$ and every $\omega_1,\omega_2\in\operationalstate$. In either formulations, both versions of relational semi-causality remain distinct conditions to the other notions of relational causality in the flat constant-field case. In summary, relational quantum field theory in Minkowski spacetime developed in \cite{fedida_foundations_2025} and Sec.~\ref{sec:RQFT Minkowski} is the constant operator field case of the more general groupoid-based theory of Sec.~\ref{sec:RQFT curved}.

\subsection{Towards a locally covariant viewpoint}

The construction above is compatible with the general philosophy of algebraic QFT and locally covariant QFT \cite{dimock_algebras_1980,brunetti_generally_2003,fewster_algebraic_2015}: to a spacetime region $\U\subseteq \M$ one may associate the algebra generated by operators of the form
\begin{equation}
\hat{\Phi}^\R(\omega), \qquad \text{Loc}_\pi(\R,\omega) \subseteq \U.
\end{equation}
The groupoid covariance law of Thm.~\ref{thm:covariance rqf} indicates that embeddings and local frame changes should act by transporting both the system field and the frame statistics. A fully functorial formulation of the resulting net of relational algebras will require a
separate treatment, but the present section already identifies its basic kinematical building blocks.

\begin{definition}
    \label{def:algebras rqf curved}
    Let $\R$ be a general relativistic QRF, $\pi = (\pi_x : \mathrm{Poin}(\M,g)^x \to \M)$ be $\mathrm{Bis}_{\mathrm{Isom}}(\mathrm{Poin}(\M,g))$-compatible, $\operationalstate \subseteq \framestate$ be convex, $\sigma \geq 0$ and $\mathcal{O}_\S \subseteq \A(\scrhis)$ be $(\operationalstate,\sigma,\pi)$-causal. Given a subset $\U \subseteq \mathcal{M}$, we call
    \begin{equation}
        \A^{(\operationalstate,\sigma,\pi)}_{\mathcal{O}_\S}(\U) := \{\hat{\Phi}^\R(\omega) \mid \phi \in \mathcal{O}_\S, \omega \in \operationalstate \text{ s.t. } \text{Loc}_\pi(\R,\omega) \subset \U\}''
    \end{equation}
    a \emph{relational local algebra for $\pi$}.
\end{definition}

This construction generalises that of relational local algebras found in Minkowski spacetime (see Def.~\ref{def:relational local algebras minkowski} and \cite{fedida_foundations_2025}) to curved spacetimes. 

\begin{proposition}
    \label{prop:operationalstate convex under bisections}
    Let $\R$ be a general relativistic QRF and $\operationalstate \subseteq \framestate$ be convex. Then for any $b \in \mathrm{Bis}(\mathrm{Poin}(\M,g))$, $b \cdot \operationalstate := \{\omega^b \mid \omega \in \operationalstate\}$ is convex. 
\end{proposition}

\begin{proof}
    This follows immediately from the fact that $\mathbf{U}_b$ is unitary so it preserves convex mixtures.
\end{proof}

\begin{proposition}
    \label{prop:Einstein causality stable under isometries}
    Let $\R$ be a general relativistic QRF, $\pi$ be $\mathrm{Bis}_{\mathrm{Isom}}(\mathrm{Poin}(\M,g))$-compatible, $\sigma \geq 0$, $\operationalstate \subseteq \framestate$ be convex and $\mathcal{O}_\S \subseteq \A(\scrhis)$ be closed under adjoints. Then for all $b \in \mathrm{Bis}_{\mathrm{Isom}}(\mathrm{Poin}(\M,g))$, $\mathcal{O}_\S$ is $(\operationalstate,\sigma,\pi)$-causal if and only if it is $(b \cdot \operationalstate,\sigma,\pi)$-causal.
\end{proposition}

\begin{proof}
    Let $\omega_1,\omega_2 \in \operationalstate$ be such that $\text{Loc}_\pi(\R,\omega_1)
    \indepsigma
    \text{Loc}_\pi(\R,\omega_2)$. Write $\omega_1^b,\omega_2^b\in b\cdot\operationalstate$. Covariance of $\mathrm{Bis}_{\mathrm{Isom}}(\mathrm{Poin}(\M,g))$-compatible marginals gives $\text{Loc}_\pi(\R,\omega_i^b) =\varphi_b\left(\text{Loc}_\pi(\R,\omega_i)\right)$. But since $\varphi_b$ is an isometry, it preserves the relation $\indepsigma$. Hence, $\text{Loc}_\pi(\R,\omega_1)
    \indepsigma
    \text{Loc}_\pi(\R,\omega_2)$ if and only if $\varphi_b\left(\text{Loc}_\pi(\R,\omega_1)\right)
    \indepsigma
    \varphi_b\left(\text{Loc}_\pi(\R,\omega_2)\right)$ if and only if    $\text{Loc}_\pi(\R,\omega_1^b)     \indepsigma     \text{Loc}_\pi(\R,\omega_2^b)$. Moreover, causality for $\operationalstate$ and bisection covariance yields, for any $\phi,\psi \in \mathcal{O}_\S$,
\begin{equation}
    \comm{\hat\Phi^\R(\omega_1^b)}{
    \hat\Psi^\R(\omega_2^b)}=\mathbf V_b
    \comm{\hat\Phi^\R(\omega_1)}{\hat\Psi^\R(\omega_2)}
    \mathbf V_b^\dagger
    =0.
\end{equation}
The adjoint commutator is identical, which finishes the proof.
\end{proof}

\begin{theorem}
    \label{thm:RAQFT}
    Let $\R$ be a general relativistic QRF, $\pi = (\pi_x : \mathrm{Poin}(\M,g)^x \to \M)$ be $\mathrm{Bis}_{\mathrm{Isom}}(\mathrm{Poin}(\M,g))$-compatible, $\operationalstate \subseteq \framestate$ be convex, $\sigma \geq 0$ and $\mathcal{O}_\S \subseteq \A(\scrhis)$ be closed under adjoints and $(\operationalstate,\sigma,\pi)$-causal. Then
    \begin{enumerate}
        \item (Isotony) For all $\U \subseteq \V \subseteq \mathcal{M}$, $\A^{(\operationalstate,\sigma,\pi)}_{\mathcal{O}_\S}(\U) \subseteq \A^{(\operationalstate,\sigma,\pi)}_{\mathcal{O}_\S}(\V)$.
        \item (Covariance) For all $b \in \mathrm{Bis}_{\mathrm{Isom}}(\mathrm{Poin}(\M,g))$ and all $\U \subset \M$,
        \begin{equation}
            \mathbf{V}_b \A^{(\operationalstate,\sigma,\pi)}_{\mathcal{O}_\S}(\U) \mathbf{V}_b^\dagger = \A^{( b \cdot \operationalstate,\sigma,\pi)}_{\mathcal{O}_\S}(\varphi_b(\U)).
        \end{equation}
        \item (Causality) For all $\U \indepsigma \V \subset \mathcal{M}$, $\comm{\A^{(\operationalstate,\sigma,\pi)}_{\mathcal{O}_\S}(\U)}{\A^{(\operationalstate,\sigma,\pi)}_{\mathcal{O}_\S}(\V)} = \{0\}$.
    \end{enumerate}
\end{theorem}

\begin{proof}
    See App. \ref{app:proof AQFT}.
\end{proof}

\begin{remark}
    We do not need to work on spacetimes with non-trivial isometries. However, when the spacetime does present isometries, then this notion of covariance  recovers that of algebraic QFT in curved spacetime \cite{dimock_algebras_1980}. Note that groupoids were approached in quantum field theory from a categorical perspective in \cite{benini_quantum_2017}, where the authors studied the conditions under which their approach satisfy the axioms of locally covariant QFT. It is an interesting prospect to further link such different approaches.
\end{remark}

\begin{remark}
    Thm.~\ref{thm:RAQFT} defines an isometry-covariant net on one fixed spacetime. It is not yet a locally covariant quantum field theory in the functorial sense: that would require assigning compatible algebras and QRF data to a category of spacetimes and causality-preserving embeddings. We keep this for future work.
\end{remark}

\section{Towards foundations of relational quantum gauge field theory}
\label{sec:towards RQGFT}

The curved-spacetime construction of Sec.~\ref{sec:RQFT curved} naturally extends to situations with internal gauge symmetry. This extension is already suggested in the recent RQFT literature \cite{fedida_foundations_2025}, where one proposes enlarging the space of frames from the Lorentz bundle to a principal bundle with structure group of the form
\begin{equation}
K:=\Lup\times G,
\end{equation}
with $G$ a locally compact Lie group describing the internal gauge symmetry. In local trivializations this leads to frame spaces of the form
\begin{equation}
F_G|_U \cong \U\times \mathcal{SO}_+^\uparrow(1,d-1)\times G
\end{equation}
and to relational observables obtained by averaging over both Lorentz and gauge variables. The purpose of the present subsection is to recast that idea in the groupoid language developed in the previous sections. The resulting object is no longer the Poincar\'e groupoid, but the Atiyah groupoid of the relevant principal bundle, \emph{cf.} Example \ref{ex:gauge-groupoid}. This is the natural groupoid encoding local gauge changes of frame. 

Let $\pi_P:P\to \M$ be a smooth principal $G$-bundle over the spacetime $\M$. If one wishes to treat spacetime and internal symmetries simultaneously, one may either work with the product bundle
\begin{equation}
Q:=F\times_\M P \to \M,
\end{equation}
which is a principal $K=\Lup\times G$-bundle, or, more generally, with any principal $K$-bundle
\begin{equation}
\pi_Q:Q\to \M
\end{equation}
encoding both Lorentz and gauge structure, that is such that its structural group $K$ is an extension of $SO_+^\uparrow(1,d-1)$ by $G$. We now consider the \emph{Atiyah groupoid of $Q$}
\begin{equation}
    \At(Q) := (Q \times Q)/K \rightrightarrows \M
\end{equation}
whose arrows are classes $[q_x,q_y]$, $q_x \in Q_x$, $q_y \in Q_y$, under the diagonal right action
\begin{equation}
    (q_x,q_y) \cdot k := (q_x \cdot k,q_y \cdot k), \qquad k \in K.
\end{equation}

In particular, the bisections of $\At(Q)$ covering the identity on $\M$ are gauge transformations. 

\begin{definition}
The group of \emph{internal gauge bisections} is
\begin{equation}
\label{eqn:internal-gauge-bisections-corrected}
    \operatorname{Bis}_G(\At(Q)):=    \left\{b\in\operatorname{Bis}(\At(Q)) \mid \varphi_b = \id_\M \; \& \; b\text{ induces the identity on }Q/G\right\}.
\end{equation}
This is the group of internal $G$-gauge transformations.
\end{definition}

Hence the passage from the Poincar\'e groupoid to $\At(Q)$ is the natural step from pure spacetime frames to gauge-covariant frames.

\begin{remark}
    The transformations of the fields below are implemented by unitary conjugation. Gauge connections and gauge potentials transform affinely and require additional connection or holonomy data; they are not included in the present construction.
\end{remark}

\subsection{Gauge quantum reference frames and relational gauge-covariant quantum fields}

A straightforward generalisation of general relativistic QRFs can be undertaken here, by simply replacing $\mathrm{Poin}(\M,g)$ by $\At(Q)$. $(\mathcal{M},g)$ is still taken to be an oriented and time-oriented $d$-dimensional Lorentzian spacetime, and $\mu_g$ the volume measure associated to the metric $g$. The quantum system of interest is described by a continuous Hilbert bundle $\scrhis=\bigsqcup_{x\in \M}\hisx \to \M$ carrying a fiber-wise ultraweakly continuous unitary representation $V:\At(Q) \curvearrowright \scrhis$.

\begin{definition}
\label{def:gauge qrf}
A \emph{gauge quantum reference frame} for $Q$ on $(\M,g)$ is a $\mu_g$-admissible principal Lie groupoid QRF $\R_G$ for $\At(Q)$.
\end{definition}

\begin{lemma}
    \label{lem:gauge QRF is bisection admissible}
    A gauge QRF is bisection-admissible.
\end{lemma}

\begin{proof}
    This is immediate from the fact that diffeomorphisms of a manifold preserve the null sets of the volume measure.
\end{proof}

We can then define relational (semi(local)) quantum fields in full analogy to Sec.~\ref{sec:RQFT curved}. Note that since the base space is still $\M$, so the operator fields $\phi = (\phi_x)_{x \in \M} \in \A(\scrhis)$ can be seen as defining absolute quantum gauge fields of the form of Eqn.~\eqref{eqn:absolute quantum field def}. We still want to work in a picture where we relativise operators rather than fields however.

\begin{definition}
    Let $\phi \in \A(\scrhis)$ and $\R_G$ be a gauge QRF for $Q$ on $(\M,g)$. For $\mu_g$-almost every $x \in \mathcal{M}$, the \emph{relational semilocal quantum gauge field at $x$} is the map
\begin{align}
    \hat{\Phi}^{\R_G}_x : \framestategauge &\to \bhsx \\
    \label{eqn:relational semilocal observable gauge}
    \hat{\Phi}^{\R_G}_x(\omega) &:= \euro^{\R_G}_{\omega}(\phi)_x = \iint_{\At(Q)^x} V_\beta \phi_{s(\beta)} V_\beta^\dagger \, d\mu^{\E_{\R_G}^x}_{\omega_x}(\beta).
\end{align}
    We call the right-hand side of Eqn.\eqref{eqn:relational semilocal observable gauge} a \emph{relational semilocal gauge observable}.
\end{definition}

Again, by disintegrating the measure $\mu^{\E_{\R_G}^x}_{\omega_x}$ with respect to its source marginal $\mu^{\F_{\R,\pi}^x}_{\omega_x}$ (Lem.~\ref{lem:disintegration exists Lie QRF}), we get a family of conditional measures $\{\nu^{\R,\pi}_{\omega_x}(\cdot \mid y)\}_{y \in \supp \mu^{\F_{\R,\pi}^x}_{\omega_x}}$.

\begin{definition}
    Let $\R$ be a general relativistic QRF, $\omega \in \framestategauge$, $\pi = (\pi_x : \At(Q)^x \to \M)$ be a family of measurable projections, and $x \in \mathcal{M}$. We call
    \begin{align}
    \hat{\phi}^{\R_G,\pi}_{x,\omega} : \supp \mu^{\F_{\R_G,\pi}^x}_{\omega_x} &\to \bhsx \\
    \hat{\phi}^{\R_G,\pi}_{x,\omega}(y) &:= \int_{\pi_x^{-1}(\{y\})} V_\beta \phi_{s(\beta)} V_\beta^\dagger d\nu^{\R_G,\pi}_{\omega_x}(\beta \mid y)
\end{align}
    a \emph{relational local quantum gauge field at $x$ given $\pi$}.
\end{definition}

Again, we obtain the local expression
\begin{equation}
    \hat{\Phi}^{\R_G}_x(\omega) = \int_{\mathcal{M}} \hat{\phi}^{\R_G,\pi}_{x,\omega}(y) \, d\mu^{\F_{\R_G,\pi}^x}_{\omega_x}(y),
\end{equation}
since
\begin{equation}
    \hat{\Phi}^{\R_G}_x(\omega)=\iint_{\At(Q)^x} V_\beta \phi_{s(\beta)} V_\beta^\dagger \, d\mu^{\E_{\R_G}^x}_{\omega_x}(\beta) = \int_{\mathcal{M}} \underbrace{\left(\int_{\pi_x^{-1}(\{y\})} V_\beta \phi_{s(\beta)} V_\beta^\dagger \, d\nu^{\R_G,\pi}_{\omega_x}(\beta \mid y)\right)}_{\hat{\phi}^{\R_G,\pi}_{x,\omega}(y)} \, d\mu^{\F_{\R_G,\pi}^x}_{\omega_x}(y).
\end{equation}

From there, we can define relational gauge-covariant quantum fields in curved spacetimes.

\begin{definition}
    Let $\R_G$ be a gauge QRF for $Q$ on $(\M,g)$, $\mu_g$ be the spacetime volume measure determined by the metric $g$, and $\phi \in \A(\scrhis)$. The \emph{relational gauge-covariant quantum field} is the map
    \begin{align}
        \hat{\Phi}^{\R_G} : \framestategauge &\to \A(\scrhis) \\
        \hat{\Phi}^{\R_G}(\omega) &:= \euro^{\R_G}_\omega(\phi)=
        \label{eqn:relational gauge observable} \int_M^\oplus \hat{\Phi}^{\R_G}_x(\omega) \, d\mu_g(x).
    \end{align}
    We call the right-hand side of Eqn.~\eqref{eqn:relational gauge observable} a \emph{relational gauge-covariant observable}. 
\end{definition}

The construction applies to bounded operator fields whose gauge action is implemented by unitary conjugation. It therefore describes gauge-covariant matter observables and, when the seed is chosen in an adjoint associated bundle, may also describe curvature or field-strength observables. It does not by itself define a gauge potential or a connection, whose transformation law is affine. The term ``relational quantum field strength" should be used only after the seed operator field has been identified with the curvature of a connection.

\subsection{Relational gauge covariance}

We can now talk about relational covariance with both spacetime symmetries and gauge symmetries.

\begin{theorem}
    \label{thm:covariance semilocal gauge}
    Let $\R_G$ be a gauge QRF, $x,z \in \mathcal{M}$, $\alpha : x \to z$ be an arrow of $\At(Q)$ and $\omega = (\omega_x) \in \framestate$. Write the transformed frame state as $\omega_z^\alpha := U_\alpha \omega_x U_\alpha^\dagger$ so that for all $x \in \mathcal{M}$, $\omega^\alpha_x = \begin{cases}
        \omega_x \text{ if } x \neq z \\
        \omega_z^\alpha \text{ if } x=z
    \end{cases}$. Then for all $\omega \in \framestate$,
    \begin{equation}
        V_\alpha \hat{\Phi}^{\R_G}_x(\omega)V_\alpha^\dagger = \hat{\Phi}^{\R_G}_z(\omega^\alpha).
    \end{equation}
\end{theorem}

\begin{proof}
    We have
    \begin{equation}
        V_\alpha \hat{\Phi}^{\R_G}_x(\omega) V_\alpha^\dagger = V_\alpha \euro^{\R_G}_{\omega}(\phi)_x V_\alpha^\dagger \stackrel{\ref{thm:covariance euro}}{=} \euro^{\R_G}_{\omega^{\alpha}}(\phi)_z = \hat{\Phi}^{\R_G}_z(\omega^\alpha).
    \end{equation}
\end{proof}

Likewise, we recover a notion of covariance at the level of the whole relational quantum field.

\begin{theorem}
    \label{thm:covariance rqf gauge}
    Let $\R_G$ be a gauge QRF, $\phi \in \A(\scrhis)$ and $\omega = (\omega_x) \in \framestategauge$. Then for all $b \in \mathrm{Bis}(\At(Q))$,
    \begin{equation}
        \mathbf{V}_{b} \hat{\Phi}^{\R_G}(\omega) \mathbf{V}_{b}^\dagger = \hat{\Phi}^{\R_G}(\omega^b)
    \end{equation}
    where $\omega^b = \left(\omega^b_x = U_{b(\varphi_b^{-1}(x))} \omega_{\varphi_b^{-1}(x)} U_{b(\varphi_b^{-1}(x))}^\dagger\right)$.
\end{theorem}

\begin{proof}
    This follows immediately from above and Lem.~\ref{lem:gauge QRF is bisection admissible} and Cor.~\ref{cor:invariance and covariance bisections euro}.
\end{proof}

This bisection covariance includes both spacetime transformations as well as gauge transformations. Finally, we can discuss relational causality at the level of the relational local quantum fields, as follows.

\begin{theorem}
    \label{thm:covariance local rqf gauge}
    Let $\R_G$ be a gauge QRF, $\pi = (\pi_x : \At(Q)^x \to \M)$ be a family of measurable projections, $x,z \in \mathcal{M}$, $\alpha : x \to z$ be an arrow of $\At(Q)$ and $\omega = (\omega_x) \in \framestate$. Let $\tau_\alpha : M \to M$ be a Borel isomorphism such that $\pi_z \circ L_\alpha = \tau_\alpha \circ \pi_x$, where $L_\alpha : \At(Q)^{s(\alpha)} \to \At(Q)^{t(\alpha)}$ is such that $L_\alpha(\beta) = \alpha \circ \beta$. Write transformed frame states as $\omega_z^\alpha := U_\alpha \omega_x U_\alpha^\dagger$ so that for all $x \in \mathcal{M}$, $\omega^\alpha_x = \begin{cases}
        \omega_x \text{ if } x \neq z \\
        \omega_z^\alpha \text{ if } x=z
    \end{cases}$. Then for all $\omega \in \framestategauge$ and all $y \in \supp \mu^{\F_{\R_G,\pi}^x}_{\omega_x}$, 
    \begin{equation}
    \label{eqn:conditional-probability-covariance gauge}
        \nu_{\omega_z^\alpha}^{\R_G,\pi}
        (\,\cdot\mid\tau_\alpha(y))
        =
        (L_\alpha)_*
        \nu_{\omega_x}^{\R_G,\pi}(\,\cdot\mid y).
    \end{equation}
     Consequently, for all $\phi \in \A(\scrhis)$,
    \begin{equation}
        V_\alpha \hat{\phi}^{\R_G,\pi}_{\omega,x}(y) V_{\alpha}^\dagger = \hat{\phi}^{\R_G,\pi}_{\omega^\alpha,z}(\tau_\alpha(y)).
    \end{equation}
    
\end{theorem}

\begin{proof}
    See App.~\ref{app:proof covariance local rqf gauge}
\end{proof}

Thus, relational gauge covariance can once again be understood at three levels: for relational local quantum gauge fields, for relational semilocal quantum gauge fields, and for relational gauge-covariant quantum fields.

\subsection{Localization limit} The relational description reduces to the ordinary one when the frame becomes sharply localized at the identity arrow.

\begin{theorem}
    \label{thm:localization rqf gauge}
    Let $\R_G$ be a localizable gauge QRF. Let $x \in \M$ and suppose that $1_x \in \supp \E_{\R_G}^x$. Let $(\omega_n^x) \subset \mathscr{D}(\hi_{\R_G}^x)$ be a localizing sequence centred around $1_x$, with $\omega_n = (\omega_n^x) \in \framestategauge$. Then for any $\phi \in \A(\scrhis)_{\mathrm{uw}}$,
    \begin{equation}
        \lim_{n \to \infty} \hat{\Phi}^{\R_G}_x(\omega_n) = \hat{\phi}(x)
    \end{equation}
    ultraweakly in $\bhsx$. Moreover, if for $\mu_g$-almost every $x \in \M$, $1_x \in \supp \E_{\R_G}^x$, we have that
    \begin{equation}
        \lim_{n \to \infty} \hat{\Phi}^{\R_G}(\omega_n) = \phi
    \end{equation}
    ultraweakly in $\A(\scrhis)$.
\end{theorem}

\begin{proof}
    This follows immediately from Thm.~\ref{thm:euro localisation limit} and the fact that $\hat{\Phi}^{\R_G}(\omega_n) = \euro^{\R_G}_{\omega_n}(\phi)$ and $\hat{\Phi}^{\R_G}_x(\omega_n) = \euro^{\R_G}_{\omega_n}(\phi)_x$.
\end{proof}

Thus, absolute quantum gauge fields are recovered in the localization limit of the relational description, just as for relational quantum field theory in Minkowski spacetime.

\begin{remark}
    The notions of localization and relational causality might be importable from Sec.~\ref{subsec:relational causality}, though the nonlocal nature of gauge theories makes it probable that some conditions might need modifications, or that causality conditions of another nature altogether are required. We keep this for future work.
\end{remark}

\begin{remark}
    The definition of algebras of observables for relational gauge-covariant quantum fields can be defined just as in Def.~\ref{def:algebras rqf curved} modulo causality conditions, with isotony and bisection covariance holding straightforwardly. We keep a more detailed discussion of those for when causality in relational quantum gauge field theories will have been understood more thoroughly.
\end{remark}

\begin{remark}
    Further note that we do not consider any spin structure here. Extending the formalism of relational quantum field theory to spinors is a work in progress. Likewise, the construction above treats bounded operator fields whose gauge transformations are implemented by unitary conjugation. It does not yet define relational gauge potentials or connections, which transform affinely. A relational quantum Yang-Mills theory therefore requires additional connection, path-groupoid, or holonomy-groupoid data.
\end{remark}

The analysis conducted in this section can be summarised by stating that the Atiyah groupoid of a principal fiber bundle extending the orthonormal frame provides the natural groupoid-valued kinematics for relational gauge-covariant quantum field theory.  It simultaneously incorporates spacetime frame data and internal gauge data, its bisections encode gauge transformations, and it yields a direct generalization of the curved-spacetime relational field construction of Section \ref{sec:RQFT curved}.  In local trivializations, the resulting formulas reduce to the bundle-theoretic expressions involving averages over $\Lup\times G$, while globally they are formulated intrinsically on the gauge groupoid.

\section{Conclusions and Outlook}
\label{sec:conclusions}

In this paper we propose a new formulation of quantum reference frames in which the kinematical role usually played by a symmetry group is taken instead by a (continuous) groupoid. The guiding motivation is clear: while the standard operational formalism of quantum reference frames works naturally on homogeneous spaces of locally compact groups, generic curved spacetimes do not admit a sufficiently rich global symmetry group. Nevertheless, they do possess canonical local kinematical structures, most notably the orthonormal frame bundle and its associated gauge groupoid. Our main claim has been that this groupoid, and more specifically the Poincar\'e groupoid of a Lorentzian spacetime, provides the natural geometric object on which an operational and relational theory of quantum reference frames should be based.

The first part of the paper established the geometric and operator-theoretic framework required for this extension. After reviewing the relevant background on continuous groupoids, action groupoids, gauge groupoids, and the Poincar\'e and Wigner groupoids associated with a spacetime, we showed how the standard group-based structures of the QRF formalism admit natural groupoid analogues. In place of a single Hilbert space carrying an ultraweakly continuous unitary representation of a group, one is led to a continuous Hilbert bundle carrying a fiber-wise ultraweakly continuous unitary representation of a continuous groupoid. In place of a single covariant POVM on a homogeneous
space, one is led to a fiber-wise covariant family of POVMs on the target fibers of the groupoid.

We also showed that, on homogeneous spaces of locally compact groups with quasi-invariant measure $\mu$, normalized covariant POVMs induce $\mu$-continuous Born probability measures, thus leading to $L^1$ smearing functions without any additional hypothesis. This is useful not only to make contact with Wightmannian quantum field theory for relational quantum field theory in Minkowski spacetime, where the frame smearing function plays the same role as the Schwartz test function, but is also important to prove $\mu$-admissibility (i.e. well-definedness) of the groupoid relativization map in the principal homogeneous space case.

The central conceptual result of the paper is that the operational notion of quantum reference frame extends naturally from groups to continuous groupoids. Concretely, we defined a ($\mu$-admissible, principal) groupoid quantum reference frame as a triple consisting of a Hilbert bundle, a (fiber-wise ultraweakly continuous) unitary representation of the groupoid, and a fiber-wise covariant frame observable. From these data we constructed a groupoid relativization map which produces invariant observables on the joint system--frame theory. This construction is the direct groupoid analogue of the usual relativization map in the standard theory of quantum reference frames. It yields a well-defined notion of relative observables, and it makes precise the idea that the observable content of a theory should be formulated relative to a quantum frame rather than with respect to an external classical structure.

A crucial consistency check was also established: our construction can reduce to the ordinary operational QRF formalism for the group $G$ acting on either the group itself or a torsor for the group. Thus the groupoid formalism is not an alternative theory, but a genuine extension of the existing one. This reduction result is especially important in flat spacetime, where the Poincar\'e groupoid of Minkowski space is canonically identified with the action groupoid of the proper orthochronous Poincar\'e group. Hence the standard relativistic QRF formalism is recovered as a special case of the present framework.

The second part of the paper applied this groupoid-QRF formalism to relational quantum field theory. We first reviewed relational quantum field theory in Minkowski spacetime, and drew some close links with Wightman's as well as algebraic quantum field theory. We then expanded this formalism to curved spacetimes, and showed that the groupoid QRF formalism can be used to understand relational quantum fields beyond the flat spacetime limit. We generalised the notion of relational covariance and relational causality conditions, and showed that the localization limit can also be understood in this more general regime. We then argued that the same logic extends naturally to internal gauge symmetries. In that setting, both spacetime frame transformations and internal gauge transformations are treated on the same footing. This suggests that the groupoid formalism is flexible enough to accommodate not only relativistic quantum reference frames, but also gauge-covariant and, potentially, spinorial or gravitationally coupled versions of the theory.

Taken together, these results support the main thesis of the paper: \emph{a natural generalization of symmetry-based quantum reference frames to generic curved spacetimes is obtained by replacing global symmetry groups with local kinematical groupoids}. This move is mathematically natural, conceptually robust, and operationally meaningful. It preserves the essential relational content of the QRF programme while freeing it from the strong homogeneity assumptions that underlie the standard group-based setting.

At the same time, the present work is best viewed as a foundational and kinematical first step. Several important issues remain open.

First, our treatment of relational quantum fields has been formulated at the level of bounded operator fields and bounded seed observables. This is sufficient to make the groupoid construction precise and to exhibit its main structural features, but it is not yet the full language of quantum field theory. A major next step will be to extend the formalism to genuinely unbounded operators, and to understand the precise analytic conditions under which the corresponding groupoid-valued relativization maps are well defined. This should involve generalising the relativization procedure to $O^*$-families \cite{schmudgen_spaces_1990}.

Second, although the present framework naturally accommodates honest unitary groupoid representations, many physically relevant situations involve projective representations, multiplier cocycles, or central extensions. This issue is already familiar from Wigner's classification in the group and groupoidal case \cite{ibort_notion_2025}. Developing a systematic theory of projective representations of the Poincar\'e and Wigner groupoids, together with their relation to particle sectors and superselection structures, is an obvious and important direction for future work (see \cite{ibort_notion_2025} for first steps in this direction).

Third, more work on understanding relational causality is needed. In particular, generalising the notion of relational causality to different QRFs, understanding how to generalise internal and external QRF transformations \cite{glowacki_towards_2024,carette_operational_2025} to groupoid QRFs and verifying whether these preserve the causal properties of relational quantum fields \cite{fedida_foundations_2025} could prove insightful. Moreover, finding non-trivial examples of causal relational quantum fields on non-homogeneous curved Lorentzian spacetimes (e.g. a Schwarzschild black hole spacetime) would be a very useful application of the groupoidal machinery developed in this paper. Likewise, a treatment of QRFs and relational quantum fields in spacetimes with boundary (e.g. the exterior Schwarzschild spacetime) would be very interesting to undertake, and to relate to the study of edge modes \cite{balachandran_edge_1996,carrozza_edge_2022,kabel_quantum_2023,carrozza_edge_2024}.

Fourth, our use of the full Poincar\'e and Atiyah groupoids is kinematically natural, but dynamics may select preferred subgroupoids. For example, a Lorentzian or gauge connection determines a holonomy subgroupoid, and one may expect physically relevant relational observables to be associated not with all possible frame comparisons, but with those selected by parallel transport or by dynamical constraints. Understanding the relation between the full kinematical groupoid and such dynamically selected subgroupoids could lead to a more refined and physically richer version of the theory. Further note that relational quantum field dynamics, understood in the sense of partial differential equations, has not been explored yet. We currently understand dynamics of covariant relational quantum fields as arising from the fixing of the representation of the Poincaré group(oid) with respect to which the relational quantum field is covariant. For example, a relational quantum field on Minkowski spacetime covariant with respect to a mass $m$, spin $0$ (Klein-Gordon) representation of the Poincaré group should evolve according to the Klein-Gordon equation (perhaps not understood strongly, but rather mildly or weakly instead). More work needs to be done in that direction.

Fifth, the gauge-theoretic extension should be developed much further. In particular, one would like to understand the role of the gauge marginal of a quantum frame in concrete models, its possible effect on vacuum sectors, and its relation to standard notions of gauge-invariant or gauge-covariant observables. The same applies, \emph{a fortiori}, to spinor fields, Yang-Mills theories, and gravitational fields.

Finally, a particularly important long-term direction concerns the relation of the present framework with the broader groupoid-based programme in the foundations of quantum theory and elementary particle physics \cite{ciaglia_groupoidal_2024,ibort_notion_2025,ibort_introduction_2021}. The Poincar\'e and Wigner groupoids were introduced as kinematical structures for the description of particles on curved spacetimes; in the present work they have also been shown to provide the correct geometric setting for quantum reference frames and relational observables. This suggests that the groupoid point of view may offer a common language in which particles, frames, and fields can be treated in a single unified formalism. Exploring this possibility in detail is one of the most interesting perspectives opened by the present paper.

In summary, the paper has advanced three main claims. First, operational quantum reference frames admit a natural and nontrivial extension from groups to groupoids. Second, the Poincar\'e groupoid of a Lorentzian spacetime provides the appropriate geometric setting for relational quantum field theory on curved backgrounds. Third, the same logic extends naturally to gauge structures via Atiyah groupoids. We hope that these results clarify the mathematical foundations of relational quantum field theory and open the way to a more general, geometrically intrinsic treatment of quantum reference frames in physics.

\subsection*{Acknowledgments and funding}

S.F. thanks Jeremy Butterfield for interesting conversations related to quantum reference frames on principal bundles, and Jan G{\l}owacki for useful feedback on an earlier version of Sec.~\ref{subsec:localization and continuity}. S.F. is funded by a studentship from the Engineering and Physical Sciences Research Council, grant number 2882481. 
A.I. and A.M.-D. acknowledge financial support from the Spanish Ministry of
Economy and Competitiveness, through the Severo Ochoa Program for Centers of Excellence in RD (SEV-2015/0554), the MINECO research project PID2024-160539NB-I00, and the Comunidad de Madrid project TEC-2024/COM-84 QUITEMAD-CM.  Portions of Appendix \ref{app:proofs} benefited from suggestions from Claude (Anthropic) and ChatGPT (OpenAI). All statements were independently drafted and verified by the authors.

\printbibliography[title={References}]

\addtocontents{toc}{\protect\setcounter{tocdepth}{-1}}

\begin{appendix}

    \section{Proofs omitted from the main text}
    \label{app:proofs}
    \subsection{Proof of Thm. \ref{thm:covariance implies mu continuity}}
\label{proof mu continuity}
\begin{proof}
    By Prop. 1 of \cite{cattaneo_mackeys_1979}, there exists a system of imprimitivity $(U_P,P)$ for $G$ based on $\Sigma$ and acting on a separable Hilbert space $\hi'$ such that
    \begin{align}
        W U(g) &= U_P(g) W \qquad \forall g \in G \\
        \E(\Delta) &= W^\dagger P(\Delta) W \qquad \forall \Delta \in \Bor(\Sigma) 
    \end{align}
    where $W : \hi \to \hi'$ is an isometry (this is a covariant form of Naimark's dilation theorem), i.e. $P$ is $G$-covariant (with respect to $U_P$). Mackey's imprimitivity theorem \cite{mackey_imprimitivity_1949,mackey_unitary_1958,kaniuth_induced_2012} then guarantees a unitary operator $Y : \hi' \to \int_\Sigma^\oplus \hi'_x d\nu(x)$ where $\nu$ is a quasi-invariant $\sigma$-finite measure on $\Sigma$, and
    \begin{equation}
        Y P(\Delta) Y^\dagger = M_{1_\Delta} \qquad \forall \Delta \in \Bor(\Sigma) \, ,
    \end{equation}
    where $M_{1_\Delta}$ is the multiplication operator $(M_{1_\Delta} f)(x) = 1_{\Delta}(x) f(x)$ and $1_\Delta$ is the indicator function in $\Delta$. Note that all quasi-invariant measures on $G/H$ have the same null sets (Corollary 1.23 in \cite{kaniuth_induced_2012}). Then $\forall \Delta \in \Bor(\Sigma)$ such that $\mu(\Delta) = 0$, we have $\nu(\Delta) = 0$ and so
    \begin{equation}
        \norm{M_{1_\Delta}f}^2_{L^2} = \int_\Sigma \norm{1_\Delta(x) f(x)}^2 d\nu(x) = \int_\Delta \norm{f(x)}^2 d\nu(x) = 0 \, .
    \end{equation}
    Thus,
    \begin{equation}
        \mu(\Delta)=0 \Rightarrow Y P(\Delta) Y^\dagger = 0 \Rightarrow P(\Delta) = 0 
    \end{equation}
    as $Y$ is unitary. Thus
    \begin{equation}
        \mu(\Delta) = 0 \Rightarrow \E(\Delta) = 0 
    \end{equation}
    i.e. $\E \ll \mu$. Thus, for any $\chi \in \thi$, $\mu(\Delta) = 0 \Rightarrow \mu^\E_\chi(\Delta) = 0$. But since $\mu$ is $\sigma$-finite and $\mu^\E_\chi$ is a complex measure of finite total variation, the Radon-Nikodym theorem applies: $\mu^\E_\chi \ll \mu$ iff $\exists \dfrac{d\mu^\E_\chi}{d\mu} \in L^1(\Sigma,\mu)$ with $\mu^\E_\chi(\Delta) = \int_\Delta \dfrac{d\mu^\E_\chi}{d\mu} d\mu$. Furthermore, if $\chi \geq 0$ then $\mu^\E_\chi \geq 0$ so by the positivity of $\mu$ it follows that $\dfrac{d\mu^\E_\chi}{d\mu} \geq 0$. Finally, if $\Tr[\chi] = 1$ then since $\E(\Sigma) = \mathbb{1}$ we have $\mu^\E_\chi(\Sigma) = \int_\Sigma \dfrac{d\mu^\E_\chi}{d\mu} d\mu = 1$.
\end{proof}
    
    \subsection{Proof of Thm. \ref{thm: absolute continuity POVMs}}
    \label{proof absolute continuity}
    \begin{proof}
        Suppose $\E \ll_\text{ac} \mu$. Then $\forall \Delta \in \Bor(\Sigma)$,
    \begin{equation}
        \abs{\mu^\E_\chi(\Delta)} = \abs{\Tr[\chi \E(\Delta)]} \leq \norm{\chi}_1 \norm{\E(\Delta)} \leq c \norm{\chi}_1 \mu(\Delta) \, .
    \end{equation}
    Then for any measurable partition $\{\Delta_k\}$ of $\Delta$,
    \begin{equation}
        \sum_k \abs{\mu^\E_\chi(\Delta_k)} \leq c \norm{\chi}_1 \sum_k \mu(\Delta_k) = c \norm{\chi}_1 \mu(\Delta)
    \end{equation}
    hence the total variation satisfies $\abs{\mu^\E_\chi}(\Delta) \leq c \norm{\chi}_1 \mu(\Delta)$ hence $\abs{\mu^\E_\chi} \ll_\text{ac} \mu$ and thus $\mu^\E_\chi \ll_\text{ac} \mu$ so the Radon-Nikodym derivative $\dfrac{d\mu^\E_\chi}{d\mu} \in L^\infty(\Sigma,\mu)$. 

    Conversely, suppose that $\forall \chi \in \thi, \exists \dfrac{d\mu^\E_\chi}{d\mu} \in L^\infty(\Sigma,\mu)$ such that
    \begin{equation}
        \mu^\E_\chi(\Delta) = \int_\Delta \frac{d\mu^\E_\chi}{d\mu} d\mu \, .
    \end{equation}
    For all measurable $\Delta$ such that $\mu(\Delta)=0$,
    \begin{equation}
        \mu^{\E}_\chi(\Delta) = \int_\Delta f_\chi d\mu = 0
    \end{equation}
    so $\Tr[\chi \E(\Delta)] = 0$ for all $\chi \in \thi$ so since $\bh \cong \thi^*$ this implies $\E(\Delta) = 0$ so $\E \ll \mu$.  For all measurable $\Delta$ such that $0 < \mu(\Delta) < \infty$, let
    \begin{equation}
        L_\Delta : \thi \ni \chi \mapsto \frac{\mu^\E_\chi(\Delta)}{\mu(\Delta)} = \frac{1}{\mu(\Delta)} \int_\Delta \frac{d\mu^\E_\chi}{d\mu} d\mu \in \mathbb{C} \, .
    \end{equation}
    Then for each fixed $\chi$,
    \begin{equation}
        \abs{L_\Delta(\chi)} \leq \frac{1}{\mu(\Delta)} \int_\Delta \abs{f_\chi} d\mu \leq \norm{f_\chi}_\infty
    \end{equation}
    so the family $\{L_\Delta : 0 < \mu(\Delta) < \infty\}$ is pointwise bounded on $\thi$. By Banach-Steinhaus, this implies that
    \begin{equation}
        c := \sup_{0 < \mu(\Delta) < \infty} \norm{L_\Delta} < \infty \, .
    \end{equation}
    But
    \begin{equation}
        \norm{L_\Delta} = \sup_{\norm{\chi}_1 \leq 1} \frac{\abs{\Tr[\chi \E(\Delta)]}}{\mu(\Delta)} = \frac{\norm{\E(\Delta)}}{\mu(\Delta)} \, .
    \end{equation}
    Hence, $\norm{\E(\Delta)} \leq c \mu(\Delta)$ for all $\Delta$ such that $0 < \mu(\Delta) < \infty$. Hence, $\E \ll_\text{ac} \mu$.

    Furthermore, if $\chi \geq 0$ then $\mu^\E_\chi \geq 0$ so by the positivity of $\mu$ it follows that $\dfrac{d\mu^\E_\chi}{d\mu} \geq 0$. Finally, if $\Tr[\chi] = 1$ then since $\E(\Sigma) = 1$ we have $\mu^\E_\chi(\Sigma) = \int_\Sigma \dfrac{d\mu^\E_\chi}{d\mu} d\mu = 1$.
    \end{proof}

    \subsection{Proof of Lem.~\ref{lem:projection marginal covariance}}
    \label{app:proof projection marginal covariance}

    \begin{proof}
        Fix $b\in\mathscr{B}$, $x\in M$ and $N\in\Bor(M)$ and write $L_{b(x)}:\Gamma^x\to\Gamma^{\varphi_b(x)}$, $\beta\mapsto b(x)\circ\beta$. The arrow $b(x):x\to\varphi_b(x)$ gives, by covariance of the frame observable, 
    \begin{equation}
    \label{eqn:proof cov step}
    U_{b(x)}\,\F_{\R,\pi}^{x}(N)\,U_{b(x)}^\dagger
    =U_{b(x)}\,\E_\R^{x}\bigl(\pi_x^{-1}(N)\bigr)\,U_{b(x)}^\dagger
    =\E_\R^{\varphi_b(x)}\bigl(L_{b(x)}\pi_x^{-1}(N)\bigr).
    \end{equation}
    We claim
    \begin{equation}
    \label{eqn:set identity}
    L_{b(x)}\bigl(\pi_x^{-1}(N)\bigr)=\pi_{\varphi_b(x)}^{-1}\bigl(\varphi_b(N)\bigr).
    \end{equation}
    Here $L_{b(x)}$ is a Borel isomorphism of target fibers, with inverse $L_{b(x)}^{-1}=L_{b(x)^{-1}}$, and $\varphi_b\in\mathrm{Diff}(M)$, so $\varphi_b(N)\in\Bor(M)$ and both sides of
    \eqref{eqn:set identity} lie in $\Bor(\Gamma^{\varphi_b(x)})$. $\mathscr{B}$-compatibility, in pointwise form, reads
    \begin{equation}
    \label{eqn:BC pointwise}
    \pi_{\varphi_b(x)}\circ L_{b(x)}=\varphi_b\circ\pi_x\qquad\text{on }\Gamma^x.
    \end{equation}
    For $(\subseteq)$: if $\gamma=L_{b(x)}\beta$ with $\beta\in\pi_x^{-1}(N)$, then by
    \eqref{eqn:BC pointwise}, $\pi_{\varphi_b(x)}(\gamma)=\varphi_b(\pi_x(\beta))\in\varphi_b(N)$,
    so $\gamma\in\pi_{\varphi_b(x)}^{-1}(\varphi_b(N))$.    
    For $(\supseteq)$: if
    $\pi_{\varphi_b(x)}(\gamma)\in\varphi_b(N)$, set $\beta:=L_{b(x)}^{-1}\gamma\in\Gamma^x$; then
    $\varphi_b(\pi_x(\beta))=\pi_{\varphi_b(x)}(\gamma)\in\varphi_b(N)$, and since $\varphi_b$ is
    injective, $\pi_x(\beta)\in N$, whence $\gamma=L_{b(x)}\beta\in L_{b(x)}(\pi_x^{-1}(N))$. This proves \eqref{eqn:set identity}.

Substituting \eqref{eqn:set identity} into \eqref{eqn:proof cov step},
\begin{equation}
U_{b(x)}\,\F_{\R,\pi}^{x}(N)\,U_{b(x)}^\dagger
=\E_\R^{\varphi_b(x)}\bigl(\pi_{\varphi_b(x)}^{-1}(\varphi_b(N))\bigr)
=\F_{\R,\pi}^{\varphi_b(x)}\bigl(\varphi_b(N)\bigr).
\end{equation}
    \end{proof}

    \subsection{Proof of Lem. \ref{lem:supp E fiber closed}}
    \label{app:proof supp E fiber closed}

    \begin{proof}
    
    For a state $\omega_x$, $\mu^{\E_\R^x}_{\omega_x}(\cdot)=\Tr[\omega_x\E_\R^x(\cdot)]$ is a Borel probability measure on $\Gamma^x$, and $\mu^{\E_\R^x}_{\ketbra{\psi}}(\cdot)=\expval{\E_\R^x(\cdot)}{\psi}$. The set $O^x_{\max}$ is open. As a topological subspace of the second-countable $\Gamma^x$ it is Lindel\"of, so the open cover $\{O \subseteq O^x_{\max}\mid\E_\R^x(O)=0\}$ of
$O^x_{\max}$ has a countable subcover $\{O^x_k\}_{k\in\mathbb{N}}$ with $O^x_{\max}=\bigcup_kO^x_k$ and $\E_\R^x(O^x_k)=0$. For every unit $\ket{\psi}$, countable subadditivity of $\mu^{\E_\R^x}_{\ketbra{\psi}}$ gives
\begin{equation}
    \expval{\E_\R^x(O^x_{\max})}{\psi}\le\sum_k\expval{\E_\R^x(O^x_k)}{\psi}=0\,.
\end{equation}
Since $\E_\R^x(O^x_{\max})\ge0$ and this holds for all $\ket{\psi}$, we get $\E_\R^x(O^x_{\max})=0$.

For any $\omega_x \in \framestatex$, we get $\mu^{\E_\R^x}_{\omega_x}(O^x_{\max})=\Tr[\omega_x\E_\R^x(O^x_{\max})]=0$; as $O^x_{\max}$ is open, $\supp\mu^{\E_\R^x}_{\omega_x}\subseteq\Gamma^x\smallsetminus O^x_{\max}$. Taking the union over $\omega_x$ yields the inclusion $\supp \E_\R^x \subseteq \Gamma^x \smallsetminus O^x_{\max}$.

Fix a countable base $\{U^x_k\}_{k\in\mathbb{N}}$ of $\Gamma^x$
and let $K:=\{k\mid\E_\R^x(U^x_k)\neq0\}$. By normalization $\E_\R^x(\Gamma^x)=\mathbb{1}_{\bhrx}\neq0$, so the
$U^x_k$ are not all null and $K\neq\varnothing$. For $k\in K$ set $A^x_k:=\E_\R^x(U^x_k)\ge0$, $A^x_k\neq0$.
Since $A^x_k\ge0$, $\expval{A^x_k}{\eta}=0\Leftrightarrow A^x_k\ket{\eta}=0$, so $\ker A^x_k$ is a proper
closed subspace, hence nowhere dense. Thus
\begin{equation}
    D^x_k:=\{\ket{\eta}\in\hirx\mid\expval{A^x_k}{\eta}>0\}=\hirx\smallsetminus\ker A^x_k
\end{equation}
is open and dense. The family $\{D^x_k\}_{k\in K}$ is countable, so by the Baire category theorem in the complete space $\hirx$ the intersection $\bigcap_{k\in K}D^x_k$ is dense, in particular nonempty; as $\bigcap_{k\in K}D^x_k$ is invariant under nonzero scaling and contains no zero vector as $0 \in \ker A_k^x$. Then $\mu^{\E_\R^x}_{\ketbra{\psi}}(U^x_k)=\expval{A^x_k}{\psi}>0$ for every $k\in K$.

We now show that $\supp\mu^{\E_\R^x}_{\ketbra{\psi}}=\Gamma^x\smallsetminus O^x_{\max}$. The inclusion $\subseteq$ has been verified above. For $\supseteq$, let $\beta\in\Gamma^x\smallsetminus O^x_{\max}$ and let $O\ni\beta$ be open; we then know that $\E_\R^x(O)\neq0$. Covering $O$ by the base sets it contains, the previous countable-subadditivity argument shows that they are not all null, so some $U^x_{k_0}\subseteq O$
has $\E_\R^x(U^x_{k_0})\neq0$, i.e. $k_0\in K$. Hence
$\mu^{\E_\R^x}_{\ketbra{\psi}}(O)\ge\mu^{\E_\R^x}_{\ketbra{\psi}}(U^x_{k_0})=\expval{A^x_{k_0}}{\psi}>0$.
As $O$ was an arbitrary neighbourhood of $\beta$, $\beta\in\supp\mu^{\E_\R^x}_{\ketbra{\psi}}$.

Thus, $\Gamma^x\smallsetminus O^x_{\max}=\supp\mu^{\E_\R^x}_{\ketbra{\psi}}\subseteq\bigcup_{\omega_x}\supp\mu^{\E_\R^x}_{\omega_x} = \supp \E_\R^x$, hence using the other inclusion this is an equality. Thus $\supp\E_\R^x=\Gamma^x\smallsetminus O^x_{\max}$, closed as the complement of the open $O^x_{\max}$, and attained by the single state $\ketbra{\psi}$.

\end{proof}

    \subsection{Proof of Prop. \ref{prop:norm 1 groupoid sequence}}
    \label{app:proof norm 1 groupoid sequence}

    \begin{proof}

    Here we adapt the proof of Prop. 2.13 of \cite{carette_operational_2025} to the case of groupoids. $\Gamma^x$ is metrizable so we choose a compatible metric $d_x$ on $\Gamma^x$. Let $B^x_n := \{\beta \in \Gamma^x \mid d_x(\beta,\beta_x) < \frac{1}{n}\}$ be a sequence of open balls of radius $1/n$. Since $B_n^x$ is an open neighbourhood of $\beta_x \in \supp \E_\R^x$, every such neighbourhood is non-null, so $\E_\R^x(B_n^x) \neq 0$ for all $n$, so by Prop. \ref{prop:norm 1 property povm} we can choose unit vectors $\ket{\phi_n^x}$ such that $\expval{\E_\R^x(B^x_n)}{\phi_n^x} > 1-1/n$. Denoting by $\omega_n^x$ the associated pure state $\omega^x_n = \ketbra{\phi_n^x}$, we have $\Tr[\omega_n^x \E_\R^x(B_n^x)] = \expval{\E_\R^x(B^x_n)}{\phi_n^x} > 1-1/n$ and thus $\mu^{\E_\R^x}_{\omega_n^x}(B_n^x) > 1-1/n$.

    We can now show convergence using the portmanteau theorem \cite{billingsley_convergence_1999}. It suffices to show $\lim_{n \to \infty}\mu^{\E_\R^x}_{\omega_n^x}(X) =\delta_{\beta_x}(X)$ for every $X \in \Bor(\Gamma^x)$ with $\beta_x \notin \partial X$. Decompose
    \begin{equation}
        \mu^{\E_\R^x}_{\omega_n^x}(X) = \mu^{\E_\R^x}_{\omega_n^x}(X \smallsetminus B_n^x) + \mu^{\E_\R^x}_{\omega_n^x}(X \cap B_n^x) \, .
    \end{equation}
    The first term obeys $\mu^{\E_\R^x}_{\omega_n^x}(X \smallsetminus B_n^x) \leq \mu^{\E_\R^x}_{\omega_n^x}(\Gamma^x \smallsetminus B_n^x) = 1 - \mu^{\E_\R^x}_{\omega_n^x}(B_n^x) < \frac{1}{n} \to 0$. For the second term, since $\beta_x \notin \partial X$,
    \begin{itemize}
        \item if $\beta_x \in X$, then $\beta_x \in \mathring{X}$, the interior of $X$, which is open, so $B_n^x \subseteq \mathring{X} \subseteq X$ for $n$ sufficiently large, so $X \cap B_n^x = B_n^x$ and $\mu^{\E_\R^x}_{\omega_n^x}(X \cap B_n^x) = \mu^{\E_\R^x}_{\omega_n^x}(B_n^x) \to 1 = \delta_{\beta_x}(X)$.
        \item if $\beta_x \notin X$, then $\beta_x \in \Gamma^x \smallsetminus \overline{X}$, which is open, so $B_n^x \subseteq \Gamma^x \smallsetminus \overline{X} \subseteq \Gamma^x \smallsetminus X$ for $n$ sufficiently large; hence $X \cap B_n^x = \varnothing$ and $\mu^{\E_\R^x}_{\omega_n^x}(X \cap B_n^x) = 0 = \delta_{\beta_x}(X)$.
    \end{itemize}
    In both cases, $\mu^{\E_\R^x}_{\omega_n^x}(X) \to \delta_{\beta_x}(X)$, so by portmanteau $\mu^{\E_\R^x}_{\omega_n^x} \to \delta_{\beta_x}$ weakly.

    \end{proof}

    \subsection{Proof of Prop.~\ref{def:groupoid-relativization-section4}}
    \label{app:proof groupoid-relativization-section4}

    \begin{proof}
        First, we start with showing that $(\euro^\R(A)_x)$ defines an essentially bounded measurable field in $\A(\scrhis\boxtimes\scrhir)$. We have
    \begin{equation}
        \norm{\euro^\R(A)_x}_{\bhsrx} \leq \norm{A}_{\A(\scrhis)} \cdot \norm{\E_\R^x(\Gamma^x)}_{\bhrx} < \infty
    \end{equation}
    with 
    \begin{equation}
        \esssup_{x\in M} \norm{\euro^\R(A)_x}_{\bhsrx} \leq \esssup_{x\in M} \left(\norm{A}_{\A(\scrhis)} \cdot \underbrace{\norm{\E_\R^x(\Gamma^x)}_{\bhrx}}_{=1}\right) < \infty
    \end{equation}
    Moreover, $\E_\R = (\E_\R^x)$ is a measurable field of normalized POVMs, $A = (A_x) \in \A(\scrhis)$ is a measurable operator field and $\beta \mapsto V_\beta$ is measurable, so $(\euro^\R(A)_x)$ defines a measurable field of operators. We continue now by showing that $(\euro^\R_\omega(A)_x)$ defines an essentially bounded measurable field in $\A(\scrhis)$. We have
    \begin{equation}
        \norm{\euro^\R_{\omega}(A)_x}_{\bhsx} \leq \left(\norm{A}_{\A(\scrhis)}\right) \cdot \Tr[\omega_x \E_\R^x(\Gamma^x)] < \infty
    \end{equation}
    with 
    \begin{equation}
        \esssup_{x\in M} \norm{\euro^\R_{\omega}(A)_x}_{\bhsx} \leq \esssup_{x\in M} \left(\norm{A}_{\A(\scrhis)} \cdot \underbrace{\Tr[\omega_x \E_\R^x(\Gamma^x)]}_{=1}\right) < \infty.
    \end{equation}
    Moreover, $\E_\R = (\E_\R^x)$ is a measurable field of normalized POVMs, $A = (A_x) \in \A(\scrhis)$ is a measurable operator field, $\omega = (\omega_x) \in \framestate$ is a measurable field of state and $\beta \mapsto V_\beta$ is measurable so $(\euro^\R_{\omega}(A)_x)$ indeed defines a measurable field of operators as above.
    \end{proof}

    \subsection{Proof of Thm.~\ref{thm:groupoid-relativization-invariant}}
    \label{app:proof properties euro}

    \begin{proof}

    Throughout the proof, all fiber-wise identities are understood for $\mu$-almost every $x\in M$.
\begin{enumerate}
\item Let $A,A'\in\A(\scrhis)$ be two representatives of the same decomposable operator, so that $A_y=A'_y$ for $\mu$-a.e. $y \in M$. Set $D:=A-A'$. There is a $\mu$-null Borel set $N\subseteq M$ such that $D_y=0$ for $y \notin N$. Fix $x$ such that $\F_{\R,s}^x\ll\mu$. Then $\E_\R^x(s^{-1}(N)\cap\Gamma^x) = \F_{\R,s}^x(N) = 0$. For $\beta\notin s^{-1}(N)\cap\Gamma^x$, one has $V_\beta D_{s(\beta)}V_\beta^\dagger=0$. 

Define $\Delta_x := \euro^\R(A)_x-\euro^\R(A')_x = \int_{\Gamma^x} \left(V_\beta D_{s(\beta)}V_\beta^\dagger\right)
\otimes d\E_\R^x(\beta)$. Let $\rho_x\in\systemstatex$ and $\omega_x\in\framestatex$. Then
\begin{equation}
    \Tr[\rho_x \otimes \omega_x \Delta_x] = \int_{\Gamma^x} \Tr[\rho_x V_\beta D_{s(\beta)}V_\beta^\dagger] \,  d\mu^{\E_\R^x}_{\omega_x}(\beta).
\end{equation}

The scalar integrand vanishes outside $s^{-1}(N)\cap\Gamma^x$. Moreover,
the Born probability measure $\mu^{\E_\R^x}_{\omega_x}$ vanishes on $s^{-1}(N)\cap\Gamma^x$, since $\E_\R^x(s^{-1}(N)\cap\Gamma^x)=0$. Therefore $\Tr[\rho_x \otimes \omega_x \Delta_x] = 0$ for any $\rho_x \in \systemstatex$ and $\omega_x \in \framestatex$. By ultraweak density of linear combinations of product states in $\thsrx$ which separates $\bhsrx$, we get $\Delta_x=0$. The fiber-wise equalities hold at every $x$ satisfying $\F_{\R,s}^x\ll\mu$,
hence for $\mu$-almost every $x$. Taking direct integrals gives the
global equalities.
\item We prove the properties of $\euro^\R|_x$. Linearity follows immediately from the linearity of $A\mapsto K_x^A(\beta)$ and from the linearity of the operator-valued integral. For unitality, let $\id_{\A(\scrhis)}$ denote the identity field, so that
\begin{equation}
(\id_{\A(\scrhis)})_y=\id_{\his^y}
\end{equation}
for every $y\in M$. Then
\begin{equation}
V_\beta \id_{\his^{s(\beta)}}V_\beta^\dagger=\id_{\hisx}.
\end{equation}
Therefore
\begin{equation}
\euro^\R(\id_{\A(\scrhis)})_x
=
\int_{\Gamma^x}\id_{\hisx}\otimes d\E_\R^x(\beta)  
=
\id_{\hisx}\otimes \E_\R^x(\Gamma^x) 
=
\id_{\hisx}\otimes \id_{\hirx} 
=
\id_{\hisx\otimes\hirx}.
\end{equation}
Thus $\euro^\R|_x$ is unital. Complete positivity and normality follow from \cite[Corollary 7.4 and Theorem 5.3]{glowacki_w-algebraic_2026}. Since $\euro^\R|_x$ is unital and completely positive, it is contractive. Since every positive linear map between $C^*$-algebras is $*$-preserving, $\euro^\R|_x$ is $*$-preserving. Finally, if $A$ is an effect in $\A(\scrhis)$, i.e.
\begin{equation}
0\leq A\leq \id_{\A(\scrhis)},
\end{equation}
then positivity and unitality imply
\begin{equation}
0
\leq
\euro^\R(A)_x
\leq
\euro^\R(\id_{\A(\scrhis)})_x
=
\id_{\hisx\otimes\hirx}.
\end{equation}
Thus $\euro^\R|_x$ is effect-preserving.

\item These properties are given in \cite{carette_operational_2025} and refs.~therein.

\item These properties follow from the identity
\begin{equation}
\euro^\R_\omega|_x
=
\Gamma_{\omega_x}^x\circ \euro^\R|_x .
\end{equation}
Indeed, the composition of two linear, normal, unital, completely positive, contractive, $*$-preserving and effect-preserving maps has the same properties. 

\item Linearity, unitality, complete positivity, $*$-preservation and effect preservation follow fiber-wise from item $1$. Contractivity follows from the fiber-wise estimate
\begin{equation}
\norm{\euro^\R(A)_x}_{\bhsx}
\leq
\norm{A}_\infty \norm{\E_\R^x(\Gamma^x)} = \norm{A}_\infty,
\end{equation}
hence
\begin{equation}
\norm{\euro^\R(A)} = \esssup_{x\in M}
\norm{\euro^\R(A)_x} \leq \norm{A}_\infty.
\end{equation}

It remains only to spell out normality. Let $(A_i)_{i\in I}$ be a bounded increasing net in $\A(\scrhis)_+$ with supremum $A$. By the fiber-wise normality already proved,
\begin{equation}
\euro^\R(A_i)_x \uparrow \euro^\R(A)_x
\end{equation}
for $\mu$-almost every $x$. Let
\begin{equation}
\chi=\int_M^\oplus \chi_x\,d\mu(x) \in \A(\scrhis\boxtimes\scrhir)_*
\end{equation}
be positive. Then
\begin{equation}
\Tr[\chi\euro^\R(A_i)] = \int_M \Tr_{\hisx\otimes\hirx} [\chi_x\euro^\R(A_i)_x] \,d\mu(x).
\end{equation}
The integrand increases pointwise to $\Tr_{\hisx\otimes\hirx}[\chi_x\euro^\R(A)_x]$
and is dominated by $\norm{A}_\infty\norm{\chi}_1$. 
Since $\chi\in\A(\scrhis\boxtimes\scrhir)_*$, the function $x\mapsto \|\chi_x\|_1$ is $L^1(M,\mu)$. The monotone convergence theorem therefore gives
\begin{equation}
\Tr[\chi\euro^\R(A_i)] \uparrow \Tr[\chi\euro^\R(A)].
\end{equation}
Hence $\euro^\R$ is normal. 

\item The same direct-integral argument as above applies to the fiber-wise slice maps. Namely,
\begin{equation}
(\Gamma_\omega B)_x = \Gamma_{\omega_x}^x(B_x).
\end{equation}
Linearity, unitality, complete positivity, contractivity, $*$-preservation and effect preservation follow fiber-wise from item $2$. For normality, if $(B_i)_{i\in I}$ is a bounded increasing net in $\A(\scrhis\boxtimes\scrhir)_+$ with supremum $B$, then
\begin{equation}
\Gamma_{\omega_x}^x(B_{i,x})
\uparrow
\Gamma_{\omega_x}^x(B_x)
\end{equation}
for $\mu$-almost every $x$. Pairing with an arbitrary positive element of $\A(\scrhis)_*$ and integrating over $M$ gives
\begin{equation}
\Gamma_\omega(B_i)
\uparrow
\Gamma_\omega(B)
\end{equation}
ultraweakly. Thus $\Gamma_\omega$ is normal. 

\item These properties follow from
\begin{equation}
\euro^\R_\omega = \Gamma_\omega\circ\euro^\R.
\end{equation}
Since both $\Gamma_\omega$ and $\euro^\R$ are linear, normal, unital, completely positive, contractive, $*$-preserving and effect-preserving, their composition have the same properties.

\item Let $\alpha:x\to y$ be an arrow of $\Gamma$ and write
\begin{equation}
W_\alpha:=V_\alpha\otimes U_\alpha:
\hisx\otimes\hirx
\to
\his^y\otimes\hir^y.
\end{equation}
Using functoriality of $V$ and covariance of $\E_\R$, we compute
\begin{equation}
\begin{aligned}
W_\alpha\euro^\R(A)_xW_\alpha^\dagger
&=
\int_{\Gamma^x}
\left(
V_\alpha V_\beta A_{s(\beta)}V_\beta^\dagger V_\alpha^\dagger
\right)
\otimes
U_\alpha\,d\E_\R^x(\beta)\,U_\alpha^\dagger \\
&=
\int_{\Gamma^x}
\left(
V_{\alpha\circ\beta}A_{s(\beta)}
V_{\alpha\circ\beta}^\dagger
\right)
\otimes
d\E_\R^y(L_\alpha\beta).
\end{aligned}
\end{equation}
The map
\begin{equation}
L_\alpha:\Gamma^x\to\Gamma^y,
\qquad
\beta\mapsto \alpha\circ\beta,
\end{equation}
is a Borel isomorphism. If
\begin{equation}
\gamma=L_\alpha\beta=\alpha\circ\beta,
\end{equation}
then
\begin{equation}
s(\gamma)=s(\beta).
\end{equation}
Changing variables in the operator-valued integral gives
\begin{equation}
W_\alpha\euro^\R(A)_xW_\alpha^\dagger = \int_{\Gamma^y} \left(V_\gamma A_{s(\gamma)}V_\gamma^\dagger \right) \otimes d\E_\R^y(\gamma) = \euro^\R(A)_y.
\end{equation}
Thus
\begin{equation}
\euro^\R(A)\in\A(\scrhis\boxtimes\scrhir)^\Gamma,
\end{equation}
and therefore
\begin{equation}
\euro^\R\bigl(\A(\scrhis)\bigr)
\subseteq
\A(\scrhis\boxtimes\scrhir)^\Gamma.
\end{equation}

\item Suppose that $\R$ is sharp. We prove multiplicativity fiber-wise. For $A\in\A(\scrhis)$ and $\mu$-almost every $x \in M$, define the function $f_x(\beta) := V_\beta A_{s(\beta)} V_\beta^\dagger$ where $\beta \in \Gamma^x$. Since $f \in L^\infty(\Gamma^x,\bhsx)$ and $\E_\R^x$ is sharp, the integration map $\int_{\Gamma^x} d\E_\R^x : L^\infty(\Gamma^x,\bhsx) \to \bhsrx$ given by $[f_x] \mapsto \int_{\Gamma^x} f_x \otimes d\E_\R^x$ is multiplicative \cite[Theorem 5.3]{glowacki_w-algebraic_2026}. This proves that $\euro^\R|_x$ is multiplicative for $\mu$-almost every $x$. Taking the direct integral gives
\begin{equation}
\euro^\R(A)\euro^\R(B)
=
\euro^\R(AB).
\end{equation}
Since $\euro^\R|_x$ and $\euro^\R$ are multiplicative, linear and $*$-preserving, they are normal unital $*$-homomorphisms for $\R$ sharp.
\item We adapt the proof of this property given in \cite{glowacki_operational_2023}, elevated from the $\yen^\R$ map to the $\euro^\R$ map. We will first show that if $\R$ is localizable then $\euro^\R$ is isometric. Given $A = (A_x) \in \A(\scrhis)_{\mathrm{uw}}$, we know from above that $\norm{\euro^\R(A)} \leq \norm{A}$, so it remains to show that $\norm{A} \leq \norm{\euro^\R(A)}$. By Thm.~\ref{thm:euro localisation limit}, if $\forall x \in M, 1_x \in \supp \E_\R^x$ then there is a localising sequence of states $\omega_n=(\omega_n^x) \in \framestate$ such that for all $\rho \in \normalsystemstate$,
\begin{align}
    \abs{\Tr[\rho A]} &= \lim_{n \to \infty} \abs{\Tr[\rho (\Gamma_{\omega_n} \circ \euro^\R)(A)]} \\
    &= \lim_{n \to \infty} \abs{\Tr[(\rho \boxtimes \omega_n)(\euro^\R(A))]} \\
    &\leq \sup_{\Omega \in \normalsystemframestate} \abs{\Tr[\Omega \euro^\R(A)]} = \norm{\euro^\R(A)}
\end{align}
where the fact that $\rho \boxtimes \omega \in \normalsystemframestate$ is shown in Prop.~\ref{prop:joint normal state}. Since $\rho \in \normalsystemstate$ above is arbitrary, we can conclude
\begin{equation}
    \norm{A} = \sup_{\rho \in \normalsystemstate} \abs{\Tr[\rho A]} \leq \norm{\euro^\R(A)}.
\end{equation}
Thus, if $\R$ is localizable then $\euro^\R$ is an isometry, and is thus injective, since for any $A \in \A(\scrhis)$ such that $\norm{\euro^\R(A)}=0$, we have $\norm{A}=\norm{\euro^\R(A)}=0 \Rightarrow A=0$ so $\ker \euro^\R = \{0\}$ which concludes the proof.
\end{enumerate}

\end{proof}

    \subsection{Proof of Thm.~\ref{thm:covariance euro}}
    \label{app:proof of covariance euro}

\begin{proof}

    By the functoriality of the representation and the covariance of the frame observable,
    \begin{equation}
        V_\alpha \euro^\R_\omega(A)_x V_\alpha^\dagger = \int_{\Gamma^x} V_\alpha V_\beta A_{s(\beta)} V_\beta^\dagger V_{\alpha}^\dagger \, d\mu^{\E_\R^x}_{\omega_x}(\beta) = \int_{\Gamma^x} V_{\alpha \circ \beta} A_{s(\beta)} V_{\alpha \circ \beta}^\dagger \, d\mu^{\E_\R^x}_{\omega_x}(\beta).
    \end{equation}
    The covariance of the frame observable implies that the measure $\mu^{\E_\R^x}_{\omega_x}$ is pushed forward by left translation $L_\alpha$ to the measure $\mu^{\E_\R^z}_{\omega_z^\alpha}$ on $\Gamma^z$, since
    \begin{multline}
        d\Tr[\omega_x \E_\R^x(\beta)] = d\Tr[\omega_x \E_\R^x(L_{\alpha^{-1}} \circ L_\alpha \circ \beta)] = d\Tr[\omega_x U_{\alpha^{-1}} \E_\R^z(L_\alpha \circ \beta) U_\alpha] = d\Tr[U_\alpha \omega_x U_\alpha^\dagger \E_\R^z(L_\alpha \circ \beta)] \\ = d\Tr[\omega_z^\alpha \E_\R^z(L_\alpha \circ \beta)].
    \end{multline}
    Thus,
    \begin{equation}
        V_\alpha \euro^\R_\omega(A)_x V_\alpha^\dagger = \int_{\Gamma^z} V_{\alpha \circ \beta} A_{s(\beta)} V_{\alpha \circ \beta}^\dagger d\mu^{\E_\R^z}_{\omega^\alpha_z}(L_\alpha \circ \beta) \stackrel{\gamma = L_\alpha \circ \beta}{=} \int_{\Gamma^z} V_\gamma A_{s(\gamma)} V_\gamma^\dagger \, d\mu^{\E_\R^z}_{\omega_z^\alpha}(\gamma) = \euro^\R_{\omega^\alpha}(A)_z.
    \end{equation}
\end{proof}

    \subsection{Proof of Thm. \ref{thm:euro localisation limit}}
    \label{app:proof yen localisation limit}

    \begin{proof}

    Let $\rho_x \in \systemstatex$. We have
    \begin{equation}
        \Tr[\rho_x \euro_{\omega_n}^\R(A)_x] = \Tr\left[\rho_x \int_{\Gamma^x} V_{\beta} A_{s(\beta)} V_\beta^\dagger \, d\mu^{\E_\R^x}_{\omega_n^x}(\beta)\right] = \int_{\Gamma^x} \Tr[\rho_x V_{\beta} A_{s(\beta)} V_\beta^\dagger] \, d\mu^{\E_\R^x}_{\omega_n^x}(\beta) \, .
    \end{equation}
    The function $\beta \mapsto \Tr_{\hisx}[\rho_x V_\beta A_{s(\beta)} V_{\beta}^\dagger]$ is continuous (by fiber-wise ultraweak continuity of $V_\beta$ and of $A$) and bounded, and by Prop. \ref{prop:norm 1 groupoid sequence} the sequence of measures $(\mu^{\E_\R^x}_{\omega_n^x})$ converges weakly to $\delta_{1_x}$, so
    \begin{equation}
        \lim_{n\to \infty} \int_{\Gamma^x}\Tr[\rho_x V_{\beta} A_{s(\beta)} V_\beta^\dagger] \, d\mu^{\E_\R^x}_{\omega_n^x}(\beta) = \Tr[\rho A_x]
    \end{equation}
    hence the sequence of operators $\left(\euro_{\omega_n}^\R(A)_x\right)$ converges to $A_x$ in the ultraweak topology of $\bhsx$ since states span the trace class operators. More generally, for any $\chi_x \in \thsx$, $\lim_{n \to \infty} \Tr[\chi_x \euro_{\omega_n}^\R(A)_x] = \Tr[\chi_x A_x]$.

    Moreover, if $\forall x \in M$, $1_x \in \supp \E_\R^x$, we have that for any $\chi \in \A(\scrhis)_*$,
    \begin{multline}
        \Tr[\chi \euro_{\omega_n}^\R(A)] =  \Tr\left[\int^\oplus_{M} \chi_x \int_{\Gamma^x} V_{\beta} A_{s(\beta)} V_\beta^\dagger \, d\mu^{\E_\R^x}_{\omega_n^x}(\beta) d\mu(x)\right] = \int_M \Tr\left[\chi_x \int_{\Gamma^x} V_{\beta} A_{s(\beta)} V_\beta^\dagger \, d\mu^{\E_\R^x}_{\omega_n^x}(\beta)\right] d\mu(x) \\ = \int_M \Tr[\chi_x \euro_{\omega_n^x}^{\R}(A)_x] d\mu(x) \, . 
    \end{multline}
    But 
    \begin{equation}
        \abs{\Tr[\chi_x \euro_{\omega_n}^{\R}(A)_x]} \leq \norm{\chi_x}_{1} \cdot \norm{\euro_{\omega_n}^{\R}(A)_x}_{\bhsx} \leq \norm{\chi_x}_{1} \cdot \norm{A} < \infty 
    \end{equation}
    so $x \mapsto \Tr[\chi_x \euro_{\omega_n}^{\R}(A)_x]$ is measurable and so
    \begin{equation}
        \lim_{n \to \infty} \Tr[\chi \euro_{\omega_n}^\R(A)] = \int_M \left(\lim_{n \to \infty}\Tr[\chi_x \euro_{\omega_n}^{\R}(A)_x]\right) d\mu(x) = \int_M \Tr[\chi_x A_x] \, d\mu(x) = \Tr[\chi A]
    \end{equation}
    which concludes the proof.

    \end{proof}

    \subsection{Proof of Thm.~\ref{thm:canonical-qrf}}
    \label{app:proof canonical qrf}
    
\begin{proof} 
We verify the three clauses of Def.~\ref{def:groupoid-qrf-section4} and then the additional properties.

\emph{(i) The field of Hilbert spaces.} Each $L^2(\Gamma^x,\nu^x)$ is a separable Hilbert space, since $\Gamma^x$ is second countable and $\nu^x$ is a Radon measure. That the family \eqref{eqn:canonical-hilbert-field} is a continuous field of Hilbert spaces, with continuous sections generated by $\{f|_{\Gamma^x}\mid f\in C_c(\Gamma)\}$, is the standard construction underlying the regular representation of a continuous groupoid with continuous Haar system; see \cite[Sec.~II.1]{renault_locally_1980} and \cite[Ch.~3]{landsman_mathematical_1998}. Continuity of $x\mapsto\norm{f|_{\Gamma^x}}^2_{L^2(\Gamma^x,\nu^x)} =\int_{\Gamma^x}\abs{f}^2\,d\nu^x$ for $f\in C_c(\Gamma)$ is precisely the continuity of the Haar system.

\emph{(ii) $U^{\mathrm{can}}$ is a fibre-wise ultraweakly continuous unitary representation.} Let $\alpha:x\to y$. If $\gamma\in\Gamma^{y}$ then $t(\gamma)=y=t(\alpha)=s(\alpha^{-1})$, so $\alpha^{-1}\circ\gamma$ is defined and $t(\alpha^{-1}\circ\gamma)=t(\alpha^{-1})=s(\alpha)=x$; hence $\alpha^{-1}\circ\gamma\in\Gamma^{x}$ and \eqref{eqn:canonical-representation} makes sense. Equivalently $U^{\mathrm{can}}_\alpha f=f\circ L_\alpha^{-1}$.

For isometry, apply the left-invariance $\alpha\circ\nu^{s(\alpha)}=\nu^{t(\alpha)}$ of the Haar system to the function $\gamma\mapsto\abs{f(\alpha^{-1}\circ\gamma)}^2$:
\begin{equation}
    \norm{U^{\mathrm{can}}_\alpha f}^2 = \int_{\Gamma^{y}}\abs{f(\alpha^{-1}\circ\gamma)}^2\,d\nu^{y}(\gamma) =  \int_{\Gamma^{x}}\abs{f(\alpha^{-1}\circ\alpha\circ\beta)}^2\,d\nu^{x}(\beta) =    \int_{\Gamma^{x}}\abs{f(\beta)}^2\,d\nu^{x}(\beta) =
    \norm{f}^2 .
\end{equation}
The same computation applied to $\alpha^{-1}$ shows that $U^{\mathrm{can}}_{\alpha^{-1}}$ is a two-sided inverse, so $U^{\mathrm{can}}_\alpha$ is unitary. Functoriality is immediate: for a composable pair $(\alpha,\beta)$ and $\gamma\in\Gamma^{t(\alpha)}$,
\begin{equation}
    \bigl(U^{\mathrm{can}}_\alpha U^{\mathrm{can}}_\beta f\bigr)(\gamma) =    \bigl(U^{\mathrm{can}}_\beta f\bigr)(\alpha^{-1}\circ\gamma) =    f\bigl(\beta^{-1}\circ\alpha^{-1}\circ\gamma\bigr) = f\bigl((\alpha\circ\beta)^{-1}\circ\gamma\bigr) = \bigl(U^{\mathrm{can}}_{\alpha\circ\beta}f\bigr)(\gamma),
\end{equation}
and $U^{\mathrm{can}}_{1_x}=\id_{L^2(\Gamma^x,\nu^x)}$. For fibre-wise ultraweak continuity it suffices, by Thm.~\ref{thm:equivalences continuity}, to check weak continuity. Let $\xi,\eta$ be continuous sections of $\scrhir^{\,\mathrm{can}}$; by density we may take $\xi=f|_{\Gamma^{\bullet}}$ and $\eta=h|_{\Gamma^{\bullet}}$ with $f,h\in C_c(\Gamma)$. Then
\begin{equation}
    \braket{U^{\mathrm{can}}_\alpha\xi_{s(\alpha)}}{\eta_{t(\alpha)}} = \int_{\Gamma^{t(\alpha)}}    \overline{f\bigl(\alpha^{-1}\circ\gamma\bigr)}\,h(\gamma)\, d\nu^{t(\alpha)}(\gamma),
\end{equation}
which is continuous in $\alpha$ because the groupoid operations are continuous, $f$ and $h$ are continuous with compact support, and $\nu$ is a continuous Haar
system. This is the usual continuity statement for the regular representation \cite[Sec.~II.1]{renault_locally_1980}.

\emph{(iii) $\E_{\R_{\mathrm{can}}}$ is a measurable field of covariant normalized POVMs.} For each $x$ the map $\Delta\mapsto M_{\mathbf 1_\Delta}$ is a projection-valued measure on $\Bor(\Gamma^x)$: idempotency and self-adjointness are clear from $\mathbf 1_\Delta^2=\mathbf 1_\Delta=\overline{\mathbf 1_\Delta}$, normalization is $M_{\mathbf 1_{\Gamma^x}}=\id$, and countable additivity in the ultraweak topology is also easily verified. It is therefore a normalized POVM with values in $\Eff(\hir^{x,\mathrm{can}})$.

Measurability of the field: for $\Delta\in\Bor(\Gamma)$ one has $\E^x_{\R_{\mathrm{can}}}(\Delta\cap\Gamma^x)=M_{\mathbf 1_\Delta|_{\Gamma^x}}$, which is multiplication by the restriction to $\Gamma^x$ of one fixed Borel function on $\Gamma$; measurability of $x\mapsto M_{\mathbf 1_\Delta|_{\Gamma^x}}$ in the sense of measurable operator fields follows, for $\Delta$ open, from continuity of $x\mapsto\nu^x(\Delta\cap\Gamma^x)$ against $C_c(\Gamma)$-sections.

Covariance: let $\alpha:x\to y$, $\Delta\in\Bor(\Gamma^{x})$ and $f\in L^2(\Gamma^{y},\nu^{y})$. Using $\bigl(U^{\mathrm{can}}_\alpha\bigr)^\dagger=U^{\mathrm{can}}_{\alpha^{-1}}$,
\begin{equation}
    \Bigl(U^{\mathrm{can}}_\alpha\,M_{\mathbf 1_\Delta}\,    \bigl(U^{\mathrm{can}}_\alpha\bigr)^\dagger f\Bigr)(\gamma) = \mathbf{1}_\Delta\bigl(\alpha^{-1}\circ\gamma\bigr)\,    \bigl(U^{\mathrm{can}}_{\alpha^{-1}}f\bigr)\bigl(\alpha^{-1}\circ\gamma\bigr) = \mathbf{1}_\Delta\bigl(\alpha^{-1}\circ\gamma\bigr)\,f(\gamma).
\end{equation}
Since $\alpha^{-1}\circ\gamma\in\Delta$ if and only if $\gamma\in\alpha\circ\Delta=L_\alpha\Delta$, we get $\mathbf 1_\Delta(\alpha^{-1}\circ\gamma)=\mathbf 1_{L_\alpha\Delta}(\gamma)$ and therefore
\begin{equation}
    U^{\mathrm{can}}_\alpha\,\E^{s(\alpha)}_{\R_{\mathrm{can}}}(\Delta)\,    \bigl(U^{\mathrm{can}}_\alpha\bigr)^\dagger =
    \E^{t(\alpha)}_{\R_{\mathrm{can}}}(L_\alpha\Delta),
\end{equation}
which is the covariance condition of Def.~\ref{def:groupoid-qrf-section4}.

\emph{Additional properties.} The frame is based on $\Gamma$ itself, hence \emph{principal}. Each $\E^x_{\R_{\mathrm{can}}}$ is projection-valued, hence $\R_{\mathrm{can}}$ is \emph{sharp}, and being principal and sharp it is \emph{ideal}. For \emph{localizability}, note that a projection $P$ satisfies $\norm{P}\in\{0,1\}$, so $\E^x_{\R_{\mathrm{can}}}(\Delta)\neq0$ forces $\norm{\E^x_{\R_{\mathrm{can}}}(\Delta)}=1$: every sharp frame is localizable. For \emph{$\mu$-admissibility}, the source marginal is
\begin{equation}
    \label{eqn:canonical-source-marginal}
    \F^x_{\R_{\mathrm{can}},s}(N) =    \E^x_{\R_{\mathrm{can}}}\bigl(s^{-1}(N)\cap\Gamma^x\bigr) = M_{\mathbf 1_{s^{-1}(N)\cap\Gamma^x}}, \qquad N\in\Bor(M),
\end{equation}
and a multiplication operator $M_{\mathbf 1_B}$ on $L^2(\Gamma^x,\nu^x)$ vanishes precisely when $\nu^x(B)=0$. Hence $\F^x_{\R_{\mathrm{can}},s}\ll\mu$ if and only if $\mu(N)=0\Rightarrow\nu^x(s^{-1}(N)\cap\Gamma^x)=0$, which is exactly source-compatibility of Def.~\ref{def:nu source compatible mu}. The last two assertions are immediate from the definitions.
\end{proof}

    \subsection{Proof of Prop.~\ref{prop:canonical-localizing-sequence}}
    \label{app:proof canonical localizing sequence}

\begin{proof}
Fix a complete Riemannian metric $d$ on the manifold $\Gamma$, which exists since $\Gamma$ is second countable and Hausdorff, and set
\begin{equation}
    B^x_n:=\bigl\{\beta\in\Gamma^x \;\big|\; d(\beta,1_x)<\tfrac1n\bigr\},
    \qquad n\in\Nn .
\end{equation}
Each $B^x_n$ is open in $\Gamma^x$ and nonempty, since $1_x\in B^x_n$, so $0<\nu^x(B^x_n)<\infty$: positivity because $\supp\nu^x=\Gamma^x$, finiteness because $\nu^x$ is Radon and $\overline{B^x_n}$ is compact for $n$ large by completeness of $d$. Define
\begin{equation}
    \psi^x_n :=    \nu^x\bigl(B^x_n\bigr)^{-1/2}\,\mathbf 1_{B^x_n}    \in L^2\bigl(\Gamma^x,\nu^x\bigr), \qquad    \omega^x_n:=\ketbra{\psi^x_n}{\psi^x_n} .
\end{equation}
Then $\norm{\psi^x_n}=1$, so $\omega^x_n\in\framestatex$. The map $(x,\beta)\mapsto\mathbf 1_{B^x_n}(\beta)$ is Borel on $\Gamma$, because the unit map $x\mapsto 1_x$ is continuous and $d$ is continuous, and $x\mapsto\nu^x(B^x_n)$ is Borel and strictly positive by continuity of the Haar system; hence $x\mapsto\psi^x_n$ is a measurable section and $\omega_n=(\omega^x_n)\in\framestate$.

Finally, for $\Delta\in\Bor(\Gamma^x)$,
\begin{equation}
    \mu^{\E^x_{\R_{\mathrm{can}}}}_{\omega^x_n}(\Delta) = \Tr\bigl[\omega^x_n\,M_{\mathbf 1_\Delta}\bigr] =
    \braket{\psi^x_n}{M_{\mathbf 1_\Delta}\psi^x_n} =   \frac{\nu^x\bigl(\Delta\cap B^x_n\bigr)}{\nu^x\bigl(B^x_n\bigr)} ,
\end{equation}
so $\mu^{\E^x_{\R_{\mathrm{can}}}}_{\omega^x_n}$ is the normalized restriction of $\nu^x$ to $B^x_n$. Since the sets $B^x_n$ shrink to $\{1_x\}$, for every bounded continuous $F$ on $\Gamma^x$ we get
\begin{equation}
    \abs*{  \int_{\Gamma^x}F\,d\mu^{\E^x_{\R_{\mathrm{can}}}}_{\omega^x_n} - F(1_x)} \leq \sup_{\beta\in B^x_n}\abs{F(\beta)-F(1_x)}    \xrightarrow[n\to\infty]{} 0
\end{equation}
by continuity of $F$ at $1_x$. Hence $\mu^{\E^x_{\R_{\mathrm{can}}}}_{\omega^x_n}\rightarrow\delta_{1_x}$ weakly.
\end{proof}

    \subsection{Proof of Thm.~\ref{thm:torsor reduction}}
    \label{app:proof torsor reduction}

    \begin{proof}
        Let $\R = (U_{\R},\E_\R,\hi_{\R})$ be a torsor QRF on $\Sigma$, with $U_{\R} : G \to \U(\hi_{\R})$ and $\E_{\R} : \Bor(\Sigma) \to \Eff(\hi_{\R})$ satisfying
\begin{equation}
    U_{\R}(g) \E_{\R}(B) U_{\R}(g)^\dagger = \E_{\R}(g \cdot B) \qquad \forall g \in G, B \in \Bor(\Sigma).
\end{equation}
Define a groupoid QRF on $\Gamma = \At(\Sigma \to M)$ by taking the trivial Hilbert bundle $\scrhi_{\underline{\R}} = M \times \hir \to M$. For each $x \in M$, define the target-fiber POVM by
\begin{equation}
    \E_{\underline{\R}}^x(\Delta) := \E_{\R}(\Pi_x^o(\Delta)), \qquad \Delta \in \Bor(\Gamma^x).
\end{equation}
This is a normalized POVM because $\Pi^o_x$ is a Borel isomorphism, and yields a measurable field of POVMs since $(x,g) \mapsto (g,g^{-1}x)$ is Borel. Let $\alpha = [p_y,p_x] : x \to y$ and $\beta \in \Gamma^x$. Then
\begin{equation}
    \Pi^o_y(L_\alpha \beta) = \kappa(\alpha) \cdot \Pi^o_x(\beta).
\end{equation}
Thus,
\begin{equation}
    U_\alpha \E_{\underline{\R}}^x(\Delta) U_\alpha^\dagger = U_{\R}(\kappa(\alpha)) \E_{\R}(\Pi^o_x(\Delta)) U_{\R}(\kappa(\alpha))^\dagger = \E_{\R}(\kappa(\alpha) \cdot \Pi^o_x(\Delta)) = \E_{\R}(\Pi^o_y(L_\alpha \Delta)) = \E_{\underline{\R}}^y(L_\alpha \Delta).
\end{equation}
The map
\begin{equation}
  \kappa:\At(\Sigma\to M)\longrightarrow G,
  \qquad [p_t,p_s]\longmapsto g_{p_t}g_{p_s}^{-1},
\end{equation}
is continuous and multiplicative. Hence the ultraweak continuity of $U_\R$ implies the fiber-wise ultraweak continuity of $U_\alpha=U_\R(\kappa(\alpha))$. In action-groupoid coordinates $\Gamma\cong G\ltimes G/H$, we have that for every $\chi \in \thr$,
\begin{equation}
 \Tr\left[\chi\,\E_{\underline\R}^{x}(\Delta\cap\Gamma^x)\right]  =\int_G \mathbf 1_\Delta(g,g^{-1}x) \,d\left((\theta_o^{-1})_*\mu_\chi^{\E_\R}\right)(g),
\end{equation}
is Borel in $x$, i.e. $(\E_{\underline\R}^{x})_{x\in M}$ is a measurable field of POVMs.

For $x\in M$ define $r_x:\Sigma\to M$ by $r_x(g\cdot o)=g^{-1}x$. Then
\begin{equation}
  \F_{\underline\R,s}^{x}(N)=\E_\R(r_x^{-1}(N)).
\end{equation}
Let $N \in \Bor(M)$ be such that $\mu_M(N)=0$ and, for $x \in M$, write
\begin{equation}
    I(x) := \int_G \mathbf{1}_N(g^{-1}x) dm_G(g).
\end{equation}
Then \begin{multline}
    \int_M I(x) \, d\mu_M(x) = \int_M \left(\int_G \mathbf{1}_N(g^{-1}x) dm_G(g) \right) \, d\mu_M(x) = \int_G\left(\int_M \mathbf{1}_N(g^{-1}x) d\mu_M(x) \right) dm_G(g) \\ = \int_G \mu_M(gN) d\lambda(g) = 0
\end{multline}
where the second equality holds by Tonelli and the last equality by quasi-invariance ($\mu_M(N)=0 \Rightarrow \mu_M(gN)=0$ for all $g \in G$). Since $I \geq 0$, $I(x) = 0$ for $\mu_M$-almost every $x \in M$. Choose $x_0 \in M$ such that $I(x_0)=0$. For arbitrary $x = a \cdot  x_0 \in M$, 
\begin{equation}
    \{g \mid g^{-1} x \in N\} = a \{h \mid h^{-1}x_0 \in N\},
\end{equation}
and left invariance gives
\begin{equation}
    m_G(\{g \mid g^{-1}x \in N\}) = 0.
\end{equation}
Thus, with $r_x(g \cdot o) = g^{-1}x$, one has
\begin{equation}
    \nu_\Sigma(r_x^{-1}(N)) = ((\theta_o)_* m_G)(\{g \cdot o \mid g^{-1}x \in N\}) = m_G(\{g \mid g^{-1} x \in N\}) = 0.
\end{equation}
By covariance, Thm.~\ref{thm:covariance implies mu continuity} gives $\E_\R\ll\nu_\Sigma$. Thus,
\begin{equation}
    \F_{\underline\R,s}^{x}(N)=\E_\R(r_x^{-1}(N))=0,
\end{equation}
proving $\mu_M$-admissibility. Thus, the ordinary torsor QRF $\R$ can be used to build a principal groupoid QRF $\underline{\R}$ on the Atiyah groupoid over $M=\Sigma/H$.

For the system, start with an ordinary $G$-system $\S=(U_{\S},\hi_{\S})$. Use the trivial system Hilbert bundle $\scrhis = M \times \his \to M$ with groupoid representation $V_\alpha := U_\S(\kappa(\alpha))$. Then the groupoid relativization recovers the ordinary torsor relativization as a constant decomposable field over $M$: we have, for all $A \in \bhs$,
\begin{equation}
    \euro^{\underline{\R}}(\iota_\S(A))_x = \int_{\Gamma^x} U_\S(\kappa(\beta)) A U_\S(\kappa(\beta))^\dagger \otimes d\E_{\underline{\R}}^x(\beta) = \int_\Sigma U_\S(g_{\eta,o}) A U_{\S}(g_{\eta,o})^\dagger \otimes d\E_\R(\eta) = \yen^\R(A)
\end{equation}
where $g_{\eta,o}$ is the unique element in $G$ satisfying $g_{\eta,o} \cdot o = \eta$ for $o,\eta \in \Sigma$. This is exactly the ordinary torsor relativization. The result is independent of $x$, so the direct-integral operator field is constant over $M$. Likewise, given $\omega \in \homframestate$, we can build the constant field of fiber states $\underline{\omega} := \iota_\R(\omega) \in \framestate$ (now an operator field of states). Then,
\begin{equation}
    \Gamma_{\underline{\omega}} \circ \iota_{\S\R} = \iota_\S \circ \Gamma_\omega
\end{equation}
and so, as above,
\begin{equation}
    \euro^{\underline{\R}}_{\underline{\omega}}(\iota_\S(A))_x = (\Gamma_{\underline{\omega}} \circ \euro^{\underline{\R}}(\iota_\S(A)))_x = \Gamma_{\omega} \circ \euro^{\underline{\R}}(\iota_\S(A))_x = \Gamma_\omega \circ \yen^\R(A) = \yen^\R_\omega(A)
\end{equation}
which is again independent of $x \in M$.
    \end{proof}

    \subsection{Proof of Prop.~\ref{prop:covariance localization}}
    \label{app:proof covariance localization}

\begin{proof}

    Fix $b\in\mathscr{B}$ and $\omega\in\framestate$. Recall $\mathrm{Loc}_\pi(\R,\omega)=\bigcup_{x\in\M}\supp\mu^{\F^x_{\R,\pi}}_{\omega_x}$ for
$\omega\in\framestate$, where $\mu^{\F^x_{\R,\pi}}_{\omega_x}(N)=\Tr[\omega_x\,\F^x_{\R,\pi}(N)]$,
$N\in\Bor(\M)$. Note first that $\omega^b\in\framestate$: for $\mu_g$-a.e.\ $x$, writing
$y:=\varphi_b^{-1}(x)$, the arrow $b(y):y\to\varphi_b(y)=x$ gives
$U_{b(y)}:\hir^{y}\to\hirx$, so $\omega^b_x=U_{b(y)}\omega_y U_{b(y)}^\dagger\in
\framestatex$; measurability is preserved by unitary conjugation and by the reindexing $x\mapsto\varphi_b^{-1}(x)$, which maps $\mu_g$-conull sets to $\mu_g$-conull sets by the bisection-admissibility of $\R$ (Lem.~\ref{lem:general relativistic QRF is bisection admissible}). We claim that for $\mu_g$-almost every $x\in\M$, with $y=\varphi_b^{-1}(x)$,
\begin{equation}
\label{eqn:born pushforward}
\mu^{\F^{x}_{\R,\pi}}_{\omega^b_{x}}
=(\varphi_b)_*\,\mu^{\F^{y}_{\R,\pi}}_{\omega_{y}} .
\end{equation}
For $N\in\Bor(\M)$, cyclicity of the trace gives
\begin{equation}
\mu^{\F^{x}_{\R,\pi}}_{\omega^b_{x}}(N)
=\Tr\!\bigl[U_{b(y)}\,\omega_y\,U_{b(y)}^\dagger\,\F^{x}_{\R,\pi}(N)\bigr]
=\Tr\!\bigl[\omega_y\;U_{b(y)}^\dagger\,\F^{x}_{\R,\pi}(N)\,U_{b(y)}\bigr].
\end{equation}
Since $\varphi_b$ is a diffeomorphism (hence a homeomorphism), $\varphi_b^{-1}(N)\in\Bor(\M)$, and the $\mathscr{B}$-compatibility
Lem.~\ref{lem:projection marginal covariance} applied at $y$ to the set
$\varphi_b^{-1}(N)$ yields
\begin{equation}
\F^{x}_{\R,\pi}(N)=U_{b(y)}\,\F^{y}_{\R,\pi}\bigl(\varphi_b^{-1}(N)\bigr)\,U_{b(y)}^\dagger
=\F^{\varphi_b(y)}_{\R,\pi}\bigl(\varphi_b(\varphi_b^{-1}(N))\bigr).
\end{equation}
Substituting,
\begin{equation}
\mu^{\F^{x}_{\R,\pi}}_{\omega^b_{x}}(N)
=\Tr\!\bigl[\omega_y\,\F^{y}_{\R,\pi}\bigl(\varphi_b^{-1}(N)\bigr)\bigr]
=\mu^{\F^{y}_{\R,\pi}}_{\omega_{y}}\bigl(\varphi_b^{-1}(N)\bigr)
=\bigl((\varphi_b)_*\mu^{\F^{y}_{\R,\pi}}_{\omega_{y}}\bigr)(N),
\end{equation}
which is \eqref{eqn:born pushforward}. For a Borel probability measure $\nu$ on the second-countable space $\M$ and a
homeomorphism $\psi:\M\to\M$,
\begin{equation}
\label{eqn:support pushforward}
\supp(\psi_*\nu)=\psi(\supp\nu).
\end{equation}
Applying \eqref{eqn:support pushforward} to
\eqref{eqn:born pushforward} with $\psi=\varphi_b$,
\begin{equation}
\label{eqn:support transformation}
\supp\mu^{\F^{x}_{\R,\pi}}_{\omega^b_{x}}
=\varphi_b\Bigl(\supp\mu^{\F^{\varphi_b^{-1}(x)}_{\R,\pi}}_{\omega_{\varphi_b^{-1}(x)}}\Bigr),
\qquad x\in\M .
\end{equation}

Since $x\mapsto\varphi_b^{-1}(x)$ is a bijection of $\M$, reindexing by
$y=\varphi_b^{-1}(x)$ and using that images under any map commute with arbitrary
unions,
\begin{equation}
\mathrm{Loc}_\pi(\R,\omega^b)
=\bigcup_{x\in\M}\supp\mu^{\F^{x}_{\R,\pi}}_{\omega^b_{x}}
=\bigcup_{y\in\M}\varphi_b\Bigl(\supp\mu^{\F^{y}_{\R,\pi}}_{\omega_{y}}\Bigr)
=\varphi_b\Bigl(\bigcup_{y\in\M}\supp\mu^{\F^{y}_{\R,\pi}}_{\omega_{y}}\Bigr)
=\varphi_b\bigl(\mathrm{Loc}_\pi(\R,\omega)\bigr).\qedhere
\end{equation}

\end{proof}

    \subsection{Proof of Thm.~\ref{thm:covariance local rqf}}
    \label{app:proof covariance local rqf}

    \begin{proof}
        For $\Delta\in\Bor(\mathrm{Poin}(\M,g)^z)$, covariance of the frame observables gives
    \begin{equation}
    \label{eqn:covariance arrow Born measure}
    \mu_{\omega_z^\alpha}^{\E_\R^z}(\Delta)
    = \Tr\left[    U_\alpha\omega_xU_\alpha^\dagger    \E_\R^z(\Delta)\right] 
    = \Tr\left[\omega_x   \E_\R^x(L_{\alpha^{-1}}\Delta) \right] 
    = \mu_{\omega_x}^{\E_\R^x}(L_{\alpha^{-1}}\Delta).
    \end{equation}
    Moreover, for all $N \in \Bor(\M)$,
    \begin{multline}
    \label{eqn:marginal-pushforward-arrow}
        \mu_{\omega^\alpha_z}^{\F_{\R,\pi}^z}(N) = \mu^{\E_\R^z}_{\omega^\alpha_z}(\pi_z^{-1}(N)) \stackrel{\eqref{eqn:covariance arrow Born measure}}{=} \mu^{\E_\R^x}_{\omega_x}\left(L_{\alpha^{-1}}\left(\pi_z^{-1}(N)\right)\right) = \mu^{\E_\R^x}_{\omega_x}((\pi_z \circ L_\alpha)^{-1}(N)) \\= \mu^{\E_\R^x}_{\omega_x}((\tau_\alpha \circ \pi_x)^{-1}(N)) = \mu^{\E_\R^x}_{\omega_x}(\pi_x^{-1}(\tau_\alpha^{-1}(N)))
        =
        (\tau_\alpha)_*
        \mu_{\omega_x}^{\F_{\R,\pi}^x}(N).
    \end{multline}

Let $y\mapsto\nu_{\omega_x}^{\R,\pi}(\,\cdot\mid y)$ disintegrate $\mu_{\omega_x}^{\E_\R^x}$ with
respect to $\pi_x$. Define, for $y'\in M$,
\begin{equation}
    \widetilde\nu_z(\Delta\mid y')
    :=
    \nu_{\omega_x}^{\R,\pi}
    \left(
    L_{\alpha^{-1}}\Delta
    \,\middle|\,
    \tau_\alpha^{-1}(y')
    \right).
\end{equation}
Since $\pi_z \circ L_\alpha = \tau_\alpha \circ \pi_x$, this kernel is concentrated on
$\pi_z^{-1}(\{y'\})$. If $N\in\Bor(M)$, then
\begin{multline}
    \int_N
    \widetilde\nu_z(\Delta\mid y')    \,d((\tau_\alpha)_*\mu^{\F_{\R,\pi}^x}_{\omega_x})(y')=
    \int_{\tau_\alpha^{-1}(N)}
    \nu_{\omega_x}^{\R,\pi}
    (L_{\alpha^{-1}}\Delta\mid y)
    \,d\mu^{\F_{\R,\pi}^x}_{\omega_x}(y)=
    \mu_{\omega_x}^{\E_\R^x}
    \left(
    L_{\alpha^{-1}}\Delta
    \cap
    \pi_x^{-1}(\tau_\alpha^{-1}(N))
    \right) \\ =
    \mu_{\omega^\alpha_z}^{\E_\R^z}
    (\Delta\cap\pi_z^{-1}(N)).
\end{multline}
Hence $\widetilde\nu_z$ is a disintegration of the transformed Born measure. Essential
uniqueness thus proves Eqn.~\eqref{eqn:conditional-probability-covariance}. Finally, using functoriality of $V$ and the fact that
$s(\alpha\circ\beta)=s(\beta)$,
\begin{multline}
    \hat\phi_{\omega^\alpha,z}^{\R,\pi}
    (\tau_\alpha(y))= \int_{\pi_{z}^{-1}(\tau_\alpha(y))} V_\beta \phi_{s(\beta)} V_\beta^\dagger d\nu^{\R,\pi}_{\omega^\alpha_z}(\beta \mid \tau_\alpha(y)) =\int_{\pi^{-1}_x(y)}    V_{\alpha\circ\beta} \phi_{s(\beta)}    V_{\alpha\circ\beta}^\dagger    \,d\nu_{\omega_x}^{\R,\pi}(\beta\mid y)\\ =
    V_\alpha\hat\phi_{\omega_x,x}^{\R,\pi}(y)    V_\alpha^\dagger.
\end{multline}
    \end{proof}
    
    \subsection{Proof of Thm.~ \ref{thm:causality conditions}}
    \label{app:proof causality}

\begin{proof}
    
    \begin{enumerate}
        \item Let us show that if $\mathcal{O}_\S$ is absolutely $(\sigma,\pi)$-microcausal then it is strongly $(\operationalstate,\sigma,\pi)$-microcausal for any $\operationalstate \subset \framestate$. 
        For $\mu_g$-almost every $z\in \M$, let $\omega^{(1)},\omega^{(2)}\in\operationalstate$,and let $x \in \supp \mu^{\F_{\R,\pi}^z}_{\omega^{(1)}_z}$, $y \in \supp \mu^{\F_{\R,\pi}^z}_{\omega^{(2)}_z}$  satisfy $x\indepsigma y$. We have, for any $\phi,\psi \in \mathcal{O}_\S$,
        \begin{equation} \hat\phi^{\R,\pi}_{z,\omega^{(1)}}(x) = \int_{\pi_z^{-1}(\{x\})} V_\beta\phi_{s(\beta)}V_\beta^\dagger\, d\nu^{\R,\pi}_{\omega^{(1)}_z}(\beta\mid x), 
        \end{equation} and 
        \begin{equation} \hat\psi^{\R,\pi}_{z,\omega^{(2)}}(y) = \int_{\pi_z^{-1}(\{y\})} V_\gamma \psi_{s(\gamma)}V_\gamma^\dagger\, d\nu^{\R,\pi}_{\omega^{(2)}_z}(\gamma\mid y). 
        \end{equation} 
        Take arbitrary arrows $\beta:x\to z$ and $\gamma:y\to z$. Then $\pi_z(\beta) \indepsigma \pi_z(\gamma)$, so from absolute $(\sigma,\pi)$-microcausality we have 
        \begin{equation}  \comm{V_\beta\phi_{s(\beta)}V_\beta^\dagger}{ V_\gamma\psi_{s(\gamma)}V_\gamma^\dagger}  = 0. \end{equation}
        Thus, 
        \begin{equation} \comm{\hat\phi^{\R,s}_{z,\omega^{(1)}}(x)}{ \hat\psi^{\R,s}_{z,\omega^{(2)}}(y)} = \iint_{\pi_z^{-1}(\{x\}) \times \pi_z^{-1}(\{y\})} \comm{V_\beta\phi_{s(\beta)}V_\beta^\dagger}{ V_\gamma\psi_{s(\gamma)}V_\gamma^\dagger} d\nu^{\R,\pi}_{\omega^{(1)}_z}(\beta\mid x) d\nu^{\R,\pi}_{\omega^{(2)}_z}(\gamma\mid y) =0. 
        \end{equation} 
        Similarly, absolute $(\sigma,\pi)$-microcausality implies
        \begin{equation} \comm{V_\beta\phi_{s(\beta)}^\dagger V_\beta^\dagger}{ V_\gamma\psi_{s(\gamma)}V_\gamma^\dagger} =  0, 
        \end{equation} 
        and therefore 
        \begin{equation}
            \comm{\hat{\phi}^{\R,\pi}_{z,\omega^{(1)}}(x)^\dagger}{ \hat{\psi}^{\R,\pi}_{z,\omega^{(2)}}(y)} = \iint_{\pi_z^{-1}(\{x\}) \times \pi_z^{-1}(\{y\})} \comm{V_\beta\phi_{s(\beta)}^\dagger V_\beta^\dagger}{ V_\gamma\psi_{s(\gamma)}V_\gamma^\dagger} d\nu^{\R,\pi}_{\omega^{(1)}_z}(\beta\mid x) d\nu^{\R,\pi}_{\omega^{(2)}_z}(\gamma\mid y)  =0. 
            \end{equation} 
            Thus $\mathcal{O}_\S$ is strongly $(\operationalstate,\sigma,\pi)$-microcausal.
        
        \item If $\omega_z^{(1)} \indepERsigma \omega_z^{(2)}$ then for any $z \in \text{Loc}_\pi(\R,\omega_z^{(1)})$ and $w \in \text{Loc}_\pi(\R,\omega_z^{(2)})$, $z \indepsigma w$ so strong $(\operationalstate,\sigma)$-microcausality implies weak $(\operationalstate,\sigma,\pi)$-microcausality.
        \item Strong semi $(\operationalstate,\sigma)$-causality satisfies the same assumptions as weak semi $(\operationalstate,\sigma,\pi)$-causality apart from the fact that it is not conditioned on whether $\omega_z^{(1)} \indepERsigma \omega_z^{(2)}$, so strong semi $(\operationalstate,\sigma)$-causality implies weak semi $(\operationalstate,\sigma,\pi)$-causality.
        \item We follow here the proof of Thm. 5.10 of \cite{fedida_foundations_2025}. Suppose $\R$ is a general relativistic QRF such that $\mathcal{O}_\S$ is weakly $(\operationalstate,\sigma,\pi)$-microcausal, and $\omega_1 \indepERsigma \omega_2$. Then for all $x \in M$ and $\phi,\psi \in \mathcal{O}_\S$,
    \begin{equation}\begin{aligned}
        \hat{\Phi}_x^\R(\omega_1) \hat{\Psi}_x^\R(\omega_2) &= \left(\int_{\supp \mu^{\F^x_{\R,\pi}}_{\omega_1}} \hat{\phi}^\R_{x,\omega_1}(x_1) d\mu^{\F^x_{\R,\pi}}_{\omega^x_1}(x_1) \right) \left( \int_{\supp \mu^{\F^x_\R}_{\omega_2}} \hat{\psi}^\R_{x,\omega_2}(x_2) d\mu^{\F^x_\R}_{\omega^x_2}(x_2)\right) \\
        &= \iint_{\supp \mu^{\F^x_{\R,\pi}}_{\omega^x_1} \times \supp \mu^{\F^x_{\R,\pi}}_{\omega^x_2}} \hat{\phi}^\R_{x,\omega_1}(x_1) \hat{\psi}^\R_{x,\omega_2}(x_2)  \, d(\mu^{\F^x_{\R,\pi}}_{\omega^x_1} \times \mu^{\F^x_{\R,\pi}}_{\omega^x_2})(x_1,x_2) \\
        &\stackrel{\omega_1 \indepERsigma \omega_2}{=} \iint_{\supp \mu^{\F^x_{\R,\pi}}_{\omega^x_1} \times \supp \mu^{\F^x_{\R,\pi}}_{\omega^x_2}} \hat{\psi}^\R_{x,\omega_2}(x_2) \hat{\phi}^\R_{x,\omega_1}(x_1) \, d(\mu^{\F^x_{\R,\pi}}_{\omega^x_1} \times \mu^{\F^x_{\R,\pi}}_{\omega^x_2})(x_1,x_2) \\
        &= \iint_{\supp \mu^{\F^x_{\R,\pi}}_{\omega^x_2} \times\supp \mu^{\F^x_{\R,\pi}}_{\omega^x_1}} \hat{\psi}^\R_{x,\omega_2}(x_2) \hat{\phi}^\R_{x,\omega_1}(x_1) \, d(\mu^{\F^x_{\R,\pi}}_{\omega^x_2} \times \mu^{\F^x_{\R,\pi}}_{\omega^x_1})(x_2,x_1) \\
        &= \left(\int_{\supp \mu^{\F^x_{\R,\pi}}_{\omega^x_2}} \hat{\psi}^\R_{x,\omega_2}(x_2) d\mu^{\F^x_{\R,\pi}}_{\omega^x_2}(x_2) \right) \left( \int_{\supp \mu^{\F^x_{\R,\pi}}_{\omega^x_1}} \hat{\phi}^\R_{x,\omega_1}(x_1) d\mu^{\F^x_{\R,\pi}}_{\omega^x_1}(x_1)\right) \\
        &= \hat{\Psi}_x^\R(\omega_2) \hat{\Phi}_x^\R(\omega_1) \, .
    \end{aligned}\end{equation}
    The same argument holds for $\hat{\Phi}^\R_x(\omega_1)^\dagger \hat{\Psi}^\R_x(\omega_2)$, and these hold fiber-wise in the direct integral, which concludes the proof.
    \end{enumerate}

\end{proof}

    \subsection{Proof of Thm. \ref{thm:RAQFT}}
    \label{app:proof AQFT}

\begin{proof}

    \begin{enumerate}
        \item Let $\C^{(\operationalstate,\sigma,\pi)}_{\mathcal{O}_\S}(\U) := \{\hat{\Phi}^\R(\omega) \mid \phi \in \mathcal{O}_\S, \,  \omega \in \operationalstate \text{ s.t. } 
 \text{Loc}_\pi(\R,\omega) \subseteq \U\}$.
If $\U \subseteq \V$, we have $\text{Loc}_\pi(\R,\omega) \subseteq \U$ $\Rightarrow$ $\text{Loc}_\pi(\R,\omega) \subseteq \V$. Thus, $\C^{(\operationalstate,\sigma,\pi)}_{\mathcal{O}_\S}(\U) \subseteq \C^{(\operationalstate,\sigma,\pi)}_{\mathcal{O}_\S}(\V)$ and since $\A^{(\operationalstate,\sigma,\pi)}_{\mathcal{O}_\S}(\U) = \C^{(\operationalstate,\sigma,\pi)}_{\mathcal{O}_\S}(\U)''$ the result follows from the isotony of the double commutant.
\item By Prop.~\ref{prop:covariance localization}, we have that for all $\omega \in \framestate$ and $b \in \mathrm{Bis}(\mathrm{Poin}(\M,g))$,
\begin{equation}
    \text{Loc}_\pi(\R,\omega^{b^{-1}}) = \varphi_b^{-1}(\text{Loc}(\R,\omega)).
\end{equation}
and hence we get
\begin{align}
    \mathbf{V}_b\C^{(\operationalstate,\sigma,\pi)}_{\mathcal{O}_\S}(\U)\mathbf{V}_b^\dagger &= \left\{\mathbf{V}_b \hat{\Phi}^\R(\omega) \mathbf{V}_b^\dagger \mid \phi \in \mathcal{O}_\S, \omega \in \operationalstate \text{ s.t. } \text{Loc}_\pi(\R,\omega) \subseteq \U\right\}\\
    &\stackrel{\ref{thm:covariance rqf}}{=} \left\{ \hat{\Phi}^\R(\omega^b)  \mid \phi \in \mathcal{O}_\S, \omega \in \operationalstate \text{ s.t. } \text{Loc}_\pi(\R,\omega) \subseteq \U\right\} \\ 
    &= \left\{\hat{\Phi}^\R(\omega) \mid \phi \in \mathcal{O}_\S, \omega^{b^{-1}} \in \operationalstate \text{ s.t. } \text{Loc}_\pi(\R,\omega^{b^{-1}}) \subseteq \U\right\} \\
    &\stackrel{\ref{prop:covariance localization}}{=} \left\{\hat{\Phi}^\R(\omega) \mid \phi \in \mathcal{O}_\S, \omega \in b \cdot \operationalstate \text{ s.t. } \varphi_{b^{-1}}\left(\text{Loc}_\pi(\R,\omega)\right) \subseteq \U\right\} \\
    &=  \left\{\hat{\Phi}^\R(\omega) \mid \phi \in \mathcal{O}_\S, \omega \in b \cdot \operationalstate \text{ s.t. } \text{Loc}_\pi(\R,\omega) \subseteq \underbrace{\varphi^{-1}_{b^{-1}}}_{=\varphi_b}(\U)\right\} \\
    &= \C^{(b \cdot \operationalstate,\sigma,\pi)}_{\mathcal{O}_\S}(\varphi_b(\U)).
\end{align}
Since unitary conjugation commutes with commutants, we get
\begin{multline}
    \mathbf{V}_b\A^{(\operationalstate,\sigma,\pi)}_{\mathcal{O}_\S}(\U)\mathbf{V}_b^\dagger = \mathbf{V}_b\C^{(\operationalstate,\sigma,\pi)}_{\mathcal{O}_\S}(\U)''\mathbf{V}_b^\dagger = \left(\mathbf{V}_b\C^{(\operationalstate,\sigma,\pi)}_{\mathcal{O}_\S}(\U)\mathbf{V}_b^\dagger\right)'' \\ = \C^{(b \cdot \operationalstate,\sigma,\pi)}_{\mathcal{O}_\S}(\varphi_b(\U))'' = \A^{(b \cdot \operationalstate,\sigma,\pi)}_{\mathcal{O}_\S}(\varphi_b(\U)).
\end{multline}
Prop.~\ref{prop:Einstein causality stable under isometries} and $\mathrm{Bis}_{\mathrm{Isom}}(\mathrm{Poin}(\M,g))$-compatibility then ensure that the elements of $\A^{(b \cdot \operationalstate,\sigma,\pi)}_{\mathcal{O}_\S}(\varphi_b(\U))$ are $(b \cdot \operationalstate,\sigma,\pi)$-causal.

\item If $\U \indepsigma \V$ then for all $A \in \C^{(\operationalstate,\sigma,\pi)}_{\mathcal{O}_\S}(\U)$ and $B \in \C^{(\operationalstate,\sigma,\pi)}_{\mathcal{O}_\S}(\V)$, the commutant vanishes $\comm{A}{B} = 0$. In particular, we have $\C^{(\operationalstate,\sigma,\pi)}_{\mathcal{O}_\S}(\V) \subseteq \C^{(\operationalstate,\sigma,\pi)}_{\mathcal{O}_\S}(\U)'$. 
        Since for two subsets $\mathfrak{A}$ and $\mathfrak{B}$ of $\bh$, $\mathfrak{A} \subseteq \mathfrak{B} \Rightarrow \mathfrak{B}' \subseteq \mathfrak{A}'$, taking the commutant of both sides gives $\A^{(\operationalstate,\sigma,\pi)}_{\mathcal{O}_\S}(\U) \subseteq \C^{(\operationalstate,\sigma,\pi)}_{\mathcal{O}_\S}(\V)'$, and applying it again yields $\A^{(\operationalstate,\sigma,\pi)}_{\mathcal{O}_\S}(\V) \subseteq \A^{(\operationalstate,\sigma,\pi)}_{\mathcal{O}_\S}(\U)'$.
    \end{enumerate}

    \end{proof}

\subsection{Proof of Thm.~\ref{thm:covariance local rqf gauge}}
    \label{app:proof covariance local rqf gauge}

    \begin{proof}
        We repeat what is essentially the same proof as Thm.~\ref{thm:covariance local rqf}. For $\Delta\in\Bor(\At(Q)^z)$, covariance of the frame observables gives
    \begin{equation}
    \label{eqn:covariance arrow Born measure gauge}
    \mu_{\omega_z^\alpha}^{\E_{\R_G}^z}(\Delta)
    = \Tr\left[    U_\alpha\omega_xU_\alpha^\dagger    \E_{\R_G}^z(\Delta)\right] 
    = \Tr\left[\omega_x   \E_{\R_G}^x(L_{\alpha^{-1}}\Delta) \right] 
    = \mu_{\omega_x}^{\E_{\R_G}^x}(L_{\alpha^{-1}}\Delta).
    \end{equation}
    Moreover, for all $N \in \Bor(\M)$,
    \begin{multline}
    \label{eqn:marginal-pushforward-arrow gauge}
        \mu_{\omega^\alpha_z}^{\F_{\R_G,\pi}^z}(N) = \mu^{\E_{\R_G}^z}_{\omega^\alpha_z}(\pi_z^{-1}(N)) \stackrel{\eqref{eqn:covariance arrow Born measure gauge}}{=} \mu^{\E_{\R_G}^x}_{\omega_x}\left(L_{\alpha^{-1}}\left(\pi_z^{-1}(N)\right)\right) = \mu^{\E_{\R_G}^x}_{\omega_x}((\pi_z \circ L_\alpha)^{-1}(N)) \\= \mu^{\E_{\R_G}^x}_{\omega_x}((\tau_\alpha \circ \pi_x)^{-1}(N)) = \mu^{\E_{\R_G}^x}_{\omega_x}(\pi_x^{-1}(\tau_\alpha^{-1}(N)))
        =
        (\tau_\alpha)_*
        \mu_{\omega_x}^{\F_{\R_G,\pi}^x}(N).
    \end{multline}

Let $y\mapsto\nu_{\omega_x}^{\R_G,\pi}(\,\cdot\mid y)$ disintegrate $\mu_{\omega_x}^{\E_{\R_G}^x}$ with
respect to $\pi_x$. Define, for $y'\in M$,
\begin{equation}
    \widetilde\nu_z(\Delta\mid y')
    :=
    \nu_{\omega_x}^{\R_G,\pi}
    \left(
    L_{\alpha^{-1}}\Delta
    \,\middle|\,
    \tau_\alpha^{-1}(y')
    \right).
\end{equation}
Since $\pi_z \circ L_\alpha = \tau_\alpha \circ \pi_x$, this kernel is concentrated on
$\pi_z^{-1}(\{y'\})$. If $N\in\Bor(M)$, then
\begin{multline}
    \int_N
    \widetilde\nu_z(\Delta\mid y')    \,d((\tau_\alpha)_*\mu^{\F_{\R_G,\pi}^x}_{\omega_x})(y')=
    \int_{\tau_\alpha^{-1}(N)}
    \nu_{\omega_x}^{\R_G,\pi}
    (L_{\alpha^{-1}}\Delta\mid y)
    \,d\mu^{\F_{\R_G,\pi}^x}_{\omega_x}(y)=
    \mu_{\omega_x}^{\E_{\R_G}^x}
    \left(
    L_{\alpha^{-1}}\Delta
    \cap
    \pi_x^{-1}(\tau_\alpha^{-1}(N))
    \right) \\ =
    \mu_{\omega^\alpha_z}^{\E_{\R_G}^z}
    (\Delta\cap\pi_z^{-1}(N)).
\end{multline}
Hence $\widetilde\nu_z$ is a disintegration of the transformed Born measure. Essential
uniqueness thus proves Eqn.~\eqref{eqn:conditional-probability-covariance gauge}. Finally, using functoriality of $V$ and the fact that
$s(\alpha\circ\beta)=s(\beta)$,
\begin{multline}
    \hat\phi_{\omega^\alpha,z}^{\R_G,\pi}
    (\tau_\alpha(y))= \int_{\pi_{z}^{-1}(\tau_\alpha(y))} V_\beta \phi_{s(\beta)} V_\beta^\dagger d\nu^{\R_G,\pi}_{\omega^\alpha_z}(\beta \mid \tau_\alpha(y)) =\int_{\pi^{-1}_x(y)}    V_{\alpha\circ\beta} \phi_{s(\beta)}    V_{\alpha\circ\beta}^\dagger    \,d\nu_{\omega_x}^{\R_G,\pi}(\beta\mid y)\\ =
    V_\alpha\hat\phi_{\omega_x,x}^{\R_G,\pi}(y)    V_\alpha^\dagger.
\end{multline}
    \end{proof}
    
\end{appendix}

\end{document}